\DocumentMetadata 
{
	lang		= en-US,
	pdfversion  = 1.7,
    pdfstandard = a-2b,
}

\documentclass[mydesign]{mitthesis}% fontset=newtx, fontset=libertine, fontset=libertinus, fontset=lmodern, 
\usepackage[english]{babel} %for the english langage
\usepackage{amsthm} %for the theorem setup
\usepackage{amssymb} %for math symbols
\usepackage[capitalize]{cleveref} %for clever references and links

\usepackage{tikz-cd} %I guess, more stuff than just tikz ?
\usepackage{graphicx} %for figures
\usepackage{amsmath} %for equations
\usetikzlibrary{angles, quotes}
\usepackage{float}

\usepackage{multibib}
\newcites{pub}{Publications and scripts}

\renewenvironment{thebibliography}[1]{
  \begin{oldthebibliography}{#1}
    \setlength{\itemsep}{0.78em} %before: 0.7
    \setlength{\parskip}{0pt plus 0.3ex}
}
{
  \end{oldthebibliography}
}

\usepackage{algorithm2e}
\RestyleAlgo{ruled}
\SetKwComment{Comment}{/* }{ */}

\definecolor{fillvertices}{RGB}{250,240,230}
\definecolor{cornflowerblue}{RGB}{100,149,237}
\definecolor{colormls1}{RGB}{15,100,255}
\definecolor{colormls2}{RGB}{20,150,150}
\definecolor{colormls3}{RGB}{106,90,205}
\definecolor[named]{lipicsYellow}{rgb}{0.99,0.78,0.07}
\colorlet{colorvertices}{cornflowerblue!40}

\tikzset{%
	vertices/.style={circle,minimum size=15pt,thick,draw,fill=colorvertices, inner sep = 0pt},
	edges/.style={draw=darkgray,very thick}
}

\newcommand{\defeq}{\vcentcolon=}

\newcommand{\ket}[1]{|{#1} \rangle}
\newcommand{\bra}[1]{\left \langle#1\right |}
\newcommand{\dloc}{\delta_\text{loc}}
\renewcommand{\iff}{\Leftrightarrow}
\newcommand{\se}{\subseteq}
\newcommand{\ls}{\leqslant}
\newcommand{\gs}{\geqslant}
\newcommand{\sm}{\setminus}
\newcommand{\ham}[1]{w(#1)}
\newcommand{\supp}{\textup{\textsf{supp}}}
\newcommand{\Zp}[1]{\mathbb{Z}/#1\mathbb{Z}}
\newcommand{\cutrk}{\textup{cutrk}}

\theoremstyle{plain}
\newtheorem{theorem}{Theorem}
\newtheorem{corollary}{Corollary}
\newtheorem{lemma}{Lemma}

\newtheorem{proposition}{Proposition}

\newtheorem*{théorème*}{Théorème}
\newtheorem*{proposition*}{Proposition}

\theoremstyle{remark}
\newtheorem{remark}{Remark}

\theoremstyle{definition}
\newtheorem{definition}{Definition}
\newtheorem{example}{Example}

\crefname{proposition}{Proposition}{Propositions}
\Crefname{proposition}{Proposition}{Propositions}

\hypersetup{%
	pdfsubject={Local equivalences of graph states},
	pdfkeywords={Nathan Claudet, quantum computing, graph states, stabilizer states, cluster states, local equivalence, local Clifford equivalence, local unitary equivalence, local complementation, vertex-minors},
	pdfurl={},
	pdfcontactemail={nathan.claudet@gmail.com},
	pdfauthortitle={},
	colorlinks = true,
    linkcolor={red},
	citecolor={red},
	urlcolor={magenta}
}

\begin{document}

\nocite{claudet2024covering, claudet2024local, claudet2025deciding, Cautres2024, claudet2023smallkpairablestates, codelulc19}

%%%%%%%%%%%%%%%%%%%%%%%%%%%%%%%%%%%%%%%%%%%%%%%%%%%%%%%%%%%%%%%%%%%%%%%%%%%%%%%%%%%%%%%%%%%%%%%%%

%%% Make titlepage
\thispagestyle{empty}

\hfill \textbf{École doctorale IAEM Lorraine}

\noindent\rule{\textwidth}{0.5pt}

    \vskip10mm plus2fil
\begin{center}
    \textbf{\Huge{Local Equivalences of Graph States}}
    {\vskip10mm plus2fil}
    \textbf{\LARGE{THÈSE}}
    {\vskip5mm plus2fil}
    présentée et soutenue publiquement le 17 novembre 2025 \\ pour l'obtention du
    {\vskip3mm plus1fil}
    \textbf{\Large{Doctorat de l'Université de Lorraine}} \\ (mention informatique)
    {\vskip5mm plus1fil}
    
    {\vskip3mm plus1fil}
    par
    {\vskip3mm plus1fil}
    \Large{Nathan Claudet}
\end{center}

%{\vskip5mm plus2fil} {\hskip1cm} \textbf{Composition du jury}
{\vskip5mm plus2fil} \textbf{Composition du jury}

\begin{center}
    \begin{tabular}{llll}%
        & \textit{Président :} & Xavier Goaoc & Université de Lorraine \\ 
        \rule{0pt}{4ex}   

        & \textit{Rapporteurs :} & Mamadou Moustapha Kanté & Université Clermont Auvergne \\ 
        && Robert Raussendorf & Leibniz Universität Hannover \\ 
        \rule{0pt}{4ex}   

        & \textit{Examinateurs  :} 
        & Otfried Gühne & Universität Siegen \\  
        && Elham Kashefi & CNRS \\           
        \rule{0pt}{4ex} 

        & \textit{Directeurs :} & Mathilde Bouvel & CNRS \\
        && Simon Perdrix & Inria   
    \end{tabular}
\end{center}
    \vskip5mm

\noindent\rule{\textwidth}{0.5pt}

\begin{center}
    \textbf{\small{Laboratoire Lorrain de Recherche en Informatique et ses Applications — UMR 7503}}
\end{center}

%%% Abstract
\pdfbookmark[0]{Abstract}{abstract}
\thispagestyle{empty}
%\chapter*{Abstract}

\begin{center}
    \textbf{\large{Résumé}}
\end{center}

{\small
Les états graphes forment une vaste famille d'états quantiques qui correspondent de manière bijective à des graphes mathématiques. Les états graphes sont utilisés dans de nombreuses applications, telles que le calcul quantique basé sur la mesure, en tant que ressources intriquées multipartites. Il est donc essentiel de comprendre quand deux états graphes ont la même intrication, c'est-à-dire quand ils peuvent être transformés l'un en l'autre en utilisant uniquement des opérations locales. Dans ce cas, on dit que les états graphes sont LU-équivalents (unitaire locale). Si les opérations locales sont restreintes au groupe de Clifford, on dit alors que les états graphes sont LC-équivalents (Clifford locale). Il est intéressant de noter qu'une règle graphique simple appelée complémentation locale capture exactement la LC-équivalence, dans le sens où deux états graphes sont LC-équivalents si et seulement si les graphes sous-jacents sont liés par une séquence de complémentations locales. Alors qu'il était autrefois conjecturé que deux états graphes LU-équivalents sont toujours LC-équivalents, il existe des contre-exemples et la complémentation locale ne parvient pas à capturer entièrement l'intrication des états graphe. 

Dans ce manuscript, nous introduisons une généralisation de la complémentation locale qui capture exactement la LU-équivalence. À l'aide de cette caractérisation, nous prouvons l'existence d'une hiérarchie infinie stricte d'équivalences locales entre la LC-équivalence et la LU-équivalence. Cela conduit également à la conception d'un algorithme quasi-polynomial permettant de déterminer si deux états graphes sont LU-équivalents, et à la preuve que deux états graphes LU-équivalents sont LC-équivalents s'ils sont définis sur au plus 19 qubits. 

De plus, nous étudions les états graphes qui sont universels dans le sens où tout état graphe plus petit, défini sur un ensemble suffisamment réduit de qubits, peut être induit en utilisant uniquement des opérations locales. Nous donnons des bornes et une construction probabiliste optimale.

\vspace{0.5em}

\noindent\textbf{Mots-clés:~} Informatique quantique, Intrication, Théorie des graphes

}

\vspace{0.7em}

\begin{center}
    \textbf{\large{Abstract}}
\end{center}

{\small

Graph states form a large family of quantum states that are in one-to-one correspondence with mathematical graphs. Graph states are used in many applications, such as measurement-based quantum computation, as multipartite entangled resources. It is thus crucial to understand when two such states have the same entanglement, i.e. when they can be transformed into each other using only local operations. In this case, we say that the graph states are LU-equivalent (local unitary). If the local operations are restricted to the so-called Clifford group, we say that the graph states are LC-equivalent (local Clifford). Interestingly, a simple graph rule called local complementation fully captures LC-equivalence, in the sense that two graph states are LC-equivalent if and only if the underlying graphs are related by a sequence of local complementations. While it was once conjectured that two LU-equivalent graph states are always LC-equivalent, counterexamples do exist and local complementation fails to fully capture the entanglement of graph states. 

We introduce in this thesis a generalization of local complementation that does fully capture LU-equivalence. Using this characterization, we prove the existence of an infinite strict hierarchy of local equivalences between LC- and LU-equivalence. This also leads to the design of a quasi-polynomial algorithm for deciding whether two graph states are LU-equivalent, and to a proof that two LU-equivalent graph states are LC-equivalent if they are defined on at most 19 qubits. 

Furthermore, we study graph states that are universal in the sense that any smaller graph state, defined on any small enough set of qubits, can be induced using only local operations. We provide bounds and an optimal, probabilistic construction.

\vspace{0.5em}

}

%%% Acknowledgments
\chapter*{Remerciements}
\pdfbookmark[0]{Remerciements}{remerciements}

Un merci tout particulier à mes parents pour leur soutien sans faille, et sans qui les études supérieures auraient été bien moins amusantes. Par ailleurs, cette thèse n'existerait pas sans vous, car je n'existerais pas moi-même, ce qui est clairement un désavantage pour faire des maths. 
Merci à tout le reste de la famille pour leur enthousiasme.

\vspace{1em}

Merci à Noé pour la Insane Coloc Max HD et pour avoir été un super camarade de thèse. Merci à lui et à Martin d'avoir fait le quantique avec moi depuis quelques années déjà.

\vspace{1em}

Bien sûr un immense merci à Simon, qui m'a proposé un super sujet de recherche, et qui a donné beaucoup de son temps pour m'aider à faire des jolis résultats et des jolis papiers. Merci à Mathilde pour nos discussions.

\vspace{1em}

Merci à ceux qui m'ont fait aimer les maths et l'informatique (et aussi un peu la physique), en prépa et à Télécom, et en particulier à Karine Brunel qui m'a appris l'existence de l'informatique quantique.

\vspace{1em}

Merci à Alessio, Anuyen, Benjamin, Colin, Joannès, Kathleen et Quentin pour les repas au labo, et merci à Isabelle pour la constante bonne humeur à midi. Merci aussi à Sophie pour son travail avant et après chaque voyage.

\vspace{1em}

Pour toutes les fiestas, les festivals et les voyages, qui ont grandement contribués à ma joie pendant ces trois dernières années, merci aux potes du Sud, de Nancy et de Paris. Je sais qu'il y en a certains qui aiment voir leur nom dans les remerciements, donc merci à: Agathe, Alexandre, Anas, Antoine, Aurore, Baptiste, Chiara, Christie, Cst, Ésa, Florence, Geoffrey, Hélène, Jodie, Joe, Josselin, Julien, Julien, Laura, Lise, Lucie, Lucille, Lucho, Mathieu, Margaux, Margot, Nathan, Mélo, Océane, Robin, Théo, et Tuyo. Merci spécial à Léoche pour tous les bons moments, tu es super.

\vspace{1em}

Enfin, merci aux petits chats sur le chemin du labo pour leur divertissement gratuit. 

%Enfin, merci à Tuyo. Même si tu es maintenant probablement un tas de métal dans une décharge, tu restes un tas de métal dans mon cœur.
\chapter*{Résumé (fr)}

\subsubsection{Informatique quantique}

Si la physique Newtonienne décrit avec précision l'univers à des échelles intermédiaires, la physique dite "classique" ne parvient pas à décrire les propriétés des objets à très petite échelle. La physique quantique est un formalisme qui tente de décrire avec précision l'évolution des objets très petits. Dire que les lois de la physique quantique diffèrent de celles de la physique classique est un euphémisme: de nombreuses propriétés fondamentales de la physique classique, telles que le déterminisme, ne s'appliquent plus.

L'idée d'utiliser les lois de la mécanique quantique pour effectuer des opérations de traitement de l'information est souvent attribuée à Richard Feymann, au début des années 80 \cite{Feynman1982}. L'une des avancées les plus importantes dans le domaine de l'informatique quantique est la conception de l'algorithme de Shor à la fin des années 90 \cite{Shor97}, qui factorise un nombre composé en temps polynomial à l'aide d'un ordinateur quantique, ce que l'on ne sait pas (pour l'instant?) faire avec un ordinateur classique. Malgré des applications prometteuses, la construction d'ordinateurs quantiques utiles reste un défi de taille en raison de la nature très délicate de l'information quantique.

Lorsque l'on tente d'évaluer l'origine de l'accélération potentielle des ordinateurs quantiques, l'intrication, c'est-à-dire ce qui ne peut être expliqué par la corrélation classique, est souvent le suspect principal. L'une des raisons pour lesquelles les états intriqués sont difficiles à exploiter dans des configurations classiques est que l'espace des états quantiques intriqués croît très rapidement avec le nombre de composants de base, les qubits, l'équivalent quantique des bits d'information classiques. Les \textbf{états graphes} sont un candidat clé pour expliquer et utiliser l'intrication multipartite. En effet, les états graphes sont des systèmes quantiques intéressants, dans le sens où ils présentent une forme complexe d'intrication, tout en étant assez faciles à exprimer et à manipuler grâce à leur correspondance concise avec les graphes (un graphe étant un ensemble de sommets reliés ou non par des arêtes).

\subsubsection{Les états graphes et leurs applications}

L'application la plus importante des états graphes est probablement comme ressource universelle pour le calcul quantique basé sur la mesure, ou MBQC ("Measurement-Based Quantum Computation") en abrégé \cite{raussendorf2001one, raussendorf2003measurement, briegel2009measurement}, introduit par Hans Briegel et Robert Raussendorf au début des années 2000. Les états graphes sont apparus comme généralisation des états clusters \cite{Briegel2001}, les resources originelles du MBQC \cite{Hein04}. Contrairement au modèle basé sur les circuits, où de l'intrication est ajoutée à l'état quantique tout au long du calcul, l'idée principale derrière le MBQC est de partir d'une ressource déjà intriquée - un état graphe - et d'effectuer un calcul consistant uniquement en des mesures  d'un seul qubit à la fois, et de décisions basées sur de l'information classique (le résultat des mesures). Le MBQC est un modèle de calcul très intéressant d'un point de vue théorique, car il permet de séparer les deux principales sources de puissance du calcul quantique: l'intrication et les opérations non-Clifford. En effet, la préparation de l'état graphe n'utilise que des opérations dites de Clifford et peut être considérée comme un problème complètement différent de celui de l'exécution du calcul. Il est donc naturel d'étudier les états graphes en eux-mêmes. Par exemple, la question de savoir comment préparer efficacement un état graphe est très importante tant pour les théoriciens que pour les expérimentateurs \cite{Perdrix06, Cabello2011, russo2018photonic, Li2022, Ghanbari2024}.

Les états graphes apparaissent aussi naturellement dans le contexte des codes stabilisateurs \cite{gottesman1997stabilizercodesquantumerror, Grassl2002}. Les codes stabilisateurs constituent une très vaste famille de codes quantiques, comprenant par exemple les codes CSS \cite{Calderbank1996, steane1996multiple}, et sont définis comme le point fixe commun d'opérateurs de Pauli qui commutent entre eux. Le formalisme des stabilisateurs a conduit au théorème de Gottesman-Knill \cite{gottesman1998heisenberg, Aaronson2004}, qui stipule que toute opération de Clifford peut être simulée efficacement sur un ordinateur classique. Il a été démontré que les codes graphe \cite{schlingemann2001quantum} - des codes quantiques basés sur les états graphes - sont des codes stabilisateurs, mais aussi que tout code stabilisateur est équivalent à un code graphe \cite{schlingemann2001stabilizer}, ce qui permet une visualisation plus explicite. En effet, les états graphes offrent un cadre plus intuitif que les tableaux de stabilisateurs pour concevoir des codes stabilisateurs utiles \cite{khesin2025quantum}.

Les états graphes sont également utiles pour créer des réseaux de communication quantiques, dans lesquels des individus distants détiennent chacun un qubit de l'état graphe. Les individus peuvent alors exécuter des protocoles cryptographiques, par exemple du partage de secret quantique \cite{markham2008graph, Keet2010, Javelle2013, gravier2013quantum, Bell2014secret}. Un autre exemple de construction est l'état graphe de répétition \cite{azuma2015all, azuma2023quantum}, un composant permettant de distribuer l'intrication dans un réseau de communication quantique. Il est courant, dans la conception de protocoles de routage de réseaux quantiques, d'utiliser les descriptions graphiques des opérations locales sur les états graphes \cite{hahn2019quantum, bravyi2024generating, meignant2019distributing, fischer2021distributing, Mannalath2023}, en tirant parti, par exemple, du riche formalisme des "vertex-minors" \cite{OumSurvey, dahlberg2020transforming}.

\subsubsection{L'intrication des états graphes}

Dans les applications décrites ci-dessus, les états graphes sont utilisés comme ressource d'intrication. Il est donc essentiel de classer les états graphes selon leur intrication. En général, on dit que deux états quantiques ont la même intrication si on peut passer de l'un à l'autre par opérations locales stochastiques et communication classique, ou SLOCC ("Stochastic Local Operations and Classical Communication") en abrégé. Pour les états graphes, cela se traduit par la notion de \textbf{LU-équivalence} ("local unitary"), ce qui signifie que l'on peut passer d'un état graphe à un autre (et vice-versa) par l'applications de portes quantiques unitaires sur un seul qubit à la fois. Si nous limitons les portes unitaires au groupe dit de Clifford, cela définit une notion plus forte d'équivalence: les états graphes sont dits \textbf{LC-équivalents} ("local Clifford"). La LC-équivalence des états graphes est particulièrement facile à caractériser, car une opération graphique simple et bien étudiée, appelée la \textbf{complémentation locale}, capture exactement la LC-équivalence d'états graphes. Cela implique notamment l'existence d'un algorithme efficace pour reconnaître les états graphes LC-équivalents. La LC-équivalence implique évidemment la LU-équivalence, et bien que l'inverse était un jour une conjecture \cite{5pb} (qui était connue sous le nom de conjecture LU=LC), il existe des paires d'états graphes qui sont LU-équivalents mais pas LC-équivalents \cite{Ji07}, laissant un écart entre ces deux notions d'équivalences locales. Dans ce qui peut-être considéré comme le papier le plus influent sur les états graphes \cite{Hein06}, écrit avant que la conjecture LU=LC ne soit réfutée, on peut lire:

"Notez qu'une coïncidence générale entre les notions de SLOCC-, LU- et LC-équivalence des états graphes serait particulièrement avantageuse pour les deux raisons suivantes: comme les trois notions d'équivalence locale correspondraient à la notion de LC-équivalence, 
\begin{enumerate}
    \item La vérification de l'équivalence locale entre états graphes donnés pourrait alors être effectuée efficacement;
    \item Les trois notions d'équivalence locale seraient entièrement décrites par la règle de complémentation locale, ce qui permettrait d'obtenir une description de l'équivalence locale des états graphes en termes simples et purement graphiques."
\end{enumerate} 
Bien qu'il soit désormais établi que ces équivalences locales ne coïncident pas et que les résultats concernant la LC-équivalence ne peuvent pas être directement transposés à la LU-équivalence, nous abordons dans ce manuscript ces deux points:
\begin{enumerate}
    \item Nous introduisons une généralisation de la complémentation locale qui décrit la LU-équivalence (et donc la SLOCC-équivalence) des états graphes en termes purement graphiques;
    \item Nous utilisons cette description graphique pour concevoir un algorithme quasi-polyno\-mial pour la LU-équivalence, ce qui signifie que la complexité en temps est de $O(n^{\log{n}})$ où $n$ est le nombre de qubits.
\end{enumerate}

Cette règle graphique nous aide également à mieux comprendre l'écart entre les deux principales équivalences locales des états graphes: nous identifions une hiérarchie infinie stricte d'équivalences locales des états graphes, entre la LC-équivalence et la LU-équivalence, ce qui nous permet de créer une image beaucoup plus précise de la situation.

\subsubsection{Contributions de la thèse}

L'objectif général de cette thèse est de contribuer à établir un lien entre la théorie des graphes et la théorie quantique, en exprimant des propriétés quantiques des états graphes en termes purement graphiques. La plupart des résultats présentés dans ce manuscrit traduisent soit une propriété quantique des états graphes en termes graphiques, soit utilisent cette traduction graphique pour étudier les propriétés quantiques des états graphes.

Nous nous intéressons principalement au problème de la LU-équivalence entre deux états graphes à $n$ qubits chacun, correspondants à deux graphes qui contiennent le même nombre $n$ de sommets. 
Notre principal outil graphique est la notion d'ensemble local minimal. Un ensemble de sommets d'un graphe est appelé un ensemble local quand il est l'union d'un ensemble non-vide $D$ de sommets - appelé générateur - et de l'ensemble $Odd(D)$ des sommets adjacents à un nombre impairs de sommets de $D$. Un ensemble local est dit minimal quand il ne contient aucun autre ensemble local. Les ensembles locaux minimaux sont utiles car ils sont invariants par LU-équivalence, et surtout, il nous donnent des informations sur l'opération quantique qui lie les deux états graphes. Plus précisément, un ensemble local minimal donne de l'information sur les sommets qui le constituent. Il est donc crucial d'avoir, pour n'importe quel graphe, une couverture des sommets par ensemble locaux minimaux, c'est-à-dire que n'importe quel sommet doit être contenu dans au moins un ensemble local minimal. Nous prouvons que c'est toujours possible.

\begin{théorème*}
    Tout graphe est couvert par ses ensembles locaux minimaux.
\end{théorème*}

La preuve de ce résultat tire parti de l'expression des ensembles locaux minimaux en termes de rang de coupe, une fonction bien étudiée, qui associe des ensembles de sommets à des entiers. Ce résultat est généralisable à des objects sur lesquels on peut définir une fonction qui partage les propriétés phares du rang de coupe. Ainsi, le théorème est valable sur les matroïdes, ou encore les multigraphes.

Nous nous intéressons également dans ce manuscript à des questions de complexité, c'est-à-dire le temps nécessaire, asymptotiquement, pour effectuer un calcul. Nous donnons notamment une procédure efficace pour calculer une couverture par ensemble locaux minimaux.

\begin{théorème*}
    Il existe un algorithme efficace pour calculer une famille d'ensembles locaux minimaux qui couvre un graphe.
\end{théorème*}

Nous prouvons aussi des résultats divers sur la taille et le nombre d'ensembles locaux minimaux dans un graphe:
\begin{itemize}
    \item Les ensembles locaux minimaux contiennent au plus la moitié des sommets du graphe;
    \item Il existe des graphes avec seulement des petits ensembles locaux minimaux (précisément, de taille 2);
    \item Il existe des graphes avec seulement des gros ensembles locaux minimaux (de taille linéaire en le nombre de sommets du graphe);
    \item Il existe des graphes avec un nombre exponentiels d'ensembles minimaux locaux.
\end{itemize}

Pour attaquer l'étude de la LU-équivalence des états graphes, nous définissons d'abord des notions d'équivalences locales intermédiaires, entre la LC-équivalence et la LU-équivalence, que l'on appelle LC$_r$-équivalence, où le niveau $r > 0$ est un entier. Le niveau 1, la LC$_1$-équivalence, coïncide avec la notion de LC-équivalence.

Nous généralisons aussi la notion de complémentation locale, et définissons la $r$-complé\-mentation locale, où $r > 0$ est un entier. La 1-complémentation locale correspond plus ou moins à la complémentation locale. La complémentation locale capture la LC-équivalence des états graphes, dans le sens où deux états graphes sont LC-équivalents si et seulement si on peut passer de l'un à l'autre par une séquence de complémentations locales. Nous prouvons que la $r$-complémentation locale généralise bien cette propriété importante:

\begin{théorème*}
    Deux états graphes sont LC$_r$-équivalents si et seulement si on peut passer de l'un à l'autre par une séquence de $r$-complémentations locales.
\end{théorème*}

De plus, la $r$-complémentation locale capture la LU-équivalence des états graphes (cela constitue le résultat central de ce manuscript), prouvant au passage que deux états graphes LU-équivalents sont LC$_r$-équivalent.

\begin{théorème*}
    Deux états graphes sont LU-équivalents si et seulement s'il existe un entier $r$ tel que l'on peut passer d'un état graphe à l'autre par des $r$-complémentations locales.
\end{théorème*}

On sait que la complémentation locale usuelle ne capture pas la LU-équivalence des états états graphes. Cependant, une question importante est de savoir s'il existe un niveau $r$ tel que la $r$-complémentation locale capture la LU-équivalence de tous les états graphes. Nous répondons par la négative.

\begin{théorème*}
    Pour tout entier $r > 1$, il existe des états graphes qui sont LC$_r$-équivalents mais pas LC$_{r-1}$-équivalents. Ainsi, pour tout entier $r > 0$, il existe des états graphes qui sont LU-équivalents mais pas LC$_r$-équivalents.
\end{théorème*}

Cela prouve que la notion de LC$_r$-équivalence constitue une hiérarchie infinie stricte d'équivalence locales entre les états graphes.

Nous utilisons notre caractérisation graphique de la LU-équivalence pour étudier la complexité du problème de décider si deux états graphes sont LU-équivalents. Comme précédemment, nous commençons par étudier le problème de la LC$_r$-équivalence, et créons un algorithme efficace (quand $r$ est fixé).

\begin{théorème*}
    Il existe un algorithme efficace pour décider si deux états graphes sont LC$_r$-équivalents. Sa complexité en temps est de l'ordre de $n^r$, où $n$ est le nombre de sommets.
\end{théorème*}

Une des étapes de cet algorithme est une généralisation de l'algorithme de Bouchet qui décide si deux états graphes sont LC-équivalents. Nous autorisons l'addition de contraintes supplémentaires, tout en conservant l'efficacité de l'algorithme, moyennant de faibles conditions sur les graphes.

\begin{proposition*}
    Il existe un algorithme efficace pour décider si deux états graphes (qui contiennent au moins un sommet de degré pair) sont LC-équivalents, avec des contraintes supplémentaires (faisant partie de l'entrée de l'algorithme) pouvant être exprimées comme des équations linéaires.
\end{proposition*}

Comme l'algorithme pour décider la LC$_r$-équivalence a une complexité de l'ordre de $n^r$, il n'implique pas directement un algorithme efficace pour décider la LU-équivalence (car $r$ n'est \emph{a priori} pas suffisamment borné). Pour y remédier, nous prouvons une borne sur le niveau des $r$-complémentations locales qui lient deux états graphes LU-équivalents.

\begin{théorème*}
    Si deux états graphes sont LU-équivalents, on peut passer de l'un à l'autre par des $r$-complémentations locales avec $r$ au plus de l'ordre de $\log(n)$, où $n$ est le nombre de sommets.
\end{théorème*}

Ainsi, notre algorithme pour décider la LC$_r$-équivalence, décide aussi la LU-équivalence des états graphes en temps quasi-polynomial. Auparavant, seules des méthodes en complexité exponentielle existaient.

\begin{théorème*}
    Il existe un algorithme quasi-polynomial pour décider si deux états graphes sont $LU$-équivalents. Sa complexité en temps est de l'ordre de $n^{\log_2(n)}$, où $n$ est le nombre de sommets.
\end{théorème*}

Bien qu'il existe des paires d'états graphes qui sont LU-équivalents mais pas LC-équivalents, il existe des états graphes dont l'orbite par LC-équivalence est la même que l'orbite par LU-équivalence. On dit que ces états graphes vérifient LU=LC. LU=LC est une propriété désirable pour un état graphe, par example il est possible d'explorer son orbite par LU-équivalence, avec seulement des complémentations locales. Nous créons un protocole pour prouver que certaines familles d'états graphes vérifient LU=LC, basé sur la $r$-local complémentation. Nous utilisons cette procédure pour prouver que certains états graphes de répétition, utiles pour les réseaux de communication quantique, vérifient LU=LC.

\begin{proposition*}
    Certaines instances d'états graphes de répétition vérifient LU=LC.
\end{proposition*}

Il était su depuis longtemps que tous les états graphes avec au plus 8 qubits, vérifient LU=LC. Ce résultat avait été obtenu par une exploration exhaustive assistée par ordinateur. Le formalisme de la $r$-complémentation locale permet de drastiquement améliorer ce résultat.

\begin{proposition*}
    Tous les états graphes avec au plus 19 sommets vérifient LU=LC.
\end{proposition*}

Enfin, nous étudions la propriété de vertex-minor universalité. Un graphe est $k$-vertex-minor universel, s'il est possible de créer n'importe quel graphe sur n'importe quel ensemble de $k$ sommets, avec des complémentations locales, et des destructions de sommets. Pour l'état graphe correspondant, cela signifie que l'on peut créer n'importe quel état graphe sur n'importe quel ensemble de $k$ qubits, avec seulement des opérations locales. Nous montrons que pour tout entier $k$, il existe des graphes $k$-vertex-minor universels de taille raisonnable.

\begin{théorème*}
    Pour toute constante $\alpha > 2$, et pour $k$ assez grand, il existe un graphe $k$-vertex-minor universel avec moins de $\alpha k^2$ sommets. 
\end{théorème*}

Nous montrons que cette construction est asymptotiquement optimale, c'est à dire qu'un graphe $k$-vertex-minor universel possède toujours au moins un nombre quadratique (en $k$) de sommets.

%%% Table of contents and lists of stuff (delete unused lists, i.e., if no tables or figures) %%%%%
\tableofcontents
%\listoffigures
%\listoftables

%%% Chapters of thesis  %%%%%%%%%%%%%%%%%%%%%%%%%%%%%%%%%%%%%%%%%%%%%%%%%%%%%%%%%%%%%%%%%%%%%%%%%%%

%% If you want to use "double spacing", you should start here...

\chapter{Introduction}

\subsubsection{Quantum computing}

While Newtonian physics accurately describe the universe at intermediate scales, the so-called \textit{classical} physics fail to describe the properties of objects at a very small scale. Quantum physics is a formalism that tries to accurately describe how very small objects evolve. It is an understatement to say that the laws of quantum physics differ from the laws of classical physics: many fundamental properties of classical physics, such as determinism, no longer apply.

The idea of utilizing the laws of quantum mechanics to perform information processing operations is often attributed to Richard Feynman, in the early 80s \cite{Feynman1982}. One of the most important breakthroughs in quantum computing is the design of Shor's algorithm in the late 90s \cite{Shor97}, which factors a composite number in polynomial-time using a quantum computer, which is not known to be possible using a classical computer. Despite promising applications, building useful quantum computers is still a huge challenge because of the very delicate nature of quantum information.

When assessing where a potential speedup of quantum computers comes from, entanglement, i.e. what can not be explained by classical correlation, is the usual suspect. Part of the reason why entangled states are difficult to harness with classical settings, is that the space of entangled quantum states grows very rapidly with the number of basic components, the qubits, the quantum equivalent of the classical bits of information. \textbf{Graph states} are a key candidate to explain and design around multipartite entanglement. Indeed, graph states are intricate quantum systems, in the sense that they exhibit a complex form of entanglement, while being fairly easy to express and manipulate thanks to their concise one-to-one correspondence with mathematical graphs.

\subsubsection{Graph states and their applications}
%With corrections suggested by Robert Raussendorf

The most prevalent use of graph states is as universal resources in measurement-based quantum computation, MBQC for short \cite{raussendorf2001one, raussendorf2003measurement, briegel2009measurement}, introduced by Hans Briegel and Robert Raussendorf in the early 2000s. Graph states first appeared as a generalization of the cluster states \cite{Briegel2001}, the original resources for MBQC \cite{Hein04}. In opposition to the circuit model, where entanglement is added to the quantum state during the computation, the main idea behind MBQC is to start with an already entangled resource - a graph state - and perform a computation consisting only of single-qubit measurements along with classical feedback and control. MBQC is a valuable model of computation from a theoretical perspective, as it provides a separation between two main sources of quantum computational power: entanglement, and non-Clifford operations. Indeed, preparing the resource graph state uses only the so-called Clifford operations and can be thought of as a completely different problem than performing the computation. It is thus natural to study graph states on their own. For example, how to efficiently prepare a graph state is a very important question for both theorists and experimentalists \cite{Perdrix06, Cabello2011, russo2018photonic, Li2022, Ghanbari2024}.

Graph states also arise naturally in the context of stabilizer codes \cite{gottesman1997stabilizercodesquantumerror, Grassl2002}. Stabilizer codes comprise a very large family of quantum codes, including for example the CSS codes \cite{Calderbank1996, steane1996multiple}, and are defined as the common fixpoint of commuting  Pauli operators. The stabilizer formalism has led to the Gottesman-Knill theorem \cite{gottesman1998heisenberg, Aaronson2004}, which states that any Clifford operation can be efficiently simulated on a classical computer. It has been shown that graph codes \cite{schlingemann2001quantum} - quantum codes based on graph states - are stabilizer codes, but also that any stabilizer code is equivalent to a graph code \cite{schlingemann2001stabilizer}, leading to a more intuitive visualization. Indeed, graph states offer a more intuitive framework than stabilizer tableaux to design useful stabilizer codes \cite{khesin2025quantum}.

Graph states are useful to create quantum networks, where distant parties hold one qubit of the graph state each. The parties can then perform cryptography protocols, for example quantum secret sharing \cite{markham2008graph, Keet2010, Javelle2013, gravier2013quantum, Bell2014secret}. Another example of a powerful construction is the repeater graph state \cite{azuma2015all, azuma2023quantum}, a component allowing entanglement to be distributed in a quantum network. It is common in the design of quantum network routing protocols to use the graphical descriptions of the local operations over graph states \cite{hahn2019quantum, bravyi2024generating, meignant2019distributing, fischer2021distributing, Mannalath2023}, taking advantage for example of the rich formalism of vertex-minors \cite{OumSurvey, dahlberg2020transforming}.

\subsubsection{Entanglement of graph states}

In the applications described above, graph states are used as a resource of entanglement. It is thus an essential question to classify graph states according to their entanglement. For general quantum states, having the same entanglement is often formalized at being related by SLOCC (stochastic local operations and classical communication). For graph states in particular, this is the same as being local unitary equivalent, or \textbf{LU-equivalent} for short, meaning that the graph states are related by single-qubit unitary operators. If we restrict the single-qubit unitaries to be in the so-called Clifford group, this defines a stronger notion of equivalence: the graph states are said local Clifford equivalent, or \textbf{LC-equivalent} for short. LC-equivalence of graph states is particularly easy to characterize, as a simple and well-studied graphical operation, called \textbf{local complementation}, exactly captures LC-equivalent graph states. This implies for example an efficient algorithm for recognizing LC-equivalent graph states. LC-equivalence obviously implies LU-equivalence, and while the converse was once conjectured to be true \cite{5pb} (this was known as the LU=LC conjecture), there exist pairs of graph states that are LU-equivalent but not LC-equivalent \cite{Ji07}, leaving a gap between these two local equivalences. In what is arguably the most influential survey on graph states \cite{Hein06}, written before the LU=LC conjecture was disproven, we can read: "Note that a general coincidence of SLOCC-, LU- and LC-equivalence for graph states would be particularly advantageous for the following two reasons: since in this case all 3 local equivalences would correspond to LC-equivalence, 
\begin{enumerate}
    \item Checking whether two given graph states are locally equivalent could then be done efficiently;
    \item All 3 locally equivalences would entirely be described by the local complementation rule, yielding a description of local equivalence of graph states in simple, purely graph theoretic terms. "
\end{enumerate} 
While it is now established that these local equivalences do not coincide, and the results about LC-equivalence can not be directly lifted to LU-equivalence, we address these two points:
\begin{enumerate}
    \item We introduce a generalization of local complementation that does describe LU-equiva\-lence (and thus SLOCC-equivalence) of graph states in purely graph theoretic terms;
    \item We use this graphical description to design a quasi-polynomial algorithm for LU-equivalence, meaning the time-complexity is $O(n^{\log_2{n}})$ where $n$ is the number of qubits.
\end{enumerate}

This graphical rule also greatly helps our understanding of the gap between the two main local equivalences of graph states: we pinpoint an infinite strict hierarchy of local equivalences of graph states, between LC- and LU-equivalence, leading to a much finer picture of the situation.

\subsubsection{Structure of the thesis}

The overall goal of this thesis is to help connect graph theory and quantum theory, by harnessing quantum properties of graph states  in purely graph-theoretical terms. Most results presented in this thesis either translate a quantum property of graph states in graph terms, or use this graphical translation to study quantum properties of graph states.\\

\cref{chap:background} is an introduction to the mathematical concepts covered by this thesis. It contains discussions on LC- and LU-equivalent graph states, and the relation to local complementation and the cut-rank function of a graph. We also review Bouchet's algorithm for recognizing LC-equivalent graphs.\\

\cref{chap:mls} is based on a WG 2024 paper \citepub{claudet2024covering} with Simon Perdrix: "Covering a Graph with Minimal Local Sets". We lay the mathematical foundations of the study of LU-equivalent graph states. We study minimal local sets, peculiar sets of vertices of a graph that are related to the stabilizers of the corresponding graph states. Importantly, minimal local sets are invariant by LU-equivalence, and they give strong constraints of the unitaries applied to each of the qubits corresponding to their vertices. It is thus important to have, for any graph, a family of minimal local sets such that each vertex is contained in at least one. We show that such a family always exists and show how to construct one in polynomial time.\\

\cref{chap:glc} is based on a STACS 2025 paper \citepub{claudet2024local} with Simon Perdrix: "Local Equivalence of Stabilizer States: a Graphical Characterisation". 
We introduce a graph rule that exactly captures LU-equivalence. This graph rule is a generalization of the local complementation that is parametrized by an integer $r$, leading to the notion of $r$-local complementation (the 1-local complementation corresponds to the usual local complementation). We use this graphical characterization of LU-equivalence to prove the existence of an infinite strict hierarchy of equivalences between graph states, that we call LC$_r$-equivalences, bridging the gap between LC- and LU-equivalence.\\

\cref{chap:algo} is based on an ICALP 2025 paper \citepub{claudet2025deciding} with Simon Perdrix: "Deciding Local Unitary Equivalence of Graph States in Quasi-Polynomial Time". 
We devise a (classical) algorithm that decides whether two graph states are LU-equivalent, based on the generalized local complementation introduced in the aforementioned paper \citepub{claudet2024local}. The algorithm runtime is shown quasi-polynomial thanks to new results on properties of the generalized local complementation.\\

\cref{chap:conditionsLULC} contains results from both \citepub{claudet2024local} and \citepub{claudet2025deciding}. In general, for graph states, LC- and LU-equivalence do not coincide. However, for some restricted classes of graph states, LU=LC i.e. LC- and LU-equivalence do coincide. We provide a general protocol to find such classes of graph states and use it to prove that LU=LC holds for repeater graph states. 
We also prove that LU=LC holds for graph states defined on at most 19 qubit, which was only previously known up to 8 qubits.\\

\cref{chap:vmu} is based on an ICALP 2024 paper with Maxime Cautrès, Mehdi Mhalla, Simon Perdrix, Valentin Savin and Stéphan Thomassé \citepub{Cautres2024}: "Vertex-Minor Universal Graphs for Generating
Entangled Quantum Subsystems". Some results are from an earlier, unpublished version of the paper \citepub{claudet2023smallkpairablestates}. We prove the existence of $n$-qubit graph states for which it is possible to induce any stabilizer state on any $\Theta(\sqrt{n})$ qubits by using only local Clifford operations and classical communication. The construction is probabilistic and is based on random bipartite graphs. We rely on of the vertex-minor formalism.\\

\subsubsection{Reading guide}

\cref{chap:glc} relies on results from \cref{chap:mls}, but both chapters can be understood independently. \cref{chap:algo} and \cref{chap:conditionsLULC} relies heavily on the results and the formalism introduced in \cref{chap:glc}. \cref{chap:conditionsLULC} also rely on some results from \cref{chap:algo}. \cref{chap:vmu} can be understood independently of the other chapters.% .tex extension is presumed
\chapter{Background}

\label{chap:background}

\section{Graphs}

A \textbf{graph} $G$ is composed two sets, a set $V$ of vertices, and a set $E$ of edges between the vertices. We write $G = (V,E)$. We consider finite graphs, so $V$ is a finite set. Graphically, the vertices are points, and the edges are links connecting points. Mathematically, if $n$ is the number of vertices, $V$ can be associated with the set $[1,n] = \{1,2, \cdots, n\}$ of integers between 1 and $n$, while $E$ is composed of couples $(u,v)$, where $u, v \in V$. Through this thesis, the objects we call graphs are actually simple and undirected graphs. A graph is \textbf{simple} if each edge appears once in $E$, i.e. $E$ is a subset of $V^2$, and for any vertex $u \in V$, $(u,u) \notin E$, i.e. there is no edge connecting a vertex to itself. A graph is \textbf{undirected} if $(u,v) \in E \iff (v,u) \in E$, i.e. edges have no direction.

Let us emphasize that the graphs considered in this thesis are labelled, meaning that each vertex is assigned a distinct identifier. For instance the graph composed of 3 vertices, 1,2 and 3, with edges $(1,2)$ and $(2,3)$, is distinct from the graph composed of 3 vertices, 1,2 and 3, with edges $(1,2)$ and $(1,3)$. Additionally, we consider that each vertex set $V$ is equipped with a canonical total order $<$.

\subsubsection{Notations}  Consider the graph $G = (V,E)$. The \textbf{order} of the graph is the number of vertices $|V|$, and will be denoted $n$. We use the notation $u \sim_G v$ when $(u,v) \in E$, and we say $u$ and $v$ are \textbf{adjacent} (or sometimes neighbors). Else we note $u \not\sim_G v$. Notice that $u \sim_G v \iff v \sim_G u$ as $G$ is undirected. Given two sets of vertices $D_1, D_2 \se V$, the symmetric difference of $D_1$ and $D_2$ is $D_1 \Delta D_2 = (D_1 \cup D_2) \sm (D_1 \cap D_2)$. Given a vertex $u \in V$, $N_G(u) = \{v\in V ~|~ u \sim_G v\}$ is the \textbf{neighborhood} of $u$, i.e. the set of vertices that are adjacent to $u$. Notice that $u \notin N_G(u)$ as $G$ is simple. The \textbf{degree} $\delta_G(u)$ of a vertex $u$ is the number of vertices adjacent to $u$, i.e. $\delta_G(u) = |N_G(u)|$. The odd and common neighborhoods are two natural generalizations of the notion  of neighborhood to sets of vertices. For a set of vertices $D \se V$, the \textbf{odd neighborhood} of $D$ is $Odd_G(D) =  \Delta_{u\in D} N_G(u) = \{v \in V ~|~|N_G(v) \cap D| = 1 \text{ mod } 2\}$. Informally, $Odd_G(D)$ is the set of vertices that are adjacent to an odd number of vertices in $D$. The \textbf{common neighborhood} of $D$ is $\Lambda_G^D=\bigcap_{u\in D}N_G(u)=\{v\in V~|~\forall u \in D, v\in N_G(u)\}$. Informally, $\Lambda_G^D$ is the set of vertices that are adjacent to every single vertex in $D$. Notice that $Odd_G(D)$ may contain vertices of $D$, but $\Lambda_G^D$ may not, as $G$ is simple. A set of vertices $D \se V$ is said \textbf{independent} if no two vertices of $D$ are adjacent, i.e. for each $u \in D$, $N_G(u) \cap D = \emptyset$. Given a set $D \se V$ of vertices, the \textbf{subgraph} of $G$ \textbf{induced} by $D$ is the graph $G[D]$, defined with $D$ as the set of vertices, % and whose edges correspond exactly to the edges in $G$. 
and where any two vertices are adjacent in $G[D]$ if and only if they are adjacent in $G$. For example, if $u$ is a vertex, the graph $G \sm u = G[V\sm\{u\}]$ is the graph obtained from $G$ after removing $u$ and all the edges adjacent to $u$. A graph is said \textbf{connected} if there is a path between any two vertices. A graph that is not connected is the union of so-called \textbf{connected components}, maximal induced subgraphs that are connected. 
A graph is said \textbf{bipartite} (or 2-colorable) is there exists a bipartition $V = L \cup R$ of the vertices such that $L$ and $R$ are independent sets, i.e. both $G[L]$ and $G[R]$ are empty graphs. Two unadjacent vertices $u$ and $v$ are said \textbf{twins} if they are adjacent to the same vertices i.e. $N_G(u) = N_G(v)$.

For each of the notations above, "$G$" may occasionally be omitted when the graph is clear form the context.

\section{Quantum computing}

\subsection{Quantum states}

In the formalism of quantum mechanics, the \textbf{qubit}, the quantum equivalent of a bit, is represented as a unit vector in $\mathbb C^2$. That is, a vector $\left(\begin{smallmatrix} a \\ b \end{smallmatrix}\right)$, where $a$ and $b$ are complex numbers that satisfy the condition $|a|^2 + |b|^2 = 1$. In quantum mechanics, a more common notation is the Dirac notation, where a qubit is written $a \ket 0 + b \ket 1$, with $\ket 0 = \left(\begin{smallmatrix} 1 \\ 0 \end{smallmatrix}\right)$ and $\ket 1 = \left(\begin{smallmatrix} 0 \\ 1 \end{smallmatrix}\right)$. Contrary to the classical bit, the qubit is not exactly either $\ket 0$ or $\ket 1$ but instead is in a weighted superposition of the basis vectors $\ket 0$ and $\ket 1$. The composition of the qubit $a \ket 0 + b \ket 1$ and the qubit $c \ket 0 + d \ket 1$ is the so-called product state $(a \ket 0 + b \ket 1) \otimes (c \ket 0 + d \ket 1) = ac \ket{00} + ad \ket{01} + bc \ket{10} + bd \ket{11}$. $\otimes$ here represents the tensor product. In general, the composition of $n$ qubits is a unit vector\footnote{A vector $v$ in $\mathbb C^{k}$ is a unit vector if $|v_1|^2+|v_1|^2+\cdots+|v_{k}|^2 =1$.} in $\mathbb C^{2^n}$. In fact, the family of unit vectors in $\mathbb C^{2^n}$ is the family of $n$-qubit \textbf{quantum states}. It happens that not any quantum state can be written as a composition of qubits. 
We say in this case that the quantum state is \textbf{entangled}, meaning  that it cannot be described as the independent sum of its parties.

The reversible operations that map quantum states to other quantum states correspond, in the formalism of quantum mechanics, to \textbf{unitary matrices}. A unitary matrix $U$ is such that $U U^\dagger = U^\dagger U = I$, where $I$ is the identity matrix and $U^\dagger$ is the conjugate transpose of $U$. A unitary matrix acting on $n$ qubits has dimension $2^n$, i.e. it is a $2^n \times 2^n$ matrix. An example of a single-qubit unitary gate is the Hadamard gate:
$$ H\defeq 
\frac{1}{\sqrt 2}
\begin{pmatrix}
    1 & 1\\
    1 & -1
\end{pmatrix} 
$$
The Hadamard gate is useful for creating states in superposition from the basis states, as it maps $\ket 0$ to $\ket + \defeq \frac{1}{\sqrt{2}}\left(\ket{0}+\ket{1}\right)$ and $\ket 1$ to $\ket - \defeq \frac{1}{\sqrt{2}}\left(\ket{0}-\ket{1}\right)$. We note $H \ket 0 = \ket +$ and $H \ket 1 = \ket -$. As for quantum states, unitary gates can be composed by the tensor product: for example $H \otimes H$ denotes the 2-qubit unitary gate that perform the Hadamard gate on both the first and second qubits. The tensor product of single-qubit unitary gates cannot create entangled states from product states (and vice versa). However, some multiple-qubit unitaries cannot be written as tensor product of single-qubit unitaries and can actually create entangled states from non-entangled product states. An example of such an entangling gate is the $CZ$ gate:
$$CZ =  \begin{pmatrix}
1 & 0 & 0 & 0\\
0 & 1 & 0 & 0\\
0 & 0 & 1 & 0\\
0 & 0 & 0 & -1\\
\end{pmatrix}$$ 
The $CZ$ gate flips the sign before the basis vector $\ket{11}$ and leave the other basis vectors invariant. As an example of creation of entanglement, applying $CZ$ to the product state $\ket + \ket + = \left(\frac{1}{\sqrt{2}}\left(\ket{0}+\ket{1}\right)\right)\otimes \left(\frac{1}{\sqrt{2}}\left(\ket{0}+\ket{1}\right)\right) = \frac{1}{2}\left(\ket{00}+\ket{01}+\ket{10}+\ket{11} \right)$ results in the state $CZ \ket + \ket + = \frac{1}{2}\left(\ket{00}+\ket{01}+\ket{10}-\ket{11} \right)$. This resulting state is entangled, i.e. it cannot be written as a product state.

In general, an $n$-qubit quantum state $\ket \psi$ is said \textbf{fully entangled} if for no bipartition A, B of its qubits, $\ket \psi$ can be written as a composition $\ket{\psi_A} \otimes \ket{\psi_B}$. Such fully entangled states include the $n$-qubit W state
$$ \ket W = \frac{1}{\sqrt{n}}\left(\ket{00\cdots01} + \ket{00\cdots10} + \cdots + \ket{10\cdots00} \right)$$
or the $n$-qubit GHZ (Greenberger-Horne-Zeilinger) state:
$$ \ket{\text{GHZ}} = \frac{1}{\sqrt{2}}\left(\ket{00\cdots00} + \ket{11\cdots11}\right)$$
To classify multiparty entanglement, it is standard to say that two quantum states are equivalent, or have the same entanglement, if they can be mapped one to the other with reversible local (i.e. qubit per qubit) operations. For example the 2-qubit GHZ state is equivalent to $CZ \ket + \ket +$, as $\left(H \otimes I \right) \frac{1}{\sqrt{2}} \left( \ket{00} + \ket{11}\right) = \frac{1}{2}\left(\ket{00}+\ket{01}+\ket{10}-\ket{11} \right)$. Conversely, for $n \gs 3$, the $n$-qubit GHZ state and the $n$-qubit W state are not equivalent.

\subsubsection{Notations}

In this thesis, $n$ always denotes the number of qubits of the quantum state. $V$ denotes the sets of qubits, and thus $n = |V|$.  
If a specified single-qubit unitary gate $U$ is applied to the qubit $u$, we write $U_u$. If the specified unitary gate $U$ is applied to a set $D$ of qubits, we write $U_D$. In some cases, we apply a tensor product of (possibly distinct) \emph{unspecified} single-qubit unitary gates to a set of qubits $D$; if we note $U_u$ the unitary gate applied on the qubit $u$, then we write $\bigotimes_{u \in D} U_u$ instead of $\prod_{u \in D} (U_u)_u$ for readability. For two-qubit unitaries such as $CZ$, we write $CZ_{uv}$ when $CZ$ is applied to the qubits $u$ and $v$. We use the Dirac notation: while $\ket \psi$ represents a vector, $\bra \psi$ represents the adjoint of $\ket \psi$. Thus, if $\ket \psi$ is a quantum state, i.e. represents a unit vector, then $\bra \psi \ket \psi = 1$. 
Given a set $B \se V$ of qubits, the \textbf{partial trace} over $B$ (also called the reduced density matrix) of a $2^n \times 2^n$ matrix $M$ is defined as $Tr_B(M) \defeq \sum_{x \in \{0,1\}^{|B|}} \bra{x}_B M \ket{x}_B $ (where $\ket x_B$ is short for $\ket x_B \otimes I_{V\sm B}$ here). The transpose of a matrix $M$ is noted $M^T$, and the conjugate transpose is noted $M^\dagger$. A global phase before a quantum state or unitary operation is often irrelevant and will be denoted $e^{i\phi}$.

\subsection{Single-qubit unitary gates}

In this thesis, we mostly care about local unitary gates, where "local" means qubit per qubit. 
These local unitary gates are represented by tensor products of single-qubit unitary matrices. Thus, we review, in this section, some useful single-qubit unitary gates.

Arguably the most prevalent single-qubit unitaries in quantum mechanics are the Pauli gates:

$$ I\defeq 
\begin{pmatrix}
    1 & 0\\
    0 & 1
\end{pmatrix} 
\qquad 
X\defeq 
\begin{pmatrix}
    0 & 1\\
    1 & 0
\end{pmatrix} 
\qquad 
Y\defeq 
\begin{pmatrix}
    0 & -i\\
    i & 0
\end{pmatrix} 
\qquad 
Z\defeq 
\begin{pmatrix}
    1 & 0\\
    0 & -1
\end{pmatrix} $$

Below are some basic properties of the Pauli gates.

\begin{proposition}
    X, Y and Z have eigenvalues +1 and -1: 
    \begin{itemize}        
        \item $X \ket + = \ket +$; $X \ket - = - \ket -$;
        \item $Z \ket 0 = \ket 0$; $Z \ket 1 = - \ket 1$;
        \item $Y \frac{1}{\sqrt{2}}\left(\ket{0}+i\ket{1}\right) = \frac{1}{\sqrt{2}}\left(\ket{0}+i\ket{1}\right)$; $Y \frac{1}{\sqrt{2}}\left(\ket{0}-i\ket{1}\right) = -\frac{1}{\sqrt{2}}\left(\ket{0}-i\ket{1}\right)$.        
    \end{itemize}
\end{proposition}

\begin{proposition}
    The Pauli gates form a basis for $2\times 2$ matrices with coefficient in $\mathbb C$, that is, for any $2\times 2$ matrix M, there exist unique complex numbers $a,b,c,s$ such that $M = a I + b X + c Y + d Z$.
\end{proposition}

The Pauli gates, augmented with a constant factor that is a multiple of $-i$, form a multiplicative group, noted $\mathcal P = \{\pm 1, \pm i\} \times \{I,X,Y,Z\}$. Pauli gates either commute or anticommute.

\begin{proposition}
    ~
    \begin{itemize}
        \item $X^2 = Z^2 = Y^2 = I$;
        \item $ XY = -YX = i Z$;
        \item $ XZ = -ZX = -i Y$;
        \item $ YZ = -ZY = iX$.
    \end{itemize}
\end{proposition}

While each possible quantum state of a single qubit is a point on the so-called Bloch sphere, every single-qubit unitary gate can be seen as a rotation in the Bloch sphere, see Figure \ref{fig:bloch}. 

\begin{figure}[H]
\centering
\scalebox{1.3}{
\begin{tikzpicture}

  % Define radius
  \def\r{3}

  % Bloch vector
  \draw[color = blue] (0, 0) node[circle, fill, inner sep=1, color=black] (orig) {} -- (-\r/3, \r/2) node[circle, fill, inner sep=0.7, label=above:$\ket \psi$] (a) {};

  % Sphere
  \draw (orig) circle (\r);
  \draw[dashed] (orig) ellipse (\r{} and \r/3);

  % Axes
  \draw[-latex, thick] (orig) -- node[left] {X} ++(-\r/5, -\r/3) node[below] (x1) {$\ket +$};
  \draw[-latex, color=black!40] (orig) -- ++(\r/5, \r/3) node[above right] (x1) {$\ket -$};
  \draw[-latex, thick] (orig) -- node[below] {Y} ++(\r, 0) node[right] (x2) {};
  \draw[-latex, color=black!40] (orig) -- ++(-\r, 0) node[right] (x2) {};
  \draw[-latex, thick] (orig) -- node[right] {Z} ++(0, \r) node[above] (x3) {$\ket 0$};
  \draw[-latex, color=black!40] (orig) -- ++(0, -\r) node[below] (x3) {$\ket 1$};

\end{tikzpicture}
}
\caption{The Bloch sphere. A single-qubit quantum state $\ket \psi$ corresponds to a point on the surface of the sphere. More precisely, a point corresponds to quantum states up to an irrelevant global phase, that is, $\ket \psi$ and $e^{i\phi}\ket \psi$ share the same point on the sphere. The axis corresponds to the Pauli gates X, Y and Z, and the endpoints corresponds to their respective eigenvectors.}
\label{fig:bloch}
\end{figure}
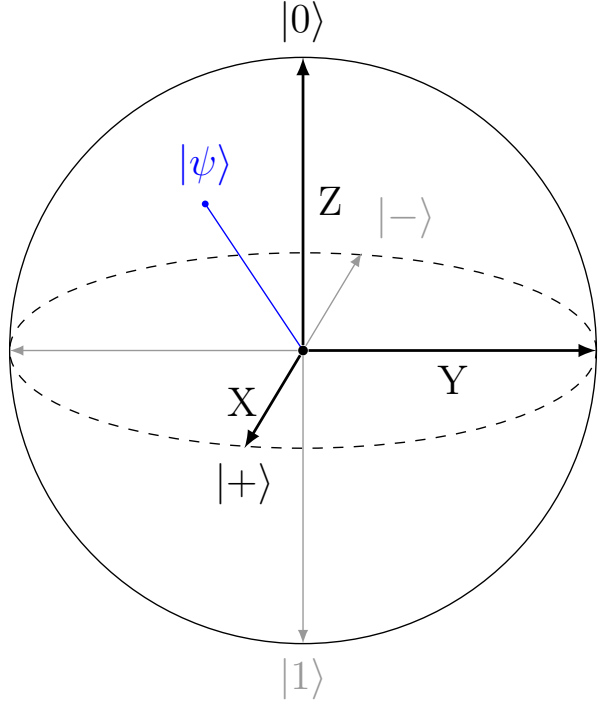

The X (resp. Y, Z) gate corresponds to a rotation of angle $\pi$ around the X (resp. Y, Z) axis of the Bloch sphere. This can be generalized by allowing rotations of arbitrary angles (in this thesis we will not need Y-rotations so they will not be introduced here). We decide to write X-rotations (resp. Z-rotations) as $X(\theta)$  (resp. $Z(\theta)$) for readability, in the literature the rotation can also be written $R_X(\theta)$ or $X^\theta$ (resp. $R_Z(\theta)$ or $Z^\theta$). The Z-rotation is sometimes called the phase shift gate. As their names imply, an X-rotation and Z-rotation of angle $\theta$ implement a rotation of angle $\theta$ around the X-axis and Z-axis, respectively, in the Bloch sphere. Below are the formal definitions of the X- and Z-rotations (for completeness we include multiple equivalent definitions).

$$ Z(\theta) \defeq 
\begin{pmatrix}
    1 & 0\\
    0 & e^{i\theta}
\end{pmatrix} 
= e^{i\frac \theta 2}\left(\cos\left(\frac \theta 2\right)I-i\sin\left(\frac \theta 2\right)Z\right) = 
\frac{1}{2} \left((1+e^{i\theta})I+(1-e^{i\theta})Z\right)
$$

\begin{equation*}
\begin{split}
X(\theta) %\defeq HZ(\alpha)H
& = e^{i\frac \theta 2} \begin{pmatrix}
    \cos\left(\frac \theta 2\right) & -i\sin\left(\frac \theta 2\right)\\
    -i\sin\left(\frac \theta 2\right) & \cos\left(\frac \theta 2\right)
\end{pmatrix} 
= e^{i\frac \theta 2}\left(\cos\left(\frac \theta 2\right)I-i\sin\left(\frac \theta 2\right)X\right)\\
& = \frac{1}{2} \begin{pmatrix}
    1 + e^{i\theta} & 1 - e^{i\theta}\\
    1 - e^{i\theta} & 1 + e^{i\theta}
\end{pmatrix}
= \frac{1}{2} \left((1+e^{i\theta})I+(1-e^{i\theta})X\right)\\
\end{split}
\end{equation*}

In particular, $X(0) = Z(0) = I$, $X(\pi)=X$ and $Z(\pi)=Z$. $Z(\pi/2)$ is often referred as $S$ in the literature. $Z(\pi/4)$ is often referred as $T$ in the literature. Some useful inequalities about rotations are the following: $Z X(\theta) = e^{i\theta} X(-\theta) Z$ and $X Z(\theta) = e^{i\theta} Z(-\theta) X$. The X-rotations and Z-rotations are universal for single-qubit unitary gates, in the sense that any single-qubit unitary gate can be written as a (finite) product of X- and Z-rotation:

\begin{proposition}[Euler's decomposition] \label{prop:euler}
    If $U$ is a single-qubit unitary gate, there exist 4 real number $\alpha, \beta, \gamma, \delta$ such that $U = e^{i\alpha} Z(\beta) X(\gamma) Z(\delta)$.
\end{proposition}

The Hadamard gate $H$ maps X-rotations to Z-rotations (and vice versa) by conjugation: $H X(\theta) H = Z(\theta)$ and $H Z(\theta) H = X(\theta)$. Following \cref{prop:euler}, $H$ can be decomposed into X- and Z-rotations of angle $\pi/2$.

\begin{proposition}
    \begin{equation*}
    \begin{split}
        H & = e^{-i\pi/4} X(\pi/2) Z(\pi/2) X(\pi/2) = e^{-i\pi/4} Z(\pi/2) X(\pi/2) Z(\pi/2)\\
        & = e^{i\pi/4} X(-\pi/2) Z(-\pi/2) X(-\pi/2) = e^{i\pi/4} Z(-\pi/2) X(-\pi/2) Z(-\pi/2)
    \end{split}
    \end{equation*}    
\end{proposition}

\subsection{Single-qubit Clifford gates}

\label{subsec:clifford}

The single-qubit Clifford group $\mathcal C$ is the group of unitaries that are the product of a sequence of Hadamard gates $H$ and Z-rotations $Z(\pi/2)$: we write $\mathcal C = \langle H, Z(\pi/2) \rangle$. $\mathcal C$ is often also defined as the group of unitaries that normalize the single-qubit Pauli group $\mathcal P$, i.e. $C \in \mathcal C$ if $\forall P \in \mathcal P$, $C P C^\dagger \in \mathcal P$. These two definitions are not exactly equivalent: under the first definition, $\mathcal C$  is finite (it contains exactly 192 elements), under the second definition, $\mathcal C$  is infinite. The two definitions are however equivalent up to global phase. In the literature, the Clifford group is sometimes defined up to global phase.

Up to global phase, there are exactly 24 single-qubit Clifford gates. Up to global phase and Pauli, there are exactly 6 single-qubit Clifford gates. We provide an exhaustive look-up table for single-qubit Clifford gates in Figure \ref{fig:clifford}. 

\begin{figure}[h!]
\bgroup
\setlength{\tabcolsep}{30pt}
\def\arraystretch{1.5}
\begin{center}
\begin{tabular}{|c|c|c|c|}
    \hline
    Decomposition of $C$    & $C ~ X ~ C^\dagger$ & $C ~ Y ~ C^\dagger$  & $C ~ Z ~ C^\dagger$    \\
    \hline
    $I$  & $X$  & $Y$  & $Z$  \\
    $X$  & $X$  & $-Y$  & $-Z$ \\
    $Y$ & $-X$  & $Y$  & $-Z$\\
    $Z$ & $-X$  & $-Y$  & $Z$ \\
    \hline
    $X(\pi/2)$  &  $X$ & $Z$  &  $-Y$ \\
    $X(-\pi/2)$  &  $X$ & $-Z$  &  $Y$ \\
    $Z X(\pi/2)$  &  $-X$ & $Z$  &  $Y$ \\
    $Z X(-\pi/2)$  &  $-X$ & $-Z$  &  $-Y$ \\
    \hline
    $Z(\pi/2)$  &  $Y$ & $-X$  &  $Z$ \\
    $Z(-\pi/2)$  &  $-Y$ & $X$  &  $Z$ \\
    $X Z(\pi/2)$  &  $-Y$ & $-X$  &  $-Z$ \\
    $X Z(-\pi/2)$  &  $Y$ & $X$  &  $-Z$ \\ 
    \hline
    $H$  &  $Z$ & $-Y$  &  $X$ \\
    $X H$  &  $-Z$ & $Y$  &  $X$ \\
    $Y H$  &  $-Z$ & $-Y$  &  $-X$ \\
    $Z H$  &  $Z$ & $Y$  &  $-X$ \\   
    \hline
    $H X(\pi/2)$  &  $Z$ & $X$  &  $Y$ \\ 
    $H X(- \pi/2)$  &  $Z$ & $-X$  &  $-Y$ \\
    $X H X(\pi/2)$  &  $-Z$ & $X$  &  $-Y$ \\
    $Y H X(\pi/2)$  &  $-Z$ & $-X$  &  $Y$ \\       
    \hline
    $H Z(\pi/2)$  &  $-Y$ & $-Z$  &  $X$ \\
    $H Z(-\pi/2)$  &  $Y$ & $Z$  &  $X$ \\
    $Y H Z(\pi/2)$  &  $-Y$ & $Z$  &  $-X$ \\
    $Z H Z(\pi/2)$  &  $Y$ & $-Z$  &  $-X$ \\ 
    \hline
\end{tabular}
\end{center}
\egroup
\caption{Each of the 24 single-qubit Clifford gates (up to global phase) and the action of their conjugation over the Pauli gates. 
}
\label{fig:clifford}
\end{figure}

\newpage

\section{Graph states}

Graph states are nothing but a family of quantum states. As their name suggests, graph states are in one-to-one correspondence with graphs. That is, a graph represents a graph state, and a graph state is represented by exactly one graph. Informally, vertices represent qubits and edges represent entanglement. In particular, in this thesis we will be using the notation $V$ for both the set of vertices of a graph, 
and for the set of qubits belonging to a quantum state, e.g. a graph state. Below we define graph states. 
Note that a great reference for learning about graph states is \cite{Hein06}, which contains the idea of many proofs presented in the following sections.

\subsection{Definition}

Starting from a graph $G=(V,E)$, there exists an easy procedure to construct the corresponding graph state. First, for each vertex $u \in V$, create a qubit in the state $\ket +$. Recall that a qubit can be prepared to be in the state $\ket +$, for instance, by preparing it in the basis state $\ket 0$ then applying a Hadamard $H$ gate on it, as $H \ket 0 = \ket +$. Then, for each edge i.e. each pair $u,v \in V$ of vertices such that $u \sim_G v$, apply a $CZ$ gate on the corresponding qubits. The order over which the $CZ$ gates are applied does not matter, as the $CZ$ gates commute and are symmetric (i.e. $CZ_{uv}=CZ_{vu}$). Below is a formal definition.

\begin{definition}\label{def:graph_state}
    Let $G = (V,E)$ be a graph of order $n$. The corresponding graph state $\ket G$ is the $n$-qubit state: $$\ket G = \left(\prod_{(u,v) \in E, u < v} CZ_{uv}\right) \ket{+}_V$$
\end{definition}

The condition $u < v$ in the product ensures that a $CZ$ gate is applied at most once per pair of vertices. The construction of a graph state is illustrated in Figure \ref{fig:graph_state_construction}. 

\begin{figure}[h!]
\centering
\scalebox{1}{
\begin{tikzpicture}[scale = 0.45]
\begin{scope}[every node/.style={circle,minimum size=15pt,thick,draw, fill=colorvertices}]

    \node (U1) at (0,4) {1};
    \node (U2) at (3,8) {2};
    \node (U3) at (6,4) {3};
    
\end{scope}

\begin{scope}[every node/.style={circle,minimum size=15pt,thick,draw, fill=colorvertices}]

    \node (U1) at (13,4) {1};
    \node (U2) at (16,8) {2};
    \node (U3) at (19,4) {3};
    
\end{scope}
\begin{scope}[every edge/.style={draw=darkgray,very thick}]
              
    \path [-] (U1) edge node {} (U2);

\end{scope}

\begin{scope}[every node/.style={circle,minimum size=15pt,thick,draw, fill=colorvertices}]

    \node (U1) at (26,4) {1};
    \node (U2) at (29,8) {2};
    \node (U3) at (32,4) {3};
    
\end{scope}
\begin{scope}[every edge/.style={draw=darkgray,very thick}]
              
    \path [-] (U1) edge node {} (U2);
    \path [-] (U2) edge node {} (U3);

\end{scope}

\draw [-stealth, very thick](8.5,6) -- (10.5,6);
\draw [-stealth, very thick](21.5,6) -- (23.5,6);

\end{tikzpicture}
}
\caption{Construction of a 3-qubit graph state. (Left) We start with a tensor product of 3 qubits in the state $\ket +$. That is, $\ket {G_1} = \ket{+++} = \frac{1}{2\sqrt 2}\left(\ket{000}+\ket{001}+\ket{010}+\ket{100}+\ket{011}+\ket{101}+\ket{110}+\ket{111}\right)$. (Middle) We create an edge between vertices 1 and 2 by applying a $CZ$ gate on the qubits 1 and 2, leading to the graph state $\ket {G_2} = CZ_{12} \ket{+++} = \frac{1}{2\sqrt 2}\left(\ket{000}+\ket{001}+\ket{010}+\ket{100}+\ket{011}+\ket{101}-\ket{110}-\ket{111}\right)$. (Right) We create an edge between vertices 2 and 3 by applying a $CZ$ gate on the qubits 2 and 3, leading to the graph state $\ket {G_3} = CZ_{23} CZ_{12} \ket{+++} = \frac{1}{2\sqrt 2}\left(\ket{000}+\ket{001}+\ket{010}+\ket{100}-\ket{011}+\ket{101}-\ket{110}+\ket{111}\right)$.}
\label{fig:graph_state_construction}
\end{figure}
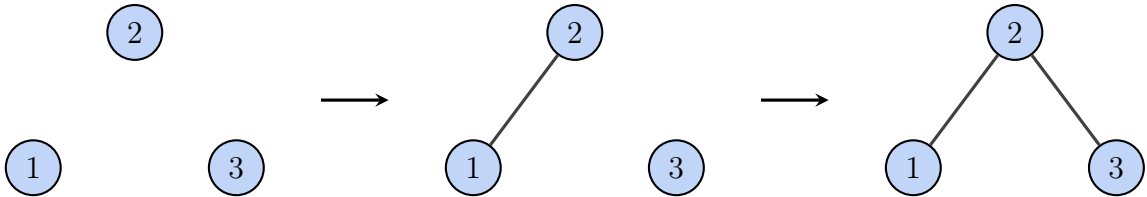

\begin{proposition} \label{prop:boolean_definition}
    Let $G=(V,E)$ be a graph of order $n$. The corresponding graph state $\ket G$ is the $n$-qubit state: $$\ket G = \frac 1{\sqrt {2^n}}\sum_{x\in \{0,1\}^n}(-1)^{|G[x]|}\ket x$$ where $|G[x]|$ denotes the number of edges in the subgraph of $G$ induced by $\{u_i~| ~x_i=1\}\subseteq V$.
\end{proposition}

The proof follows directly from \cref{def:graph_state}, as $\ket{+}_V = \frac 1{\sqrt {2^n}}\sum_{x\in \{0,1\}^n}\ket x$, and $CZ_{uv}$ toggles the sign before each basis state $\ket x$ such that $x_u = x_v = 1$.

\subsection{Graph states are stabilizer states}

An $n$-qubit stabilizer state is the unique fixpoint (i.e. the unique +1 eigenvector), up to global phase, of $n$ independent commuting Pauli operators of the form $\pm \bigotimes_{u\in V} P_u$ where $P_u \in \{I,X,Y,Z\}$. Put differently, an $n$-qubit stabilizer states is the unique fixpoint of $2^n$ distinct commuting Pauli operators, up to global phase. Note that the stabilizers of a stabilizer states (i.e. the Pauli operators which the stabilizer state is a fixpoint of) cannot contain $-I$, because it has no +1 eigenvectors. We say that a Pauli operator $P$ \emph{stabilizes} a state $\ket \psi$ when $P \ket \psi = \ket \psi$, i.e. $\ket \psi$ is a fixpoint of $P$. Graph states are a subfamily of stabilizer states, as they are stabilized by some local Pauli operators:

\begin{proposition}[\cite{Hein06}] \label{prop:stabilizer_of_graph_state}
    A graph state $\ket G$ is the unique quantum state (up to global phase) stabilized by the $n$ Pauli operators $X_u Z_{N_G(u)}$ for every vertex $u \in V$.
\end{proposition}

\begin{proof}[Proof Sketch]
    As $X_u CZ_{uv} = CZ_{uv} X_u Z_v$ and $Z_u CZ_{uv} = CZ_{uv} Z_u$, for any vertex $a \in V$, 

    \begin{align*}
        X_a \ket G  & = X_a \left(\prod_{(u,v) \in E, u < v} CZ_{uv}\right) \ket{+}_V = \left(\prod_{(u,v) \in E, u < v} CZ_{uv}\right) X_a Z_{N_G(a)} \ket{+}_V\\
        & = Z_{N_G(a)} \left(\prod_{(u,v) \in E, u < v} CZ_{uv}\right) \ket{+}_V = Z_{N_G(a)} \ket G
    \end{align*}

    The unicity follows from the fact that the set of common eigenstates of $X_u Z_{N_G(u)}$ is a basis for $(\mathbb C^2)^V$ (see \cite{Hein06}). More precisely, for a set of vertices $D \se V$, the quantum state $Z_D \ket G$ is the +1 eigenstate of every $X_u Z_{N_G(u)}$ when $u \notin D$, and the -1 eigenstate of every $X_u Z_{N_G(u)}$ when $u \in D$. 
\end{proof}

Below we characterize the complete set of stabilizers of a graph state.

\begin{proposition} \label{prop:stabilizer_definition}
    A graph state $\ket G$ is the unique quantum state (up to global phase) stabilized by the $2^n$ Pauli operators $(-1)^{|G[D]|}X_D Z_{Odd(D)}$ for any $D\subseteq V$, where $|G[D]|$ denotes the number of edges of the subgraph of $G$ induced by $D$.
\end{proposition}

\begin{proof}
    We show the proposition by induction on the size of $D$. When $D=\emptyset$ the property is trivially true. Since X and Z anticommute, we have for any sets $A$ and $B$, $X_AZ_B = (-1)^{|A\cap B|}Z_BX_A$. As a consequence, for any non-empty $D$ and for some $u\in D$, 
    \begin{equation*}
        \begin{split}
            (-1)^{|G[D]|}X_D Z_{Odd(D)}\ket G & = (-1)^{|G[D\setminus \{u\}]|}(-1)^{|N_G(u)\cap D\setminus \{u\}|}X_u X_{D\setminus \{u\}}Z_{N_G(u)}Z_{Odd(D\setminus \{u\})}\ket G\\
            & =(-1)^{|G[D\setminus \{u\}]|}X_uZ_{N_G(u)}X_{D\setminus \{u\}}Z_{Odd(D\setminus \{u\})}\ket G = X_uZ_{N_G(u)}\ket G \!=\! \ket G \qedhere                 
        \end{split}
    \end{equation*}    
\end{proof}

Any stabilizer state is related to at least one graph state by a tensor product of single-qubit Clifford operators. Thus, studying the entanglement of graph states is enough to study the entanglement of stabilizer states.

\begin{proposition}[\cite{VandenNest04}] \label{prop:stabilizer_LC}
    Any stabilizer state is related to at least one graph state by local Clifford operations.
\end{proposition}

\begin{proof}[Proof Sketch]
    A stabilizer state $\ket S$ is the fixpoint of $n$ independent commuting Pauli operators. Applying a single-qubit Clifford $C$ on $\ket S$ at qubit $a$ changes the stabilizers: if $(\pm \bigotimes_{u \in V} P_u) \ket S = \ket S$ then $\pm C_a P_a {C_a}^\dagger\bigotimes_{u \in V\sm\{a\}} P_u (C_a \ket S) = (C_a \ket S)$. As seen in the look-up table in Figure \ref{fig:clifford}, any permutation of the Pauli gates can be performed by the conjugation of a Clifford gate. Thus, with a proper choice of single-qubit Clifford operators, the stabilizers can be transformed to be those of a graph state. The proof in \cite{VandenNest04} is constructive and shows that the graph state can be computed in polynomial-time.
\end{proof}

\subsection{Basic properties of graph states}

Graph states are in one-to-one correspondence with (simple, undirected) graphs, hence we can represent a graph state in a non-ambiguous way by a graph, and vice versa.

\begin{proposition}
    $\ket{G_1} = \ket{G_2}$ if and only if $G_1 = G_2$.
\end{proposition}

The proof follows directly from \cref{prop:boolean_definition}.

\begin{proposition}\label{prop:meas}
    For any graph $G=(V,E)$ and any vertex $a \in V$,  $\bra 0_a \ket G = \frac1{\sqrt 2} \ket{G\setminus a}$. Moreover,  if $a$ is an isolated vertex (i.e. $N_G(a)=\emptyset$), then $\bra +_a\ket G = \ket{G\setminus a}$.
\end{proposition}

\begin{proof}
    ${\bra 0}_u CZ_{uv} = {\bra 0}_u I_v$. Thus, 
    \begin{align*}
        {\bra 0}_a \ket G & = {\bra 0}_a \left(\prod_{(u,v) \in E,~ u < v} CZ_{uv}\right) \ket{+}_V\\
        & = \left(\prod_{(u,v) \in E,~ u < v,~u,v \neq a} CZ_{uv}\right) {\bra 0}_a  \ket{+}_V\\
        & = \frac1{\sqrt 2} \left(\prod_{(u,v) \in E,~ u < v,~ u,v \neq a} CZ_{uv}\right) \ket{+}_{V \sm \{a\}} = \frac1{\sqrt 2} \ket{G\setminus a} \qedhere  
    \end{align*}
\end{proof}
    
\begin{remark}
A standard basis measurement consists in applying either $\bra 0$ or $\bra 1$ which correspond the classical outcomes $0$ and $1$ respectively. While the action in the $0$ case is directly described in \cref{prop:meas}, the action in the $1$ case can be recovered thanks to \cref{prop:stabilizer_definition}: $\bra 1_u\ket G = \bra 0_u X_u\ket G = \bra 0_u Z_{N_G(u)}\ket G =  \frac1{\sqrt 2}Z_{N_G(u)}\ket{G\setminus u}$. Thus, it also corresponds to a vertex deletion up to some Z corrections on the neighborhood of the measured qubit. 
\end{remark}

Global phase operators and Pauli gates are not sufficient to map one graph state to a different graph state.

\begin{proposition}
    If two graph states $\ket{G_1}$ and $\ket{G_2}$ are related by Paulis and a global phase, i.e. $\ket{G_2} = e^{i\phi}\bigotimes_{u\in V}P_u \ket{G_1}$ where $P_u \in \{I,X,Y,Z\}$, then $\ket{G_1} = \ket{G_2}$.
\end{proposition}

\begin{proof}
    Let $D\se V$ be the set of vertices over which X or Y is applied, that is, $D = \{u \in V ~|~ P_u \in \{X,Y\}\}$. According to \cref{prop:stabilizer_definition}, $\ket{G_1} = (-1)^{|G_1[D]|}X_D Z_{Odd_{G_1}(D)} \ket{G_1}$. Thus, $\ket{G_2} = e^{i\phi'}\bigotimes_{u\in V}P'_u \ket{G_1}$ where $P'_u \in \{I,Z\}$. In other words $\ket{G_2} = e^{i\phi'}Z_K \ket{G_1}$ for some set $K \se V$ of vertices. Notice that $\bra{0}_V \ket{G_2} = 1/\sqrt{2^n}$ and $\bra{0}_V e^{i\phi'}Z_K \ket{G_1} = e^{i\phi'}\bra{0}_V \ket{G_1}=  e^{i\phi'}/\sqrt{2^n}$. Thus, $e^{i\phi'} = 1$. Let us prove that $K = \emptyset$. By contradiction, suppose $K$ non-empty and let $u \in K$. According to \cref{prop:meas}, $\bra{0}_{V \sm \{u\}} \ket{G_2} = \ket +$ and $\bra{0}_{V \sm \{u\}} Z_K \ket{G_1} = Z \ket + = \ket -$.
\end{proof}

\begin{proposition}[\cite{Hein06}] \label{prop:marginal}
    The partial trace of a graph state is proportional to a sum of stabilizers of the graph state. More precisely, if $A$ is a set of vertices:

    \begin{align*}
        {Tr}_{V \sm A} \left(\ket G \bra G\right) & = \frac{1}{2^{|A|}}\sum_{D \se V ~|~ D \cup Odd_G(D) \se A} (-1)^{|G[D]|}X_D Z_{Odd_G(D)}%\\
    \end{align*}
\end{proposition}

\begin{proof}
    First, $\ket +_V \bra +_V = \frac{1}{2^n}\sum_{D\se V} X_D $, as $\ket 0 \bra 0 + \ket 1 \bra 1 = I$ and $\ket 0 \bra 1 + \ket 1 \bra 0 = X$. Then, as $X_u CZ_{uv} = CZ_{uv} X_u Z_v$ and $Z_u CZ_{uv} = CZ_{uv} Z_u$,
    \begin{align*}
        \ket G \bra G & = \left(\prod_{(u,v) \in E,~ u < v} CZ_{uv}\right)\ket +_V \bra +_V \left(\prod_{(u,v) \in E,~ u < v} CZ_{uv}\right)\\
        & = \frac{1}{2^n}\sum_{D\se V} \left(\prod_{(u,v) \in E,~ u < v} CZ_{uv}\right) X_D \left(\prod_{(u,v) \in E,~ u < v} CZ_{uv}\right) \\
        & = \frac{1}{2^n}\sum_{D\se V} (-1)^{|G[D]|}X_D Z_{Odd_G(D)}
    \end{align*}
    The proposition then follows from the fact that the trace of X, Y and Z is 0, while the trace of the identity is 2.
\end{proof}

\section{Different notions of local equivalence}

We care about graph states that have the same entanglement. A standard way of defining "having the same entanglement" is through SLOCC-equivalence. In the particular case of graph states, the notion of SLOCC-equivalence coincides with an arguably simpler notion, LU-equivalence.

\subsection{SLOCC-equivalence}

Two quantum states are said to be SLOCC-equivalent if they can be transformed into each other with stochastic local operations and classical communication (SLOCC), with some non-zero probability of success. "Local" here means applied on each qubit independently.

\begin{definition}
    Two quantum states $\ket{\psi_1}$ and $\ket{\psi_2}$ are SLOCC-equivalent if there exist determinant 1 gates $Q_u$ such that $\ket{\psi_2} = \bigotimes_{u\in V}Q_u \ket{\psi_1}$. 
\end{definition}

\subsection{LU-equivalence}

Two quantum states are said to be LU-equivalent if they can be transformed into each other with local unitary transformation (with a probability 1 of success).

\begin{definition}
    Two quantum states $\ket{\psi_1}$ and $\ket{\psi_2}$ are LU-equivalent if there exist single-qubit unitary operators $U_u$ such that $\ket{\psi_2} = \bigotimes_{u\in V}U_u \ket{\psi_1}$.
\end{definition}

For convenience, we lift the definition of LU-equivalence to graphs: two graphs are said LU-equivalent if the two corresponding graph states are LU-equivalent. We use the notation $G_1 =_{LU} G_2$.

\begin{proposition}[\cite{Verstraete2003}]
    If the local density operators of two $n$-qubit quantum states $\ket {\psi_1},\ket {\psi_2} $ are all proportional to the identity (i.e. $\forall u \in V, {Tr}_{V\sm\{u\}}\left(\ket {\psi_i} \bra {\psi_i}\right) \propto I$), then $\ket {\psi_1}$ and $\ket {\psi_2}$ are SLOCC-equivalent if and only if they are LU-equivalent.
\end{proposition}

Following \cite{Hein06}, if for some vertex $u$ of a graph $G$, ${Tr}_{V\sm\{u\}}\left(\ket G \bra G\right) \not\propto I$, then, according to \cref{prop:marginal}, $u$ is an isolated vertex. Thus, if $G_1$ and $G_2$ are connected graphs of order at least 2, $\ket {G_1}$ and $\ket{G_2}$ are SLOCC-equivalent if and only if they are LU-equivalent. This property can be extended easily to any pair of graphs by considering each connected component, and noticing that there is only one possible graph of order 1.

\begin{corollary}
    Two graph states are SLOCC-equivalent if and only if they are LU-equivalent.
\end{corollary}

Thus, LU-equivalence is the most natural framework to study graph states having the same entanglement. Actually, the local unitaries relating two LU-equivalent graph states are always Z-rotations up to Clifford.

\begin{proposition}[\cite{gross2007lu, zeng2011transversality}]\label{prop:czc}
    If $G_1 =_{LU} G'_2$ such that $\ket{G_2} = \bigotimes_{u\in V}U_u \ket{G_1}$, then there exist single-qubit Clifford operators  $C_u, C'_u$ and angles $\theta_u$ such that for every vertex $u \in V$, $U_u = C_u Z\left(\theta_u\right) C'_u$ up to global phase.    
\end{proposition}

\subsection{LC-equivalence}

Two quantum states are said to be LC-equivalent if they can be transformed into each other with local Clifford operations.

\begin{definition}
    Two quantum states $\ket{\psi_1}$ and $\ket{\psi_2}$ are LC-equivalent if there exist single-qubit Clifford operators $C_u$ such that $\ket{\psi_2} = e^{i\phi} \bigotimes_{u\in V}C_u \ket{\psi_1}$.
\end{definition}

For convenience, we lift the definition of LC-equivalence to graphs: two graphs are said LC-equivalent if the two corresponding graph states are LC-equivalent.
Note that the global phase $e^{i\phi}$ in the definition makes it so LC-equivalence is equivalently defined no matter which definition of the single-qubit Clifford group we choose (see \cref{subsec:clifford}). 

Obviously, two LC-equivalent graphs are LU-equivalent. Conversely, while it was once conjectured that two LU-equivalent graphs are always LC-equivalent \cite{5pb}, a counter-example of order 27 was discovered \cite{Ji07}. It is to this day still an open question whether this counter-example is minimal in order, or if there exists a counter-example of order 26 or less. 

The original 27-vertex counterexample consists of a pair of graphs that differ by only one edge (see \cite{Ji07}). Since, an equivalent counter-example has been found, with a much nicer form \cite{Tsimakuridze17}, which we provide in Figure \ref{fig:ce}. A formal explanation and a generalization of this counter-example can be found in \cref{chap:glc}. 

\begin{figure}[h!]
\centering
\scalebox{0.77}{
\begin{tikzpicture}
\begin{scope}[every node/.style={circle, fill=black, inner sep=3pt}]

    \node (u1) at (0,0) {};
    \node (u2) at (1,-0.5) {};
    \node (u3) at (2,-0.75) {};
    \node (u4) at (3,-0.75) {};
    \node (u5) at (4,-0.5) {};
    \node (u6) at (5,0) {};

    \node (v12345) at (-2,5) {};
    \node (v12346) at (-1.25,5) {};
    \node (v12356) at (-0.5,5) {};
    \node (v12456) at (0.25,5) {};
    \node (v13456) at (1,5) {};
    \node (v23456) at (1.75,5) {};

    \node (v1234) at (4,5) {};
    \node (v1235) at (4.75,5) {};
    \node (v2456) at (6.25,5) {};
    \node (v3456) at (7,5) {};    
    
\end{scope}
\begin{scope}[every edge/.style={draw=darkgray,very thick}]
              
    \path [-] (v12345) edge node {} (u1);
    \path [-] (v12345) edge node {} (u2);
    \path [-] (v12345) edge node {} (u3);
    \path [-] (v12345) edge node {} (u4);
    \path [-] (v12345) edge node {} (u5);

    \path [-] (v12346) edge node {} (u1);
    \path [-] (v12346) edge node {} (u2);
    \path [-] (v12346) edge node {} (u3);
    \path [-] (v12346) edge node {} (u4);
    \path [-] (v12346) edge node {} (u6);

    \path [-] (v12356) edge node {} (u1);
    \path [-] (v12356) edge node {} (u2);
    \path [-] (v12356) edge node {} (u3);
    \path [-] (v12356) edge node {} (u5);
    \path [-] (v12356) edge node {} (u6);

    \path [-] (v12456) edge node {} (u1);
    \path [-] (v12456) edge node {} (u2);
    \path [-] (v12456) edge node {} (u4);
    \path [-] (v12456) edge node {} (u5);
    \path [-] (v12456) edge node {} (u6);

    \path [-] (v13456) edge node {} (u1);
    \path [-] (v13456) edge node {} (u3);
    \path [-] (v13456) edge node {} (u4);
    \path [-] (v13456) edge node {} (u5);
    \path [-] (v13456) edge node {} (u6);

    \path [-] (v23456) edge node {} (u2);
    \path [-] (v23456) edge node {} (u3);
    \path [-] (v23456) edge node {} (u4);
    \path [-] (v23456) edge node {} (u5);
    \path [-] (v23456) edge node {} (u6);

    \path [-] (v1234) edge node {} (u1);
    \path [-] (v1234) edge node {} (u2);
    \path [-] (v1234) edge node {} (u3);
    \path [-] (v1234) edge node {} (u4);

    \path [-] (v1235) edge node {} (u1);
    \path [-] (v1235) edge node {} (u2);
    \path [-] (v1235) edge node {} (u3);
    \path [-] (v1235) edge node {} (u5);

    \path [-] (v2456) edge node {} (u2);
    \path [-] (v2456) edge node {} (u4);
    \path [-] (v2456) edge node {} (u5);
    \path [-] (v2456) edge node {} (u6);

    \path [-] (v3456) edge node {} (u3);
    \path [-] (v3456) edge node {} (u4);
    \path [-] (v3456) edge node {} (u5);
    \path [-] (v3456) edge node {} (u6);

\end{scope}
\draw (5.5,5) node(){\large\dots};
\draw [stealth-stealth,draw=black,very thick ](8,2.5) -- (9,2.5);
\begin{scope}[shift={(12,0)}]
\begin{scope}[every node/.style={circle, fill=black, inner sep=3pt}]

    \node (u1) at (0,0) {};
    \node (u2) at (1,-0.5) {};
    \node (u3) at (2,-0.75) {};
    \node (u4) at (3,-0.75) {};
    \node (u5) at (4,-0.5) {};
    \node (u6) at (5,0) {};

    \node (v12345) at (-2,5) {};
    \node (v12346) at (-1.25,5) {};
    \node (v12356) at (-0.5,5) {};
    \node (v12456) at (0.25,5) {};
    \node (v13456) at (1,5) {};
    \node (v23456) at (1.75,5) {};

    \node (v1234) at (4,5) {};
    \node (v1235) at (4.75,5) {};
    \node (v2456) at (6.25,5) {};
    \node (v3456) at (7,5) {};    
    
\end{scope}
\begin{scope}[every edge/.style={draw=darkgray,very thick}]
              
    \path [-] (v12345) edge node {} (u1);
    \path [-] (v12345) edge node {} (u2);
    \path [-] (v12345) edge node {} (u3);
    \path [-] (v12345) edge node {} (u4);
    \path [-] (v12345) edge node {} (u5);

    \path [-] (v12346) edge node {} (u1);
    \path [-] (v12346) edge node {} (u2);
    \path [-] (v12346) edge node {} (u3);
    \path [-] (v12346) edge node {} (u4);
    \path [-] (v12346) edge node {} (u6);

    \path [-] (v12356) edge node {} (u1);
    \path [-] (v12356) edge node {} (u2);
    \path [-] (v12356) edge node {} (u3);
    \path [-] (v12356) edge node {} (u5);
    \path [-] (v12356) edge node {} (u6);

    \path [-] (v12456) edge node {} (u1);
    \path [-] (v12456) edge node {} (u2);
    \path [-] (v12456) edge node {} (u4);
    \path [-] (v12456) edge node {} (u5);
    \path [-] (v12456) edge node {} (u6);

    \path [-] (v13456) edge node {} (u1);
    \path [-] (v13456) edge node {} (u3);
    \path [-] (v13456) edge node {} (u4);
    \path [-] (v13456) edge node {} (u5);
    \path [-] (v13456) edge node {} (u6);

    \path [-] (v23456) edge node {} (u2);
    \path [-] (v23456) edge node {} (u3);
    \path [-] (v23456) edge node {} (u4);
    \path [-] (v23456) edge node {} (u5);
    \path [-] (v23456) edge node {} (u6);

    \path [-] (v1234) edge node {} (u1);
    \path [-] (v1234) edge node {} (u2);
    \path [-] (v1234) edge node {} (u3);
    \path [-] (v1234) edge node {} (u4);

    \path [-] (v1235) edge node {} (u1);
    \path [-] (v1235) edge node {} (u2);
    \path [-] (v1235) edge node {} (u3);
    \path [-] (v1235) edge node {} (u5);

    \path [-] (v2456) edge node {} (u2);
    \path [-] (v2456) edge node {} (u4);
    \path [-] (v2456) edge node {} (u5);
    \path [-] (v2456) edge node {} (u6);

    \path [-] (v3456) edge node {} (u3);
    \path [-] (v3456) edge node {} (u4);
    \path [-] (v3456) edge node {} (u5);
    \path [-] (v3456) edge node {} (u6);

    \path [-] (u1) edge node {} (u2);
    \path [-] (u1) edge node {} (u3);
    \path [-] (u1) edge node {} (u4);
    \path [-] (u1) edge node {} (u5);
    \path [-] (u1) edge node {} (u6);
    \path [-] (u2) edge node {} (u3);
    \path [-] (u2) edge node {} (u4);
    \path [-] (u2) edge node {} (u5);
    \path [-] (u2) edge node {} (u6);
    \path [-] (u3) edge node {} (u4);
    \path [-] (u3) edge node {} (u5);
    \path [-] (u3) edge node {} (u6);
    \path [-] (u4) edge node {} (u5);
    \path [-] (u4) edge node {} (u6);
    \path [-] (u5) edge node {} (u6);

\end{scope}
\draw (5.5,5) node(){\large\dots};
\end{scope}

\end{tikzpicture}
}
\caption{A 27-vertex counter-example to the LU=LC conjecture, that is, a pair of graphs that are LU-equivalent but not LC-equivalent. The graphs have 6 bottom vertices. There is one upper vertex per set of 5 bottom vertices, and one upper vertex per set of 4 bottom vertices; leading to $\binom{6}{5} + \binom{6}{4} = 21$ upper vertices. In the leftmost graph, the bottom vertices form an independent set, while in the rightmost graph, the bottom vertices are fully connected. Applying $X(\pi/4)$ on the upper vertices and $Z(\pi/4)$ on the bottom vertices maps one graph state to the other. Proving that those two graphs are not LC-equivalent is more involved, a proof can be found in \cite{Tsimakuridze17}.}
\label{fig:ce}
\end{figure}
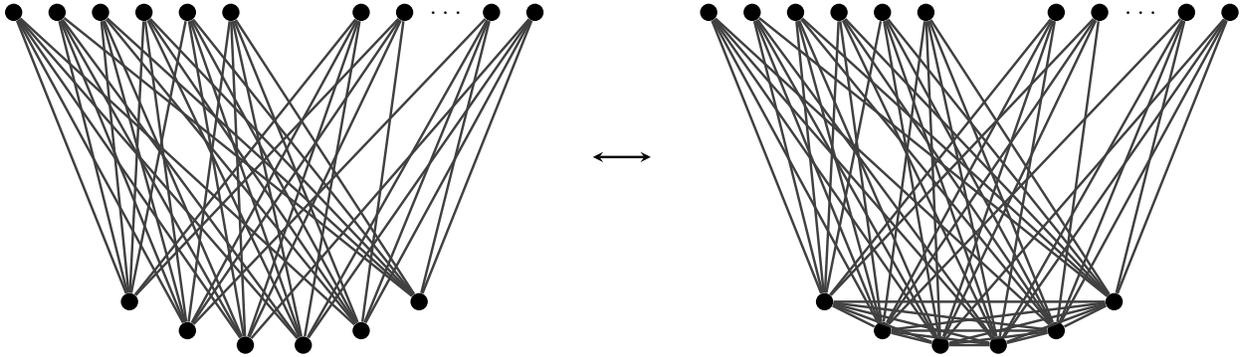

\section{Local complementation}

Local complementation is a graphical operation that exactly captures the LC-equivalence of graph states.

\subsection{Definition and properties}

Local complementation on a vertex $u$ of a graph consists in complementing the subgraph induced by the neighborhood of $u$. Formally, a local complementation on $u$ maps the graph $G$ to the graph $G\star u= G\Delta K_{N_G(u)}$ where $\Delta$ denotes the symmetric difference on edges and $K_A$ is the complete graph on the vertices of $A \se V$. Local complementation is illustrated in Figure \ref{fig:local_complementation}.

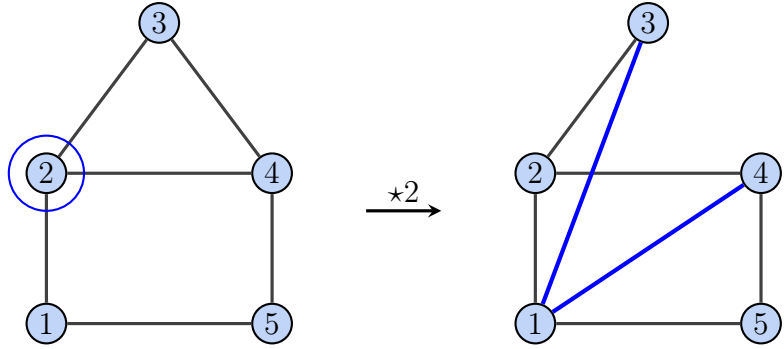
\begin{figure}[h!]
\centering
\scalebox{1}{
\begin{tikzpicture}[scale = 0.5]
\begin{scope}[every node/.style=vertices]

    \node (U1) at (0,0) {1};
    \node (U2) at (0,4) {2};
    \node (U3) at (3,8) {3};
    \node (U4) at (6,4) {4};
    \node (U5) at (6,0) {5};
    
\end{scope}
\begin{scope}[every edge/.style=edges]
              
    \path [-] (U1) edge node {} (U5);
    \path [-] (U1) edge node {} (U2);
    \path [-] (U2) edge node {} (U3);
    \path [-] (U2) edge node {} (U4);
    \path [-] (U3) edge node {} (U4);
    \path [-] (U4) edge node {} (U5);

\end{scope}

\draw[blue, thick] (0,4) circle (1);

\draw [-stealth, label=above:a, very thick](8.5,3) -- (10.5,3);
\draw (9.5,3.5) node(){$\star 2$};

\begin{scope}[every node/.style=vertices]

    \node (U1) at (13,0) {1};
    \node (U2) at (13,4) {2};
    \node (U3) at (16,8) {3};
    \node (U4) at (19,4) {4};
    \node (U5) at (19,0) {5};
    
\end{scope}
\begin{scope}[every edge/.style=edges]
              
    \path [-] (U1) edge node {} (U5);
    \path [-] (U1) edge node {} (U2);
    \path [-] (U2) edge node {} (U3);
    \path [-] (U2) edge node {} (U4);
    \path [-] (U4) edge node {} (U5);    

\end{scope}

\begin{scope}[every edge/.style={draw=blue,ultra thick}]

    \path [-] (U1) edge node {} (U3);
    \path [-] (U1) edge node {} (U4);

\end{scope}

\end{tikzpicture}
}
\caption{Example of a local complementation. Vertex 2 is related to vertices 1,3 and 4. Thus, the local complementation on vertex 2 toggles each edge between vertices 1,3 and 4. That is, the edge between vertices 3 and 4 is removed; while vertices 1 and 3 (resp. 1 and 4) share no edge, thus an edge is created. Formally, the graph $G$ is mapped to $G \star 2 = G\Delta K_{N_G(2)} = G\Delta K_{1,3,4}$.}
\label{fig:local_complementation}
\end{figure}

Local complementation was formalized and thoroughly studied by André Bouchet in the 80s \cite{Bouchet1987,Bouchet1987digraph,Bouchet1987reducing,Bouchet1988,Bouchet1988TransformingTB,Bouchet1989,Bouchet1991,Bouchet1993,Bouchet1994,Bouchet2001}. According to Bouchet, local complementation was prior introduced in the 60s by Anton Kotzig \cite{kotzig1968eulerian,kotzig1977quelques}.

Below are three basic properties of local complementation that follow from the definition.

\begin{proposition}
    Local complementation is self-inverse i.e. $G \star u \star u = G$.
\end{proposition}

\begin{proposition}\label{prop:lc_neighborhood}
    Local complementation on a vertex $u$ preserves the neighborhood of $u$ i.e. $N_{G\star u}(u) = N_{G}(u)$. More generally, if $u \not\sim_{G} v$, then $N_{G\star v}(u) = N_{G}(v)$. If $u \sim_{G} v$, $N_{G\star u}(v) = N_G(u) \Delta N_G(v) \Delta \{v\}$.
\end{proposition}

\begin{proposition}
    Two local complementations on two unadjacent vertices commute i.e. if $u \not\sim_{G} v$ then $G \star u \star v = G \star v \star u$.
\end{proposition}

\subsection{Pivoting}

\begin{definition}
    A pivoting on an edge $(u,v)$ maps a graph $G$ to $G\wedge uv \defeq G \star u \star v \star u = G \star v \star u \star v$.
\end{definition}

An illustration of the action of pivoting is provided in Figure \ref{fig:pivoting}.

\begin{figure}[h!]
\centering
\scalebox{1}{
\begin{tikzpicture}
\begin{scope}[every node/.style={circle,minimum size=10pt,thick,draw, fill=colorvertices}]

    \node (u1) at (1,2) {$u$};
    \node (v1) at (3,2) {$v$};

    \node (u2) at (9,2) {$v$};
    \node (v2) at (11,2) {$u$};
    
\end{scope}
\begin{scope}[every node/.style={circle,minimum size=25pt,thick,draw, fill=colorvertices}]
    \node (A1) at (0,0) {$A$};
    \node (B1) at (2,0) {$B$};
    \node (C1) at (4,0) {$C$};

    \node (A2) at (8,0) {$A$};
    \node (B2) at (10,0) {$B$};
    \node (C2) at (12,0) {$C$};
\end{scope}
\begin{scope}[every edge/.style={draw=darkgray,very thick}]

    \path [-] (u1) edge node {} (v1);
    \path [-] (u2) edge node {} (v2);
              
    \path [-] (u1) edge node {} (A1);
    \path [-] (u1) edge node {} (B1);
    \path [-] (v1) edge node {} (B1);
    \path [-] (v1) edge node {} (C1);   

    \path [-] (u2) edge node {} (A2);
    \path [-] (u2) edge node {} (B2);
    \path [-] (v2) edge node {} (B2);
    \path [-] (v2) edge node {} (C2);  

    \path [-] (A2) edge node {} (B2);
    \path [-] (B2) edge node {} (C2);     
    \path [-] (A2) edge[bend right=30] node {} (C2);

\end{scope}
\begin{scope}[style={draw=black,very thick}]

    \draw [-stealth](5.5,1) -- (6.5,1);
    \draw (6,1.3) node(){$\wedge_{uv}$};

\end{scope}
\end{tikzpicture}
}
\caption{Illustration of a pivoting on an edge $(u,v)$ in a graph $G$. $A$ denotes $N_G(u)\setminus \left( \{v\} \cup N_G(v) \right)$, $B$ denotes $N_G(u)\cap N_G(v)$, and $C$ denotes $N_G(v)\setminus  \left( \{u\} \cup N_G(u) \right)$. The vertices adjacent neither to $u$ nor $v$ are not represented, as the pivoting does not modify their neighborhood. To simplify the drawing, we assume that the vertices of the sets $A$, $B$ and $C$ are not adjacent. In general, the pivoting performs a symmetric difference with the complete bipartite graph whose bipartition is $(A,B)$ (resp. $(A,C)$, $(B,C)$). Also, the pivoting swaps the neighborhoods of $u$ and $v$.}
\label{fig:pivoting}
\end{figure}
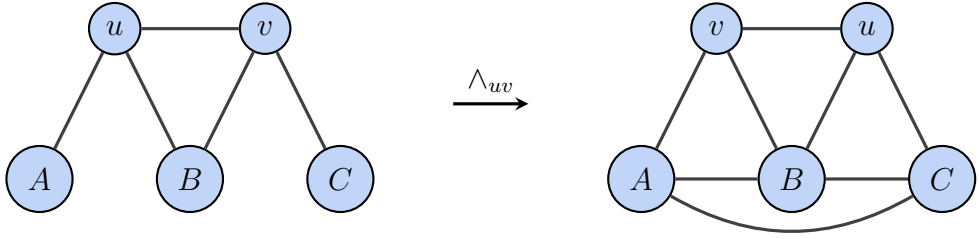

\begin{proposition} \label{prop:pivoting_preserves_bipartite}
    If $G$ is bipartite with respect to a bipartition $V = L \cup R$ of the vertices, then, given $u \in L$ and $v \in R$ such that $u \sim_G v$, $G \wedge uv$ is bipartite with respect to the bipartition $L \Delta \{u,v\}$, $R \Delta \{u,v\}$. More precisely, $G \wedge uv$ is obtained from $G$ by toggling all edges between $N_G(u)\sm \{v\}$ and $N_G(v)\sm \{u\}$, followed by swapping $u$ and $v$.
\end{proposition}

An illustration of the action of pivoting in bipartite graphs is provided in Figure \ref{fig:pivoting_bipartite}.

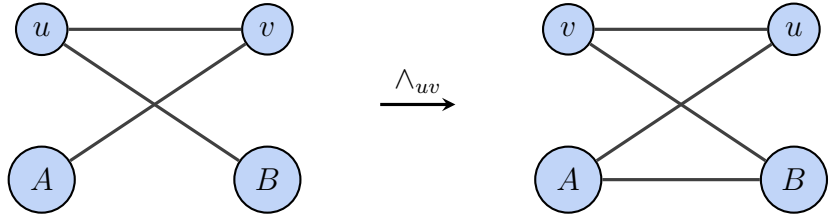
\begin{figure}[h!]
\centering
\scalebox{1}{
\begin{tikzpicture}
\begin{scope}[every node/.style={circle,minimum size=10pt,thick,draw, fill=colorvertices}]

    \node (u1) at (1,2) {$u$};
    \node (v1) at (4,2) {$v$};

    \node (u2) at (8,2) {$v$};
    \node (v2) at (11,2) {$u$};
    
\end{scope}
\begin{scope}[every node/.style={circle,minimum size=25pt,thick,draw, fill=colorvertices}]
    \node (A1) at (1,0) {$A$};
    \node (B1) at (4,0) {$B$};

    \node (A2) at (8,0) {$A$};
    \node (B2) at (11,0) {$B$};
\end{scope}
\begin{scope}[every edge/.style={draw=darkgray,very thick}]
              
    \path [-] (u1) edge node {} (v1);
    \path [-] (u2) edge node {} (v2);

    \path [-] (u1) edge node {} (B1);
    \path [-] (v1) edge node {} (A1); 

    \path [-] (u2) edge node {} (B2);
    \path [-] (v2) edge node {} (A2);
    \path [-] (A2) edge node {} (B2);
\end{scope}
\begin{scope}[style={draw=black,very thick}]

    \draw [-stealth](5.5,1) -- (6.5,1);
    \draw (6,1.3) node(){$\wedge_{uv}$};

\end{scope}
\end{tikzpicture}
}
\caption{Illustration of a pivoting on an edge $(u,v)$ in a bipartite graph $G$. $A$ denotes $N_G(u)\sm\{v\}$ and  $B$ denotes $N_G(v)\sm\{u\}$. The vertices adjacent neither to $u$ nor $v$ are not represented, as the pivoting does not modify their neighborhood. To simplify the drawing, we assume that the vertices of the sets $A$ and $B$ are not adjacent. In general, the pivoting performs a symmetric difference with the complete bipartite graph whose bipartition is $(A,B)$. Also, the pivoting swaps the neighborhoods of $u$ and $v$.}
\label{fig:pivoting_bipartite}
\end{figure}

\subsection{Bouchet's algorithm}\label{subsec:bouchet}

Equivalence by local complementation can be efficiently decided thanks to an algorithm found by André Bouchet.

\begin{theorem}[\cite{Bouchet1991}]
    There exists an algorithm that decides if two graphs are related by local complementations with runtime $O(n^{4})$, where $n$ is the order of the graphs.
\end{theorem}

In this section we explain the algorithm. 
First, Bouchet proved that the existence of a sequence of local complementations relating two graphs defined on the same vertex set $V$, reduces to the existence of subsets of vertices satisfying some equations.

\begin{proposition}[\cite{Bouchet1991}] \label{prop:Bouchet}
Two graphs $G=(V,E)$, $G'=(V,E')$ are related by local complementations if and only if there exist $A,B,C,D\subseteq V$ such that
\begin{itemize}
\item[(i)] $\forall u,v \in V$,\\
$|B\cap N_G(u)\cap N_{G'}(v)| + |A\cap N_G(u)\cap \{v\}| +  |D\cap \{u\}\cap N_{G'}(v)| + |C\cap \{u\}\cap \{v\}| = 0\bmod 2$
\item[(ii)]
$(A\cap D)~\Delta~ (B\cap C) = V$ 
\end{itemize}

\end{proposition}

The original proof involves isotropic systems~\cite{Bouchet1987}. Below is an alternative, self-contained proof, we have introduced in \citepub{claudet2025deciding}.

\begin{proof}
We begin by proving by induction the "only if" part of the statement. First, notice that equations $(i)$ and $(ii)$ are satisfied when $G = G'$ with $A = D = V$ and $B = C =  \emptyset$. Indeed, let $u,v \in V$:
\begin{align*}
    & |B\cap N_G(u)\cap N_{G}(v)| + |A\cap N_G(u)\cap \{v\}| +  |D\cap \{u\}\cap N_{G}(v)| + |C\cap \{u\}\cap \{v\}|\\
    &= |N_G(u)\cap \{v\}| +  |\{u\}\cap N_{G}(v)|\\
    &= 0 \bmod 2
\end{align*}
Furthermore, $A\cap D~\Delta~ B\cap C = V$.

Now, suppose that $G$ and $G'$ are related by local complementations and there exist $A,B,C,D\subseteq V$  satisfying $(i)$ and $(ii)$. Applying a local complementation on some vertex $a$ in $G$ results in the graph $G \star a$, which is also related by local complementations to $G'$. Define 
\begin{itemize}
    \item $A' = A \Delta ( \{a\}\cap C)$
    \item $B' = B \Delta ( \{a\}\cap D)$
    \item $C' = C \Delta ( N_{G}(a)\cap A)$
    \item $D' = D \Delta ( N_{G}(a)\cap B)$
\end{itemize}
Let us show that $A',B',C',D'$  satisfy $(i)$ and $(ii)$ for the graphs $G \star a$ and $G'$.

\subparagraph{Proof that $A',B',C',D'$  satisfy $(i)$.}
Let $u,v \in V$.
\begin{align*}
    & |B'\cap N_{G \star a}(u)\cap N_{G'}(v)| + |A'\cap N_{G \star a}(u)\cap \{v\}| +  |D'\cap \{u\}\cap N_{G'}(v)| + |C'\cap \{u\}\cap \{v\}|\\
    &=|(B \Delta ( \{a\}\cap D))\cap N_{G \star a}(u)\cap N_{G'}(v)| + |(A \Delta ( \{a\}\cap C))\cap N_{G \star a}(u)\cap \{v\}|\\
    &~~+ |(D \Delta ( N_{G}(a)\cap B))\cap \{u\}\cap N_{G'}(v)| + |(C \Delta ( N_{G}(a)\cap A))\cap \{u\}\cap \{v\}|
\end{align*}
If $u\not\sim_{G} a$, then $N_{G \star a}(u) = N_{G}(u)$ and $N_G(a)\cap \{u\} = N_G(u)\cap \{a\} = \emptyset$:
\begin{align*}
    &= |B\cap N_G(u)\cap N_{G'}(v)| + |A\cap N_G(u)\cap \{v\}| +  |D\cap \{u\}\cap N_{G'}(v)| + |C\cap \{u\}\cap \{v\}|\\
    &~~+ | \{a\}\cap D\cap N_G(u)\cap N_{G'}(v)| + |\{a\}\cap C\cap N_G(u)\cap \{v\}|\\
    &~~+ |N_{G}(a)\cap B\cap \{u\}\cap N_{G'}(v)| + |N_{G}(a)\cap A \cap \{u\}\cap \{v\}| \bmod 2\\
    &= 0 \bmod 2
\end{align*}
If $u\sim_{G} a$, then $N_{G \star a}(u) = N_{G}(u) \Delta N_{G}(a) \Delta \{u\}$, $N_G(a)\cap \{u\} = \{u\}$ and $N_G(u)\cap \{a\} = \{a\}$, thus:
\begin{align*}
    &=|(B \Delta ( \{a\}\cap D))\cap (N_{G}(u) \Delta N_{G}(a) \Delta \{u\})\cap N_{G'}(v)|\\
    &~~+ |(A \Delta ( \{a\}\cap C))\cap (N_{G}(u) \Delta N_{G}(a) \Delta \{u\})\cap \{v\}|\\
    &~~+ |(D \Delta ( N_{G}(a)\cap B))\cap \{u\}\cap N_{G'}(v)| + |(C \Delta ( N_{G}(a)\cap A))\cap \{u\}\cap \{v\}|\\
    &= |B\cap N_G(u)\cap N_{G'}(v)| + |A\cap N_G(u)\cap \{v\}| +  |D\cap \{u\}\cap N_{G'}(v)| + |C\cap \{u\}\cap \{v\}|\\
    &~~+|B\cap N_G(a)\cap N_{G'}(v)| + |A\cap N_G(a)\cap \{v\}| +  |D\cap \{a\}\cap N_{G'}(v)| + |C\cap \{a\}\cap \{v\}|\\
    &~~+|B\cap\{u\}\cap N_{G'}(v)|+|A\cap\{u\}\cap\{v\}|+|B\cap\{u\}\cap N_{G'}(v)|+|A\cap\{u\}\cap \{v\}|\\
    &~~+ |\{a\} \cap D \cap N_G(a) \cap N_{G'}(v)| + |\{a\} \cap C \cap N_G(a) \cap \{v\}|\\
    &~~+ |\{a\} \cap D \cap \{u\} \cap N_{G'}(v)| + |\{a\} \cap C \cap \{u\}\cap \{v\}|\\
    &= 0 \bmod 2
\end{align*}

\subparagraph{Proof that $A',B',C',D'$  satisfy $(ii)$.}
\begin{align*}
    & (A'\cap D')\Delta(B'\cap C')\\
    &= \left((A \Delta ( \{a\}\cap C))\cap (D \Delta ( N_{G}(a)\cap B))\right)\Delta\left((B \Delta ( \{a\}\cap D))\cap (C \Delta ( N_{G}(a)\cap A))\right)\\
    &= (A\cap D)\Delta(A\cap N_{G}(a)\cap B)\Delta(D\cap\! \{a\}\!\cap C)\Delta(B\cap C)\Delta(B \cap  N_{G}(a)\cap A)\Delta(C \!\cap\{a\}\!\cap D)\\
    &=(A\cap D)\Delta(B\cap C) = V
\end{align*}

Now we prove the "if" part of the statement. Let $G$ and $G'$ be two graphs defined on the same vertex set $V$ along with $A, B, C, D$ satisfying $(i)$ and $(ii)$. Condition $(ii)$ implies that for some vertex $u \in V$, 6 cases can occur:
\begin{enumerate}
    \item $u \in A \cap \overline B \cap \overline C \cap D$ %\Ncom{$I$}
    \item $u \in A \cap B \cap \overline C \cap D$ %\Ncom{$X(\pi/2)$}
    \item $u \in A \cap \overline B \cap C \cap D$ %\Ncom{$Z(\pi/2)$}
    \item $u \in \overline A \cap B \cap C \cap \overline D$ %\Ncom{$H$}
    \item $u \in A \cap B \cap C \cap \overline D$ %\Ncom{$X(\pi/2) H = H Z(\pi/2) = Z(\pi/2)X(\pi/2)$}
    \item $u \in \overline A \cap B \cap C \cap D$ %\Ncom{$Z(\pi/2) H = H X(\pi/2) = X(\pi/2)Z(\pi/2)$}
\end{enumerate}
We call $V_1$ (resp. $V_2$, $V_3$, $V_4$, $V_5$, $V_6$) the set of vertices in case 1 (resp. 2, 3, 4, 5, 6). Notice that $V = V_1$ implies $G = G'$, indeed condition $(ii)$ implies that for any $u,v \in V$, $|N_G(u)\cap \{v\}| +  |\{u\}\cap N_{G'}(v)| = 0 \bmod 2$ i.e.~$u \sim_G v \Leftrightarrow u \sim_{G'} v$. Furthermore, applying a local complementation on a vertex $a$ of $G$ changes the sets $A,B,C,D$, thus it changes in which case a vertex is. The changes are given in the following table (the case in which unwritten vertices remain unchanged). 

\begin{center}
    \begin{tabular}{|c|c|}
    \hline
    \multicolumn{2}{|c|} {Case of $a$ in}\\
    $~~G~~$& $G\star a$\\
    \hline
    1&2\\
    \hline
    2&1\\
    \hline
    3&6\\
    \hline
    4&5\\
    \hline
    5&4\\
    \hline
    6&3\\
    \hline
    \end{tabular}\qquad\begin{tabular}{|c|c|}
    \hline
    \multicolumn{2}{|c|}{Case of $u {\in}N_G(a)$ in}\\
    $~~~G~~~$& $G\star a$\\
    \hline
    1&3\\
    \hline
    2&5\\
    \hline
    3&1\\
    \hline
    4&6\\
    \hline
    5&2\\
    \hline
    6&4\\
    \hline
    \end{tabular}
    \qquad \begin{tabular}{|c|c|}
    \hline
    \multicolumn{2}{|c|}{Case of $a$ (or $b$) in}\\
    $~~G~~$&$G\wedge a b$\\
    \hline
    1&4\\
    \hline
    2&6\\
    \hline
    3&5\\
    \hline
    4&1\\
    \hline
    5&3\\
    \hline
    6&2\\
    \hline
    \end{tabular}
\end{center}

The table indicates that if $G$ and $G'$ are related by local complementations and $A, B, C, D$ satisfy $(i)$ and $(ii)$, then, for $G \wedge a b$ and $G'$, $A', B', C', D'$ satisfy $(i)$ and $(ii)$ where:
\begin{itemize}
    \item $A' = (A \sm \{a,b\}) \cup (C \cap \{a,b\})$
    \item $B' = (B \sm \{a,b\}) \cup (D \cap \{a,b\})$
    \item $C' = (C \sm \{a,b\}) \cup (A \cap \{a,b\})$
    \item $D' = (D \sm \{a,b\}) \cup (B \cap \{a,b\})$
\end{itemize}

Let us design an algorithm that produces a sequence of (possibly repeating) vertices $s = (a_1, \cdots, a_m)$ such that $G' = G \star a_1 \star \cdots \star a_m$. Initialize $G_0 = G$, $s_0 = [~]$ an empty sequence of vertices and $A_0 = A$, $B_0 = B$, $C_0 = C$, $D_0 = D$.

\begin{enumerate}
    \item If there is a vertex $u$ in case 2 or 6: let $s_0 \leftarrow s_0 + [u]$, $G_0 \leftarrow G_0 \star u$, $A_0 \leftarrow A_0 \Delta ( \{u\}\cap C_0)$, $B_0 \leftarrow B_0 \Delta ( \{u\}\cap D_0)$, $C_0 \leftarrow C_0 \Delta ( N_{G_0}(u)\cap A_0)$, $D_0 \leftarrow D_0 \Delta ( N_{G_0}(u)\cap B_0)$. Repeat until there is no vertex in case 2 or 6 left. 
    \item If there is a vertex $u$ in case 4 or 5: let $v \in N_{G_0}(u)$ such that $v$ is also in case 4 or 5. 
    Let $s_0 \leftarrow s_0 + [u, v, u]$, $G_0 \leftarrow G_0 \wedge u v$, $A_0 \leftarrow (A_0 \sm \{u,v\}) \cup (C_0 \cap \{u,v\})$, $B_0 \leftarrow (B_0 \sm \{u,v\}) \cup (D_0 \cap \{u,v\})$, $C_0 \leftarrow (C_0 \sm \{u,v\}) \cup (A_0 \cap \{u,v\})$, $D_0 \leftarrow (D_0 \sm \{u,v\}) \cup (B_0 \cap \{u,v\})$. Then go to step 1.
\end{enumerate}

\subparagraph{Correctness.} The evolution of $A_0$, $B_0$, $C_0$ and $D_0$ at each iteration of the algorithm ensures that $(i)$ and $(ii)$ are satisfied for $G \star s_0$~\footnote{If $w$ is a list of vertices, say $w = w_0, w_1, \cdots, w_k$, we define $G \star w = G \star w_0 \star w_1 \star \cdots \star w_k$.} and $G'$. In step 2, if there is a vertex $u$ in case 4 or 5, let us show that there exists $v \in N_{G_0}(u)$ such that $v$ is also in case 4 or 5. Notice that in step 2, no vertex is in case 2 or 6. Suppose by contradiction that every vertex in $N_{G_0}(u)$ is in case 1 or 3, i.e.~for every $v \in N_{G_0}(u)$, $v \in A \cap \overline B  \cap D$. Then $|B_0\cap N_{G_0}(u)\cap N_{G'}(u)| + |A_0\cap N_{G_0}(u)\cap \{u\}| +  |D_0\cap \{u\}\cap N_{G'}(u)| + |C_0\cap \{u\}\cap \{u\}| = |C_0 \cap \{u\}| = 1 \bmod 2$, contradicting $(ii)$. At the end of the algorithm, every vertex is in case 1 or 3. Actually, every vertex is in case 1. Suppose by contradiction there is a vertex $u$ in case 3. Then $|B_0\cap N_{G_0}(u)\cap N_{G'}(u)| + |A_0\cap N_{G_0}(u)\cap \{u\}| +  |D_0\cap \{u\}\cap N_{G'}(u)| + |C_0\cap \{u\}\cap \{u\}|  = |C_0 \cap \{u\}| = 1 \bmod 2$, contradicting $(ii)$. Thus, at the end of the algorithm, $G' = G \star s_0$.

\subparagraph{Termination.} The number of vertices in case 1 or 3 strictly increases at each iteration of the algorithm.
\end{proof}

\begin{remark}
    The algorithm in the proof of \cref{prop:Bouchet} always outputs a sequence of m local complementations where $m \ls \left\lfloor 3n/2\right\rfloor$. Indeed, once a vertex is in case 1 or 3, no local complementation is applied on it again. Furthermore, it takes either one local complementation to change the case of one vertex to 1 or 3 (step 1) or 3 local complementations to change the case of 2 vertices to 1 and 3 (step 2). 
\end{remark}

Notice that Equation $(i)$ is  actually a linear equation: the set 
$\mathcal S\subseteq (2^V)^4$
of solutions to  $(i)$ is a vector space, indeed given two solutions $S=(A,B,C,D)$ and $S'=(A',B',C',D')$ to $(i)$, so is $S+S'=(A\Delta A', B\Delta B', C\Delta C', D\Delta D')$. The linearity of equation $(i)$ can be emphasized using the following encoding. A set $A\se V$ can be represented by $n$ binary variables $a_1, \ldots, a_n\in \mathbb F_2$ such that $a_v=1 \Leftrightarrow v\in A$, moreover, with a slight abuse of notations, we identify any set $A\subseteq V$ with the corresponding diagonal $\mathbb F_2$ matrix of dimension $n\times n$ in which diagonal elements are the $(a_v)_{v\in V}$. Following \cite{Hein06} and \cite{VdnEfficientLC}, equation $(i)$ is equivalent to 
\begin{equation*}
\Gamma  B \Gamma' + \Gamma  A+  D\Gamma' +  C = 0
\end{equation*}
and equation $(ii)$ to 
\begin{equation*}
AD+BC=I
\end{equation*}
where $\Gamma$ and $\Gamma'$ are the adjacency matrices of $G$ and $G'$ respectively. Summing up, $G$ and $G'$ are related by local complementations if and only if $(i)$ and $(ii)$ have a solution. Solving only for $(i)$ is easy as $(i)$ is a linear system of equation. More precisely, a basis $\mathcal B$ of $\mathcal S$ of solutions to $(i)$ can be calculated efficiently in $O(n^4)$ operations by Gaussian elimination. Then, one can search among $\mathcal S$ solutions to $(ii)$. However, in general one cannot exclude that $\mathcal S$ is of dimension $O(n)$, thus this approach is not guaranteed to be polynomial. The tour de force of Bouchet's algorithm is to point out that fundamental properties of the solutions to both equations $(i)$ and $(ii)$ imply that a set of solutions $\mathcal C\subseteq \mathcal S$ that satisfies both $(i)$ and $(ii)$, is either small or it contains an affine subspace of $\mathcal S$ of small co-dimension.

\begin{lemma}[\cite{Bouchet1991}]
    Either the set $\mathcal S$ of solutions to $(i)$ is of dimension at most $4$, or the set  $\mathcal C\subseteq \mathcal S$ that additionally satisfies $(ii)$ is either empty or an affine subspace of $\mathcal S_L$ of codimension at most $2$.   
\end{lemma}

\begin{corollary} \label{cor:Bouchet}
    If $dim(\mathcal S) > 4$, the set  $\mathcal C\subseteq \mathcal S$ that additionally satisfies  $(ii)$ has a solution if and only if the set $\{b + b' ~|~ b,b' \in \mathcal B\}$ contains a vector satisfying $(ii)$.
\end{corollary}

\cref{cor:Bouchet} show that if a solution to $(i)$ and $(ii)$ does exist, this solution can be found by enumerating at most the 16 elements of $\mathcal S$ if $dim(\mathcal S) \ls 4$, or the $O(n^2)$ elements of $\{b + b' ~|~ b,b' \in \mathcal B\}$ if $dim(\mathcal S) > 4$, and checking if $(ii)$ is satisfied. The overall complexity of the algorithm is $O(n^4)$.

\subsection{Local complementation = local Clifford}

As we shall see in this section, the reason why local complementation is useful for the study of graph states is that local complementation corresponds to the actions of local Clifford operators on graph states. First, local complementation can be implemented with local Clifford operators on a graph state.

\begin{proposition}[\cite{VandenNest04}]For any graph $G=(V,E)$ and any $u\in V$,  
    $$ X\left(\frac \pi 2\right)_u Z\left(-\frac \pi 2\right)_{N_G(u)} \ket {G} = \ket{G\star u}$$     
    \label{prop:implementation_lc}
\end{proposition}

\begin{remark}
    As local complementation is self-inverse, the inverse local Clifford operator also implements a local complementation, i.e. 
    $X\left(-\frac\pi 2\right)_uZ\left(\frac \pi 2\right)_{N_G(u)}\ket {G} = \ket{G\star u}$.
\end{remark}

\begin{proof}

    According to \cref{prop:stabilizer_of_graph_state}, a graph state $\ket G$ is the unique quantum state, up to global phase, stabilized by $X_u Z_{N_G(u)}$ for every vertex $u \in V$.
    
    Let us prove that $X\left(\frac\pi 2\right)_a Z\left(-\frac \pi 2\right)_{N_G(a)} \ket {G}$ is stabilized by $X_u Z_{N_{G\star a}(u)}$ for any vertex $u \in V$, implying $X\left(\frac\pi 2\right)_a Z\left(-\frac \pi 2\right)_{N_G(a)} \ket {G} = e^{i\phi}\ket{G\star u}$. One can then check that $e^{i\phi}=1$ for example by keeping track of the phase before the basis state $\ket{00\cdots 0}$. 
    
    According to \cref{prop:lc_neighborhood}, for some vertex $u \in V$, $N_{G\star a}(u) = N_G(u) \Delta N_G(a) \Delta \{u\}$ if $u \sim_G a$, and $N_{G\star a}(u) = N_G(u)$ else. Fix a vertex $u \in V$. If $u \not\sim_G a$, then 
    \begin{align*}
        & \left(X_u Z_{N_{G\star a}(u)}\right)\left( X\left(\frac\pi 2\right)_a Z\left(-\frac \pi 2\right)_{N_G(a)}\right) \ket {G} = \left(X_u Z_{N_{G}(u)}\right)\left( X\left(\frac\pi 2\right)_a Z\left(-\frac \pi 2\right)_{N_G(a)}\right) \ket {G}\\
        & = \left( X\left(\frac\pi 2\right)_a Z\left(-\frac \pi 2\right)_{N_G(a)}\right) \left(X_u Z_{N_{G}(u)}\right) \ket{G} = \left( X\left(\frac\pi 2\right)_a Z\left(-\frac \pi 2\right)_{N_G(a)}\right) \ket{G}
    \end{align*}

    If $u \sim_G a$, then
    \begin{align*}
        & \left(X_u Z_{N_{G\star a}(u)}\right)\left( X\left(\frac\pi 2\right)_a Z\left(-\frac \pi 2\right)_{N_G(a)}\right) \ket {G}\\
        & = \left(X_u Z_{N_{G}(u)}Z_{N_{G}(a)\sm \{u\}}\right)\left( X\left(\frac\pi 2\right)_a Z\left(-\frac \pi 2\right)_{N_G(a)}\right) \left(X_a Z_{N_{G}(a)}\right)\ket {G}\\
        & = \left(X_u Z_{N_{G}(u)}\right)\left( X\left(-\frac\pi 2\right)_a  Z\left(\frac\pi 2\right)_uZ\left(-\frac \pi 2\right)_{N_G(a)\sm \{u\}}\right)\ket {G}\\
        & = \left( X\left(\frac\pi 2\right)_a Z\left(-\frac \pi 2\right)_{N_G(a)}\right) \left(X_u Z_{N_{G}(u)}\right)\ket {G} = \left( X\left(\frac\pi 2\right)_a Z\left(-\frac \pi 2\right)_{N_G(a)}\right) \ket{G} \qedhere  
    \end{align*}

\end{proof}

Likewise, pivoting can be implemented by means of Hadamard transformations (and some Pauli corrections):
\begin{proposition}[\cite{mhalla2012graph,van2005edge}]For any graph $G=(V,E)$ and any $(u,v)\in E$, 
$\ket{G\wedge uv} = H_uH_vZ_{N_G(u)\cap N_G(v)}\ket {G}$. 
\end{proposition}

\begin{proof}
    Recall $H = e^{i\pi/4} X X\left(\frac \pi 2\right) Z\left(- \frac \pi 2\right) X\left(\frac \pi 2\right) X = e^{-i\pi/4}  Z Z\left(- \frac \pi 2\right) X\left(\frac \pi 2\right) Z\left(- \frac \pi 2\right) Z$ and $N_{G \star v}(u) = N_G(u) \Delta N_G(v) \Delta \{u\}$  if $u \sim_G v$. Thus, according to \cref{prop:implementation_lc},
    \begin{align*}
        & \ket{G\wedge uv} = \ket{G \star u \star v \star u}\\
        & = (XHX)_u (ZHZ)_v ~ Z\left(-\frac \pi 2\right)_{N_{G\star u \star v}(u) \sm \{v\}}Z\left(-\frac \pi 2\right)_{N_{G\star u}(v) \sm \{u\}}Z\left(-\frac \pi 2\right)_{N_{G}(u) \sm \{v\}} \ket{G}\\
        & = (XHX)_u (ZHZ)_v ~ Z\left(-\frac \pi 2\right)_{N_G(v) \sm \{u\}}Z\left(-\frac \pi 2\right)_{N_G(v) \Delta N_G(u) \sm \{u,v\}}Z\left(-\frac \pi 2\right)_{N_{G}(u) \sm \{v\}} \ket{G}\\
        & = (XHX)_u (ZHZ)_v ~ Z_{N_G(u) \cup N_G(v) \sm \{u,v\}}\ket{G}
    \end{align*}
    According to \cref{prop:stabilizer_of_graph_state}, $X_u Z_{N_G(u)} \ket G = \ket G$ and $ X_u Z_{N_{G \star u \star v \star u}(u)} \ket{G \star u \star v \star u} \\= X_u Z_{N_G(v) \Delta \{u,v\}} \ket{G \star u \star v \star u} =\ket{G \star u \star v \star u}$. Thus, $(XHX)_u (ZHZ)_v \ket{G}\\ = H_u H_v Z_{N_G(u) \cap N_G(v)}\ket{G}$.
\end{proof}

As local complementation can be implemented by local Clifford operations on a graph state, it is obvious that if two graphs are related by a sequence of local complementations, the two corresponding graph states are LC-equivalent. Very conveniently, the converse is also true: if two graph states are LC-equivalent, the two corresponding graphs are related by a sequence of local complementations. The proof is a bit more involved, and was presented originally in \cite{VandenNest04}. Below we show that this fact can be directly derived from the one-to-one correspondence between the solutions to  equations $(i)$ and $(ii)$ and the local Clifford operators (up to Pauli operators) that maps $\ket G$ to $\ket {G'}$.

\begin{proposition}[\cite{Hein06, VdnEfficientLC}] \label{prop:clifford_ABCD}
     $A,B,C,D$ satisfy equations $(i)$ and $(ii)$ if and only if $\ket{G'} \!=\! e^{i\phi}\bigotimes_{u\in V}C_u \ket{G}$ where for any $u \in V$, $C_u$ is equal,  up to a Pauli operator, to:\\

     \centerline{\begin{tabular}{clccl}
    $I$ & if  $u\in A\cap\overline{B}\cap \overline{C}\cap D$&$\qquad\qquad$& $H$ & if $u\in \overline{A}\cap B \cap C \cap \overline{D}$\\
    $X(\pi/2)$ & if $u\in A\cap B \cap \overline{C} \cap D$&& $H X(\pi/2)$ & if $u\in \overline{A}\cap B \cap C \cap {D}$\\
    $Z(\pi/2)$ &if $u\in A\cap\overline{B}\cap C \cap D$&&   $HZ(\pi/2)$& if $u\in {A}\cap B \cap C \cap \overline{D}$
\end{tabular}}
\end{proposition}

\begin{proof}[Proof Sketch]
    Following \cite{Hein06} and \cite{VdnEfficientLC}, the stabilizer formalism has an equivalent formulation in terms of binary linear algebra. A graph state can be represented by a $2n \times n$ matrix, called the generator matrix, with coefficients in $\mathbb F_2$: if $\Gamma$ is its adjacency matrix, the graph state is represented by the matrix $(I|\Gamma)$. The generator matrix represents the stabilizers of the graph. More precisely, the left part represents the Pauli X and the right part represents the Pauli Z. Two graphs with adjacency matrices $\Gamma, \Gamma'$ are LC-equivalent if and only if there exists a $2n \times 2n$ matrix $Q$ such that $(I|\Gamma')PQ(I|\Gamma)^T=0$, where
    $$ P\defeq 
    \begin{pmatrix}
        0 & I\\
        I & 0
    \end{pmatrix} 
    \qquad \qquad 
    Q\defeq 
    \begin{pmatrix}
        A & B\\
        C & D
    \end{pmatrix} $$
    with $A,B,C,D$ being diagonal matrices such that $AD+BC=I$. This translates into the condition $\Gamma  B \Gamma' + \Gamma  A+  D\Gamma' +  C = 0$. These two conditions correspond respectively to equations $(ii)$ and $(i)$, where a vertex is in a set $A,B,C,D$ if and only if its diagonal value in the corresponding matrix is 1. To find the correspondence between $A,B,C,D$ and the local Clifford operators, we compare the action on the Pauli gates.
\end{proof}

In other words, there is a one-to-one mapping between the sets $A,B,C,D$ and the local Clifford unitaries mapping one graph state to the other, up to Pauli operators and global phase. Following the formalism of the proof of \cref{prop:Bouchet}, case 1 corresponds to $I$, case 2 to $X(\pi/2)$, case 3 to $Z(\pi/2)$, case 4 to $H$, case 5 to $X(\pi/2)H$, and case 6 to $Z(\pi/2)H$. There is also a correspondence between the action of local complementation in both formalisms. Recall that a local complementation can be implemented by X- and Z-rotations: $\ket{G\star u} = X(\frac \pi 2)_u Z(-\frac \pi 2)_{N_G(u)} \ket {G}$. If $G$ and $G'$ are LC-equivalent such that $\ket{G'} = e^{i\phi}\bigotimes_{u\in V}C_u \ket{G}$, then $\ket{G'} = e^{i\phi}\left(\bigotimes_{u\in V}C_u\right)\left( X\left(\frac \pi 2\right)_a Z\left(-\frac \pi 2\right)_{N_G(a)}\right) \ket{G \star a}$. Basic calculations show that the local Clifford operators, up to Pauli, are updated as the cases in the proof of \cref{prop:Bouchet}. Thus, \cref{prop:Bouchet} can be reformulated as such:

\begin{proposition}[\cite{VandenNest04}] \label{prop:lclc}
    If $G$ and $G'$ are LC-equivalent such that $\ket{G'} = C \ket{G}$, then there exists a sequence of (possibly repeating) vertices $a_1, \cdots, a_m$ such that $G' = G\star a_1 \star a_2 \star \cdots \star a_{m}$ with $m \ls \left\lfloor 3n/2\right\rfloor$, and the local Clifford operator that implements these local complementations is $C$ up to Paulis.
\end{proposition}

We provided a self-contained proof in \cite{claudet2024local}. In particular, two graphs are related by local complementations if and only if the two corresponding graph states are LC-equivalent. This motivates the denomination "LC-equivalence" for graphs and graph states as, conveniently, "LC" stands for both "local Clifford" and "local complementation". Thus, from now on, we use LC-equivalence for both graph and graph states, defined equivalently as "related by local Clifford operations" and "related by a sequence of local complementations". We use the notation $G_1 =_{LC} G_2$.

\section{The cut-rank function} \label{sec:cutrank}

\subsection{Definition}

Given a graph $G$ and a set $A \se V$ of vertices, one can define the map $\lambda_A: 2^{A}\to 2^{V\setminus A} = D\mapsto Odd_G(D)\setminus A$. Intuitively, $\lambda_A$ maps $D$ to its odd neighborhood on the other side of the cut (i.e. vertex bipartition) $A, V\sm A$. $\lambda_A$ is linear with respect to the symmetric difference: $\forall D,D'\subseteq A, \lambda_G(D\Delta D') = \lambda_A(D)\Delta\lambda_A(D')$. Notice that $D$ is in the kernel of $\lambda_A$, i.e. $\lambda_A(D)=\emptyset$, if and only if $Odd_G(D)\subseteq A$. 
The rank of $\lambda_A$ is nothing but the so-called cut-rank of $A$. The cut-rank function, introduced by Bouchet under the name "connectivity function" \cite{Bouchet1993,Bouchet1987}, and coined "cut-rank" by Oum and Seymour \cite{OUM2006},  is usually defined as follows:

\begin{definition}[Cut-rank function]
    For $A \se V$, let the cut-matrix $\Gamma_A = ((\Gamma_A)_{ab}: a \in A\text{, } b \in V\sm A)$ be the matrix with coefficients in $\mathbb{F}_2$ such that $\Gamma_{ab} = 1$ if and only if $(a, b) \in E$. The cut-rank function of $G$ is defined as
    \begin{align*}
        \cutrk\colon 2^V & \longrightarrow \mathbb{N}\\
        A &\longmapsto \textbf{rank}_{\mathbb{F}_2}(\Gamma_A)
    \end{align*}
\end{definition}

\subsection{Basic properties}

The proposition below lists essential properties of the cut-rank function.

\begin{proposition}[\cite{OUM2006}]
    \label{prop:cutrank_prop}
    The cut-rank function satisfies the following properties:
    \begin{itemize}
        \item \textbf{symmetry: } $\forall A \se V,~\cutrk(V\sm A) =  \cutrk(A)$;
        \item \textbf{linear boundedness: } $\forall A \se V,~\cutrk(A)  \ls |A|$;
        \item \textbf{submodularity: } $\forall A,B \se V,~\cutrk(A \cup B) + \cutrk(A \cap B)  \ls \cutrk(A) + \cutrk(B)$.
    \end{itemize}   
\end{proposition}

A set $A$ is said \emph{full cut-rank} when $\cutrk(A)=|A|$. The conjunction of symmetry, linear boundedness and submodularity, together with the fact that $\cutrk$ has values in $\mathbb{N}$, implies the following useful properties:

\begin{proposition} \label{prop:cutrank_prop2}
    For any graph $G$ of order $n$, 
    \begin{itemize}
        \item[$(i)$] $\cutrk(\emptyset)= \cutrk(V) = 0$;
        \item[$(ii)$]  $\forall A \se V,\cutrk(A) \ls \lfloor n/2 \rfloor$, so a full-cut-rank set is of size at most $\lfloor n/2 \rfloor$;
        \item[$(iii)$]  If $\cutrk(K)=|K|$ then for any $A\se K$, $\cutrk(A)=|A|$, i.e. any subset of a full-cut-rank set is full cut-rank.    
    \end{itemize} 
\end{proposition}

\begin{proof}
    $(i)$ Using {linear boundedness}, $\cutrk(\emptyset) \ls |\emptyset| = 0$ so $\cutrk(\emptyset) = 0$ as $\cutrk$ has values in $\mathbb{N}$. Then, using {symmetry}, $\cutrk(V) = \cutrk(\emptyset) = 0$. $(ii)$
  Let $A \se V$. Using {linear boundedness}, $\cutrk(A) \ls |A|$. Also, using {symmetry}, $\cutrk(A) = \cutrk(V\sm A) \ls |V\sm A| = n - |A|$. Thus, $\cutrk(A) \ls \min(|A|, n - |A|) \ls \lfloor n/2 \rfloor$.
    $(iii)$ Let $A\se K$. Using {submodularity}, $\cutrk(A)+\cutrk(K\sm A)\gs \cutrk(K)+\cutrk(\emptyset)=|K|+0$. Moreover, using {linear boundedness}, $\cutrk(A)\ls |A|$ and $\cutrk(K\sm A)\ls |K|-|A|$, so  $\cutrk(A)= |A|$.
\end{proof}

\subsection{Interpretation for graph states}

The cut-rank function of a graph gives information on the entanglement properties of the corresponding graph state. Actually, the cut-rank function coincides with some well-known entanglement measures.

\begin{proposition}[\cite{Hein06, Hein04}]
    Given a graph state $\ket G$ and a bipartition $(A , V \sm A)$ of vertices, all these quantities coincide:
    \begin{itemize}
        \item The cut-rank $cutrk(A)$;
        \item the Schmidt rank of the bipartition, i.e. the logarithm in base 2 of the minimal number $R$ of terms in the decomposition \footnote{In the literature, the Schmidt rank is sometimes defined without the logarithm, that is, the Schmidt rank is simply defined as the minimal number $R$ in the decomposition $\ket G = \sum_{i=1}^R \mu_i \ket{\psi_i}_A \otimes \ket{\psi'_i}_{V\sm A}$. Also, the notion of Schmidt rank coincides with the notion of Schmidt measure with respect to a bipartition.} $\ket G = \sum_{i=1}^R \mu_i \ket{\psi_i}_A \otimes \ket{\psi'_i}_{V\sm A}$;
        \item the entanglement entropy of the bipartition (or von Neumann entropy) i.e. the quantity\\ $$-Tr\left(Tr_{V\sm A} \ket{G}\bra{G} \log_2 (Tr_{V\sm A} \ket{G}\bra{G})\right);$$
        \item the purity of the reduced density matrix $Tr_{V\sm A} \ket{G}\bra{G}$ i.e. the quantity $$Tr\left(Tr_{V\sm A} \left(\ket{G}\bra{G}\right)^2\right).$$
    \end{itemize}
\end{proposition}

\subsection{Relationship with LU-equivalence}

As the Schmidt rank is invariant under local unitaries \cite{nielsenchuang}, so is the cut-rank.

\begin{proposition} \label{prop:cutrank_lu}
    Two LU-equivalent graphs have the same cut-rank function.
\end{proposition}

Also, as local complementation can be implemented using local Clifford operation, we recover the well-established fact that local complementation preserves the cut-rank \cite{OumCliqueWidth}.

\begin{corollary}
    Local complementation does not change the cut-rank function i.e. two LC-equivalent graph state have the same cut-rank function. 
\end{corollary}

While two LU-equivalent graphs have the same cut-rank function, there exist pairs of graphs that have the same cut-rank function, but that are not LU-equivalent: a counterexample involves two isomorphic Petersen graphs \cite{Fon-Der-Flaass1996,Hein04}, see Figure \ref{fig:petersen}. This is however not the smallest counter-example: the smallest counter-example is a pair of graphs on 9 vertices that was found by computer search \cite{OumSurvey}.

\begin{figure}[H]
\centering
\begin{tikzpicture}    
\begin{scope}[every node/.style=vertices]
    \node (U0) at (-1.427,0.464) {2};
    \node (U1) at (0,1.5) {3};
    \node (U2) at (1.427,0.464) {4};
    \node (U3) at (0.881,-1.124) {5};
    \node (U4) at (-0.881,-1.124) {1};
    \node (U5) at (-1.427*2,0.464*2) {7}; 
    \node (U6) at (0,1.5*2) {8};  
    \node (U7) at (1.427*2,0.464*2) {9};
    \node (U8) at (0.881*2,-1.124*2) {10};
    \node (U9) at (-0.881*2,-1.124*2) {6};  
\end{scope}
\begin{scope}[every node/.style={},
                every edge/.style=edges]              
    \path [-] (U0) edge node {} (U3);
    \path [-] (U3) edge node {} (U1);
    \path [-] (U1) edge node {} (U4);
    \path [-] (U4) edge node {} (U2);
    \path [-] (U2) edge node {} (U0);

    \path [-] (U5) edge node {} (U6);
    \path [-] (U6) edge node {} (U7);
    \path [-] (U7) edge node {} (U8);
    \path [-] (U8) edge node {} (U9);
    \path [-] (U9) edge node {} (U5);

    \path [-] (U7) edge node {} (U8);
    \path [-] (U8) edge node {} (U9);
    \path [-] (U0) edge node {} (U5);
    \path [-] (U1) edge node {} (U6);
    \path [-] (U2) edge node {} (U7);
    \path [-] (U3) edge node {} (U8);
    \path [-] (U4) edge node {} (U9);

\end{scope}
\end{tikzpicture}\caption{The Petersen graph. This graph is not LU-equivalent to the graph obtained from it by permuting the vertices 1,2,3,4,5 with respectively 6,7,8,9,10. However, the two graphs have the same cut-rank function.}
\label{fig:petersen}
\end{figure}
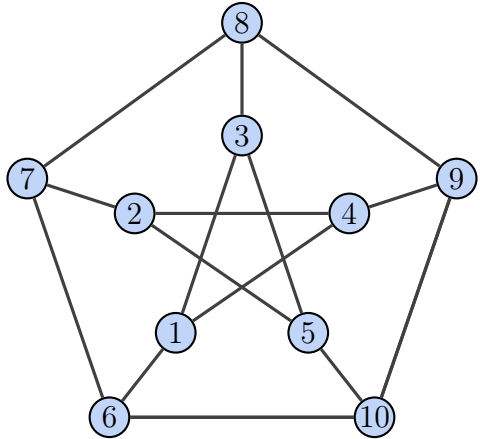

\chapter{Minimal local sets}

Local sets, and most importantly, minimal local sets, are our main graphical tool to study the LU-equivalence of graphs. The main reason behind their usefulness is that (minimal) local sets are LU-invariant, i.e. two LU-equivalent graphs have the same (minimal) local sets. Also, minimal local sets give constraints on the local unitaries that relate two LU-equivalent graph states (see \cref{subsec:constraints_type}). Thus, it is crucial to be able to construct a MLS cover: a family of minimal local sets that cover each vertex of a graph. This is actually always possible (and efficient), as we show in \cref{sec:mls_cover}.

\label{chap:mls}

\section{Definitions}

\subsection{Graphical definition}

\begin{definition}[Local set]Given $G=(V,E)$, a \emph{local set} $L$ is a non-empty subset of $V$ of the form $L = D \cup Odd_G(D)$ for some $D \se V$ called a \emph{generator}.
\end{definition}

\begin{definition}[Minimal local set]
    A minimal local set is a local set that is minimal by inclusion, i.e. it does not contain any smaller local set. 
\end{definition}

Local sets and minimal local sets are illustrated in Figure \ref{fig:MLS}.

\begin{figure}
    \centering
    
    \scalebox{1}{
    \begin{tikzpicture}[scale = 0.7]    
        \begin{scope}[every node/.style=vertices]
            \node (U0) at (0,0) {1};
            \node (U1) at (0,4) {2};
            \node (U2) at (4,4) {3};
            \node (U3) at (4,0) {4};
        \end{scope}
        \begin{scope}[every node/.style={},
                        every edge/.style=edges]              
            \path [-] (U0) edge node {} (U1);
            \path [-] (U1) edge node {} (U2);
            \path [-] (U2) edge node {} (U3);
            \path [-] (U3) edge node {} (U0);
        \end{scope}
        \draw[cornflowerblue, ultra thick] (4,-0.75) -- (0,-0.75) arc(-90:-180:0.75)-- (-0.75,0) -- (-0.75,4) arc(180:45:0.75);
        \draw[cornflowerblue, ultra thick] (4,-0.75)  arc(-90:45:0.75) -- (0.52,4.54);
    \end{tikzpicture}\qquad\qquad\qquad\qquad\raisebox{0cm}{
        \begin{tikzpicture}[scale = 0.7]
        
        \begin{scope}[every node/.style=vertices]
            \node (U0) at (0,0) {1};
            \node (U1) at (0,4) {2};
            \node (U2) at (4,4) {3};
            \node (U3) at (4,0) {4}; 
        \end{scope}
        
        \begin{scope}[every edge/.style=edges]              
            \path [-] (U0) edge node {} (U1);
            \path [-] (U1) edge node {} (U2);
            \path [-] (U2) edge node {} (U3);
            \path [-] (U3) edge node {} (U0);        
        \end{scope}
        \begin{scope}[shift={(3.4,-0.42)},rotate=45]
            \draw[cornflowerblue, ultra thick] (0,0) -- (0,5.5) arc(180:0:0.75) -- (1.5,0) arc(0:-180:0.75);
        \end{scope}
    \end{tikzpicture}}
    }
    
    \caption{(Left) A local set generated by $D = \{1\}$: $Odd(D) = \{2,4\}$. This is not a minimal local set. (Right) A local set generated by $D = \{2,4\}$: $Odd(D) = \emptyset$. In particular, this is a minimal local set, as neither $\{2\}$, $\{4\}$ nor $\emptyset$ is a local set.}
    \label{fig:MLS}
    \end{figure}
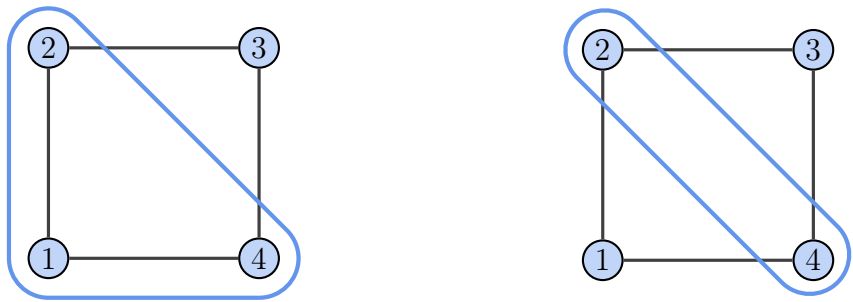

\subsection{A stabilizer characterization}

Local sets are strongly related to the stabilizers of the graph state, precisely to their support. A support of a Pauli operator is the set of positions where the Pauli is not the identity, that is, $\supp(S) = \{u \in S ~|~ S_u \not\propto I\}$.

\begin{proposition} \label{prop:mls_stabilizers}
    The local sets of a graph $G$ are precisely the supports of the stabilizers $(-1)^{|G[D]|}X_D Z_{Odd_G(D)}$ of $\ket G$. Minimal local sets correspond to the stabilizers with minimal support.
\end{proposition}

\begin{proof}
    Follows directly from \cref{prop:stabilizer_definition}.
\end{proof}

\subsection{A cut-rank characterization}

Arguably the most fundamental way of defining minimal local sets is through the cut-rank function.

\begin{proposition}
    \label{prop:characMLS}Given a graph $G=(V,E)$ and $A\se V$: 
    \begin{itemize}
        \item  If $A$ is a local set, then\footnote{This is not an equivalence. For example, in the complete graph of order 4 $K_4$, any set of vertices of size 3 is not a local set, however, such a set has a cut-rank of 1, and every set of size 1 or 2 also has a cut-rank of 1.} $\forall a \in A, \cutrk(A) \ls \cutrk(A \setminus \{a\})$;
        \item $A$ is a minimal local set if and only if $A$ is not full cut-rank, but each of its proper subset is, i.e.~$\forall a \in A, \cutrk(A) \ls \cutrk(A \setminus \{a\}) = |A|-1$.  
\end{itemize}
\end{proposition}

\begin{proof} Notice that $|A| - \cutrk(A)$, the dimension of the kernel of $\lambda_A: 2^{A}\to 2^{V\setminus A} = D\mapsto Odd_G(D)\setminus A$, counts the number of sets that generate a local set in $A$: $2^{|A|-\cutrk(A)} = |\{D \se A ~|~ D \cup Odd_G(D) \se A\}|$. As a consequence, $A \se V$ being a local set implies that for any $a \in A$, there is more local sets (counted with their generators) in $A$ than in $A \setminus \{a\}$. Then $|A\sm \{a\}| - \cutrk(A\sm \{a\}) < |A| - \cutrk(A)$. This translates to $\cutrk(A) \ls \cutrk(A \setminus \{a\})$.

Conversely, if $\forall a \in A, \cutrk(A) \ls \cutrk(A \setminus \{a\}) = |A|-1$, then any proper subset of $A$ contains no local set, and $A$ is a local set, thus $A$ is a minimal local set.
\end{proof}

\section{Basic properties}

\subsection{Invariance}

Two graphs that have the same cut-rank function have the same local sets, with the same number of generators, and this actually is an equivalence: 

\begin{proposition}
    \label{prop:corres_mls_cutrankalt}
    Two graphs have the same cut-rank function if and only if they have the same local sets with the same number of generators.   
\end{proposition}

\begin{proof}
    The cut-rank function can be computed from the number of generators of each local set: for every $A \se V$, $$\cutrk(A) = |A| - log_{2}\left( 1 + \sum_{\text{$L$ local set in $A$}} |\{D \se L ~|~ D \cup Odd_G(D) = L\}|\right)$$
    The "1" corresponds to the empty set $\emptyset$, which is part of the kernel of $\lambda_A: 2^{A}\to 2^{V\setminus A} = D\mapsto Odd_G(D)\setminus A$. Conversely, the number of generators of each local set  can be recursively computed from the values of the cut-rank function: 
    \begin{equation*}
        |\{D \se A | D \cup Odd_G(D) = A\}| = 2^{|A| - \cutrk(A)} - \sum_{B \varsubsetneq A} |\{D \se B | D \cup Odd_G(D) = B\}| \qedhere
    \end{equation*}
\end{proof}

Recall that local unitaries and local complementation preserve the cut-rank (see \cref{prop:cutrank_lu}). Then local sets are invariant by LC- and LU-equivalence. Historically, local sets were called local because of their invariance by local complementation \cite{Perdrix06}.

\begin{corollary} \label{cor:localset_invariant}
    Two LU-equivalent (or LC-equivalent) graphs have the same local sets with the same number of generators. In particular, local sets are invariant by local complementation.
\end{corollary}

A discussion on how the structure of a local set changes under local complementation can be found in \cite{Perdrix06}. For example, one can use local complementation to remove vertices from a generator of a minimal local set, leading to the following proposition.

\begin{proposition}[\cite{Perdrix06}]
    For any minimal local set $L$ of a graph $G$, for any $u\in L$, there exists $G' =_{LC}G$ such that $L = \{u\} \cup N_{G'}(u)$. 
\end{proposition}

Thus, the size of the smallest minimal local set  
can be related to the notion of minimum degree up to local complementation (sometimes called the local minimum degree \cite{Javelle12,CattaneoP15}).

\begin{definition} \label{def:dloc}
    The minimum degree of a graph up to local complementation is $\dloc(G) \defeq min_{\{u \in G' ~|~ G' =_{LC} G\}} \delta_{G'}(u)$ where the degree $\delta$ is the number of neighbors of a vertex.
\end{definition}

\begin{proposition} \label{prop:dloc_localset}
    The minimum degree up to local complementation coincides with the size of the smallest minimal local set $ - 1$, i.e.  
    $\dloc(G) = min_{D \se V, D \neq \emptyset}(|D \cup Odd_G(D)|) -1$.
\end{proposition}

\subsection{Minimal local sets have either 1 or 3 generators}

Minimal local sets have either 1 or 3 generators, and in the latter case, there are of even size. This was proved in \cite{VandenNest05} using the stabilizer formalism. Here we provide a graph-theoretical proof. Recall that $\Delta$ denotes the symmetric difference on vertices. $\Delta$ is commutative and associative, satisfies $A \Delta \emptyset = A$, $A \Delta A = \emptyset$ and $Odd(D_0\Delta D_1) =  Odd(D_0) \Delta Odd(D_1)$.

\begin{proposition}
    Given a minimal local set $L$, only two cases can occur:
    \begin{itemize}
        \item $L$ has exactly one generator;
        \item $L$ has exactly three (distinct) generators, of the form $D_0$, $D_1$, and $D_0 \Delta D_1$. This can only occur if $|L|$ is even.
    \end{itemize}
    In other words, $|L|-\cutrk(L)=$ 1 or 2, and the latter can only occur if $|L|$ is even.
\label{prop:MLS2cases}
\end{proposition}

\begin{proof}
First, let us prove that $L$ has either 1 or 3 generators.

Suppose there exist $D_0$ and $D_1$ ($D_0 \neq D_1$) such that $D_0\cup Odd(D_0) = L$ and $D_1\cup Odd(D_1) = L$. We know that $D_0 \Delta D_1 \neq \emptyset$. Besides, $(D_0 \Delta D_1)\cup Odd(D_0 \Delta D_1) = (D_0 \Delta D_1)\cup (Odd(D_0) \Delta Odd(D_1)) \se D_0 \cup D_1\cup Odd(D_0) \cup Odd(D_1) \se L$. And, as $L$ is a minimal local set, $D_0 \Delta D_1$ is a generator of $L$.

Let us show that $L = D_0 \cup D_1$. 
Suppose that there exists a vertex $u \in L \sm (D_0 \cup D_1)$.
Then, $u \in Odd(D_0)$ and $u \in Odd(D_1)$ so $u \in Odd(D_0) \cap Odd(D_1)$.
So $u \notin Odd(D_0) \Delta Odd(D_1) = Odd(D_0 \Delta D_1)$.
So, as $(D_0 \Delta D_1)\cup Odd(D_0 \Delta D_1) = L$, $u \in D_0 \Delta D_1 \se D_0 \cup D_1$, leading to a contradiction.
    
Now, we will prove that the three subsets, $D_0$, $D_1$, and $D_0 \Delta D_1$, are the only generators of $L$.
Suppose that there exists $D_2 \se L$ such that $D_0 \neq D_2$, $D_1 \neq D_2$, and $D_2 \cup Odd(D_2) = L$. Same as before, we can prove that $(D_0 \Delta D_2)\cup Odd(D_0 \Delta D_2) = L$ and $(D_1 \Delta D_2)\cup Odd(D_1 \Delta D_2) = L$. We also have that $L = D_0 \cup D_2$ and $L = D_1 \cup D_2$. Then $D_0 \cup D_1 = D_0 \cup D_2 = D_1 \cup D_2$, implying $D_1 \sm D_0 \se D_2$ and $D_0 \sm D_1 \se D_2$.
So $D_0 \Delta D_1 = (D_0 \sm D_1) \cup (D_1 \sm D_0) \se D_2$.
Suppose there exists $u \in D_2$ such that $u \in L \sm (D_0 \Delta D_1) = D_0 \cap D_1$. Then $u \in D_0 \cap D_2$, so $u \notin D_0 \Delta D_2$, and $u \in Odd(D_0 \Delta D_2)$.
For the same reason, $u \in Odd(D_1 \Delta D_2)$.
So $u \notin Odd(D_0 \Delta D_2) \Delta Odd(D_1 \Delta D_2) = Odd(D_0 \Delta D_1)$ and thus $u \in D_0 \Delta D_1$, leading to a contradiction.
Therefore,  $D_2 = D_0 \Delta D_1$.\\

Second, let us prove that the case where $L$ has 3 generators can only occur when $|L|$ is even.
All calculations are done modulo 2: "$\equiv$" means "is of same parity as".
Suppose that $L$ has three generators  $D_0$, $D_1$ and $D_0 \Delta D_1$. Then $L = D_0 \cup D_1 = \left(D_0 \sm D_1\right) \cup \left(D_1 \sm D_0\right) \cup \left(D_0 \cap D_1\right)$.
Given a subset $K$ of $V(G)$ and a vertex $u \in V$, let $\delta_{K}(u)$ denote the degree of $u$ inside $K$, i.e. $\delta_{K}(u) = |N(u)\cap K|$. Observe that for any $K \se V$, $\sum_{u\in K}\delta_{K}(u) = 2 \times \text{number of edges in $K$}\equiv 0$.
    
Given $u \in D_0 \sm D_1$, we know that $u \in Odd(D_1)$. Then $\delta_{L}(u) \equiv \delta_{D_1}(u) + \delta_{D_0 \sm D_1}(u) \equiv 1 + \delta_{D_0 \sm D_1}(u)$.
So $\sum_{u\in D_0 \sm D_1}\delta_{L}(u) \equiv |D_0 \sm D_1| + \sum_{u\in D_0 \sm D_1}\delta_{D_0 \sm D_1}(u) \equiv |D_0 \sm D_1|$.
Similarly, $\sum_{u\in D_1 \sm D_0}\delta_{L}(u) \equiv |D_1 \sm D_0|$ and $\sum_{u\in D_0 \cap D_1}\delta_{L}(u) \equiv |D_0 \cap D_1|$.   
Summing all three equations, we have: $0 \equiv \sum_{u\in L}\delta_{L}(u) \equiv \sum_{u\in D_0 \sm D_1}\delta_{L}(u) + \sum_{u\in D_1 \sm D_0}\delta_{L}(u) + \sum_{u\in D_0 \cap D_1}\delta_{L}(u) \equiv |D_0 \sm D_1| + |D_1 \sm D_0| + |D_0 \cap D_1| \equiv |L|$.
So $|L|$ is even.
\end{proof}

From now on, we say that a minimal local set has \textbf{dimension 1} if it has 1 generator, and \textbf{dimension 2} if it has 3 generators. The term "dimension" refers to the dimension of the kernel of the map $\lambda_L: 2^{A}\to 2^{V\setminus A} = D\mapsto Odd(D)\setminus A$, which is equal to $|A|-\cutrk(A)$ (see \cref{sec:cutrank}).

\subsection{Size of minimal local sets}

According to \cref{prop:cutrank_prop2} and \cref{prop:characMLS}, minimal local sets are of size at most half the order of the graph. We slightly refine this bound and show it is tight:

\begin{proposition}
    For any minimal local set $L$ in a graph of order $n$,
     \[|L|\ls \begin{cases}   \hspace{0.15cm}n/2&\text{if $n= 0 \bmod 4$}\\   \lfloor n/2 \rfloor + 1 &\text{otherwise}\end{cases}\]
    This bound is tight in the sense that for every $n>0$ there exists a graph of order $n$ that contains a minimal local set of this particular size.
    \label{prop:MLSbounds}
\end{proposition}

\begin{proof}

    \newcommand{\valuescale}{0.7}

    First note that any minimal local set has size at most $\lfloor n/2 \rfloor + 1$. Indeed, by \cref{prop:characMLS}, a necessary condition for a set $L \se V$ to be a minimal local set is that any proper subset of $L$ is full cut-rank. And, by \cref{prop:cutrank_prop2}, any full-cut-rank set has size at most $ \lfloor n/2 \rfloor$. This proves the result for $n \neq 0 \bmod 4$. Let us now suppose that $n = 0 \bmod 4$, and suppose by contradiction that there exists a minimal local set $L$ of size $n/2 + 1$. Using {symmetry} and {linear boundedness}, $\cutrk(L) = \cutrk(V \sm L) \ls |V \sm L| = n/2 - 1$. Then $|L| - \cutrk(L) \gs 2$. As $L$ is a minimal local set, $|L| - \cutrk(L) =$ 1 or 2, and the latter can only occur when $|L|$ is even (see \cref{prop:MLS2cases}).
    This leads to a contradiction, as $|L| = n/2 + 1$ is odd.\\
    We show below that this bound is tight, in the sense that for any $n$ there exists a graph of order $n$ that has a minimal local set whose size matches the bound. For this purpose, we explicitly construct graphs of arbitrary order that contain a minimal local set whose size matches the bound.
    
    \subparagraph{Case $n$ odd.}
    Let $n = 2 m + 1$ be an odd number. We provide in Figure \ref{fig:bounds1} a construction of a graph containing a minimal local set of size $m+1$.
    
    \begin{figure}[h!]
    \centering
    \scalebox{\valuescale}{
    \begin{tikzpicture}[scale = 0.7]
    \begin{scope}[every node/.style={circle,minimum size=30pt,thick,draw, fill=colorvertices}]
        \node (A) at (0,4) {$x$};
        \node (B) at (7,8) {$v_1$};
        \node (C) at (7,6) {$v_2$};
        \node (D) at (7,0) {\scalebox{1}{$v_{m}$}};
        \node (E) at (14,8) {$u_1$};
        \node (F) at (14,6) {$u_2$};
        \node (G) at (14,0) {\scalebox{1}{$u_{m}$}};
    \end{scope}
    \begin{scope}[every node/.style={},
                    every edge/.style=edges]
        \path [-] (A) edge node {} (B);
        \path [-] (A) edge node {} (C);
        \path [-] (A) edge node {} (D);
        \path [-] (B) edge node {} (E);
        \path [-] (C) edge node {} (F);
        \path [-] (D) edge node {} (G);
    \end{scope}
    \begin{scope}[style=edges]
        \draw (7.1,3.2) node[rotate = 90](){\scalebox{2}{$\boldsymbol{\ldots}$}};
        \draw (14.1,3.2) node[rotate = 90](){\scalebox{2}{$\boldsymbol{\ldots}$}};
    \end{scope}
    \begin{scope}[style={draw=cornflowerblue, ultra thick}]
        \draw[rounded corners=15mm] (-2,4)--(8,9.8)--(8,-1.8)--cycle;
        \draw (8.35,1.5) node[text=cornflowerblue](){$\boldsymbol{L}$};
    \end{scope}
    \end{tikzpicture}
    }
    \caption{Illustration of a minimal local set of size $\lceil n/2 \rceil$ when $n$ is odd.}
    \label{fig:bounds1}
    \end{figure}
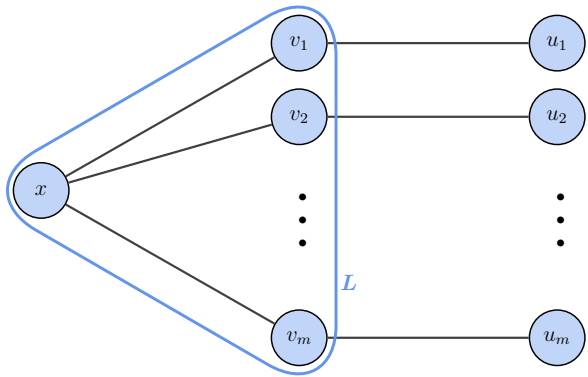
    
    $L$ is a local set, as $L = D \cup Odd(D)$ with $D = \{x\}$.
    In particular, $L$ is a minimal local set. Indeed, assume that there exists a subset $D'$ such that $D' \cup Odd(D') \varsubsetneq L$. $D'$ cannot contain any of the $v_i$'s as we would then also have $u_i \in Odd(D')$. So $D' = \emptyset$ or $\{x\}$, leading to a contradiction.
    
    \subparagraph{Case $n$ even and $n/2$ odd.}
    Let $n = 2 m$ be an even number such that $m$ is odd. We provide in Figure \ref{fig:bounds2} a construction  of a graph containing a minimal local set of size $m+1$.
    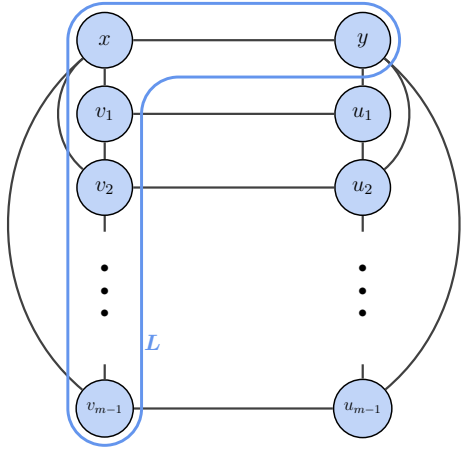
\begin{figure}[H]
    \centering
    \scalebox{\valuescale}{
    \begin{tikzpicture}[scale = 0.7]
    \begin{scope}[every node/.style={circle,minimum size=30pt,thick,draw,fill=colorvertices}]
        \node (X) at (0,10) {$x$};
        \node (U) at (7,10) {$y$};
    
        \node (B) at (0,8) {$v_1$};
        \node (C) at (0,6) {$v_2$};
        \node (D) at (0,0) {\scalebox{0.8}{$v_{m-1}$}};
    
        \node (E) at (7,8) {$u_1$};
        \node (F) at (7,6) {$u_2$};
        \node (G) at (7,0) {\scalebox{0.8}{$u_{m-1}$}};
    \end{scope}
    \begin{scope}[every node/.style={},
                    every edge/.style=edges]
        \path [-] (X) edge node {} (U);
        \path [-] (B) edge node {} (E);
        \path [-] (C) edge node {} (F);
        \path [-] (D) edge node {} (G);
    
        \path [-] (X) edge node {} (B);
        \path [-] (U) edge node {} (E);
        \path [-] (B) edge node {} (C);
        \path [-] (E) edge node {} (F);
    
        \path [-] (X) edge[bend right=50] node {} (C);
        \path [-] (X) edge[bend right=50] node {} (D);
    
        \path [-] (U) edge[bend left=50] node {} (F);
        \path [-] (U) edge[bend left=50] node {} (G);
    
        \path (C) edge node {} (0,4.8);
        \path (F) edge node {} (7,4.8);
        \path (D) edge node {} (0,1.2);
        \path (G) edge node {} (7,1.2);
    
    \end{scope}
    \begin{scope}[style={draw=blue, ultra thick}]
        \draw (0,3.2) node[rotate = 90](){\scalebox{2}{$\boldsymbol{\ldots}$}};
        \draw (7,3.2) node[rotate = 90](){\scalebox{2}{$\boldsymbol{\ldots}$}};
    \end{scope}
    \begin{scope}[style={draw=cornflowerblue, ultra thick}]
        \draw (7,9) arc(-90:90:1) -- (-0,11) arc(90:180:1) -- (-1,0) arc(-180:-90:1) -- (0,-1) arc(-90:0:1) -- (1,8) arc(180:90:1) -- (7,9);
        \draw (1.3,1.8) node[text=cornflowerblue](){$\boldsymbol{L}$};
    \end{scope}
    \end{tikzpicture}
    }
    \caption{Illustration of a minimal local set of size $ n/2+1$ when $n$ is even and $n/2$ is odd.}
    \label{fig:bounds2}
    \end{figure}
    
    $L$ is a local set, as $L = D \cup Odd(D)$ with $D = \{x\}$. 
    In particular, $L$ is a minimal local set. Indeed, assume that there exists a subset $D'$ such that $D' \cup Odd(D') \varsubsetneq L$. If $D'$ contains $y$, it also needs to contain every $v_i$, so that none of the $u_i$'s enter $Odd(D')$. At this point, either $x \in D'$, or $x \notin D'$: in both cases, $D' \cup Odd(D') = L$ (because $m$ is odd), leading to a contradiction. If $D'$ does not contain $y$, it cannot contain any of the $v_i$'s, so that none of the $u_i$'s enter $Odd(D')$. So $D' = \{x\} = D$, leading to a contradiction.
    
    \subparagraph{Case $n$ even and $n/2$ even.}
    Let $n = 2 m$ be an even number such that $m$ is even. We provide in Figure \ref{fig:bounds3} a construction of a graph containing a minimal local set of size $m$.
    
    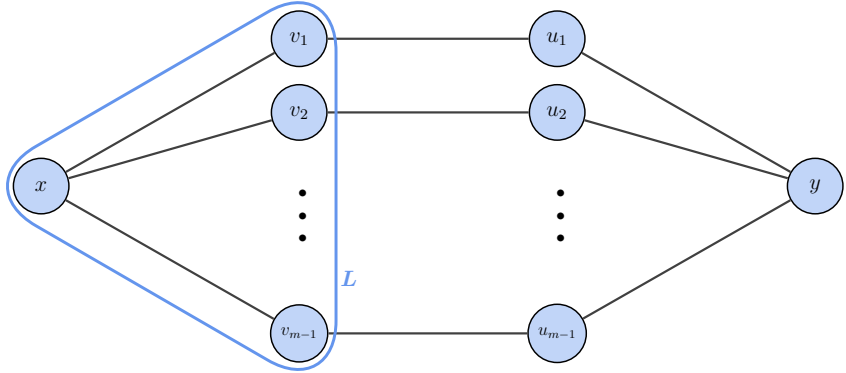
\begin{figure}[H]
    \centering
    \scalebox{\valuescale}{
    \begin{tikzpicture}[scale = 0.7]
    \begin{scope}[every node/.style={circle,minimum size=30pt,thick,draw,fill=colorvertices}]
        \node (A) at (0,4) {$x$};
        \node (B) at (7,8) {$v_1$};
        \node (C) at (7,6) {$v_2$};
        \node (D) at (7,0) {\scalebox{0.8}{$v_{m-1}$}};
        \node (E) at (14,8) {$u_1$};
        \node (F) at (14,6) {$u_2$};
        \node (G) at (14,0) {\scalebox{0.8}{$u_{m-1}$}};
        \node (H) at (21,4) {$y$};
    \end{scope}
    \begin{scope}[every node/.style={},
                    every edge/.style=edges]
        \path [-] (A) edge node {} (B);
        \path [-] (A) edge node {} (C);
        \path [-] (A) edge node {} (D);
        \path [-] (B) edge node {} (E);
        \path [-] (C) edge node {} (F);
        \path [-] (D) edge node {} (G);
        \path [-] (E) edge node {} (H);
        \path [-] (F) edge node {} (H);
        \path [-] (G) edge node {} (H);
    \end{scope}
    \begin{scope}[style={draw=blue, ultra thick}]
        \draw (7.1,3.2) node[rotate = 90](){\scalebox{2}{$\boldsymbol{\ldots}$}};
        \draw (14.1,3.2) node[rotate = 90](){\scalebox{2}{$\boldsymbol{\ldots}$}};
    \end{scope}
    \begin{scope}[style={draw=cornflowerblue, ultra thick}]
        \draw[rounded corners=15mm] (-2,4)--(8,9.8)--(8,-1.8)--cycle;
        \draw (8.35,1.5) node[text=cornflowerblue](){$\boldsymbol{L}$};
    \end{scope}
    \end{tikzpicture}
    }
    \caption{Illustration of a minimal local set of size $n/2$ when both $n$ and $n/2$ are even.}
    \label{fig:bounds3}
    \end{figure}
    
    $L$ is a local set, as $L = D \cup Odd(D)$ with $D = \{x\}$.
    In particular, $L$ is a minimal local set. Indeed, assume that there exists a subset $D'$ such that $D' \cup Odd(D') \varsubsetneq L$. It cannot contain any of the $v_i$'s as we would then also have $u_i \in Odd(D')$ (since $y \notin D'$). So $D' = \emptyset$ or $\{x\}$, leading to a contradiction.
\end{proof}

Small and large minimal local sets can coexist in a graph.

\begin{proposition}
    A path graph $P_n$ of order $n>2$ has minimal local sets of every size ranging from $2$ to $\lceil n/2 \rceil$.
\end{proposition}

An illustration is provided in Figure \ref{fig:path}.

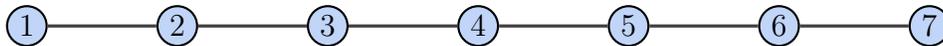
\begin{figure}[H]
    \centering
    
    \scalebox{1}{    
        \begin{tikzpicture}[scale = 1]
        
        \begin{scope}[every node/.style=vertices]
            \node (U1) at (0,0) {$1$};
            \node (U2) at (2,0) {$2$};
            \node (U3) at (4,0) {$3$};
            \node (U4) at (6,0) {$4$};
            \node (U5) at (8,0) {$5$};
            \node (U6) at (10,0) {$6$}; 
            \node (U7) at (12,0) {$7$};  
        \end{scope}
        
        \begin{scope}[every node/.style={},
                        every edge/.style=edges]              
            \path [-] (U1) edge node {} (U2);
            \path [-] (U2) edge node {} (U3);
            \path [-] (U3) edge node {} (U4);   
            \path [-] (U4) edge node {} (U5); 
            \path [-] (U5) edge node {} (U6); 
            \path [-] (U6) edge node {} (U7);      
        \end{scope}
    \end{tikzpicture}}    
    
    \caption{The path $P_7$ of order 7. $P_7$ has minimal local sets from every size ranging from 2 to 4, for example $\{1,2\}$, $\{1,3,4\}$, and $\{1,3,5,6\}$, generated respectively by $\{1\}$, $\{1,3\}$, and $\{1,3,5\}$.}
    \label{fig:path}
\end{figure}

\begin{proof}
    Let $P_n$ be a path of order $n>2$ on vertices $1, \ldots, n$ such that, for any $i$, there is an edge between $i$ and $i+1$. For any $k\in [0, \lceil n/2\rceil-2]$, let $D_k:=\{1,3,5,\ldots, 2k+1\}$. The local set $L_k$ generated by $D_k$ is $L_k:=D_k\cup Odd(D_k) = D_k\cup \{2k+2\}$. We show that $L_k$ is minimal by inclusion: by contradiction, assume there exists a non-empty $D\subseteq L_k$, $D\neq D_k$ and $D\cup Odd(D) \varsubsetneq L_k$. First notice that $2k+2\notin D$, otherwise $2k+3\in Odd(D)$ (notice that $2k+3\ls n$, so $2k+3$ is actually a vertex of the path). So $D\varsubsetneq D_k$, as a consequence there exists  $r<k$  s.t. $2r+2$ has exactly one neighbor in $D$, implying $2r+2\in Odd(D)$. This is a contradiction since $2r+2\notin L$. 

    As a consequence, for any $k\in [0, \lceil n/2\rceil-2]$, $L_k$ is a minimal local set of size $k+2$. 
\end{proof}

Notice that we have exhibited minimal local sets of any size from $2$ to $\lceil n/2\rceil$, but this is not a complete description of all minimal local sets of the path, for instance $\{2,3,4\}$ (when $n>4$) and $\{1, 3, 5, \ldots, n\}$ (when $n=1\bmod 2$) are also minimal local sets.

Some graphs only have small minimal local sets.

\begin{proposition}
    Complete graphs and bipartite complete graphs (that contain no isolated vertex) only have minimal local sets of size 2.
\end{proposition}

An illustration is provided in Figure \ref{fig:complete_and_complete_bipartite}.%??. \Ncom{TODO: figure}

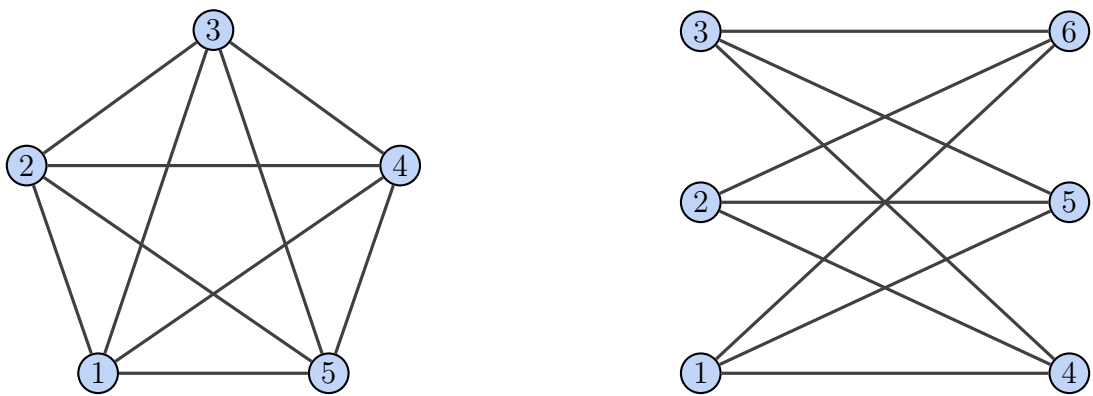
\begin{figure}[h!]
    \centering
    
    \scalebox{1}{
    \begin{tikzpicture}[scale = 0.87]    
        \begin{scope}[every node/.style=vertices]
            \node (U1) at (-0.881*2,-1.124*2) {1};
            \node (U2) at (-1.427*2,0.464*2) {2}; 
            \node (U3) at (0,1.5*2) {3};  
            \node (U4) at (1.427*2,0.464*2) {4};
            \node (U5) at (0.881*2,-1.124*2) {5};              
        \end{scope}
        \begin{scope}[every node/.style={},
                        every edge/.style=edges]              
            \path [-] (U1) edge node {} (U2);
            \path [-] (U1) edge node {} (U3);
            \path [-] (U1) edge node {} (U4);
            \path [-] (U1) edge node {} (U5);
            \path [-] (U2) edge node {} (U3);
            \path [-] (U2) edge node {} (U4);
            \path [-] (U2) edge node {} (U5);
            \path [-] (U3) edge node {} (U4);
            \path [-] (U3) edge node {} (U5);
            \path [-] (U4) edge node {} (U5);
        \end{scope}
    \end{tikzpicture}\qquad\qquad\qquad\qquad\raisebox{0cm}{
        \begin{tikzpicture}[xscale = 0.7, yscale = 0.65]
        
        \begin{scope}[every node/.style=vertices]
            \node (U1) at (0,-3) {1};
            \node (U2) at (0,0.5) {2};
            \node (U3) at (0,4) {3};
            \node (U4) at (7,-3) {4}; 
            \node (U5) at (7,0.5) {5}; 
            \node (U6) at (7,4) {6}; 
        \end{scope}
        
        \begin{scope}[every edge/.style=edges]              
            \path [-] (U1) edge node {} (U4);
            \path [-] (U1) edge node {} (U5);
            \path [-] (U1) edge node {} (U6);
            \path [-] (U2) edge node {} (U4);
            \path [-] (U2) edge node {} (U5);
            \path [-] (U2) edge node {} (U6);
            \path [-] (U3) edge node {} (U4);
            \path [-] (U3) edge node {} (U5);
            \path [-] (U3) edge node {} (U6);
        \end{scope}
    \end{tikzpicture}}
    }
    
    \caption{(Left) The complete graph $K_5$ of order 5. The minimal local sets are exactly the pairs of vertices (for example $\{1,2\}$, $\{3,5\}$). (Right) The bipartite complete graph $K_{3,3}$ of order 6. The minimal local sets are exactly the pairs of two leftmost vertices ($\{1,2\}$, $\{1,3\}$, and $\{2,3\}$), and the pairs of two rightmost vertices ($\{4,5\}$, $\{4,6\}$, and $\{5,6\}$).}
    \label{fig:complete_and_complete_bipartite}
\end{figure}

\begin{proof}
    First notice that a graph contains minimal local sets of size 1 if and only if it contains isolated vertices.

    Let $K_n$ be a complete graph of order $n \gs 2$. Let $u,v$ be two vertices of $K_n$. $Odd_{K_n}\!(\{u,v\})\! =\{u,v\}$ thus $\{u,v\}$ is a minimal local set. Any set of vertices of size 3 or more cannot be a minimal local set as it contains a pair of vertices forming a minimal local set.

    Let $K_{n_L,n_R}$ be a bipartite complete graph, i.e. a graph $K_{n_L,n_R}=(V,E)$ such that $V$ is partitioned into $V = L \cup R$ where $|L|=n_L\gs 2$ and $|R|=n_R \gs 2$, $L$ and $R$ are independent, and any two vertices $l,r$ where $l \in L$ and $r \in R$ share an edge. Let $l_1, l_2 \in L$. $Odd_{K_{n_L,n_R}}(\{l_1, l_2\}) =  \emptyset$ thus $\{l_1, l_2\}$ is a minimal local set. The same goes for any set $r_1, r_2 \in R$. Any set $l,r$ where $l \in L$ and $r \in R$, is not a minimal local set. Any set of vertices of size 3 or more cannot be a minimal local set as it contains a pair of vertices forming a minimal local set. The case $n_L = 1$ (or $n_R = 1$, equivalently) follows from the fact that the star graph is LC-equivalent to the complete graph, and minimal local sets are invariant under LC-equivalence according to \cref{cor:localset_invariant}.
\end{proof}

There exist also graphs with only "large" minimal local sets. Indeed, there exist graphs with linear minimum degree up to local complementation.

\begin{proposition}[\cite{Javelle12}] \label{prop:big_dloc}
    There exists $n_0 \in \mathbb N$ such that for all $n > n_0$ there exists a graph of order n whose minimum degree up to local complementation is greater than $0.189n$.
\end{proposition}

\begin{corollary}
    There exists $n_0 \in \mathbb N$ such that for all $n > n_0$ there exists a graph of order n whose minimal local sets all contain at least $0.189n$ vertices.
\end{corollary}

\begin{proof}
    Follows directly from the conjunction of \cref{prop:dloc_localset} and \cref{prop:big_dloc}.
\end{proof}

\subsection{Number of minimal local sets}

The size of the minimal local sets is, to some extent, related to the number of minimal local sets in a graph. For instance, complete graphs and bipartite complete graphs have a number of minimal local sets quadratic in their order. More generally, if the size of the minimal local sets is upper bounded by $k$ then there are obviously $O(n^k)$ minimal local sets in a graph of order $n$. Maybe more surprisingly, a lower bound on the size of the minimal local sets implies a lower bound on their number: 

\begin{proposition}  
Given a graph $G$ of order $n$, if all minimal local sets of $G$ are of size at least $m$, then the number of minimal local sets in $G$ is at least \[\frac{1-2r}{3\sqrt n}2^{n\left(1-(1-r)H_2\big(\frac1{2(1-r)}\big)\right)}\]
where $r=\frac m n$ and $H_2(x) = -x\log_2(x) - (1-x)\log_2(1-x)$ is the binary entropy. 
\label{prop:exp_number_MLS}
\end{proposition}

The proof uses the fact that any set of vertices of size $\lfloor n/2 \rfloor +1$ contains at least one minimal local set (by the conjunction of \cref{prop:cutrank_prop2} and \cref{prop:characMLS}), as a counting argument. The calculations are rather technical and are provided in \cite{claudet2024covering}. 
\cref{prop:exp_number_MLS} implies that a graph of order $n$, in which the size of the minimal local sets is lower bounded by $cn$ for some constant $c>0$,  has an exponential number of minimal local sets. In particular, graphs of order $n$ and minimum degree up to local complementation being at least $0.189n$, have at least $1.165^n$ minimal local sets. Notice that there also exist explicit graphs with exponentially many minimal local sets.

\begin{proposition}
    The graph $K_{k,k} \Delta M_k$ (of order $n = 2k$), defined as the symmetric difference of a complete bipartite graph and a matching, has more than $2^{k-1}=\frac12\sqrt {2^n}$ minimal local sets.
\end{proposition}

An illustration is provided in Figure \ref{fig:k_delta_m}.

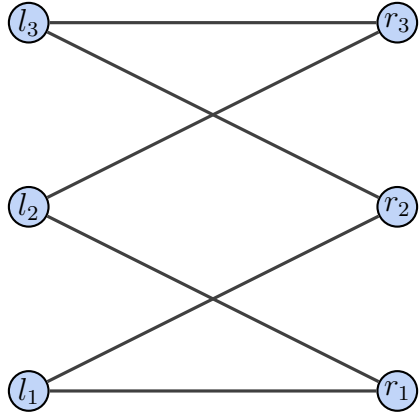
\begin{figure}[H]
    \centering
    
    \scalebox{1}{    
        \begin{tikzpicture}[scale = 0.7]
        
        \begin{scope}[every node/.style=vertices]
            \node (U1) at (0,-3) {$l_1$};
            \node (U2) at (0,0.5) {$l_2$};
            \node (U3) at (0,4) {$l_3$};
            \node (V1) at (7,-3) {$r_1$}; 
            \node (V2) at (7,0.5) {$r_2$}; 
            \node (V3) at (7,4) {$r_3$}; 
        \end{scope}
        
        \begin{scope}[every edge/.style=edges]              
            \path [-] (U1) edge node {} (V1);
            \path [-] (U1) edge node {} (V2);
            \path [-] (U2) edge node {} (V3);
            \path [-] (U2) edge node {} (V1);
            \path [-] (U3) edge node {} (V2);
            \path [-] (U3) edge node {} (V3);
        \end{scope}
    \end{tikzpicture}    
    }
    
    \caption{The graph $K_{3,3} \Delta M_3$ of order 6. Every set containing an odd number of leftmost vertices generates a minimal local set: $\{l_1\}$, $\{l_2\}$, $\{l_3\}$, and $\{l_1, l_2, l_3\}$ generate respectively the minimal local sets $\{l_1, r_1, r_2\}$, $\{l_2, r_1, r_3\}$, $\{l_3, r_2, r_3\}$, and $\{l_1, l_2, l_3\}$.}
    \label{fig:k_delta_m}
\end{figure}

\begin{proof}

    We write $K_{k,k} \Delta M_k = (L\cup R, E)$ with $L=\{l_1, \ldots, l_{k}\}$ and $R=\{r_1, \ldots, r_{k}\}$, such that for any $i$ and $j$, $(l_i,l_j)\notin E$, $(r_i,r_j)\notin E$ and $(l_i,r_j)\in E$ if $i \neq j$. Let $\omega \se [1,k]$ be a subset of the integers between 1 and $k$ such that $|\omega| = 1 \bmod 2$. The local set $L_\omega$ generated by $D_\omega = \bigcup_{i\in \omega}\{l_i\}$ is $L_\omega:=D_\omega\cup Odd(D_\omega)=D_\omega\cup \bigcup_{i\in [1,k] \sm \omega}\{r_i\}$. We show that $L_\omega$ is minimal by inclusion: let $D \se L_\omega$ such that $D\cup Odd(D) \se L_\omega$, we prove that $D \cup Odd(D)$ is either $\emptyset$ or $L_\omega$.

    \begin{itemize}
        \item Assume $|D \cap L| = 0 \bmod 2$. If $D \cap L \neq \emptyset$, i.e. there exists $l_i \in D \cap L$, then $r_i \in Odd(D)$, contradicting $r_i\notin L_\omega$. Thus, $D \cap L = \emptyset$.
        \item Assume $|D \cap L| = 1 \bmod 2$. If $D \cap L \neq D_\omega \cap L$, i.e. there exists $l_i \in D_\omega \sm D$, then $r_i \in Odd(D)$. Thus, $D \cap L = D_\omega \cap L$. 
        \item Assume $|D \cap R| = 0 \bmod 2$. If $D \cap R \neq \emptyset$, i.e. there exists $r_i \in D \cap R$, then $l_i \in Odd(D)$. Thus, $D \cap R = \emptyset$.
        \item Assume $|D \cap R| = 1 \bmod 2$. If $D \cap R \neq D_\omega \cap R$, i.e. there exists $r_i \in Odd(D_\omega) \sm D$, then $l_i \in Odd(D)$. Thus, $D \cap R = D_\omega \cap R$. 
    \end{itemize}

    To sum up, $D$ is either $\emptyset$, $D_\omega$, $Odd(D_\omega)$, or $L_\omega$. In any case, $D \cup Odd(D)$ is either $\emptyset$ or $L_\omega$. As a consequence, for any $\omega \se [1,k]$ such that $|\omega| = 1 \bmod 2$, $L_\omega$ is a minimal local set of size $k$. These minimal local sets are distinct, as $L_\omega \cap L = \omega$. Thus, $G_k$ contains at least $2^{k-1}$ minimal local sets.
\end{proof}

\section{An MLS cover always exists}

\label{sec:mls_cover}

We are interested in families of minimal local sets such that each vertex is contained in at least one minimal local set. We call such a family an \emph{MLS cover}.

\begin{definition}[MLS cover]
    Given a graph $G$, $\mathcal M\se 2^V$ is an MLS cover if:
    \begin{itemize}
        \item $\forall L\in \mathcal M$, $L$ is a minimal local set of $G$;
        \item $\forall u\in V$, $\exists L \in \mathcal M$ such that $u\in L$.
    \end{itemize}
\end{definition}

MLS covers are illustrated in Figure \ref{fig:MLScover}.

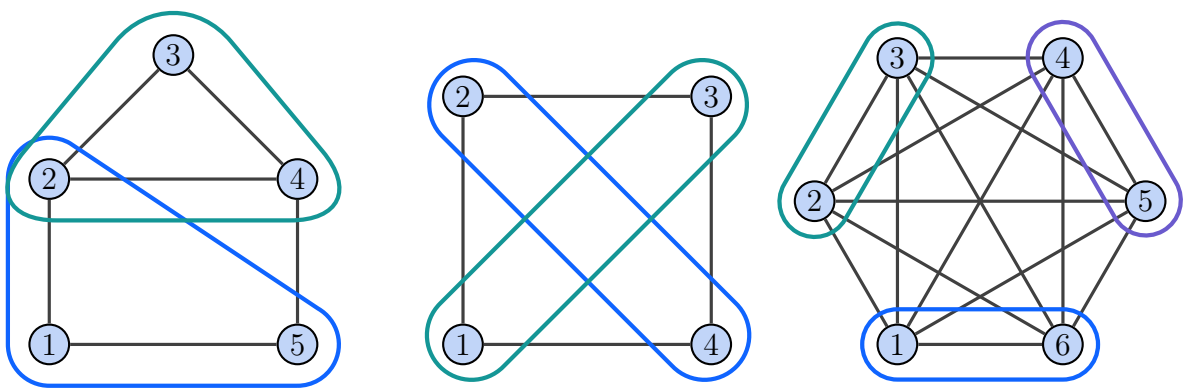
\begin{figure}[H]

\centering
\scalebox{1}{

%Graph 1
\begin{tikzpicture}[scale = 0.55]
\begin{scope}[every node/.style=vertices]
    \node (U1) at (0,0) {$1$};
    \node (U2) at (0,4) {$2$};
    \node (U3) at (3,7) {$3$};
    \node (U4) at (6,4) {$4$};
    \node (U5) at (6,0) {$5$};    
\end{scope}
\begin{scope}[every edge/.style=edges]              
    \path [-] (U1) edge node {} (U5);
    \path [-] (U1) edge node {} (U2);
    \path [-] (U2) edge node {} (U3);
    \path [-] (U2) edge node {} (U4);
    \path [-] (U3) edge node {} (U4);
    \path [-] (U4) edge node {} (U5);
\end{scope}

\draw[colormls1, ultra thick] (6,-1) -- (0,-1) arc(-90:-180:1)-- (-1,0) -- (-1,4) arc(180:60:1) -- (6.5,0.866) arc(60:-90:1);
\draw[colormls2, ultra thick, rounded corners=12mm] (-2,3)--(8,3)--(3,9)--cycle;

%Graph 2
\begin{scope}[shift={(10,0)},every node/.style=vertices]
    \node (U1) at (0,0) {$1$};
    \node (U2) at (0,6) {$2$};
    \node (U3) at (6,6) {$3$};
    \node (U4) at (6,0) {$4$};   
\end{scope}
\begin{scope}[every edge/.style=edges]              
    \path [-] (U1) edge node {} (U2);
    \path [-] (U2) edge node {} (U3);
    \path [-] (U3) edge node {} (U4);
    \path [-] (U1) edge node {} (U4);
\end{scope}

\begin{scope}[shift={(10,0)}]
\begin{scope}[rotate=0, scale = 1.08, shift={(4.7,-0.5)},rotate=45]
\draw[colormls1, ultra thick] (0,0) -- (0,7.3) arc(180:0:1) -- (2,0) arc(0:-180:1);
\end{scope}
\begin{scope}[xscale = -1, rotate=0, scale = 1.08, shift={(-0.9,-0.5)},rotate=45]
\draw[colormls2, ultra thick] (0,0) -- (0,7.3) arc(180:0:1) -- (2,0) arc(0:-180:1);
\end{scope}
\end{scope}

%Graph 3
\begin{scope}[shift={(20.5,0)},every node/.style=vertices]
    \node (U1) at (0,0) {1};
    \node (U2) at (4,0) {6};
    \node (U3) at (-2,3.464) {2};
    \node (U4) at (6,3.464) {5};
    \node (U5) at (0,6.928) {3}; 
    \node (U6) at (4,6.928) {4};  
\end{scope}
\begin{scope}[every edge/.style=edges]              
    \path [-] (U1) edge node {} (U2);
    \path [-] (U1) edge node {} (U3);
    \path [-] (U1) edge node {} (U4);
    \path [-] (U1) edge node {} (U5);
    \path [-] (U1) edge node {} (U6);
    \path [-] (U2) edge node {} (U3);
    \path [-] (U2) edge node {} (U4);
    \path [-] (U2) edge node {} (U5);
    \path [-] (U2) edge node {} (U6);
    \path [-] (U3) edge node {} (U4);
    \path [-] (U3) edge node {} (U5);
    \path [-] (U3) edge node {} (U6);
    \path [-] (U4) edge node {} (U5);
    \path [-] (U4) edge node {} (U6);
    \path [-] (U5) edge node {} (U6);

\end{scope}

\begin{scope}[shift={(20.5,0)}]

\draw[colormls1, ultra thick] (0,0.85) -- (4,0.85) arc(90:-90:0.85) -- (0,-0.85) arc(-90:-270:0.85);

\begin{scope}[shift={(-2,3.45)},rotate=60]
\draw[colormls2, ultra thick] (0,0.85) -- (4,0.85) arc(90:-90:0.85) -- (0,-0.85) arc(-90:-270:0.85);
\end{scope}

\begin{scope}[shift={(4,6.95)},rotate=-60]
\draw[colormls3, ultra thick] (0,0.85) -- (4,0.85) arc(90:-90:0.85) -- (0,-0.85) arc(-90:-270:0.85);
\end{scope}

\end{scope}

\end{tikzpicture}
}   
    
\caption{MLS covers of 3 graphs.
}
\label{fig:MLScover}
\end{figure}

To cover a vertex $u$ of a graph $G$ with a minimal local set, one can consider the local set generated by $u$, i.e. its closed neighborhood $\{u\}\cup N(u)$. However, such a  local set is not always minimal, worse yet, it does not necessarily contain a minimal local set that includes $u$. Indeed, in the cycle $C_4$ of order 4, for every vertex $u$, $\{u\}\cup N(u)$ contains a single minimal local set that does not include $u$ (see Figure \ref{fig:MLS}). 

In this particular $C_4$ example however, one can notice that, roughly speaking, minimizing local sets of the form $\{u\}\cup N(u)$ is enough to get an MLS cover (see Figure \ref{fig:MLScover}, middle graph): considering the minimal local sets included in closed neighborhoods is enough to produce an MLS cover.

This strategy does not work in general. There exist graphs (see Figure \ref{fig:MLSCover_is_hard}) for which an MLS cover requires minimal local sets that are not subsets of any closed neighborhood. Thus, one can wonder what is the best strategy to produce an MLS cover and even whether it exists for every graph. We show in this section that every graph admits an MLS cover (\cref{thm:MLS_cover}, proved in  \cref{sec:MLS_cover_proof}) and then introduce an efficient algorithm for finding an MLS cover (\cref{sec:algo_mls}).

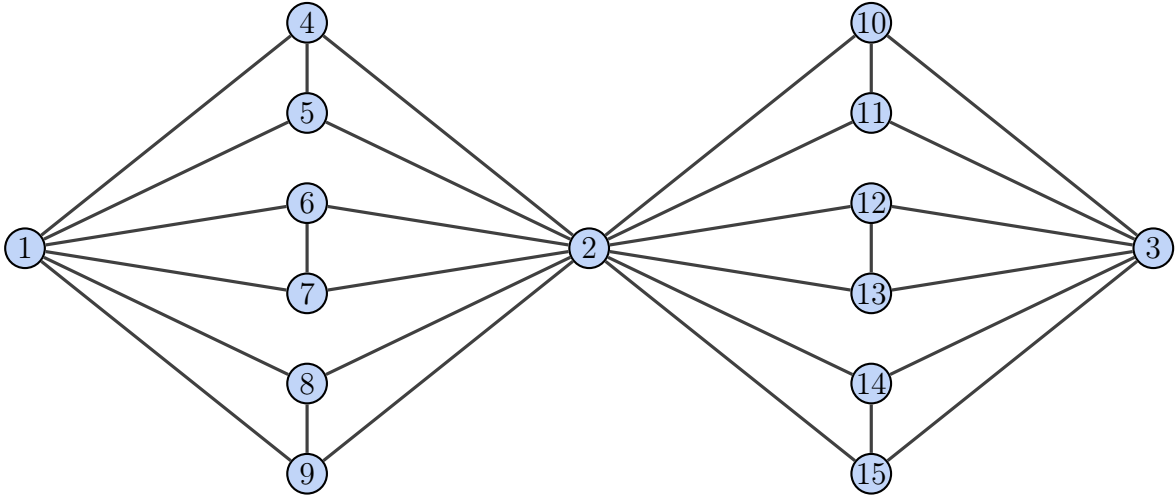
\begin{figure}[h]
    \centering
    \scalebox{1}{
    \begin{tikzpicture}[xscale = 0.75, yscale =0.8]    
        \begin{scope}[every node/.style=vertices]
            \node (U1) at (0,0) {2};
            \node (U0) at (-10,0) {1};
            \node (U2) at (10,0) {3};
            \node (U3) at (-5,3.75) {4};
            \node (U4) at (-5,2.25) {5};
            \node (U5) at (-5,0.75) {6};
            \node (U6) at (-5,-0.75) {7};
            \node (U7) at (-5,-2.25) {8};
            \node (U8) at (-5,-3.75) {9};
            \node (U9) at (5,3.75) {10};
            \node (U10) at (5,2.25) {11};
            \node (U11) at (5,0.75) {12};
            \node (U12) at (5,-0.75) {13};
            \node (U13) at (5,-2.25) {14};
            \node (U14) at (5,-3.75) {15};
        \end{scope}
        \begin{scope}[every node/.style={},
                        every edge/.style=edges]              
            \path [-] (U0) edge node {} (U3);
            \path [-] (U0) edge node {} (U4);
            \path [-] (U0) edge node {} (U5);
            \path [-] (U0) edge node {} (U6);
            \path [-] (U0) edge node {} (U7);
            \path [-] (U0) edge node {} (U8);
            \path [-] (U1) edge node {} (U3);
            \path [-] (U1) edge node {} (U4);
            \path [-] (U1) edge node {} (U5);
            \path [-] (U1) edge node {} (U6);
            \path [-] (U1) edge node {} (U7);
            \path [-] (U1) edge node {} (U8);
            \path [-] (U1) edge node {} (U9);
            \path [-] (U1) edge node {} (U10);
            \path [-] (U1) edge node {} (U11);
            \path [-] (U1) edge node {} (U12);
            \path [-] (U1) edge node {} (U13);
            \path [-] (U1) edge node {} (U14);
            \path [-] (U2) edge node {} (U9);
            \path [-] (U2) edge node {} (U10);
            \path [-] (U2) edge node {} (U11);
            \path [-] (U2) edge node {} (U12);
            \path [-] (U2) edge node {} (U13);
            \path [-] (U2) edge node {} (U14);
            \path [-] (U3) edge node {} (U4);
            \path [-] (U5) edge node {} (U6);
            \path [-] (U7) edge node {} (U8);
            \path [-] (U9) edge node {} (U10);
            \path [-] (U11) edge node {} (U12);
            \path [-] (U13) edge node {} (U14);

        \end{scope}
    \end{tikzpicture}}
    
    \caption{Example of a graph of order 15 where a naive approach for finding an MLS cover, based on the closed neighborhood of every vertex, cannot work. The vertices 1, 2 and 3 are contained in several minimal local sets, the most obvious one being $\{1,2,3\}$. However, no local set of the form $\{u\}\cup N(u)$ contains a minimal local set that contains either 1, 2 or 3. Indeed, the minimal local sets in $\{1\}\cup N(1)$ are $\{4,5\}$, $\{6,7\}$ and $\{8,9\}$. Symmetrically, the minimal local sets in $\{3\}\cup N(3)$ are $\{10,11\}$, $\{12,13\}$ and $\{14,15\}$. The minimal local sets in $\{2\}\cup N(2)$ are $\{4,5\}$, $\{6,7\}$, $\{8,9\}$, $\{10,11\}$, $\{12,13\}$ and $\{14,15\}$. The only minimal local set in $\{4\}\cup N(4)$ is $\{4,5\}$, and the same goes for 5, 6, 7, 8, 9, 10, 11, 12, 13 and 14, 15.}
    \label{fig:MLSCover_is_hard}
\end{figure}

\begin{theorem}
    \label{thm:MLS_cover}
    Any graph has an MLS cover.   
\end{theorem}

\subsection{Proof}
\label{sec:MLS_cover_proof}

In this section, we prove \cref{thm:MLS_cover}. First, we introduce a cut-rank-based characterization of the existence of a minimal local set covering a given vertex:

\begin{lemma}
    \label{lemma:characinMLS} Given a graph $G=(V,E)$, 
    $a\in V$ is contained in a minimal local set if and only if there exists $A\se V$ such that $A$ is full cut-rank but $A\cup \{a\}$ is not.
\end{lemma}

\begin{proof}
Let $a \in V$. Suppose $a$ is contained in a minimal local set $L \se V$. Then $L \sm \{a\}$ is full cut-rank but $L$ is not.
Conversely, suppose that there exists $A\se V$ such that $A$ is full cut-rank but $A\cup \{a\}$ is not (obviously, $a \notin A$). According to \cref{prop:cutrank_prop2}, any set $B \se A$ is full cut-rank. Among all sets $B \se A$ such that $B \cup \{a\}$ is not full cut-rank (there is at least one: $A$), take one of them that is minimal by inclusion. Call it $C$. Let us show that $C \cup \{a\}$ is a minimal local set. By hypothesis, $C \cup \{a\}$ is not full cut-rank. Every subset of $C$ is full cut-rank, because such a set is also a subset of $A$. Also, every subset of $C \cup \{a\}$ containing $a$ is full cut-rank, by minimality of $C$. So, according to \cref{prop:characMLS},  $C \cup \{a\}$ is a minimal local set. So $a$ is contained in some minimal local set.
\end{proof}

Notice that a set is full cut-rank when the cut-rank function  is locally strictly increasing in the following sense:

\begin{lemma}
    \label{lemma:localdec}Given a graph $G=(V,E)$, 
   $C \se V$ is full cut-rank if and only if $\forall a\in C$, $\cutrk(C\sm \{a\}) < \cutrk(C)$.
\end{lemma}
\begin{proof}
$(\Rightarrow)$ follows from \cref{prop:cutrank_prop2}.
The proof of $(\Leftarrow)$ is by induction on the size of $C$. The property is obviously true when $|C|\ls 1$. 
Assume $|C|\gs 2$, for any $b\in C$ and any $a\in C\sm \{b\}$, using {submodularity},
\begin{equation*}
    \begin{split}
        \cutrk((C\sm \{a\}) \cup (C\sm \{b\}))\\ +~\cutrk((C\sm \{a\}) \cap (C\sm \{b\})) & \ls \cutrk(C\sm \{a\}) + \cutrk(C\sm \{b\}) \\
        \cutrk(C) + \cutrk(C\sm \{a,b\}) & <  \cutrk(C) + \cutrk(C\sm \{b\})\\
        \cutrk((C\sm \{b\})\sm\{a\}) & <  \cutrk(C\sm \{b\})  
    \end{split}
\end{equation*}
 By induction hypothesis, $C\sm \{b\}$ is full cut-rank, so $C$ is also full cut-rank as $\cutrk(C)>\cutrk(C\sm \{b\})=|C\sm \{b\}| = |C|-1$.
\end{proof}

When a set is not full cut-rank, one can find a sequence of nested subsets with a larger cut-rank, leading to a full-cut-rank subset:

\begin{lemma} \label{lemma:exists_full}
    Given a graph $G=(V,E)$,  for any $A\se V$, there exist $B_0\varsubsetneq \ldots \varsubsetneq B_{|A|-\cutrk(A)-1} \!\varsubsetneq B_{|A|-\cutrk(A)} = A$ such that $\forall i$, $|B_i| = \cutrk(A)+i$ and $\cutrk(B_i)\gs \cutrk(A)$. In particular, $B_0$ is full cut-rank.
\end{lemma}

\begin{proof}
    Note this makes sense because of the {linear boundedness} of $\cutrk$, which ensures that $|A|-\cutrk(A) \gs 0$. It is enough to show that for any $B\se A$, if $\cutrk(B)\gs \cutrk(A)$ and $|B| >\cutrk(A)$ then  $\exists a\in B$ such that $\cutrk(B\sm \{a\})\gs \cutrk(A)$. There are two cases: $(i)$ if $B$ is full cut-rank then according to \cref{prop:cutrank_prop2} for any $a\in B$, $B\sm\{a\}$ is also full cut-rank so $\cutrk(B\sm\{a\}) = |B|-1\gs \cutrk(A)$; $(ii)$ if $B$ is not full cut-rank then according to \cref{lemma:localdec}, $\exists a\in B$ such that $\cutrk(B\sm \{a\}) \gs \cutrk(B)$, so $\cutrk(B\sm \{a\}) \gs \cutrk(A)$.
\end{proof}

An interesting consequence of \cref{lemma:exists_full} is that for any full-cut-rank set, there exists a disjoint full-cut-rank set of same cardinality.

\begin{corollary}
    \label{cor:full_avoid}Given a graph $G=(V,E)$, 
    for any full-cut-rank set $A \se V$, $\exists B \in {V \sm A \choose |A|}$ that is full cut-rank\footnote{Given a set $K$ and an integer $k$, ${K \choose k}$ refers to $\{B \se K ~|~|B|=k\}$.}.
\end{corollary}

\begin{proof}
Let $A \subseteq V$ be a full-cut-rank set, i.e. $\cutrk(A) = |A|$. Using {symmetry}, $\cutrk(V \sm A) = |A|$. Then applying \cref{lemma:exists_full} on $V \sm A$ yields the result.
\end{proof}

We are now ready to conclude our reasoning and prove that any vertex $a$ of an arbitrary graph $G$ is contained in some minimal local set.
Let $r$ be the maximal cardinality of a full-cut-rank set in $V$, i.e. $r=\max\{|A|~|~\cutrk(A)=|A|\}$, and let $A\se V$ be a full-cut-rank set of size $r$. For any vertex $a\in V$:
\begin{itemize}
\item if $a\notin A$ then by maximality of $r$, $A\cup\{a\}$ is not full cut-rank, thus by \cref{lemma:characinMLS}, $a$ is contained in a minimal local set;
\item if $a\in A$, then according to \cref{cor:full_avoid}, there exists a full-cut-rank set $B\in {V \sm A \choose r}$. By maximality of $r$, $B\cup\{a\}$ is not full cut-rank, so, according to \cref{lemma:characinMLS}, $a$ is contained in a minimal local set. 
\end{itemize}

This concludes the proof of \cref{thm:MLS_cover}.

\subsection{An efficient algorithm to find an MLS cover} \label{sec:algo_mls}

In this section, we turn \cref{thm:MLS_cover} into a polynomial-time algorithm that generates an MLS cover for any given input graph.

\begin{theorem}\label{thm:mls_cover_algo}
    There exists an algorithm that finds a MLS cover in $O(n^4)$ evaluations of the cut-rank function, i.e. with runtime $O(n^{6.38})$, where $n$ is the order of the graph.
\end{theorem}

\subparagraph{Description of the algorithm.}
%%%%%
The core of the algorithm consists, given a graph $G$ and a vertex $a$, in producing a minimal local set that contains $a$. One can then iterate on uncovered vertices to produce an MLS cover of the graph. 

Given a vertex $a$, a minimal local set that contains $a$ is produced as follows, using essentially two stages:

($i$) First, the algorithm produces a full-cut-rank set $A$ such that $A \cup\{a\}$ is not full cut-rank. The procedure consists in starting with an empty set $A$ -- which is full cut-rank -- and then increasing the size of $A$, in a full-cut-rank preserving manner, until $A\cup \{a\}$ is not full cut-rank. To increase the size of $A$, notice that  if $A\cup \{a\}$ is full cut-rank then,  
according to \cref{lemma:exists_full}, there exists a disjoint set $A'$ -- so in particular $a\notin A'$ -- such that $|A'|=|A|+1$ and $A'$ is full cut-rank. The proof of \cref{lemma:exists_full} is constructive: starting from $C= V\setminus (A\cup \{a\})$ some vertices are removed  one-by-one  from $C$ to get the set $A'$. To produce such a set $A'$ from $C$, there are $O(n)$ vertices to remove, at each step there are $O(n)$ candidates, and deciding whether a vertex can be removed  costs a constant number  of evaluations of the cut-rank function. 

($ii$) Given a full-cut-rank set $A$ such that $A\cup \{a\}$ is not full cut-rank, \cref{lemma:characinMLS} guarantees the existence of a minimal local set that contains $a$. Notice that the proof of \cref{lemma:characinMLS} consists in finding, among all  subsets $B$ of $A$ such that $B$ is full cut-rank but $B\cup \{a\}$ is not, one that is minimal by inclusion. To do so, one can start from $A$ and remove one-by-one vertices from $A$ until reaching a set that is minimal by inclusion. It takes $O(n)$ steps\footnote{We only need to consider each vertex at most once. Indeed, at each step, if $A \sm \{b\}$ is full cut-rank, for any $B \se A$, $B \sm \{b\}$ is also full cut-rank by \cref{prop:cutrank_prop2}.}, each step involving a constant number of evaluations of the cut-rank function. 

\subparagraph{Complexity.} Stage $(i)$ uses $O(n^3)$ evaluations of the cut-rank function, and stage $(ii)$ only $O(n)$, so the overall algorithm uses $O(n^4)$ evaluations of the cut-rank function. Using Gaussian elimination, the cut-rank can be computed in $O(n^\omega)$ field operations \cite{Bunch1974,Ibarra198245} where $\omega < 2.38$.

A pseudocode of the algorithm is provided in Figure \ref{fig:pseudo_code_mls}.

\begin{figure}[h!]
\begin{algorithm}[H]
    \caption{Construct an MLS cover}\label{alg:MLS_cover}
    \SetKwInOut{Input}{Input}
    \SetKwInOut{Output}{Output}
    \Input{A graph $G=(V,E)$.}
    \Output{An MLS Cover $\mathcal{M}$ for $G$.}
    $\mathcal{M} \gets \{\}$\;
    \While{some vertex $a \in V$ is not covered, i.e. not in a minimal local set in $\mathcal{M}$}{
        $A \gets \emptyset$\;    
        \While{$A \cup \{a\}$ is full cut-rank}{
            $B \gets V \sm (A\cup\{a\})$\;
            \While{$|B| > |A|+1$}{
                Find $b \in B$ such that $\cutrk(B \sm \{b\}) \gs |A|+1$\;
                $B \gets B \sm \{b\}$\;
            }
            $A \gets B$\;
        }
        \While{$\exists b \in A$ such that $(A \sm \{b\})\cup\{a\}$ is not full cut-rank}{
            $A \gets A \sm \{b\}$\;
        }
        $\mathcal{M} \gets \mathcal{M}\cup\{A \cup \{a\}\}$\;

    }
\end{algorithm}
\caption{Pseudocode of the efficient algorithm to construct an MLS cover.}
\label{fig:pseudo_code_mls}
\end{figure}

\section{Extensions}

It has to be noted that the existence of an MLS cover and the fact that one can be computed efficiently remains true for any notion equivalent to minimal local sets defined with a function that shares some properties of the cut-rank function. Precisely, given a set $V$, and a function $\mu:2^V\to \mathbb{N}$ that satisfies symmetry, linear boundedness, and submodularity, define a $\mu$-minimal local set $L\se V$ as a set that satisfies $\mu(L) \ls |L| - 1$ and $\forall B \varsubsetneq L$, $\mu(B) = |B|$. Then $V$ is covered by its $\mu$-minimal local sets, and a "$\mu$-MLS cover" can be computed efficiently (given that the values of $\mu$ can be computed efficiently). Indeed, the proof of \cref{thm:MLS_cover} (well as the proof of \cref{prop:cutrank_prop2}, which is used in the proof of the theorem) uses only the fact that $\cutrk:2^V\to \mathbb{N}$ satisfies symmetry, linear boundedness, and submodularity.

An example of a function satisfying such properties is the connectivity function of a matroid. Given a matroid $(E,r)$ where $E$ is the ground set and $r$ is the rank function, the connectivity function is the function $\lambda$ that maps any set $X \subseteq E$ to $\lambda(X) = r(X)+r(E\sm X)-r(E)$. Another example is the generalization of the cut-rank function to multigraphs, as we review below.

While (simple, undirected) graphs correspond to quantum qubit graphs states, their natural higher dimension extension are multigraphs, which correspond to quantum qudit graph states \cite{Beigi2006,Ketkar2006}. A quick introduction to multigraphs and qudit graph states can be found in \cite{marin2013}. 

\begin{definition}[$q$-multigraphs]
    Given a prime number $q$, a $q$-multigraph $G$ is a pair $(V, \Gamma)$ where $V$ is the set of vertices and $\Gamma: V \times V \longrightarrow \mathbb{F}_q$ is the adjacency matrix of $G$: for any $u,v \in G$, $\Gamma(u,v)$ is the multiplicity of the edge $(u,v)$ in $G$.
\end{definition}

Here we consider undirected simple $q$-multigraphs, i.e. $\forall u,v \in V$, $\Gamma(u,v) = \Gamma(v,u)$ and $\Gamma(u,u)=0$.
The cut-rank function on $q$-multigraphs is defined similar to graphs (which correspond to $2$-multigraphs): 

\begin{definition}[Cut-rank function on $q$-multigraphs]
    Let $G=(V, \Gamma)$ be a $q$-multigraph. For $A \se V$, let the cut-matrix $\Gamma_A = ((\Gamma_A)_{ab}: a \in A\text{, } b \in V\sm A)$ be the matrix with coefficients in $\mathbb{F}_q$ such that $\Gamma_{ab} = \Gamma(a,b)$. The cut-rank function of $G$ is defined as
    \begin{align*}
        \cutrk\colon 2^V & \longrightarrow \mathbb{N}\\
        A &\longmapsto \textbf{rank}_{\mathbb{F}_q}(\Gamma_A)
    \end{align*}
\end{definition}

It is not hard to see that the cut-rank function on $q$-multigraphs satisfies the properties of \cref{prop:cutrank_prop}, namely {symmetry}, {linear boundedness}, and {submodularity}. 
To extend minimal local sets to $q$-multigraphs, the easier is to use their characterization in terms of the cut-rank function (see \cref{prop:characMLS}).

\begin{definition}[Minimal local sets and MLS covers on $q$-multigraphs]
    Given $G=(V, \Gamma)$ a $q$-multigraph:
    \begin{itemize}
        \item    $A \se V$ is a \emph{minimal local} set if $A$ is a non-full-cut-rank set whose proper subsets are all full cut-rank, i.e. $\forall a \in A, \cutrk(A) \ls \cutrk(A \sm \{a\}) = |A|-1$;
        \item    $\mathcal M\se 2^V$ is an \emph{MLS cover} if $\forall L\in \mathcal M$, $L$ is a minimal local set of $G$, and $\forall u\in V$, $\exists L \in \mathcal M$ such that $u\in L$.
    \end{itemize}
\end{definition}

Local complementation naturally extends to $q$-multigraphs (see \cite{marin2013} for a definition), and local complementation on a $q$-multigraph leaves the cut-rank function invariant \cite{Kante2007}. This implies that again, in the case of $q$-multigraphs, minimal local sets are invariant under local complementation.
The existence of a MLS cover also extends to $q$-multigraphs.

\begin{theorem}
    Any q-multigraph has an MLS cover. 
\end{theorem}
\chapter{A graphical characterization of LU-equivalence}

\label{chap:glc}

In this chapter, we derive a generalization of local complementation, called $r$-local complementation. $r$-local complementation is parametrized by an integer $r>0$, called the level of the generalized local complementation. We show that $r$-local complementation exactly captures the LU-equivalence of graphs, the same way local complementation captures the LC-equivalence of graphs.  To prove it, we use minimal local sets to give types to each vertex of any graph, then use these types to define a standard form for graphs. The LU-equivalence of graphs in standard form is then easier to study. $r$-local complementation helps better understand the full picture of the gap between LC- and LU-equivalence: we introduce intermediate equivalences, called LC$_r$-equivalences, related to the levels of the generalized complementation, and show that the LC$_r$-equivalences form an infinite strict hierarchy. More precisely, for any $r \gs 2$, there exists a pair of graphs that are LC$_r$-equivalent but not LC$_{r-1}$-equivalent, generalizing the known counterexamples to the LU=LC conjecture.

\section{\texorpdfstring{LC$_r$}{LCr}-equivalence}

As an intermediate local equivalence between LC- and LU-equivalence, we define the LC$_r$-equivalence, where $r>0$ is an integer, as a generalization of LC-equivalence.

\begin{definition}
    Two quantum states $\ket{\psi_1}$ and $\ket{\psi_2}$ are LC$_r$-equivalent if there exist single-qubit unitaries gates $U_u$, that are a product of finitely many Hadamard gates $H$ and Z-rotations $Z(\pi/2^r)$, such that $\ket{\psi_2} = e^{i\phi}\bigotimes_{u\in V}U_u \ket{\psi_1}$. 
\end{definition}

For convenience, we lift the definition of LC$_r$-equivalence to graphs: two graphs are said LC$_r$-equivalent if the two corresponding graph states are LC$_r$-equivalent. We use the notation $G_1 =_{LC_r} G_2$.

Note that LC$_1$-equivalence is nothing but LC-equivalence. LC$_r$-equivalence is related to the levels of the so-called single-qubit Clifford hierarchy. The single-qubit Clifford hierarchy is defined inductively: the set $\mathcal C_1$ is the Pauli group $\{\pm 1, \pm i\} \times \{I,X,Y,Z\}$, and $\mathcal C_{k+1} = \{U ~|~ U \mathcal C_1 U^\dagger \se \mathcal C_k\}$. Single-qubit unitary gates in the level $r+1$ (where $r>0$)  of the Clifford hierarchy are actually exactly those that can be written as a product $e^{i\phi} C ~ Z\!\left(\frac{m\pi}{2^{r}}\right)  C'$, where $C, C'$ are Clifford gates and $m$ is an integer (see for example \cite{Zeng08,Cui2016}).

\begin{proposition}\label{prop:lcr_equals_clifford_hierarchy}
    $G_1 =_{LC_r} G_2$
    if and only if there exist single-qubit unitaries $U_u \in \mathcal C_{r+1}$ in the level $r+1$ of the Clifford hierarchy, such that $\ket{G_2} = e^{i\phi}\bigotimes_{u\in V}U_u \ket{G_1}$.
\end{proposition}

\begin{proof}
    Suppose $G_1 =_{LC_r} G_2$. 
    According to \cref{prop:czc}, $\ket{G_2} = e^{i\phi}\bigotimes_{u\in V}C_{u} Z(\theta_u) C'_{u} \ket{G_1}$, where $C_{u}, C'_{u}$ are single-qubit Clifford gates, and every $C_{u} Z(\theta_u) C'_{u}$ is a product of finitely many $H$ and $Z(\pi/2^r)$ up to global phase. $H$ and $Z(\pi/2)$ generate single-qubit Clifford gates hence,  
    $Z\left(\theta_u\right)$ itself is a product of finitely many $H$ and $Z(\pi/2^r)$ up to global phase.
    This occurs only when $\theta_u = 0 \mod \pi/2^{r}$ (see for example \cite{amy2024exact}). Conversely, it is clear that single-qubit unitary gates in the level $r+1$ of the Clifford hierarchy are a product of finitely many $H$ and $Z(\pi/2^r)$ up to global phase.    
\end{proof}

\section{Weighted hypergraph states}

Weighted hypergraph states are not studied directly in this thesis, however they provide a natural framework to understand the forthcoming definition of the generalized local complementation over graph states.

While graph states are in one-to-one correspondence with graphs, the so-called weighted hypergraph states are in one-to-one correspondence with weighted hypergraphs. Thus, we begin by giving the definition of weighted hypergraphs. In our setting, a weighted hypergraph is nothing but a hypergraph with weights in $[0, 2\pi)$ on the hyperedges. Formally, a hypergraph H is composed of a set $V$ of vertices, and of a function $E\colon 2^V  \longrightarrow [0, 2\pi)$ giving the weights of the hyperedges. By convention, we set $E(\emptyset)=0$. Under this formalism, a graph is associated with a hypergraph $(V,E)$ where $E(K)=0$ or $\pi$ when $|K|=2$ and $E(K)=0$ when $|K|\neq 2$. 
Weighted graph states were originally invented as the so-called LME (locally maximally entangleable) states.  
More precisely, every weighted hypergraph state is LME, and a state is LME if and only if it is LU-equivalent to a weighted hypergraph state.

The construction of a graph state starts with a tensor product of qubits in the state $\ket +$ and is followed by the application of some $CZ$ gates. Similarly, the construction of a weighted hypergraph state also starts with a tensor product of qubits in the state $\ket +$. It is followed by the applications of a generalization of the $CZ$ gates. While $CZ_{uv}= I_{uv} - 2\ket{11}\bra{11} $, we introduce the multi-controlled phase gate $CZ_K(\theta) = I_K - (1 - e^{i\theta})\ket{11...1}\bra{11...1}$ applied on some set of vertices $K \se V$. Like for $CZ$, the $CZ(\theta)$ gates are symmetric and commute. Notice that if $K=\{u\}$, then $CZ_K(\theta) = Z(\theta)_u$.

\begin{definition}
    Let $H = (V,E)$ be a weighted hypergraph of order $n$. The corresponding weighted hypergraph state $\ket H$ is the $n$-qubit state: $$\ket H = \left(\prod_{\emptyset \neq K \se V} CZ_{K}(E(K))\right) \ket{+}_V$$
\end{definition}

\begin{proposition} \label{prop:whs_angles}
    Weighted hypergraph states are exactly those quantum states that can be written as:
    $$\ket H = \frac 1{\sqrt {2^n}}\sum_{x\in \{0,1\}^n}e^{i\theta_x}\ket x$$ where $\theta_x \in [0, 2\pi)$ and $\theta_0 = 0$. The angles are given by the equation $\theta_x \! =\! \sum_{K \se \{u\in V \!~|~\! x_u = 1\}}\! E(K) \mod 2\pi$. Conversely, the weights can be recovered from the angles $\theta_x$.
\end{proposition}

\begin{proof}
    As $\ket +_V = \frac 1{\sqrt {2^n}}\sum_{x \in {0,1}^n}\ket x$, and as the $CZ_K$ can only change the phases, it is obvious that a weighted hypergraph state can be written as $\frac 1{\sqrt {2^n}}\sum_{x\in \{0,1\}^n}e^{i\theta_x}\ket x$.

    Conversely, the weights of the multi-controlled phase gates can be computed from the angles $\theta_x$: for every $K \se V$, let $ \omega_K = \theta_{x^{K}} - \sum_{K' \varsubsetneq K} \omega_{K'} \mod 2\pi$
    where $x^{K}_u=1$ if and only if $u \in K$. Then,
    \begin{equation*}    
    \ket H = \left(\prod_{\emptyset \neq K \se V} CZ_{K}(\omega_K)\right) \ket{+}_V \qedhere
    \end{equation*}
\end{proof}

\cref{prop:whs_angles} implies in particular that weighted hypergraph states are uniquely defined by the weights of the multi-controlled phase gates.

\begin{corollary}
    Weighted hypergraph states are in one-to-one correspondence with weighted hypergraphs. In other words, $\ket{H_1} = \ket{H_2}$ if and only if $H_1 = H_2$.
\end{corollary}

\section{The action of X-rotations over graph states}

In general, an X-rotation over a qubit does not map graph states to graph states. However, it always maps graph states to weighted hypergraph states.

\begin{proposition}[\cite{Tsimakuridze17}] \label{prop:action_X_rotation}
    If $G=(V,E)$ is a graph, for any vertex $u \in V$, $X(\theta)_u\ket G = \left(\prod_{K\subseteq N_G(u), K \neq \emptyset} CZ_{K}\left((-2)^{|K|-1}\theta\right)\right)\ket G = \ket H$, where $H$ is the weighted hypergraph obtained from $G$ by adding (modulo $2\pi$) a weight $(-2)^{k-1}\theta$ to each hyperedge containing exactly $k$ vertices adjacent to $u$. 
\end{proposition}

\cref{prop:action_X_rotation} is illustrated in Figure \ref{fig:Xrotation}.

\begin{figure}[h!]
\centering
\vspace{-50pt}
\scalebox{1}{
\begin{tikzpicture}[yscale = 0.5, xscale = 0.5]
\begin{scope}[every node/.style=vertices]
    \node (U1) at (0,0) {1};
    \node (U2) at (6,3) {2};
    \node (U3) at (10,0) {3};
    \node (U4) at (6,-3) {4};    
\end{scope}
\begin{scope}[every edge/.style=edges]              
    \path [-] (U1) edge node {} (U2);
    \path [-] (U1) edge node {} (U3);
    \path [-] (U1) edge node {} (U4);
\end{scope}

\draw [-stealth, label=above:a, very thick](12.5,0) -- (15.5,0);

\draw (-0.5,1.3) node(){$X(\frac \pi 4)$};

\begin{scope}[shift={(18,0)}, every node/.style=vertices]
    \node (U1) at (0,0) {1};
    \node (U2) at (6,3) {2};
    \node (U3) at (10,0) {3};
    \node (U4) at (6,-3) {4};  
    
    \draw[colormls2, thick, loosely dashed] (6,3) circle (0.8);
    \draw[colormls2, thick, loosely dashed] (10,0) circle (0.8);
    \draw[colormls2, thick, loosely dashed] (6,-3) circle (0.8);

    \draw [rounded corners=30mm,black] (4,8)--(14,0)--(4,-8)--cycle;    
\end{scope}
\begin{scope}[shift={(18,0)}]
    \draw [black] (9.8,4) node(){$\pi$};

    \draw [colormls3] (5,1) node(){$-\frac \pi 2$};
    \draw [colormls3] (8.5,1.9) node(){$-\frac \pi 2$};
    \draw [colormls3] (8.5,-1.9) node(){$-\frac \pi 2$};

    \draw [colormls2] (5,4) node(){$\frac \pi 4$};
    \draw [colormls2] (5,-4) node(){$\frac \pi 4$};
    \draw [colormls2] (10.5,1.3) node(){$\frac \pi 4$};
\end{scope}
\begin{scope}[every edge/.style=edges]              
    \path [-] (U1) edge node {} (U2);
    \path [-] (U1) edge node {} (U3);
    \path [-] (U1) edge node {} (U4);    
\end{scope}
\begin{scope}[every edge/.style={draw=colormls3, ultra thick,  dashed}]              
    \path [-] (U2) edge node {} (U3);
    \path [-] (U2) edge node {} (U4);
    \path [-] (U3) edge node {} (U4);    
\end{scope}
\end{tikzpicture}
}
\vspace{-60pt}
\caption{Illustration of the action of an X-rotation over a vertex of a graph state. Applying an X-rotation of angle $\pi/4$ on vertex 1 maps the graph state to a weighted hypergraph state. More precisely, the X-rotations produces one 3-vertex edge $\{2,3,4\}$ of weight $\pi$, three 2-vertex edges $\{2,3\}$, $\{2,4\}$ and $\{3,4\}$ of weight $-\pi/2$, and three 1-vertex edges $\{2\}$, $\{3\}$ and $\{4\}$ of weight $\pi/4$.}
\label{fig:Xrotation}
\end{figure}
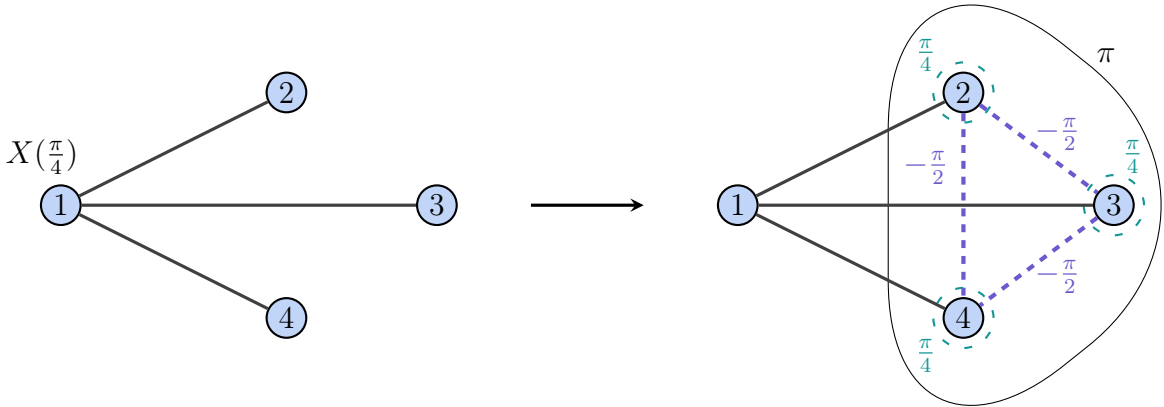

\cref{prop:action_X_rotation} can be generalized to independent sets of vertices. Recall that $\Lambda_G^K$ denotes the common neighborhood of $K$ in $G$.

\begin{corollary} \label{cor:action_X_rotation_indep_set}
    Let $A$ be an independent set of vertices in $G=(V,E)$. Then, 
    $$ \bigotimes_{u\in A}X(\theta_u)\ket G = \left(\prod_{K\subseteq V, K \neq \emptyset} CZ_{K}\left((-2)^{|K|-1}\sum_{u \in A \cap \Lambda_G^K} \theta_u\right)\right)\ket G$$ 
\end{corollary}

The condition of independence is necessary to obtain a weighted hypergraph state. Indeed, X-rotations on adjacent vertices may map graph states (or weighted hypergraph states) to LME states that are not weighted hypergraph states. For example, 
\begin{align*}
    \left(X(\pi/2) \otimes X(\pi/2) \right) \ket{K_2} & = \left(X(\pi/2) \otimes X(\pi/2) \right) \frac{1}{2} \left( \ket{00} + \ket{00} +\ket{00} - \ket{11}\right)\\
    & = e^{i\pi/4}\ket{00} + e^{-i\pi/4}\ket{11}
\end{align*}

which is not a weighted hypergraph state.

\section{Generalizing local complementation}

To ease the definition of the generalized local complementation, we begin by introducing 1-local complementation then 2-local complementation.

\subsection{1-local complementation}

As a first step towards a generalization of  local complementation, we introduce a natural shortcut $G  \star^1 S:=G\star a_1\ldots \star a_k$ to denote the graph obtained after a series of local complementations when $S=\{a_1, \ldots, a_k\}$ is an independent set in $G$. The notation $\star_1$ instead of $\star$ will make sense in the following sections when we generalize the notation to other integers. 
The independence condition  is important as local complementations do not commute in general when applied on adjacent vertices.

Notice that the action of a local complementation over $S$ can be described directly as toggling edges between vertices with an odd number of common neighbors in $S$:  
\begin{equation*}\label{eq:LCIS}u\sim_{G\star S} v ~\Leftrightarrow~\left(u\sim_{G} v ~~\oplus~~ |S\cap N_G(u)\cap N_G(v)| = 1\bmod 2\right)\end{equation*}

1-local complementation is illustrated in Figure \ref{fig:1lc}.

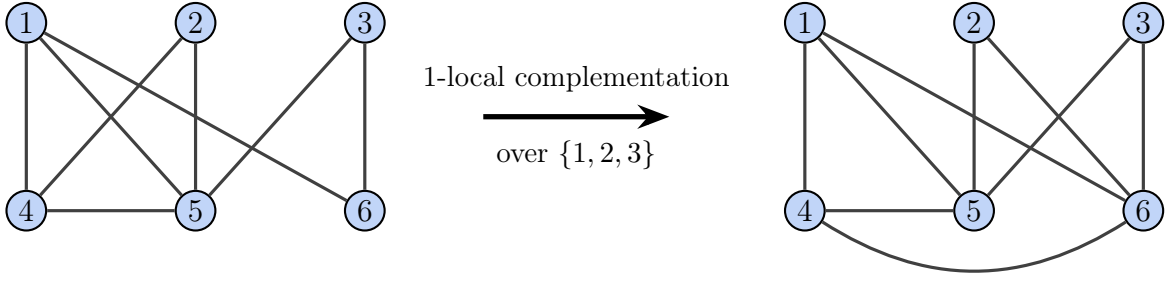
\begin{figure}[h!]
\centering
\begin{tikzpicture}[xscale = 0.45 ,yscale=0.5]

\begin{scope}[every node/.style=vertices]
    \node (U1) at (-5,5) {1};
    \node (U2) at (0,5) {2};
    \node (U3) at (5,5) {3};
    \node (U4) at (-5,0) {4};
    \node (U5) at (0,0) {5};
    \node (U6) at (5,0) {6};
\end{scope}
\begin{scope}[every node/.style={},
                every edge/.style=edges]                   
    \path [-] (U1) edge node {} (U4);
    \path [-] (U1) edge node {} (U5);
    \path [-] (U1) edge node {} (U6);
    \path [-] (U2) edge node {} (U4);
    \path [-] (U2) edge node {} (U5);
    \path [-] (U3) edge node {} (U5);
    \path [-] (U3) edge node {} (U6);

    \path [-] (U4) edge node {} (U5);
\end{scope}

\begin{scope}[shift={(23,0)},every node/.style=vertices]
    \node (U1) at (-5,5) {1};
    \node (U2) at (0,5) {2};
    \node (U3) at (5,5) {3};
    \node (U4) at (-5,0) {4};
    \node (U5) at (0,0) {5};
    \node (U6) at (5,0) {6};
\end{scope}
\begin{scope}[every edge/.style=edges]                   
    \path [-] (U1) edge node {} (U4);
    \path [-] (U1) edge node {} (U5);
    \path [-] (U1) edge node {} (U6);
    \path [-] (U2) edge node {} (U5);
    \path [-] (U2) edge node {} (U6);
    \path [-] (U3) edge node {} (U5);
    \path [-] (U3) edge node {} (U6);

    \path [-] (U4) edge[bend right] node {} (U6);
    \path [-] (U4) edge node {} (U5);

\end{scope}

\draw[-Stealth, line width=2pt] (8.5,2.5) -- (14,2.5);
\node (t) at (11.25,3.5) {\small 1-local complementation};
\node (t2) at (11.25,1.5) {\small  over $\{1,2,3\}$};

\end{tikzpicture}
\caption{Illustration on a 1-local complementation over the set $S = \{1,2,3\}$. The effect on the graph is the same as local complementations on 1, 2 and 3 in any order.}
\label{fig:1lc}    
\end{figure}

As a local complementation can be implemented with local Clifford operators (see \cref{prop:implementation_lc}), so can a 1-local complementation:

$$\ket{G\star^1 S} = \bigotimes_{u\in S}X\left(\frac {\pi}{2}\right)\bigotimes_{v\in V}Z\left(-\frac {\pi}{2} |S \cap N_G(v)| \right)\ket{G}$$

\subsection{2-local complementation}

\subsubsection{2-local complementation over an independent set}
We introduce 2-local complementation as a refinement of \emph{idempotent} 1-local complementation,~i.e.~when $G\star^1 S=G$. 
An idempotent local complementation occurs when each pair $(u,v)$ of (distinct) vertices has an even number of common neighbors in the independent set $S$, one can then consider a new graph transformation consisting in toggling an edge $(u,v)$ when its number of common neighbors in $S$ is equal to 2 modulo 4:

$$u\sim_{G\star^2 S} v ~\Leftrightarrow~\left(u\sim_{G} v \oplus |N_G(u)\cap N_G(v) \cap S|= 2\bmod 4\right)$$

However, such an action may not be implementable using local operations in the graph state formalism. To guarantee an implementation by means of local unitary transformations on the corresponding graph states, we add the condition, called $2$-incidence, that any set of 2 or 3 (distinct) vertices has an even number of common neighbors in $S$. For instance, in the following example, a 2-local complementation over $\{1,2\}$ performs the following transformation:

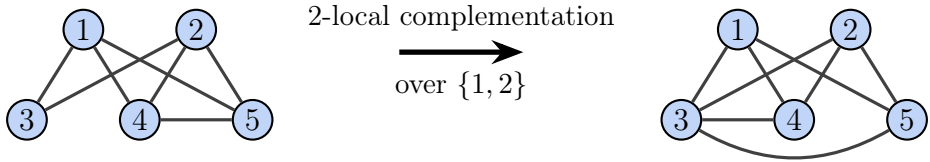
\[
\begin{tikzpicture}[scale = 0.3]

\begin{scope}[every node/.style=vertices]
    \node (a) at (-2.5,4) {1};
    \node (b) at (2.5,4) {2};
    \node (c) at (-5,0) {3};
    \node (d) at (0,0) {4};
    \node (e) at (5,0) {5};
\end{scope}
\begin{scope}[every edge/.style=edges]                   
    \path [-] (a) edge node {} (c);
    \path [-] (a) edge node {} (d);
    \path [-] (a) edge node {} (e);
    \path [-] (b) edge node {} (c);
    \path [-] (b) edge node {} (d);
    \path [-] (b) edge node {} (e);
    \path [-] (d) edge node {} (e);
\end{scope}

\begin{scope}[shift={(29,0)},every node/.style=vertices]
    \node (a) at (-2.5,4) {1};
    \node (b) at (2.5,4) {2};
    \node (c) at (-5,0) {3};
    \node (d) at (0,0) {4};
    \node (e) at (5,0) {5};
\end{scope}
\begin{scope}[every edge/.style=edges]                   
    \path [-] (a) edge node {} (c);
    \path [-] (a) edge node {} (d);
    \path [-] (a) edge node {} (e);
    \path [-] (b) edge node {} (c);
    \path [-] (b) edge node {} (d);
    \path [-] (b) edge node {} (e);
    \path [-] (d) edge node {} (c);
     \path [-] (c) edge [bend right]node {} (e);
\end{scope}

\draw[-Stealth, line width=2pt] (11.5,3) -- (17,3);
\node (t) at (14.25,4.5) {\small 2-local complementation};
\node (t2) at (14.25,1.5) {\small  over $\{1,2\}$};   

\end{tikzpicture}\]
In the left-hand side graph $G$, $S=\{1,2\}$ is an independent set. Moreover, a $1$-local complementation over $S$ is idempotent: $G\star^1 S = G$, as 1 and 2 are twins. $S$ is also $2$-incident as any pair and triplet of vertices has an even number of common neighbors in $S$. Thus, a $2$-local complementation over $S$ is valid and consists in toggling the edges $(3,4)$, $(3,5)$, and $(4,5)$ as each of them has a number of common neighbors in $S$ equal to $2 \bmod 4$. 

\subsubsection{2-local complementation over an independent multiset}

For generality, we consider that $S$ is actually a multiset, more precisely an independent multiset (i.e. no vertices counted once or more in $S$ are adjacent). When $S$ is a multiset, the number of common neighbors in $S$ should be counted with their multiplicity. With a slight abuse of notation we identify multisets of vertices with their multiplicity function $V\to \mathbb N$. \footnote{Hence, we also identify sets of vertices with their indicator functions $V\to \{0,1\}$.} We also introduce the shortcut $S\bullet \Lambda_G^K = \sum_{u \in \Lambda_G^K}S(u)$.\footnote{$.\bullet.$ is the scalar product: $A\bullet B = \sum_{u\in V}A(u).B(u)$ and $\Lambda_G^K$ is the common neighborhood of $K$: $\Lambda_G^K=\bigcap_{u\in K}N_G(u)=\{v\in V~|~\forall u \in K, v\in N_G(u)\}$.} For example, $S\bullet \Lambda_G^{\{a,b\}} = \sum_{u \in N_G(a)\cap N_G(b)}S(u)$ and $S\bullet \Lambda_G^{\{a,b,c\}} = \sum_{u \in N_G(a)\cap N_G(b) \cap N_G(c)}S(u)$. An independent multiset is thus 2-incident when both $S\bullet \Lambda_G^{\{a,b\}}$ and $S\bullet \Lambda_G^{\{a,b,c\}}$ are even for any (distinct) vertices\footnote{It suffices to check for vertices not counted in $S$, as $S$ is independent. Indeed, $S\bullet \Lambda_G^{K}=0$ if $K$ contains a vertex $a$ such that $S(a) \gs 1$.} $a$, $b$ and $c$. For independent 2-incident multisets, the definition of 2-local complementation becomes:
$$u\sim_{G\star^2 S} v ~\Leftrightarrow~\left(u\sim_{G} v \oplus S \bullet\Lambda_G^{u,v}= 2\bmod 4\right)$$

A detailed example of a $2$-local complementation over an independent multiset is given in Figure \ref{fig:generalized_lc}. 

\begin{figure}[h!]
\centering
\begin{tikzpicture}[xscale = 0.45 ,yscale=0.5]

\begin{scope}[every node/.style=vertices]
    \node (U1) at (-5,5) {1};
    \node (U2) at (0,5) {2};
    \node (U3) at (5,5) {3};
    \node (U4) at (-5,0) {4};
    \node (U5) at (0,0) {5};
    \node (U6) at (5,0) {6};
\end{scope}
\begin{scope}[every node/.style={},
                every edge/.style=edges]                   
    \path [-] (U1) edge node {} (U4);
    \path [-] (U1) edge node {} (U5);
    \path [-] (U1) edge node {} (U6);
    \path [-] (U2) edge node {} (U5);
    \path [-] (U2) edge node {} (U6);
    \path [-] (U3) edge node {} (U5);
    \path [-] (U3) edge node {} (U6);

    \path [-] (U4) edge node {} (U5);
    \path [-] (U5) edge node {} (U6);
\end{scope}

\begin{scope}[shift={(23,0)},every node/.style=vertices]
    \node (U1) at (-5,5) {1};
    \node (U2) at (0,5) {2};
    \node (U3) at (5,5) {3};
    \node (U4) at (-5,0) {4};
    \node (U5) at (0,0) {5};
    \node (U6) at (5,0) {6};
\end{scope}
\begin{scope}[every node/.style={},
                every edge/.style=edges]                   
    \path [-] (U1) edge node {} (U4);
    \path [-] (U1) edge node {} (U5);
    \path [-] (U1) edge node {} (U6);
    \path [-] (U2) edge node {} (U5);
    \path [-] (U2) edge node {} (U6);
    \path [-] (U3) edge node {} (U5);
    \path [-] (U3) edge node {} (U6);

    \path [-] (U4) edge[bend right] node {} (U6);
    \path [-] (U5) edge node {} (U6);

\end{scope}

\draw[-Stealth, line width=2pt] (8.5,2.5) -- (14,2.5);
\node (t) at (11.25,3.5) {\small 2-local complementation};
\node (t2) at (11.25,1.5) {\small  over $\{1,1,2,3\}$};

\end{tikzpicture}
\caption{Illustration of a 2-local complementation over the multiset $S = \{1,1,2,3\}$. $S$ is 2-incident: indeed $S\bullet \Lambda_G^{\{4,5,6\}} = 2$, which is a multiple of $2^{2-1-0} = 2$. Similarly, $S\bullet \Lambda_G^{\{4,5\}} = S\bullet \Lambda_G^{\{4,6\}} = 2$ and $S\bullet \Lambda_G^{\{5,6\}} = 4$. Edges $(4,5)$ and $(4,6)$ are toggled as $S\bullet \Lambda_G^{\{4,5\}} = S\bullet \Lambda_G^{\{4,6\}} = 2\bmod 4$, but not edge $(5,6)$ as $S\bullet \Lambda_G^{\{5,6\}} = 0\bmod 4$.}
\label{fig:generalized_lc}    
\end{figure}
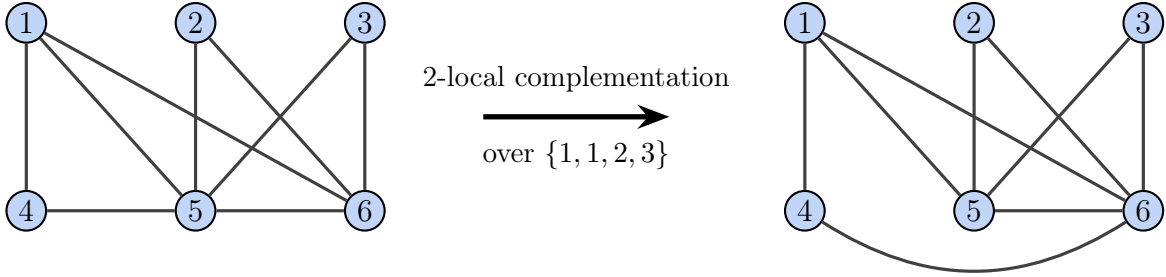

A 2-local complementation can be implemented on a graph state with local unitaries (a proof in the general case is provided in the next section):

$$\ket{G\star^2 S} = \bigotimes_{u\in S}X\left(\frac {S(u)\pi}{4}\right)\bigotimes_{v\in V}Z\left(-\frac {\pi}{4} \sum_{u \in N_G(v)}S(u)\right)\ket{G}$$

\subsubsection{Weighted hypergraph state formalism}

The weighted hypergraph state formalism is useful to understand that the condition of 2-incidence is actually a sufficient and necessary condition for this local unitary transformation to map a graph state to another graph state. According to \cref{prop:action_X_rotation}, an X-rotation of angle $\pi/4$ over a vertex of a graph state may create hyperedges, some of weight $\pi$ containing 3 vertices, some of weight $\pi/2$ (or $-\pi/2$) containing 2 vertices (and some containing one vertex, but they can be removed with Z-rotations). Applying an X-rotation of angle $\pi/4$ not only over one vertex, but over an independent (multi)set of vertices sometimes allow weights of hyperedges to cancel out. In particular, when the (multi)set is 2-incident, the weights of the hyperedges cancel out so that a 2-local complementation may only toggle edges, i.e. add weight $\pi$ to hyperedges containing 2 vertices (see Figure \ref{fig:generalized_lc}).

\subsubsection{27-vertex counterexample}

2-local complementation is useful to formally explain the 27-vertex counter-example to the LU=LC conjecture (see Figure \ref{fig:ce}). If $S$ is the set containing every upper vertex, one can check that $S$ is 2-incident: every pair of bottom vertices is adjacent to exactly $\binom{4}{3}+\binom{4}{2}=10$ vertices in $S$, and every triplet of bottom vertices is adjacent to exactly $\binom{3}{2}+\binom{3}{1}=6$ vertices in $S$. A 2-local complementation over $S$ is thus valid and maps the leftmost graph the rightmost graph, as $10 = 2 \bmod 4$.

\vspace{1em}

When a 2-local complementation is idempotent, one can similarly refine the 2-local complementation into a 3-local complementation, leading to the general definition of generalized local complementation provided in the next section.

\subsection{r-local complementation}

We introduce a generalization of local complementation, that we call $r$-local complementation, where $r$ is a (non-zero) positive integer. This generalized local complementation denoted $G\star^r S$ is parametrized by a (multi)set $S$, that has to be independent and also  $r$-\emph{incident}, which is the following counting condition on the number of common neighbors in $S$ of any set that does not intersect $\supp(S)$,  the support\footnote{The support $\supp(S)$ of $S$ is the set of vertices $u$ such that $S(u) \gs 1$.} of $S$: 

\begin{definition}[$r$-incidence]
Given a graph $G$, a multiset $S$ of vertices is  $r$-incident, if for any $k\in [0,r)$, and any $K\subseteq V\setminus \supp(S)$ of size $k+2$, $S\bullet \Lambda_G^K$ is a multiple of $2^{r-k-\delta(k)}$,
where $\delta$ is the Kronecker delta\footnote{$\delta(x)\in \{0,1\}$ and $\delta(x)=1 \Leftrightarrow x=0$.
} and $S\bullet \Lambda_G^K$ is the number of vertices of $S$, counted with their multiplicity, that are adjacent to all vertices of $K$. 
\end{definition}

\begin{remark}
    If $S$ is independent, the condition $K\subseteq V\setminus \supp(S)$ can be lifted from the definition. Indeed, if $K$ contains a vertex in $\supp(S)$, then $S\bullet \Lambda_G^K = 0$.
\end{remark}

\begin{definition}[$r$-local complementation]
Given a graph $G$ and an $r$-incident independent multiset $S$, let $G\star^rS$ be the graph defined as
\[u\sim_{G\star^r S} v ~\Leftrightarrow~\left(u\sim_{G} v ~~\oplus~~ S \bullet\Lambda_G^{u,v} = 2^{r-1}\bmod 2^{r}\right)\]
\end{definition}

$r$ is said to be the \emph{level} of the $r$-local complementation. We say that $G\star^r S$ is valid when $S$ is an $r$-incident independent multiset in $G$. The following proposition explains how an $r$-local complementation is implemented on a graph states with X- and Z-rotations, and motivates the $r$-incidence condition. As expected, this is a generalization of the implementation of the usual local complementation (see \cref{prop:implementation_lc}).

\begin{proposition}\label{prop:implementation_rlc}
Given a graph $G=(V,E)$ and an $r$-incident independent multiset $S$ of vertices, \begin{equation*}\label{eq:LCr}\ket{G\star^r S} = \bigotimes_{u\in V}X\left(\frac {S(u)\pi}{2^r}\right)\bigotimes_{v\in V}Z\left(-\frac {\pi}{2^r}\sum_{u \in N_G(v)}S(u)\right)\ket{G}\end{equation*}
In general, given an independent multiset $S$, $\bigotimes_{u\in V}X\!\left(\frac {S(u)\pi}{2^r}\right)\!\bigotimes_{v\in V}Z\!\left(\!-\frac {\pi}{2^r}\sum_{u \in N_G(v)}S(u)\!\right)\!\ket{G}$ is a graph state if and only if $S$ is $r$-incident.
\end{proposition}

\begin{proof}
    According to \cref{cor:action_X_rotation_indep_set},
    \begin{align*}
        &\bigotimes_{u\in V}X\left(\frac {S(u)\pi}{2^r}\right)\bigotimes_{v\in V}Z\left(-\frac {\pi}{2^r}\sum_{u \in N_G(v)}S(u)\right)\ket{G}\\
        &= \left(\prod_{K\subseteq V, K \neq \emptyset} CZ_{K}\left((-2)^{|K|-r-1}\pi ~ S\bullet \Lambda_G^K\right)\right) \bigotimes_{v\in V}Z\left(-\frac {\pi}{2^r}\sum_{u \in N_G(v)}S(u)\right)\ket{G}\\
        &= \left(\prod_{K\subseteq V, |K| \gs 2} CZ_{K}\left((-2)^{|K|-r-1}\pi ~S\bullet \Lambda_G^K\right)\right) \bigotimes_{v\in V}\ket{G}
    \end{align*}
    $S$ is $r$-incident if and only if for any $K\subseteq V$ such that $|K|\gs 2$, $S\bullet \Lambda_G^K$ is a multiple of $2^{r-|K|+2-\delta(|K|-2)}$, i.e. $(-2)^{|K|-r-1} ~S\bullet \Lambda_G^K$ is a multiple of $2^{r-|K|+2-\delta(|K|-2) + |K|-r-1} = 2^{1-\delta(|K|-2)}$. Thus, if $S$ is $r$-incident, only $CZ$ gates are applied on $\ket G$, and thus it remains a graph state. Else, $\ket G$ is mapped to a weighted hypergraph state that is not a graph state. 
\end{proof}

Below we show a few basic properties of generalized local complementations. First, it is easy to double-check that they are self inverse: $(G\star^r S)\star^r S =G$. Moreover, $r$-local complementations can be related to $(r+1)$ and $(r-1)$-local complementations: 

\begin{proposition}
\label{prop:monotonicity}
If $G\star^r S$ is valid then:
\begin{itemize}
\item  $G\star^{r+1} (2S)$ is valid and induces the same transformation: 
$G\star^{r+1} (2S) = G\star^rS$, where $2S$ is the multiset obtained from $S$ by doubling the multiplicity of each vertex;
\item $G\star^{r-1} S$ is valid (when $r>1$) and  $G\star^{r-1} S = G$. 
\end{itemize}
\end{proposition}

\begin{proof}
$2S$ is $(r+1)$-incident as for any $k\in [0,r+1)$ and any $K\subseteq V\setminus \supp(S)$ of size $k+2$, $2S\bullet \Lambda_G^K$ is a multiple of $2\times 2^{r-k-\delta(k)} = 2^{r+1-k-\delta(k)}$. Moreover, $S \bullet\Lambda_G^{u,v} = 2^{r-1}\bmod 2^{r}$ if and only if $2S \bullet\Lambda_G^{u,v} = 2^{r}\bmod 2^{r+1}$. $S$ is $(r-1)$-incident as for any $k\in [0,r-1)$ and any $K\subseteq V\setminus \supp(S)$ of size $k+2$, $S\bullet \Lambda_G^K$ is a multiple of $2^{r-k-\delta(k)}$, so is a multiple of $2^{r-1-k-\delta(k)}$. Moreover, for any vertices $u,v$, 
$S \bullet\Lambda_G^{u,v} = 0\bmod 2^{r-1}$    
\end{proof}

It implies in particular that two graphs related by $r$-local complementations are related by $(r+1)$-local complementations. An $r$-local complementation over $S$ preserves the neighborhood of the vertices in $\supp(S)$:

\begin{proposition}
    \label{loccompneighS}
If $G\star^r S$ is valid, then for any $u\in \supp(S)$, $N_{G\star^r S}(u)=N_G(u)$. 
\end{proposition}

\begin{proof}
        If $u \in \supp(S)$, for any vertex $v$, $S \bullet\Lambda_G^{u,v} = 0$. Indeed, $N_G(u) \cap \supp(S) = \emptyset$ thus $\Lambda_G^{u,v} \cap \supp(S) = \emptyset$. 
\end{proof}

Notice that $r$-local complementations can be composed:

\begin{proposition}
If $G\star^r S_1$ and $G\star^r S_2$ are valid and the sum\footnote{For any vertex $u$, $(S_1 + S_2)(u) = S_1(u)+S_2(u)$.}  $S_1 + S_2$ is independent in $G$, then $G\star^r (S_1 + S_2) = (G\star^r S_1)\star^r S_2$. 
\end{proposition}

\begin{proof}
$S_1 + S_2$ is $r$-incident as for any $k\in [0,r)$ and any $K\subseteq V\setminus \supp(S_1+S_2)$ of size $k+2$, $(S_1 + S_2)\bullet \Lambda_G^K = S_1\bullet \Lambda_G^K + S_2\bullet \Lambda_G^K$ is a multiple of $2^{r-k-\delta(k)}$. $S_2$ is $r$-incident in $G\star^r S_1$ as for any set $K\subseteq V\setminus \supp(S_2)$ of size $k+2$, $S_2\bullet \Lambda_{G\star^r S_1}^K = S_2\bullet \Lambda_{G}^K$. Besides, for any vertices $u,v$, $(S_1 + S_2)\bullet\Lambda_G^{u,v} = S_1\bullet \Lambda_G^K + S_2\bullet \Lambda_G^K$ and 
\begin{align*}
    u\sim_{(G\star^r S_1)\star^r S_2} v~ & \Leftrightarrow~\left(u\sim_{G\star^r S_1} v ~~\oplus~~ S_2\bullet\Lambda_{G\star^r S_1}^{u,v} = 2^{r-1}\bmod 2^{r}\right) ~\\
    & \Leftrightarrow~\left(u\sim_{G} v ~~\oplus~~ S_1\bullet\Lambda_{G}^{u,v} = 2^{r-1}\bmod 2^{r} ~~\oplus~~ S_2\bullet\Lambda_{G}^{u,v} = 2^{r-1}\bmod 2^{r}\right) \\
    & \Leftrightarrow~\left(u\sim_{G} v ~~\oplus~~ S_1\bullet\Lambda_{G}^{u,v} + S_2\bullet\Lambda_{G}^{u,v} = 2^{r-1}\bmod 2^{r}\right) \qedhere     
\end{align*}
\end{proof}

Finally, it is easy to see that the multiplicity in $S$ can be upperbounded by $2^r$: if $G\star^r S$ is valid, then $G\star^r S= G\star^r S'$, where, for any vertex $u$, $S'(u)=S(u)\bmod 2^r$.

\section{Graphical tools}

\subsection{Types with respect to an MLS cover}

\begin{definition} \label{def:type}Given a graph $G$, a vertex $u$ is of type P $\in$ \{X, Y, Z, $\bot$\} with respect to a MLS cover $\mathcal M$, where P is
    \begin{itemize}
        \item  X if for any generator $D$ of a minimal local set of $\mathcal M$ containing $u$, $u\in D \sm Odd(D)$;
        \item  Y if for any generator $D$ of a minimal local set of $\mathcal M$ containing $u$, $u\in D \cap Odd(D)$;
        \item Z if for any generator $D$ of a minimal local of $\mathcal M$ set containing $u$, $u\in Odd(D) \sm D$;
        \item $\bot$ otherwise. 
    \end{itemize}
\end{definition}

For notational simplicity, "with respect to the MLS cover $\mathcal M$" may be left out when the MLS cover is clear from the context. We then simply say that a vertex is e.g. of type X. The names X, Y and Z are chosen to match the corresponding Pauli operator in the stabilizers of the graph state. Indeed, a minimal local set $L$ is the support of a stabilizer $(-1)^{|G[D]|}X_D Z_{Odd(D)}$ if it is of dimension 1, and the support of 3 stabilizers $(-1)^{|G[D]|}X_D Z_{Odd(D)}$, $(-1)^{|G[D']|}X_{D'} Z_{Odd(D')}$ and $(-1)^{|G[D \Delta D']|}X_{D \Delta D'} Z_{Odd(D \Delta D')}$ if it is of dimension 2 (see \cref{prop:mls_stabilizers} and \cref{prop:MLS2cases}). It is direct to see that a vertex $u$ is attached to a Pauli X (resp. Y, Z) if $u\in D \sm Odd(D)$ (resp. $D \cap Odd(D)$, $Odd(D) \sm D$). In the case where $L$ is of dimension 2, any vertex $u$ is attached to the 3 Paulis X, Y and Z (this is proved in \cite{VandenNest05}, and is also direct from the proof of \cref{prop:MLS2cases}).

\begin{proposition}
    Given a graph $G$ and a MLS cover, 
    a vertex $u$ is of type 
    \begin{itemize}
        \item X if it is attached only to Paulis X in the corresponding stabilizers;
        \item Y if it is attached only to Paulis Y in the corresponding stabilizers;
        \item Z if it is attached only to Paulis Z in the corresponding stabilizers;
        \item $\bot$ otherwise.
    \end{itemize}
\end{proposition}

More precisely, a vertex $u$ is of type $\bot$ if it appears in a minimal local set of dimension 2 in $\mathcal M$, or appears in at least two minimal local sets of dimension 1, but attached to different Paulis. 

The set of vertices of type $\bot$ is invariant under LU-equivalence. To show that, we first prove a technical lemma.

\begin{lemma}\label{prop:comU}
    Given two graphs $G_1$, $G_2$ and a local unitary $U$ such that $\ket {G_2} = U\ket {G_1}$, if $L$ is a minimal local set of dimension 1 with generators respectively $D_1$ in $G_1$ and $D_2$ in $G_2$, then 
    $$X_{D_2} Z_{Odd_{G_2}(D_2)} = (-1)^{b} U X_{D_1} Z_{Odd_{G_1}(D_1)} U^\dagger$$
    with some $b \in \{0,1\}$.
\end{lemma}
    
\begin{proof}
    According to \cref{prop:marginal}, 
    ${Tr}_{V \sm L} \left(\ket {G_2} \bra {G_2}\right) = \frac{1}{|L|}\left(I+(-1)^{|G_2[D_2]|}X_{D_2} Z_{Odd_{G_2}(D_2)}\right)$ and ${Tr}_{V \sm L} \left(U \ket {G_1} \bra {G_1} U^\dagger\right) = \frac{1}{|L|}~ U \left(I+(-1)^{|G_1[D_1]|}X_{D_1} Z_{Odd_{G_1}(D_1)}\right)U^\dagger$.
\end{proof}

\begin{lemma}
    \label{lemma:LUbottom}
    Two LU-equivalent graphs have the same vertices of type $\bot$ (with respect to any MLS cover).
\end{lemma}

\begin{proof}

Let $U$ be a local unitary such that $\ket {G_2} = U\ket{G_1}$, and let $u\in V$ a vertex of type $\bot$ in $G_1$. If $u$ is in a minimal local set of dimension 2 in $G_1$, so is in $G_2$ as the dimension of minimal local sets is invariant under LU-equivalence. Indeed, two LU-equivalent graphs have the same cut-rank function (see \cref{prop:cutrank_lu}) and the dimension of a minimal local set depends solely on the cut-rank function (see \cref{prop:MLS2cases}). 

Otherwise, in $G_1$, $u$ appears in two distinct minimal local sets $L$, $L'$ of dimension 1, 
such that $u$ is attached to different Paulis $P$ and $P'$ respectively in $L$ and $L'$. 
According to \cref{prop:comU}, in $G_2$, $v$ is attached to Paulis $U_u P U_u^\dagger$ and $U_u P' U_u^\dagger$ in respectively $L$ and $L'$. Moreover, two Paulis are different when they anticommute, so $U_u P U_u^\dagger$ and $U_u P' U_u^\dagger$ are different Paulis if and only if $P$ and $P'$ are different Paulis.
\end{proof}

In \cite{VandenNest05} it is shown that for graphs with only vertices of type $\bot$, LU-equivalence implies LC-equivalence, i.e. LU=LC (see \cref{subsec:LULC_msc}). Types X, Y and Z are an attempt of generalizing the techniques used in \cite{VandenNest05} to all graphs.

\subsection{Constraints given by the types} \label{subsec:constraints_type}

We first show that local unitaries mapping one graph state to another are Clifford for vertices contained in a minimal local set of dimension 2. Recall that according to \cref{prop:MLS2cases}, a minimal local set of dimension 2 contains an even number of vertices. 

\begin{proposition}\label{prop:mls_clifford}
    Given two graphs $G_1$, $G_2$ and a local unitary $U$ such that $\ket {G_2} = U\ket {G_1}$, if $L$ is a minimal local set of dimension 2 containing at least 4 vertices, then for any vertex $u \in L$, $U_u$ is a Clifford operator.
\end{proposition}

This was first proven in \cite{Rains97} (see also \cite{zeng2011transversality}) but the original proof involves the qutrit formalism. An alternative proof, that does not use the qutrit formalism, is given in \cite{burchardt2024algorithm}. In the following, we prove \cref{prop:mls_clifford} similarly than in \cite{burchardt2024algorithm}. First we prove two technical lemmata.

\begin{lemma} \label{lemma:marginal_mls_dim_2}
    If $L$ is a minimal local set of dimension 2 of a graph $G=(V,E)$, then there exists a local Clifford operator $C$ such that $2^{|L|} {Tr}_{V \sm L} \left(\ket G \bra G\right) = C \left( I_L + X_L +(-1)^{|L|/2}Y_L +Z_L \right) C^\dagger$.
\end{lemma}

\begin{proof}
    According to \cref{prop:marginal} (see also \cref{prop:mls_stabilizers} and \cref{prop:MLS2cases}), there exist $D , D' \!\se L$ such that $2^{|L|} {Tr}_{V \sm L} \!\left(\ket G \!\bra G\right) \allowbreak \!=\! (-1)^{|G[D]|}X_D Z_{Odd(D)} \allowbreak +(-1)^{|G[D \Delta D']|}X_{D \Delta D'} Z_{Odd(D \Delta D')} \allowbreak +(-1)^{|G[D']|}X_{D'} Z_{Odd(D')}$. The conjugation of Pauli gates by single-qubit Clifford operators acts as a permutation of the Pauli gates (see Figure \ref{fig:clifford}), up to a sign change. Also, every vertex of $L$ is attached to the 3 Paulis X, Y, Z. Thus, there exists a local Clifford operator $C$ such that $2^{|L|} {Tr}_{V \sm L} \left(\ket G \bra G\right) = C \left( I_L + X_L +(-1)^{b}Y_L +Z_L \right) C^\dagger$ where $b\in\{0,1\}$. Finally, $(-1)^{|L|/2} Y_L \allowbreak =  X_L Z_L \allowbreak = C^\dagger (-1)^{|G[D]|}X_D Z_{Odd(D)} (-1)^{|G[D']|}X_{D'} Z_{Odd(D')} C \allowbreak = C^\dagger (-1)^{|G[D \Delta D']|}X_{D \Delta D'} Z_{Odd(D \Delta D')} C \allowbreak = (-1)^b Y_L$. Thus, $(-1)^b = (-1)^{|L|/2}$.    
\end{proof}

\begin{lemma} \label{lemma:sum_conjugation}
    If $U$ is a single-qubit unitary and $P \in \{X,Y,Z\}$, then $U P U^\dagger = c_{PX} X + c_{PY} Y + c_{PZ} Z$, where the $c_{PP'}$ are complex numbers satisfying $|c_{PX}|^2 + |c_{PY}|^2 + |c_{PZ}|^2 = 1$. Additionally, $U$ is Clifford if and only if, for any $P \in \{X,Y,Z\}$, there is $P' \in \{X,Y,Z\}$ such that $|c_{PP'}| = 1$ (implying $c_{PP''} = 0$ for $P'' \in \{X,Y,Z\} \sm \{P'\}$).
\end{lemma}

\begin{proof}
    $\{I,X,Y,Z\}$ form a basis of the $2 \times 2$ matrices, thus $U P U^\dagger = c_{PI} I + c_{PX} X + c_{PY} Y + c_{PZ} Z$. As the trace is linear and $Tr(U P U^\dagger) = Tr(X) = Tr(Y) = Tr(Z) = 0$, then $c_{PI} = 0$.  Also, $Tr((U P U^\dagger)(U P U^\dagger)^\dagger) = Tr(P^2) = Tr(I) = 2$ and $Tr((c_{PX} X + c_{PY} Y + c_{PZ} Z)(c_{PX} X + c_{PY} Y + c_{PZ} Z)^\dagger) = |c_{PX}|^2 Tr(I) + |c_{PY}|^2 Tr(I) + |c_{PZ}|^2 Tr(I) = 2(|c_{PX}|^2 + |c_{PY}|^2 + |c_{PZ}|^2)$.
\end{proof}

\begin{proof}[Proof of \cref{prop:mls_clifford}]
    According to \cref{prop:MLS2cases}, $L$ contains an even number of vertices. Thus, we write $L = 2m$ (where $m \gs 2$). Thus, ${Tr}_{V \sm L} \left(U \ket {G_1} \bra {G_1} U^\dagger\right) = {Tr}_{V \sm L} \left(\ket {G_2} \bra {G_2}\right)$. According to \cref{lemma:marginal_mls_dim_2}, \begin{align*}
        U C_1 \left( I_L + X_L +(-1)^{m}Y_L +Z_L \right) C_1^\dagger U^\dagger & = C_2 \left( I_L + X_L +(-1)^{m}Y_L +Z_L \right) C_2^\dagger\\
        U' \left( I_L + X_L +(-1)^{m}Y_L +Z_L \right) U'^\dagger & = I_L + X_L +(-1)^{m}Y_L +Z_L\\
        U' \left(X_L +(-1)^{m}Y_L +Z_L \right) U'^\dagger & = X_L +(-1)^{m}Y_L +Z_L
    \end{align*}
    where $C_1, C_2$ are local Clifford operators and $U' = C_2^\dagger U C_1$. Note that $U'$ is Clifford if and only if $U$ is Clifford. With the notation $U'=e^{i\phi}\bigotimes_{u\in V} U'_u$, we can write, according to \cref{lemma:sum_conjugation}, for $P,P' \in \{X,Y,Z\}$, $U'_u P  {U'}_u^\dagger = c^{(u)}_{PX} X + c^{(u)}_{PY} Y + c^{(u)}_{PZ} Z$ where $|c^{(u)}_{PX}|^2 + |c^{(u)}_{PY}|^2 + |c^{(u)}_{PZ}|^2 = 1$.
    \begin{align*}
        & U' \left(X_L +(-1)^{m}Y_L +Z_L \right) U'^\dagger\\
        & = U' X_L U'^\dagger +(-1)^{m} U' Y_L U'^\dagger + U' Z_L U'^\dagger \\
        & = \sum_{{\textbf P}\in \{X,Y,Z\}^{L}} \left(\prod_{u\in L}c^{(u)}_{X {\textbf P_u}} + (-1)^{m}\prod_{u\in L}c^{(u)}_{Y {\textbf P_u}} + \prod_{u\in L}c^{(u)}_{Z {\textbf P_u}}\right){\textbf P}\\
        &= X_L +(-1)^{m}Y_L +Z_L
    \end{align*}
    As $\{I,X,Y,Z\}$ form a basis of the $2 \times 2$ matrices, 
    $$ \sum_{P'\in \{X,Y,Z\}} \left|\prod_{u\in L}c^{(u)}_{X P'} + (-1)^{m}\prod_{u\in L}c^{(u)}_{Y P'} + \prod_{u\in L}c^{(u)}_{Z P'}\right| = 3$$
    The triangle inequality implies
    $$ \sum_{P, P' \in \{X,Y,Z\}} \prod_{u\in L} |c^{(u)}_{PP'}| \gs 3 $$  
    Conversely, the Hölder's inequality (see for example \cite{hardy1952inequalities}) implies
    \begin{align*}
        \sum_{P, P' \in \{X,Y,Z\}} \prod_{u\in L} |c^{(u)}_{PP'}| &= \sum_{P \in \{X,Y,Z\}} \left( \prod_{u\in L} |c^{(u)}_{PX}|+\prod_{u\in L} |c^{(u)}_{PY}|+\prod_{u\in L} |c^{(u)}_{PZ}| \right)\\
        &\ls \sum_{P \in \{X,Y,Z\}} \left(\prod_{u\in L} \left(|c^{(u)}_{PX}|^{2m}+|c^{(u)}_{PY}|^{2m}+|c^{(u)}_{PZ}|^{2m}\right)\right)^{\frac{1}{2m}}\\
        &\ls \sum_{P \in \{X,Y,Z\}} \left(\prod_{u\in L} \left(|c^{(u)}_{PX}|^2+|c^{(u)}_{PY}|^2+|c^{(u)}_{PZ}|^2\right)\right)^{\frac{1}{2m}} = 3
    \end{align*}
    As $m\gs 2$, this is an equality if and only $|c^{(u,v)}_{PP'}|=$ 0 or 1 for every $u,v,P,P'$.  Thus, according to \cref{lemma:sum_conjugation}, for any vertex $u \in L$, $U_u$ is a Clifford operator.
\end{proof}

Both minimal local sets of dimension 1 and 2 give constraints on local unitaries mapping one graph state to the other.

\begin{proposition}
    \label{lemma:LUconstraints}
    If $G_1 =_{LU} G_2$, then there exists $U=e^{i\phi}\bigotimes_{u\in V} U_u$ such that $\ket{G_2} = U \ket{G_1}$ and:
    \begin{itemize}
        \item $U_u$ is a Clifford operator if  $u$ is of type $\bot$ in both $G_1$ and $G_2$;
        \item $U_u = X(\theta_u) Z^{b_u}$ if $u$ is of type X in both $G_1$ and $G_2$;
        \item  $U_u = Z(\theta_u) X^{b_u}$ if $u$ is of type Z in both $G_1$ and $G_2$.
    \end{itemize}
    with $b_u \in \{0,1\}$. Additionally, if 
    $G_1 =_{LC_r} G_2$, $U$ can be chosen such that every angle satisfies $\theta_u = 0 \bmod \pi/2^r$.
\end{proposition}

\begin{proof}
    If $u$ is of type X, let $L$ be a minimal local set such that $u \in L$. According to \cref{prop:comU}, $U_u X U_u^{\dagger} = e^{i \phi} X$. $U_u \ket +$ and $U_u \ket -$ are eigenvectors of $U_u X U_u^{\dagger}$ with eigenvalues respectively 1 and -1, implying $U_u X U_u^{\dagger} = \pm X$. If $U_u X U_u^{\dagger} = X$, then $U_u = e^{i \phi} X(\theta)$. If $U_u X U_u^{\dagger} = -X$, define ${U'} \defeq U_u Z$. ${U'} X {U'}^{\dagger} = X$, so ${U'} = e^{i \phi} X(\theta)$, thus $U_u = e^{i \phi} X(\theta) Z$. The proof is similar when $u$ is of type Z.

    If $u$ is of type $\bot$, $u$ appears in a minimal local set of dimension 2, or it appears in two minimal local sets of dimension 1, but with different Paulis. If $u$ appears in a minimal local set of dimension 2 containing at least 4 vertices, according to \cref{prop:mls_clifford}, $U_u$ is Clifford. If $u$ appears in a minimal local set of dimension 2 containing exactly 2 vertices $u$ and $v$, then $u$ is part of a connected component containing exactly 2 adjacent vertices. Indeed, the only possible generators are $\{u\}$, $\{v\}$, and $\{u,v\}$, implying that $u$ is the only neighbor of $v$, and vice versa. As the only graph LU-equivalent to $K_2$ (the complete graph with two vertices) is $K_2$ itself, $U_u$ (and $U_v$) can be chosen to be the identity without loss of generality.

    Suppose now that $u$ appears in two minimal local sets $L$ and $L'$ of dimension 1 with different Paulis. Then ${U} P_1 {U}^{\dagger} = \pm P_2$ and ${U} P'_1 {U}^{\dagger} = \pm P'_2$ with Paulis $P_1, P'_1, P_2, P'_2$ such that $P_1 \neq P'_1$ and thus $P_2 \neq P'_2$. Also, $U P_1 P'_1 U^\dagger = \pm P_2 P'_2$. This means that $U_u$ is Clifford.
\end{proof}

\subsection{Standard form}

When a local complementation is applied on a vertex $u$, its type remains unchanged if it is X or $\bot$, while types Y and Z are swapped. For the neighbors of $u$, types Z and  $\bot$ remain unchanged, whereas  X and  Y are exchanged. This leads to a notion of standard form:

\begin{definition}
    A graph $G$ is in standard form with respect to a MLS cover $\mathcal M$ if:
    \begin{itemize}
        \item There are no vertices of type Y with respect to $\mathcal M$;
        \item For every vertex $u$ of type X with respect to $\mathcal M$, any neighbor $v$ of $u$ is of type Z with respect to $\mathcal M$ and satisfies $u < v$, in particular the vertices of type X with respect to $\mathcal M$ form an independent set;
        \item For every vertex $u$ of type X with respect to $\mathcal M$, $\{u\} \cup N_G(u) \in \mathcal M$.
    \end{itemize}
\end{definition}

\noindent
An illustration of the standard form is given in Figure \ref{fig:standard_form}.

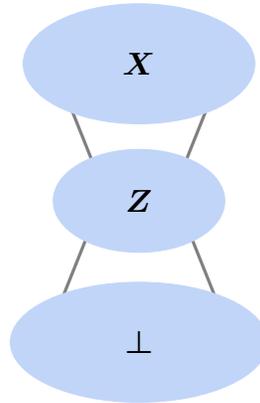
\begin{figure}[H]
    \centering
    
    \scalebox{1}{    
    \begin{tikzpicture}[scale = 0.57]

        %%%%%%%%%% X(alpha), Z(beta), I %%%%%%%%%%

        \begin{scope}[shift={(18,13)}]

        \draw[gray, very thick] (-2,0) -- (-1,-2.5);
        \draw[gray, very thick] (2,0) -- (1,-2.5);

        \begin{scope}[shift={(0,-3.5)}]
        \draw[gray, very thick] (-1,0) -- (-2,-2.5);
        \draw[gray, very thick] (1,0) -- (2,-2.5);
        \end{scope}

        \draw[cornflowerblue!40,fill=cornflowerblue!40] (0,0) ellipse (2.7cm and 1.4cm);
        \draw[cornflowerblue!40,fill=cornflowerblue!40] (0,-3.2) ellipse (2cm and 1.2cm);
        \draw[cornflowerblue!40,fill=cornflowerblue!40] (0,-6.5) ellipse (3cm and 1.4cm);

        \draw (0,0) node[text=black](){$\boldsymbol{X}$};
        \draw (0,-3.2) node[text=black](){$\boldsymbol{Z}$};
        \draw (0,-6.5) node[text=black](){$\boldsymbol{\bot}$};

        \end{scope}

    \end{tikzpicture}    
    }
    
    \caption{Illustration of a graph in standard form. Vertices of type X form an independent set, and are only adjacent to vertices of type Z. There are no vertices of type Y.}
    \label{fig:standard_form}
\end{figure}

\begin{remark}This notion of standard form is a generalization of the one introduced in \cite{claudet2024local}, where the MLS cover considered consists of every minimal local set of the graph: $\mathcal M_\text{max} \defeq \{L \se V ~|~ L \text{~is a minimal local set}\}$. Since there can be exponentially many minimal local sets, using $\mathcal M_\text{max}$ does not lead to an efficient procedure, for instance when computing the type of each vertex.
\end{remark}

Given a pair of LU-equivalent graphs, one can efficiently compute a (common) MLS cover and put both graphs in standard form with respect to this MLS cover, by means of local complementations:

\begin{lemma} \label{lemma:standardform}
    There exists an efficient algorithm that takes as inputs two graphs $G_1$ and $G_2$% of order $n$
    , and either claim that they are not LU-equivalent, or compute an MLS cover $\mathcal M$ and two graphs $G'_1 =_{LC} G_1$ and $G'_2 =_{LC} G_2$, such that $G'_1$ and $G'_2$ are both in standard form with respect to $\mathcal M$, in runtime $O(n^{6.38})$.
\end{lemma}

\begin{proof}
    To prove the proposition, we introduce an algorithm that transforms the input graphs $G_1, G_2$ into graphs in standard form with respect to the same MLS cover by means of local complementations. The notation $\wedge $ refers to the pivoting operation: $G\wedge uv \defeq G \star u \star v = G \star v \star u$.  
    The action of the local complementation and the pivoting on the type of the vertices is given in the following table (the types of the unwritten vertices remain unchanged). 
    \begin{center}
    \begin{tabular}{|c|c|}
    \hline
    \multicolumn{2}{|c|} {Type of $u$ in}\\
    $~~G~~$& $G\star u$\\
    \hline
    X&X\\
    \hline
    Y&Z\\
    \hline
    Z&Y\\
    \hline
    $\bot$&$\bot$\\
    \hline
    \end{tabular}\qquad\begin{tabular}{|c|c|}
    \hline
    \multicolumn{2}{|c|}{Type of $v {\in}N_G(u)$ in}\\
    $~~~G~~~$& $G\star u$\\
    \hline
    X&Y\\
    \hline
    Y&X\\
    \hline
    Z&Z\\
    \hline
    $\bot$&$\bot$\\
    \hline
    \end{tabular}
    \qquad \begin{tabular}{|c|c|}
    \hline
    \multicolumn{2}{|c|}{Type of $u$ (or $v$) in}\\
    $~~G~~$&$G\wedge uv$\\
    \hline
    X&Z\\
    \hline
    Y&Y\\
    \hline
    Z&X\\
    \hline
    $\bot$&$\bot$\\
    \hline
    \end{tabular}
    \end{center}

    The algorithm reads as follows:

    \begin{enumerate}
        \item Compute an MLS cover $\mathcal M$ of $G_1$ (see \cref{thm:mls_cover_algo}). If $\mathcal M$ is not an MLS cover of $G_2$, then $G_1$ and $G_2$ are not LU-equivalent.
        \item If there is an XX-edge (i.e.~an edge $uv$ such that both $u$ and $v$ are of type X with respect to $\mathcal M$) in $G_1$ or $G_2$: apply a pivoting on it.\\
        Repeat until there is no XX-edge left.
        \item If there is an XY-edge in $G_1$ or $G_2$: apply a local complementation on the vertex of type X, then go to step 2.
        \item If there is a vertex of type Y in $G_1$ or $G_2$: apply a local complementation on it, then go to step 2.
        \item If there is an X$\bot$-edge in $G_1$ or $G_2$: apply a pivoting on it.\\
        Repeat until there is no X$\bot$-edge left.

        \item If there is an XZ-edge $uv$ in $G_1$ or $G_2$ such that $v < u$: apply a pivoting on $u v$.\\
        Repeat until for every XZ-edge $uv$ in $G_1$ or $G_2$, $u < v$.
 
        \item If there is a vertex $u$ of type X in $G_1$ (resp. $G_2$) such that $\{u\} \cup N_{G_1}(u)$ (resp.  $\{u\} \cup  N_{G_2}(u)$) is not a minimal local set of dimension 1: find a minimal local set $M$ contained in $\{u\} \cup N_{G_1}(u)$ (resp.  $\{u\} \cup  N_{G_2}(u)$) and check that $M$ is a minimal local set of same dimension in both graphs (if not, they are not LU-equivalent). If this is the case, add $M$ to $\mathcal M$ then go to step 5.

        \item For every vertex $u$ of type X in $G_1$ (resp. $G_2$), add $\{u\} \cup N_{G_1}(u)$ (resp.  $\{u\} \cup N_{G_2}(u)$) to $\mathcal M$.
    \end{enumerate}

    \subparagraph{Correctness.} When step 2 is completed, there is no XX-edge. Step 3 transforms the neighbors of type Y into vertices of type Z. No vertex of type Y is created as there is no XX-edge before the local complementation. When step 3 is completed, there is no XX-edge nor any XY-edge. Step 4 transforms the vertex of type Y into a vertex of type Z. No vertex of type Y is created as there is no XY-edge before the local complementation. When step 4 is completed, there is no vertex of type Y nor any XX-edge. In step 5, applying a pivoting on an X$\bot$-edge transforms the vertex of type X into a vertex of type Z. No XX-edge is created, as the vertex of type X has no neighbor of type X before the pivoting. When step 5 is completed, there is no vertex of type Y and each neighbor of a vertex of type X is of type Z. In step 6, applying a pivoting on an XZ-edge permutes the type of the two vertices, and preserves the fact that each neighbor of a vertex of type X is of type Z. In step 7, adding a minimal local set to  $\mathcal M$ may only change the type of some vertices to $\bot$. When step 7 is completed, for every vertex $u$ of type X in $G_1$ (resp. $G_2$), $\{u\} \cup N_{G_1}(u)$ (resp.  $\{u\} \cup N_{G_2}(u)$) is a minimal local set. Thus, in step 8, adding those minimal local sets leave the types invariant. When step 8 is completed, $G_1$ and $G_2$ are in standard form with respect to $\mathcal M$.
    
    \subparagraph{Termination.} The quantity $2|V^{G_1}_Y|+|V^{G_1}_X| + 2|V^{G_2}_Y|+|V^{G_2}_X|$, where $|V^{G_i}_X|$ (resp. $|V^{G_i}_Y|$) denotes the number of vertices of type X (resp. Y) with respect to $\mathcal M$ in $G_i$, strictly decreases at steps 2 to 5, which guarantees to reach step 6. At step 6, $V^{G_i}_X$ is updated as follows: exactly one vertex $u$ is removed from the set and is replaced by a vertex $v$ such that $v < u$, which guarantees the termination of step 6. Each time a minimal local set is added to $\mathcal M$ in step 7, at least one vertex not of type $\bot$ becomes of type $\bot$. Indeed, without loss of generality, let $M$ be a minimal local set in $\{u\} \cup N_{G_i}(u)$, assuming $\{u\} \cup N_{G_i}(u)$ is not a minimal local set of dimension 1 generated by $\{u\}$. If $M$ is of dimension 2, every vertex of $M$ becomes of type $\bot$ when adding $M$ to $\mathcal M$. Else, $M$ is generated by a set containing at least one vertex of type Z, which becomes of type $\bot$ when adding $M$ to $\mathcal M$.

    \subparagraph{Complexity.} The time-complexity of the algorithm is given by the time-complexity of step 1, as it is asymptotically the most computationally expensive step. The MLS cover can be computed in runtime $O(n^{6.38})$, where $n$ is the order of the graph. It should be noted that giving the type of a vertex with respect to some minimal local set can be done in time-complexity $O(n^3)$. Indeed, in the case of a minimal local set $L$ of dimension 1, finding $D$ such that $L = D \cup Odd_G(D)$ reduces to finding the kernel of some matrix with coefficients in $\mathbb F_2$ (see details in \cite{claudet2024covering}), which can be done using Gaussian elimination.
\end{proof}

Standard form with respect to a common MLS cover implies some strong similarities in the structure of graphs:

\begin{lemma} \label{lemma:same_types}
    If two graphs $G_1 = _{LU} G_2$ are in standard form (with respect to the same MLS cover $\mathcal M$), then every vertex has the same type in $G_1$ and $G_2$, and  every vertex $u$ of type X satisfies $N_{G_1}(u) = N_{G_2}(u)$. 
\end{lemma}

\begin{proof}

    First, by \cref{lemma:LUbottom}, $G_1$ and $G_2$ have the same vertices of type $\bot$. Let $u$ be a vertex of type X in $G_1$ and suppose by contradiction that $u$ is of type Z in $G_2$. By definition of the standard form, $L \defeq \{u\} \cup N_{G_1}(u)$ is a minimal local set in $\mathcal M$ and $u$ is the smallest element in $L$ (in both $G_1$ and $G_2$). Moreover, $L = D \cup Odd_{G_2}(D)$ and $u \in Odd_{G_2}(D) \sm D$ by definition of type Z. Thus, $u$ is adjacent to at least one vertex $v \in D$ in $G_2$. $v$ is of type X in $G_2$ as, by definition of the standard form, there is no vertex of type Y, and $N_{G_1}(u)$ contains no vertex of type $\bot$. This is a contradiction with the definition of the standard form, as $u < v$. By symmetry, any vertex of type X in $G_2$ is also of type X in $G_1$. Overall, each vertex has the same type in $G_1$ and $G_2$. By definition of the standard form, for any vertex $u$ of type X, $L \defeq \{u\} \cup N_{G_1}(u)$ is a minimal local set in $\mathcal M$, thus there exists $D\subseteq L$ such that $L = D\cup Odd_{G_2}(D)$. As $D$ is non-empty and does not contain any vertex of type Z, $D=\{u\}$, implying that $N_{G_1}(u) = N_{G_2}(u)$.
\end{proof}

After performing the algorithm described in \cref{lemma:standardform}, one can check in quadratic time\footnote{In the order of the graphs, assuming the information of the types of the vertices with respect to the MLS cover is conserved.} whether each vertex has the same type  in $G_1$ and $G_2$, and whether every vertex of type X has the same neighborhood in both graphs. If either condition is not met, the graphs are not LU-equivalent.

\section{Capturing LU-equivalence}

Following \cref{lemma:same_types}, as two LU-equivalent graphs in standard form (with respect to the same MLS cover) have the same types, we use the notation $V_X$, $V_Z$ and $V_\bot$ for the set of vertices of type X, Z and $\bot$ respectively. 

Thanks to the standard form, one can accommodate the types of two LU-equivalent graphs, to simplify the local unitaries mapping one  to the other:

\begin{lemma}\label{lemma:standardform_implies_rotations}
    If $G_1 =_{LU} G_2$ are both in standard form (with respect to the same MLS cover),  
    there exists $G'_1 =_{LC} G_1$ in standard form such that $\ket {G_2} = \bigotimes_{u\in V_X}X(\alpha_u)\bigotimes_{v\in V_Z} Z(\beta_v) \ket{G'_1}$.
\end{lemma}

\begin{remark}
    If $G_1 =_{LC_r} G_2$, we can choose the angles such that $\alpha_u, \beta_v = 0 \bmod \pi/2^r$. Additionally, $G_1$ and $G'_1$ are related by local complementations only on vertices of type $\bot$.
\end{remark}

\begin{proof} 
    By \cref{lemma:LUconstraints} there exists $U$ such that $\ket {G_2} = U \ket {G_1}$, with $U=e^{i\phi}\bigotimes_{u\in V} U_u$ where:
    \begin{itemize}
        \item $C_u \defeq U_u$ is a Clifford operator if  $u\in V_\bot$;
        \item $U_u = X(\theta_u) Z^{b_u}$ if $u \in V_X$;
        \item $U_u = Z(\theta_u) X^{b_u}$ if $u \in V_Z$.
    \end{itemize}
    with $b_u \in \{0,1\}$. Additionally, 
    if $\ket {G_2} = U \ket {G_1}$, we choose $U$ so that the angles satisfy $\theta_u = 0\bmod \pi/2^r$. 

    $$ \ket {G_2} = e^{i \phi} \bigotimes_{u\in V_X}X(\alpha_u) Z^{b_u} \bigotimes_{u\in V_Z} Z(\beta_u) X^{b_u} \bigotimes_{u\in \bot}C_u \ket{G_1} $$
    As for any vertex $u \in V$, $X_u Z_{N_{G_1}(u)} \ket{G_1} = \ket{G_1}$, then for possibly different angles $\beta_u$ (albeit equal modulo $\pi/2$), values $b_u$ and Clifford operators $C_u$, $$ \ket {G_2} = e^{i \phi} \bigotimes_{u\in V_X}X(\alpha_u) Z^{b_u} \bigotimes_{u\in V_Z} Z(\beta_u) \bigotimes_{u\in \bot}C_u \ket{G_1} $$

    By applying $\bra{+}_{V_X}\bra{0}_{V_Z}$ on both sides of the previous equation we get, on the left-hand side, $\bra{+}_{ V_X}\bra{0}_{V_Z} \ket {G_2} = \frac1{\sqrt{2^{|V_Z|}}} \bra{+}_{V_X} \ket{+}_{V_X} \ket{G_2[\bot]} = \frac1{\sqrt{2^{|V_Z|}}} \ket{G_2[\bot]}$ as $V_X$ is an independent set in $G_1$; and, on the right-hand side, $\begin{cases}\frac{e^{i\phi}}{\sqrt{2^{|V_Z|}}}\otimes_{u\in\bot} C_u \ket{G_1[\bot]} & \text{if $\forall u\in V_X, b_u=0$}\\0&\text{otherwise}%\tag{A}
    \end{cases}$, since $\bra 0 Z(\theta)=\bra 0$, $\bra +X(\theta) = \bra +$ and $\bra +Z = \bra -$. Hence, every $b_u = 0$, and $\ket {G_2[\bot]} = e^{i \phi} \bigotimes_{u\in \bot}C_u \ket{G_1[\bot]}$ (if $\bot = \emptyset$ then $e^{i\phi} = 1$).

    Using \cref{prop:lclc}, there exist a sequence of (possibly repeating) vertices $a_1, \cdots, a_m \in \bot$ and a set $D \se \bot$ such that $G_2[\bot] = G_1[\bot] \star a_1 \star a_2 \star \cdots \star a_m$ and $e^{i \phi}\bigotimes_{u\in \bot}C_u = L_{a_m} \cdots L_{a_2} L_{a_1}$, 
    where $L_{a_i} \defeq X\left(\frac{\pi}{2}\right)_{a_i} Z\left(-\frac{\pi}{2}\right)_{N_{{G_1}[\bot] \star a_1 \star \cdots \star a_{i-1}}(a_i)}$.

    These operators do not implement a local complementation on the full graph. For any $i \in [1, m]$, let us introduce $$L'_{a_i} = L_{a_i} Z\left(- \frac{\pi}{2}\right)_{N_{G_1 \star a_1 \star \cdots \star a_{i-1}}(a_i)\sm \bot} = X\left(\frac{\pi}{2}\right)_{a_i} Z\left(-\frac{\pi}{2}\right)_{N_{G_1 \star a_1 \star \cdots \star a_{i-1}}(a_i)}$$

    Then, for possibly different angles $\beta_u$ (albeit equal modulo $\pi/2$),
    \begin{equation*}
        \begin{split}
            \ket {G_2} &= \bigotimes_{u\in V_X}X(\alpha_u)  \bigotimes_{u\in V_Z} Z(\beta_u) L'_{a_m} \cdots L'_{a_2} L'_{a_1} \ket{G_1}\\
            & = \bigotimes_{u\in V_X}X(\alpha_u)  \bigotimes_{u\in V_Z} Z(\beta_u) \ket{G_1 \star a_1 \star \cdots \star a_m } \qedhere
        \end{split}
    \end{equation*}
\end{proof}

They are strong constraints relating the angles of the X- and Z-rotations acting on different qubits:

\begin{lemma} \label{lemma:cond_angles}
Given $G_1$, $G_2$ in standard form (with respect to the same MLS cover), 
    if $\ket {G_2} = \bigotimes_{u\in  V_X}X(\alpha_u)\bigotimes_{v\in V_Z} Z(\beta_v) \ket{G_1}$:
    \begin{itemize}
        \item $\forall v \in V_Z,~ \beta_v  = - \sum_{u \in N_{G_1}(v)\cap V_X}\alpha_u \bmod 2\pi$;
        \item $\forall k\in \mathbb N,\forall K \se V_Z$ of size $k+2$, $\sum_{u \in \Lambda_{G_1}^K \cap V_X}\alpha_u  = 0\bmod \dfrac{\pi}{2^{k+\delta(k)}}$.
    \end{itemize}
\end{lemma}

\begin{proof}
    According to \cref{cor:action_X_rotation_indep_set}:
    \begin{align*}
        & \bigotimes_{u\in  V_X}X(\alpha_u)\bigotimes_{v\in V_Z} Z(\beta_v) \ket{G_1}\\
        &= \left(\prod_{K\subseteq V, K \neq \emptyset} CZ_{K}\left((-2)^{|K|-1}\sum_{u \in \Lambda_{G_1}^K \cap V_X} \alpha_u\right)\right) \bigotimes_{v\in V_Z} Z(\beta_v) \ket{G_1}\\
        &= \left(\prod_{K\subseteq V, ~|K|\gs 2} CZ_{K}\left((-2)^{|K|-1}\sum_{u \in \Lambda_{G_1}^K \cap V_X} \alpha_u\right)\right) \bigotimes_{v\in V_Z} Z(\beta_v + \sum_{u \in N_{G_1}(v)\cap V_X}\alpha_u) \ket{G_1}\\
    \end{align*}
    $\bigotimes_{u\in  V_X}X(\alpha_u)\bigotimes_{v\in V_Z} Z(\beta_v) \ket{G_1}$ is a weighted hypergraph state. It is in particular a graph state if and only if the constraints described in \cref{lemma:cond_angles} are met.
\end{proof}

The constraints described in \cref{lemma:cond_angles} coincide with $r$-incidence, so, when angles are multiples of $\pi/2^r$, this unitary transformation implements an $r$-local complementation.

\begin{lemma} \label{lemma:implements_lc}
    Given $G_1$, $G_2$ in standard form (with respect to the same MLS cover), \\if $\ket {G_2} = \bigotimes_{u\in  V_X}X(\alpha_u)\bigotimes_{v\in V_Z} Z(\beta_v) \ket{G_1}$ and $\forall u \in V_X$, $\alpha_u= 0\bmod \pi/2^r$, then \\ $\bigotimes_{u\in V_X}X(\alpha_u)\bigotimes_{v\in V_Z} Z(\beta_v)$ implements an $r$-local complementation over the vertices of $V_X$.
\end{lemma}

\begin{proof}
    Let us construct an $r$-incident independent multiset S such that $G_2 =  G_1 \star^r S$. We define $S$ on the vertices of $V_X$ such that $\forall u \in V_X$, $\alpha_u = \frac{S(u)\pi}{2^r}\bmod 2\pi$ with $S(u) \in [1,2^r)$. Note that $\sum_{u \in \Lambda_{G_1}^K \cap V_X}\alpha_u = \frac{\pi}{2^r}S\bullet \Lambda_{G_1}^K$. Hence, by \cref{lemma:cond_angles}, For any $k\in [0,r)$, and any $K\subseteq V\setminus S$ of size $k+2$, $S\bullet \Lambda_{G_1}^K$ is a multiple of $2^{r-k-\delta(k)}$ meaning that S is $r$-incident.    
    Also, by \cref{lemma:cond_angles}, $\beta_v = -\frac{\pi}{2^r} \sum_{u \in N_{G_1}(v)}S(u)\bmod 2\pi$. Thus, $\bigotimes_{u\in V_X}X(\alpha_u)\bigotimes_{v\in V_Z} Z(\beta_v)$ implements an $r$-local complementation over $S$, according to \cref{prop:implementation_rlc}.
\end{proof}

We can now easily relate LC$_r$-equivalence to $r$-local complementations for graphs in standard form:

\begin{lemma}
    \label{lemma:standardform_LCr_lc}
    If $G_1$ and $G_2$ are both in standard form (with respect to the same MLS cover) and $G_1 =_{LC_{r}} G_2$, then $G_1$ and $G_2$ are related by a sequence of local complementations on the vertices of type $\bot$ along with a single $r$-local complementation over the vertices of type X.
\end{lemma}

\begin{remark}
    The sequence of local complementations commutes with the $r$-local complementation, as vertices of type X and vertices of type $\bot$ do not share edges by definition of the standard form.
\end{remark}

\begin{proof}
    By \cref{lemma:standardform_implies_rotations}, there exists $G'_1 =_{LC} G_1$ in standard form such that $\ket {G_2} \\ = \bigotimes_{u\in V_X}X(\alpha_u)\bigotimes_{v\in V_Z} Z(\beta_v) \ket{G'_1}$ where 
    $\forall u \in V$ of type X or Z, $\alpha_u, \beta_u = 0\bmod \pi/2^r$. By \cref{lemma:implements_lc}, $\bigotimes_{u\in V_X}X(\alpha_u)\bigotimes_{v\in V_Z} Z(\beta_v)$ implements an $r$-local complementation over the vertices of type X.
\end{proof}

Notice in particular that when there is no vertex of type $\bot$, a single $r$-local complementation is required.

\begin{corollary}\label{cor:constraint_r_lc}
    If  two $\bot$-free LC$_r$-equivalent graphs $G_1$ and $G_2$ are both in standard form (with respect to the same MLS cover), they are related by a single $r$-local complementation over the vertices of type X.
\end{corollary}

We are ready to prove that $r$-local complementation captures the LC$_r$-equivalence of graph states.

\begin{theorem}\label{thm:LCr_lc}
    Given two graphs $G_1, G_2$, the following statements are equivalent:
    \begin{enumerate}
        \item $\ket{G_1}$ and $\ket{G_2}$ are LC$_r$-equivalent.
        \item $G_1$ and $G_2$ are related by $r$-local complementations.
        \item $G_1$ and $G_2$ are related by a sequence of local complementations along with a single $r$-local complementation.
    \end{enumerate}
\end{theorem}

\begin{proof}
    We proceed by cyclic proof.
    ($2\Rightarrow1$):~ Follows from \cref{prop:implementation_rlc}. 
    ($3  \Rightarrow  2$):~ A local complementation is, in particular, an $r$-local complementation. 
    ($1  \Rightarrow  3$):~ Follows from \cref{lemma:standardform} along with \cref{lemma:standardform_LCr_lc}.  
\end{proof}

From now, we call two graphs LC$_r$-equivalent when the corresponding graph states are LC$_r$-equivalent or, equivalently, when the graphs are related by a sequence of $r$-local complementations. More than just LC$_r$-equivalence, $r$-local complementations can actually characterize the LU-equivalence of graph states.

\begin{theorem} \label{thm:LU_imply_LCr}
    If two graph states are LU-equivalent, they are LC$_r$-equivalent for some $r \ls n$, i.e. the corresponding graphs are related by $r$-local complementations.
\end{theorem}

\begin{proof} 
    Suppose 
    $G_1 =_{LU} G_2$. By \cref{lemma:standardform} along with \cref{lemma:standardform_implies_rotations}, there exist $G'_1 =_{LC} G_1$ and $G'_2 =_{LC} G_2$ both in standard form such that $\ket {G'_2} = \bigotimes_{u\in V_X}X(\alpha_u)\bigotimes_{v\in V_Z} Z(\beta_v) \ket{G'_1}$. 
    If two vertices $u$ and $v$ are twins, i.e.~$N_G(u)=N_G(v)$, then $X(\theta_1)_u X(\theta_2)_v\ket G = X(\theta_1 + \theta_2)_u \ket G$. Indeed, $X(\theta)_u\ket G \!=\! e^{i\frac\theta 2}(\cos(\frac \theta 2)\ket G + i\sin(\frac \theta 2)X_u\ket G) \!=\! e^{i\frac\theta 2}(\cos(\frac \theta 2)\ket G + i\sin(\frac \theta 2)Z_{N_G(u)}\ket G) \allowbreak = X(\theta)_v\ket G $. Hence, without loss of generality, we suppose that if $u,v \in V_X$ are twins, then at most one of the angles $\alpha_u, \alpha_v$ is non-zero.

    Moreover, if a vertex $u$ is a leaf, i.e.~there exists a unique $v \in V$ such that $u \sim_{G} v$, then $X(\theta)_u Z(-\theta)_v \ket G = \ket G$. Indeed, $X(\theta)_u \ket G = e^{i\frac\theta 2}(\cos(\frac \theta 2)\ket G + i\sin(\frac \theta 2)Z_{v}\ket G) = Z(\theta)_v \ket G$. Hence, without loss of generality, we suppose that if $u \in V_X$ is a leaf, then $\theta = 0$. 
    
    If a vertex $u$ is an isolated vertex, $X(\theta)_u \ket G = \ket G$. Hence, without loss of generality, we suppose that if a vertex $u \in V_X$ is an isolated vertex, then $\alpha_u = 0$.

    Let us show that for every $u \in V_X$, $\alpha_u=0\bmod \pi/2^{|V_Z|-2+\delta(|V_Z|-2)}$. 
    If $|V_Z|\ls 1$, the proof is trivial as every vertex in $V_X$ is a leaf or an isolated vertex. Else, we prove by induction over the size of a set $K \se V_Z$ that for every $u\in \Lambda_{G'_1}^{V_Z\setminus K} \cap V_X$, $\alpha_u = 0\bmod \pi/2^{|V_Z|-2+\delta(|V_Z|-2)}$.

    \begin{itemize}
        \item If $K = \emptyset$: there is at most one vertex $u$ in $\Lambda_{G'_1}^{V_Z\setminus K} \cap V_X = \Lambda_{G'_1}^{V_Z} \cap V_X$ such that $\alpha_u$ is non-zero. By \cref{lemma:cond_angles}, $\alpha_u = 0\bmod \pi/2^{|V_Z|-2+\delta(|V_Z|-2)}$.
        \item Else, there is at most one vertex $u$ in $\left(\Lambda_{G'_1}^{V_Z\setminus K} \sm \bigcup_{K' \varsubsetneq K} \Lambda_{G'_1}^{V_Z\setminus K'} \right) \cap V_X$ such that $\alpha_u$ is non-zero. For any vertex $v \in \Lambda_{G'_1}^{V_Z \sm K'} \cap V_X$ such that $K' \varsubsetneq K$, $\alpha_u = 0\bmod \pi/2^{|V_Z|-2+\delta(|V_Z|-2)}$ by hypothesis of induction. If $|V_Z|-|K|=1$, then $u$ is a leaf and thus $\alpha_u = 0$. Else, by \cref{lemma:cond_angles}, $\sum_{w \in \Lambda_{G'_1}^{V_Z \sm K} \cap V_X}\alpha_w  = 0\bmod \dfrac{\pi}{2^{|V_Z|-2+\delta(|V_Z|-2)}}$, implying that $\alpha_u = 0\bmod \dfrac{\pi}{2^{|V_Z|-2+\delta(|V_Z|-2)}}$.
    \end{itemize}

    A vertex $u \in V_X$ such that for every $K \se V_Z$, $u \notin \Lambda_{G'_1}^{V_Z\setminus K}$, is an isolated vertex and thus $\alpha_u = 0$. Hence, for every $u \in V_X$, $\alpha_u=0\bmod \dfrac{\pi}{2^{|V_Z|-2+\delta(|V_Z|-2)}}$. By \cref{lemma:cond_angles}, for every vertex $v \in V_Z$, $\beta_v  = - \sum_{u \in N_{G'_1}(v)\cap V_X}\alpha_u = 0\bmod \dfrac{\pi}{2^{|V_Z|-2+\delta(|V_Z|-2)}}$. 
    
    By \cref{lemma:implements_lc}, $\bigotimes_{u\in V_X}X(\alpha_u)\bigotimes_{v\in V_Z} Z(\beta_v)$ implements an $r$-local complementation where $r \ls |V_Z|-1$.
\end{proof}

Remarkably, \cref{thm:LU_imply_LCr} (along with \cref{prop:lcr_equals_clifford_hierarchy}) implies that for graph states, LU-equivalence reduces to equivalence up to local unitaries in the Clifford hierarchy. This has been hinted previously in \cite{gross2007lu}, and similar restrictions occur in the symmetries of graph states \cite{Englbrecht2020}.

In \cref{chap:algo} we strengthen \cref{thm:LU_imply_LCr} by proving that $r$ can be actually be bounded by $log_2(n) + O(1)$. This bound reveals very useful for algorithmic purposes. 

\section{An infinite strict hierarchy of local equivalences}

Two graphs related by $r$-local complementations are also related by $(r+1)$-local complementations. This can be seen graphically as a direct consequence of \cref{prop:monotonicity}, or using graph states through \cref{thm:LCr_lc}, as two LC$_{r}$-equivalent states are obviously LC$_{r+1}$-equivalent. The known counter examples to the LU-LC conjecture introduced in \cite{Ji07,Tsimakuridze17} are actually LC$_2$-equivalent graphs that are not LC-equivalent. There was however no known examples of graphs that are LC$_3$-equivalent but not LC$_2$-equivalent, and more generally no known examples of graphs that are LC$_{r+1}$-equivalent but not LC$_{r}$-equivalent for $r\gs 2$. We introduce such examples in this section, by showing that for any $r$ there exist pairs of graphs that are LC$_{r+1}$-equivalent but not LC$_{r}$-equivalent, leading to an infinite  hierarchy of generalized local equivalences.  Our proof is constructive.

We introduce a family of bipartite graphs $C_{t,k}$ parametrized by two integers $t$ and $k$. This family is a variant of the (bipartite) Kneser graphs and uniform subset graphs. The graph $C_{t,k}$ is bipartite: the first independent set of vertices corresponds to the integers from $1$ to $t$, and the second independent set is composed of all the subsets of $[1,t]$ of size $k$ (we note $\binom{[1,t]}{k}$). There is an edge between an integer $u\in [1,t]$ and a subset $A\subseteq [1,t]$ of size $k$ if and only if $u\in A$:

\begin{definition} For any $k\gs 1$ and $t\gs k$ two integers, let $C_{t,k} = (V,E)$ be a bipartite graph defined as: $V = [1,t] \cup \binom{[1,t]}{k}$~~and~~ $E = \{(u,A) \in [1,t]\times \binom{[1,t]}{k} ~|~ u\in A\}$. 
\end{definition}

\noindent

We introduce a second family of graphs $C'_{t,k}$, defined similarly to $C_{t,k}$, the independent set $[1,t]$ being replaced by a clique:

\begin{definition} For any $k\gs 1$ and $t\gs k$ two integers, let $C'_{t,k}=(V,E)$ be a graph defined as:  $V = [1,t] \cup \binom{[1,t]}{k}$ and $E = \{(u,A) \in [1,t]\times \binom{[1,t]}{k} ~|~ u\in A\} \cup \{(u,v)\in [1,t]\times [1,t]~|~u\neq v\}$.
\end{definition}

Examples of such graphs with parameters $t=4$ and $k=2$ are given in Figure \ref{fig:bikneser}.

\begin{figure}[h]
\centering

\scalebox{1}{
\begin{tikzpicture}[xscale = 0.65,yscale=0.6]    
    \begin{scope}[every node/.style=vertices]
        \node (U1) at (-3,0) {$1$};
        \node (U2) at (-1,0) {$2$};
        \node (U3) at (1,0) {$3$};
        \node (U4) at (3,0) {$4$};
        \node (U12) at (-5,8) {\footnotesize $\{1,2\}$};
        \node (U13) at (-3,8) {\footnotesize$\{1,3\}$};
        \node (U14) at (-1,8) {\footnotesize$\{1,4\}$};
        \node (U23) at (1,8) {\footnotesize$\{2,3\}$};
        \node (U24) at (3,8) {\footnotesize$\{2,4\}$};
        \node (U34) at (5,8) {\footnotesize$\{3,4\}$};
    \end{scope}
    \begin{scope}[every edge/.style=edges]              
        \path [-] (U1) edge node {} (U12);
        \path [-] (U1) edge node {} (U13);
        \path [-] (U1) edge node {} (U14);
        \path [-] (U2) edge node {} (U12);
        \path [-] (U2) edge node {} (U23);
        \path [-] (U2) edge node {} (U24);
        \path [-] (U3) edge node {} (U13);
        \path [-] (U3) edge node {} (U23);
        \path [-] (U3) edge node {} (U34);
        \path [-] (U4) edge node {} (U14);
        \path [-] (U4) edge node {} (U24);
        \path [-] (U4) edge node {} (U34);
    \end{scope}
\end{tikzpicture}\qquad\raisebox{0cm}{
    \raisebox{-1.31cm}{
    \begin{tikzpicture}[xscale = 0.65,yscale=0.6]        
        \begin{scope}[every node/.style=vertices]
        \node (U1) at (-3,0) {$1$};
        \node (U2) at (-1,0) {$2$};
        \node (U3) at (1,0) {$3$};
        \node (U4) at (3,0) {$4$};
        \node (U12) at (-5,8) {\footnotesize$\{1,2\}$};
        \node (U13) at (-3,8) {\footnotesize$\{1,3\}$};
        \node (U14) at (-1,8) {\footnotesize$\{1,4\}$};
        \node (U23) at (1,8) {\footnotesize$\{2,3\}$};
        \node (U24) at (3,8) {\footnotesize$\{2,4\}$};
        \node (U34) at (5,8) {\footnotesize$\{3,4\}$};
        \end{scope}
        \begin{scope}[every edge/.style=edges]              
            \path [-] (U1) edge node {} (U12);
            \path [-] (U1) edge node {} (U13);
            \path [-] (U1) edge node {} (U14);
            \path [-] (U2) edge node {} (U12);
            \path [-] (U2) edge node {} (U23);
            \path [-] (U2) edge node {} (U24);
            \path [-] (U3) edge node {} (U13);
            \path [-] (U3) edge node {} (U23);
            \path [-] (U3) edge node {} (U34);
            \path [-] (U4) edge node {} (U14);
            \path [-] (U4) edge node {} (U24);
            \path [-] (U4) edge node {} (U34);

            \path [-] (U1) edge node {} (U2);
            \path [-] (U2) edge node {} (U3);
            \path [-] (U3) edge node {} (U4);
            \path [-] (U4) edge[bend left=75] node {} (U1);
            \path [-] (U1) edge[bend right=68] node {} (U3);
            \path [-] (U2) edge[bend right=68] node {} (U4);
        \end{scope}
\end{tikzpicture}}}
}

\caption{(Left) The graph $C_{4,2}$. (Right) The graph $C'_{4,2}$.}
\label{fig:bikneser}
\end{figure}
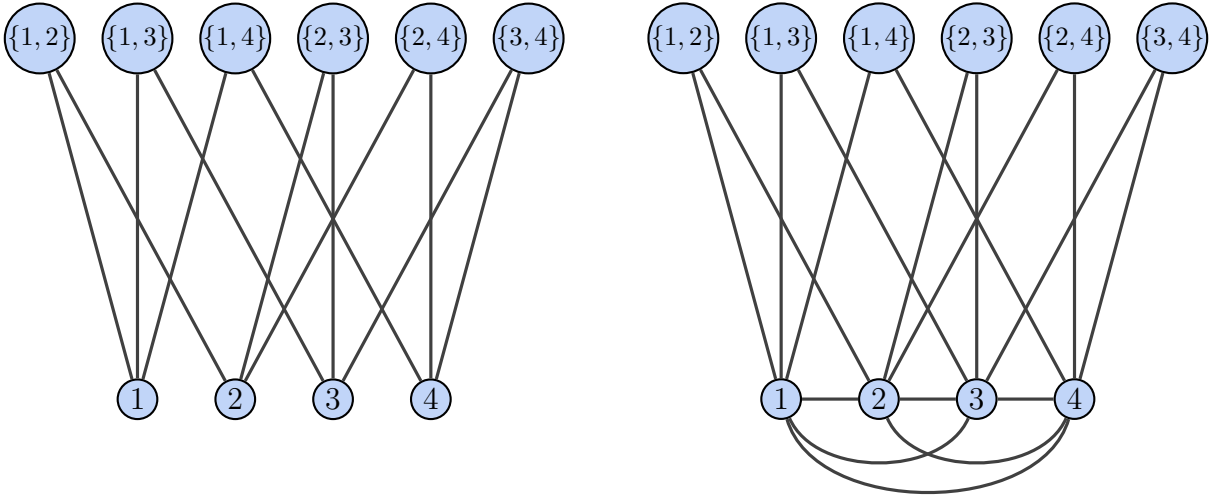

For some parameters $t$ and $k$, $C_{t,k}$ and $C'_{t,k}$ are in standard form with respect to the maximal MLS cover $\mathcal M_{max} = \{L \se V ~|~ L \text{ is a minimal local set}\}$.

\begin{proposition}
    For any odd $k \gs 3$, $t\gs k+2$, $C_{t,k}$ and $C'_{t,k}$ are in standard form, with respect to $\mathcal M_{max}$, for any ordering such that, for any $u \in \binom{[1,t]}{k}$ and for any $v \in [1,t]$, $u < v$. More precisely the vertices in $\binom{[1,t]}{k}$ are of type X and the vertices in $[1,t]$ are of type Z.
\end{proposition}

\begin{proof}
    Fix two integers $k$ and $t$ such that $k \gs 3$ is odd and $t\gs k+2$. 
    First, let us prove by contradiction that a minimal local set cannot be generated by a set containing a vertex in $[1,t]$. Suppose that a set $D \se V$ such that $D \cap [1,t] \neq \emptyset$, generates a minimal local set $L$ in $C_{t,k}$. By symmetry, we suppose without loss of generality that $D \cap [1,t] = [1,m]$ for some $m\in[1,t]$. Let us exhibit a local set strictly contained in $L$, implying a contradiction with the fact that $L$ is a minimal local set:
    \begin{itemize}
        \item If $m=1$ i.e. $D \cap [1,t] = \{1\}$: $L$ contains every vertex of the form $\{\{1\}\cup x\}$ for some $x\in \binom{[2,t]}{k-1}$. Let $D' = \{\{1\}\cup x ~|~ x \in \binom{[2,k+1]}{k-1}\}$. $Odd_{C_{t,k}}(D') = \{1\}$, so $D' \cup Odd_{C_{t,k}}(D')$ is a local set contained strictly in $L$ ($L$ contains other vertices, e.g. the vertex $\{1,4,5,\cdots,k+2\}$).
        \item If $m\gs 2$: let $m_\text{odd}$ be the minimum between $k$ and the largest odd integer strictly less than $m$. Let $D' = \{\{1,\cdots,m_\text{odd},m+1,\cdots,m+k-m_\text{odd}\},\{2,\cdots,m_\text{odd}+1,m+1,\cdots,m+k-m_\text{odd}\}\}$. $Odd_{C_{t,k}}(D') = \{1,m_\text{odd}+1\}$, so $D' \cup Odd_{C_{t,k}}(D')$ is a local set contained strictly in $L$ ($L$ contains other vertices, e.g. the vertex $\{1,3,\cdots,m_\text{odd}+1,m+1,\cdots,m+k-m_\text{odd}\}$).
    \end{itemize} 
    Thus, every minimal local set is generated by a set in $\binom{[1,t]}{k}$. As every vertex is covered by a minimal local set by \cref{thm:MLS_cover}, it follows that each vertex in $\binom{[1,t]}{k}$ is of type X and each vertex in $[1,t]$ is of type Z. The proof is the same for $C'_{t,k}$ as $C_{t,k}$ and $C'_{t,k}$ differ only by edges between vertices of $[1,t]$ thus for any set $D \se V$, $(D \cup Odd_{C_{t,k}}(D)) \cap \binom{[1,t]}{k} = (D \cup Odd_{C'_{t,k}}(D)) \cap \binom{[1,t]}{k}$.
\end{proof}

The following proposition gives a sufficient condition on $t$ and $k$ for the LC$_{r}$-equivalence of $C_{t,k}$ and $C'_{t,k}$.

\begin{proposition}\label{lemma:cond_lcr}
    For any $t\gs k\gs r+2$, $C_{t,k}$ and $C'_{t,k}$ are LC$_{r}$-equivalent if
    \begin{equation*}
            \binom{t - 2}{k - 2}  = 2^{r-1}\bmod 2^r
          \text{~~~and~~~}  \forall i \in [ 1, r-1 ]\text{, ~~}\binom{t - i -2}{k - i -2}  = 0\bmod 2^{r-i}
    \end{equation*}
\end{proposition}

\begin{proof}
    $C_{t,k}$ and $C'_{t,k}$ are related by an $r$-local complementation over the set $\binom{[1,t]}{k}$ (which can be seen as the multiset where each vertex of $\binom{[1,t]}{k}$ appears once). Let $K \se [1,t]$ of size $k'+2$. $\binom{[1,t]}{k}\bullet \Lambda_{C_{t,k}}^K = \left|\left\{x \in \binom{[1,t]}{k}~|~K\se x\right\}\right| = \binom{t-k'-2}{k-k'-2}$ is a multiple of $2^{r-k'-\delta(k')}$ by hypothesis. Thus, $\binom{[1,t]}{k}$ is $r$-incident. Besides, for any $u,v \in [1,t]$, $\binom{[1,t]}{k} \bullet \Lambda_{C_{t,k}}^{u,v} = \binom{t-2}{k-2} = 2^{r-1}\bmod 2^{r}$ by hypothesis. Thus, $C_{t,k}\star^r\binom{[1,t]}{k} = C'_{t,k}$.
\end{proof}

The following proposition provides a sufficient condition on $t$ and $k$ for the non LC$_{r}$-equivalence of $C_{t,k}$ and $C'_{t,k}$.

\begin{proposition}
    \label{lemma:cond_not_lcr}
    For any odd $k \gs 3$, $t\gs k+2$, $C_{t,k}$ and $C'_{t,k}$ are not LC$_{r}$-equivalent if $\binom{t}{2}$ is odd and $\binom{k}{2} = 0\bmod 2^{r}$.
\end{proposition}

\begin{proof}       
    Let us suppose by contradiction that $C_{t,k}$ and $C'_{t,k}$ are LC$_{r}$-equivalent. 
    By \cref{cor:constraint_r_lc}, they are related by a single $r$-local complementation over vertices in $\binom{[1,t]}{k}$. Let $S$ be the multiset in $\binom{[1,t]}{k}$ such that $C_{t,k}\star^rS = C'_{t,k}$. For each $u,v \in [1,t]$, $S \bullet\Lambda_{C_{t,k}}^{u,v} = \sum_{x \in \Lambda_{C_{t,k}}^{u,v}}S(x) = 2^{r-1}\bmod 2^{r}$. Summing over all pairs $u,v \in [1,t]$, and, as by hypothesis $\binom{t}{2}$ is odd, $\sum_{u,v\in [1,t]} S \bullet\Lambda_{C_{t,k}}^{u,v} = \binom{k}{2}\sum_{x \in \binom{[1,t]}{k}}S(x)  = \binom{t}{2}2^{r-1}\bmod 2^r= 2^{r-1} \bmod 2^{r}$. The first part of the equation is true because $\sum_{u,v\in [1,t]} S \bullet\Lambda_{C_{t,k}}^{u,v} = \sum_{u,v\in [1,t]} \sum_{x \in \Lambda_{C_{t,k}}^{u,v}}S(x) = \sum_{x \in \binom{[1,t]}{k}} \sum_{u,v\in x} S(x) $. This leads to a contradiction as $\binom{k}{2} = 0 \bmod 2^{r}$ by hypothesis. 
\end{proof}

It remains to find, for any $r$, parameters $t$ and $k$ such that the corresponding graphs are LC$_{r}$-equivalent but not LC$_{r-1}$-equivalent. Such parameters exist.

\begin{theorem}
    For any $r \gs 2$, $\ket{C_{t,k}}$ and $\ket{C'_{t,k}}$ are LC$_{r}$-equivalent but not LC$_{r-1}$-equivalent when $k = 2^r +1$ and $t = 2^r + 2^{\lceil\log_2(r+2) \rceil} -1$. 
\label{thm:existence_Ctk}
\end{theorem}

We begin by proving the following technical lemma.

\begin{lemma}
    Given any integers $s \ls m$, $\binom{m}{s} = 2^{w(s)+w(m-s)-w(m)}q$ where $q$ is an odd integer and $w(i)$ denotes the Hamming weight of the integer $i$.
\label{lemma:paritybinom}
\end{lemma}

\begin{proof}
    
    We prove by induction that for any $n \gs 0$, $n! = 2^{n-\ham n}q$ where $q$ is an odd integer. The property is trivially true for $n=0$. Given $n>0$, by induction hypothesis, there exists $q=1\bmod 2$ s.t. $n! = n (n-1)! = n2^{n-1-w(n-1)}q=2^kq'2^{n-1-w(n-1)}q= 2^{n-1+k-w(n-1)}qq'$, where $n=2^kq'$ and $q'=1\bmod 2$. Notice that $w(n)=w(q')$, as the binary notation of $n$ is nothing but the binary notation of $q'$ concatenated with $k$ zeros. Moreover, $w(n-1)=w(q')-1+k = w(n)-1+k$ as $n-1= 2^k(q'-1)+ \sum_{j=0}^{k-1}2^j$. Therefore, $n!=2^{n-w(n)}qq'$. 
   
    As a consequence there exist three odd integers $q_1$, $q_2$ and $q_3$ such that \begin{equation*}
    \binom{m}{s} = \dfrac{m!}{s! (m-s)!}=\dfrac{2^{m-w(m)}q_1}{2^{s-w(s)}q_2 2^{m-s-w(m-s)}q_3} = 2^{w(s)+w(m-s)-w(m)}\dfrac{q_1}{q_2 q_3} \qedhere
    \end{equation*}
\end{proof}

\begin{proof}[Proof of \cref{thm:existence_Ctk}]
    \cref{lemma:cond_lcr} along with \cref{lemma:cond_not_lcr} imply that for any odd $k \gs 3$ and $t\gs k+2$, $C_{t,k}$ and $C'_{t,k}$ are LC$_{r}$-equivalent but not LC$_{r-1}$-equivalent if
    \begin{equation*}
        \begin{split}
            \binom{t}{2} & = 1\bmod 2\\  
            \binom{k}{2} & = 0\bmod2^{r-1}\\ 
            \binom{t - 2}{k - 2} & = 2^{r-1}\bmod 2^r\\   
            \forall i \in [ 1, r-1 ]\text{, ~~}\binom{t - i -2}{k - i -2} & = 0\bmod 2^{r-i}                     
        \end{split}
    \end{equation*}  

    By \cref{lemma:paritybinom}, it is equivalent to prove the following equations:

    \begin{equation*}
        \begin{split}
            & \ham{t} - \ham{t-2}  = 1\\
            & \ham{k-2} - \ham{k}  \gs r-2 \\
            & \ham{k-2}+\ham{t-k}-\ham{t-2} =r -1\\   
            \forall i \in [ 1, r-1 ]\text{, ~~}& \ham{k-2 -i}+\ham{t-k} -\ham{t-2-i} \gs r -i            
        \end{split}
    \end{equation*}

    Let us prove that these equations are satisfied for $k = 2^r +1$ and $t = 2^r + 2^{\lceil \log_2(r+2) \rceil} -1$. We will make use of the following property of the Hamming weight: if $x\ls 2^z-1$, then $\ham{2^z-1-x}=z-\ham{x}$. In particular, for any $x \ls r+1 \ls 2^{\lceil \log_2(r+2) \rceil}-1$, $\ham{t-x}=\lceil \log_2(r+2) \rceil +1-\ham{x}$.
    \begin{itemize}
        \item $\ham{t}=\lceil \log_2(r+2) \rceil +1$, $\ham{t-2}=\lceil \log_2(r+2) \rceil$, thus $\ham{t} - \ham{t-2}  = 1$.
        
        \item $\ham{k}=2$ and $\ham{k-2}=\ham{2^r-1}=r$, thus $\ham{k-2} - \ham{k} = r-2$.
        
        \item $\ham{t-k} = \ham{2^{\lceil \log_2(r+2) \rceil}-2} = \lceil \log_2(r+2) \rceil -1$. Thus, $\ham{k-2}+\ham{t-k}-\ham{t-2}= r + \lceil \log_2(r+2) \rceil -1 - \lceil \log_2(r+2) \rceil = r -1$.
        
        \item Let $i \in [ 1, r-1 ]$. $\ham{k-2-i} = \ham{2^r-1-i}=r-\ham{i}$ and $\ham{t-2-i} = \ham{t -(i+2)} =  \lceil \log_2(r+2) \rceil +1 - \ham{i+2}$. Thus, $\ham{k-2 -i}+\ham{t-k} -\ham{t-2-i} = r-\ham{i} + \lceil \log_2(r+2) \rceil -1 - \left(\lceil \log_2(r+2) \rceil +1 - \ham{i+2}\right) = r - \left(\ham{i} - \ham{i+2} + 2\right)$. It is easy to check that $\ham{i} - \ham{i+2} + 2 \ls i$. Thus, $\ham{k-2 -i}+\ham{t-k} -\ham{t-2-i} \gs r -i$.
    \end{itemize}

Hence, $C_{t,k}$ and $C'_{t,k}$ are LC$_{r}$-equivalent but not LC$_{r-1}$-equivalent. By \cref{thm:LCr_lc}, $\ket{C_{t,k}}$ and $\ket{C'_{t,k}}$ are LC$_{r}$-equivalent but not LC$_{r-1}$-equivalent.    
\end{proof}

\begin{remark}
    For the case $r=2$, we obtain the pair $(C_{7,5},C'_{7,5})$, which is the 28-vertex counter-example to the LU-LC conjecture from \cite{Tsimakuridze17}. This is an example of graphs that are LC$_2$-equivalent but not LC-equivalent.
\end{remark}

While Ji et al. proved that there exist pairs of graphs that are LU-equivalent but not LC-equivalent \cite{Ji07}, we showed a finer result -- the existence of an infinite strict hierarchy of graph states equivalence between LC- and LU-equivalence. There exist LC$_2$-equivalent graphs that are not LC-equivalent, LC$_3$-equivalent graphs that are not LC$_2$-equivalent, and so on. On the other end, for any integer $r$, there exist LU-equivalent graphs that are not LC$_r$-equivalent.
\chapter{A quasi-polynomial algorithm for LU-equivalence}

\label{chap:algo}

In this chapter, we address the problem of deciding whether two given graphs are LU-equivalent. Additionally, we  consider a variant of this problem consisting in deciding whether two graphs are LC$_r$-equivalent for a fixed $r$. Since LU-equivalent graphs are necessarily LC$_r$-equivalent for some $r$ (see \cref{thm:LU_imply_LCr}), the difference lies in whether the level $r$ is fixed or not. 

Thanks to \cref{thm:LCr_lc}, if $G_1$ is LC$_r$-equivalent to $G_2$, there exists a single $r$-local complementation, together with usual local complementations, that transforms $G_1$ into $G_2$. We introduce an algorithm that builds such a sequence of generalized local complementations, in essentially four stages: 
\begin{itemize}
\item[$(i)$] Both $G_1$ and $G_2$ are turned in standard forms $G_1'$ and $G'_2$ by means of (usual) local complementations.  
\item[$(ii)$] We then focus on the single $r$-local complementation: all the possible actions of a single $r$-local complementation over $G'_1$ are described as a vector space, for which we compute a basis $\mathcal B$.  
\item[$(iii)$] It remains to find, if it exists, the $r$-local complementation to apply on $G'_1$ that leads to $G'_2$ up to some additional usual local complementations. With an appropriate construction depending on $G_1'$, $G'_2$ and $\mathcal B$, we reduce this problem to deciding whether two graphs are LC-equivalent under some additional requirements  on the sequence of local complementations to apply. These requirements can be expressed as linear constraints.
\item [$(iv)$] Finally, to find such a sequence of local complementations, we apply a variant of Bouchet's algorithm, generalized to accommodate the additional linear constraints. 
\end{itemize}

Stages $(i)$, $(iii)$ and $(iv)$ can be performed in polynomial time in the order $n$ of the graphs. Stage $(ii)$ has essentially an $O(n^r)$ time complexity, thus deciding LC$_r$-equivalence for a fixed $r$ can be done in polynomial time. Regarding LU-equivalence, \cref{thm:LU_imply_LCr} implies $r\ls n$. We improve this upperbound and show that $r$ is at most logarithmic in $n$, leading to a quasi-polynomial time algorithm for LU-equivalence.

This chapter is dedicated to the description of the algorithm, its correctness and complexity, beginning with the generalization of Bouchet's algorithm to decide, in polynomial time, LC-equivalence with additional constraints.

\section{Extending Bouchet's algorithm for LC-equivalence}

In this section, we consider an extension of Bouchet's algorithm (see \cref{subsec:bouchet}) where an additional set of linear constraints on $A,B,C$ and $D$ is added as input of the problem. Such additional linear equations can reflect constraints on the applied  local complementations, e.g.~deciding whether two graphs are LC-equivalent under the additional constraint that all local complementations are applied on a fixed set $V_0$ of vertices (see \cref{ex:V0}).

Recall that Bouchet's algorithm for deciding LC-equivalence of $G$ and $G'$ consists in solving for sets $A,B,C,D\subseteq V$ such that
\begin{itemize}
\item[$(i)$] $\forall u,v \in V$,\\
$|B\cap N_G(u)\cap N_{G'}(v)| + |A\cap N_G(u)\cap \{v\}| +  |D\cap \{u\}\cap N_{G'}(v)| + |C\cap \{u\}\cap \{v\}| = 0\bmod 2$
\item[$(ii)$]
$(A\cap D)~\Delta~ (B\cap C) = V$ 
\end{itemize}

While solving linear equations is easy, equation $(ii)$ is not linear. 
Bouchet showed that a set of solutions $\mathcal C\subseteq \mathcal S$ that satisfies both $(i)$ and $(ii)$,  is either small or it contains an affine subspace of $\mathcal S$ of small co-dimension. In  the latter case, the entire set $\mathcal S$ is  actually an affine space except for  some particular cases that can be avoided by assuming that the graphs contain vertices of even degree. We extend this result as follows:

\begin{lemma}\label{lemma:codim2}
Given $G$, $G'$ two connected graphs with an even-degree vertex, and a set $L$  of linear constraints on $V^4$, then either the set $\mathcal S_L$ of solutions to both $ L$ and $(i)$ is of dimension at most $4$, or the set  $\mathcal C_L\subseteq \mathcal S_L$ that additionally satisfies  $(ii)$ is either empty or an affine subspace of $\mathcal S_L$ of codimension at most $2$.   
\end{lemma}

\begin{proof} It is enough to consider the case $\dim(\mathcal S_L)>4$ and $\mathcal C_L\neq \emptyset$. The set  $\mathcal S$ of solutions to $(i)$ contains $\mathcal S_L$, so $\dim(\mathcal S)>4$. According to \cite{Bouchet1991}, $\mathcal C$, the solutions to $(i)$ and $(ii)$, is an affine subspace of $\mathcal S$ of codimension at most $2$. For any  $a\in \mathcal C_L$, $\mathcal C = a+ \mathcal N$ where $ \mathcal N$ is a subvector space of $\mathcal S$. Notice that $\mathcal C_L = \mathcal C\cap \mathcal S_L = a+ \mathcal N\cap \mathcal S_L $. We have $\dim(\mathcal C_L)=\dim(\mathcal N\cap \mathcal S_L ) = \dim(\mathcal N)+\dim(\mathcal S_L )-\dim(\mathcal N + \mathcal S_L )\gs \dim(\mathcal N)+\dim(\mathcal S_L )-\dim(\mathcal S)\gs \dim(\mathcal S_L)-2$ as $\mathcal N$ is of codimension at most $2$ in $\mathcal S$. 
\end{proof}

\begin{remark} \cref{lemma:codim2} holds actually for any graph that is not in the so-called "Class $\alpha$" of graphs with only odd-degree vertices together with a few additional properties\footnote{(a) any pair of unadjacent vertices should have an even number of common neighbors; (b) for any cycle $C$, the number of triangles having an edge in $C$ is equal to the size of $C$ modulo $2$.}.  When the graphs are in  "Class $\alpha$", and in the absence of additional constraints, Bouchet proved that there is at most $2$ solutions in $\mathcal C$ that do not belong to the affine subspace of small codimension, and these two solutions can be easily computed (see \cite{Bouchet1991}, section 7).  We leave as an open question the description of the set of solutions for graphs in "Class $\alpha$", in particular when the two particular solutions pointed out by Bouchet do not satisfy $L$, the set of additional constraints. 
\end{remark}
From an algorithmic point of view, \cref{lemma:codim2} leads to a straightforward generalization of Bouchet's algorithm to  efficiently decide LC-equivalence of graphs, under a set of additional linear constraints: 

\begin{proposition}\label{prop:extendedBouchet}
Given $G$, $G'$ two connected graphs of order $n$ with an even-degree vertex, and a set $L$ of $\ell$ linear constraints on $V^4$, one can compute a solution to both $(i)$, $(ii)$ and $L$ when it exists, or claim there is no solution, in runtime $O((n^2+\ell)n^2)$.
\end{proposition}

\begin{proof} 
A Gaussian elimination can be used to compute a basis $\mathcal B = \{S_0, \ldots , S_k\}$ of $\mathcal S_L$, with $k<n$, in $O((n^2+\ell)n^2)$ operations as there are $n^2$ equations in $(ii)$. If the dimension of $\mathcal S_L$ is at most $4$ (so $|\mathcal S_L|\ls 16$), we check in $O(n)$ operations, for each element of $\mathcal S_L$ whether equation $(ii)$ is satisfied. Otherwise, when $\dim(\mathcal S_L)>4$, if $\mathcal C_L$ is non-empty, at least one element of $\mathcal C_L$ is the sum $S_i +  S_j$ of two  elements of $\mathcal B$ (see Lemma 4.4 in \cite{Bouchet1991}). For each of the $O(n^2)$ candidates we check  whether condition $(ii)$ is satisfied.  If no solution is found, it implies that $\mathcal C_L=\emptyset$. 
\end{proof}

As the algorithm described in the proof of \cref{prop:Bouchet} translates a solution to $(i)$ and $(ii)$, into a sequence of local complementations relating two graphs, some constraints on the sequence of local complementations may be encoded as additional linear constraints. We give a fairly simple example below. A more intricate example is presented in \cref{lemma:LC_new_graphs} (in \cref{sec:algo_lcr}).

\begin{example} \label{ex:V0}
    Let $G$, $G'$ be two connected graphs of order $n$ with an even-degree vertex, and $V_0$ a set of vertices. One can decide in runtime $O(n^4)$ whether there exists a sequence of (possibly repeating) vertices $a_1, \cdots, a_m\in V_0 $ such that $G' = G \star a_1 \star \cdots \star a_m$. Roughly speaking, the idea is to consider the linear constraint $b_u=0$ (i.e.~$u\in   \overline B$) for any $u\notin V_0$, to reflect the constraints that local complementations should not be applied outside $V_0$.
\end{example}

From a graph state point of view, \cref{prop:extendedBouchet} provides an efficient algorithm to decide whether two graph states are LC-equivalent under some constraints on the Clifford operators. Notice that such constraints are expressible as a linear equation through the correspondence given in \cref{prop:clifford_ABCD}. We give below a non-exhaustive family of constraints expressible as linear equations (in the following, $k$ denotes an integer).

\begin{itemize}
    \item $u \notin B$:~ $U_u$ is $Z(k\pi/2)$ up to Pauli;
    \item $u \notin C$:~ $U_u$ is $X(k\pi/2)$ up to Pauli;
    \item $u \notin A$:~ $U_u$ is $Z(k\pi/2) H$ up to Pauli;
    \item $u \notin D$:~ $U_u$ is $X(k\pi/2) H$ up to Pauli;
    \item $u \notin \overline{B}\cap \overline{C}$:~ $U_u$ is a Pauli;
    \item $u \notin \overline{A}\cap \overline{D}$:~ $U_u$ is $H$ up to a Pauli;
    \item $u \in A$ iff $u \in D$:~ $U_u$ is I, $X(\pi/2)$, $Z(\pi/2)$ or $H$ up to Pauli, i.e.~$u^2$ is a Pauli;
    \item $u \in A$ iff $u \in B$:~ $U_u$ is $X(\pi/2)$ or $X(k\pi/2) H$ up to Pauli;
    \item $u \in A$ iff $v \in A$, $u \in B$ iff $v \in B$, $u \in C$ iff $v \in C$, $u \in D$ iff $v \in D$:~ $U_u=U_v$ up to Pauli.
\end{itemize}

\section{An algorithm to recognize \texorpdfstring{LC$_r$}{LCr}-equivalent graph states} \label{sec:algo_lcr}

We are now ready to describe the algorithm that recognizes two LC$_r$-equivalent graphs. We consider a level $r \gs 1$, and two graphs $G_1$ and $G_2$ of order $n$, defined on the same vertex set $V$. Following \cref{lemma:standardform}, assume, without loss of generality, that $G_1$ and $G_2$ are both in standard form with respect to the same MLS cover. Then, it is valid (see \cref{lemma:same_types}) to define $V_X, V_Z \se V$, the sets of vertices respectively of type X and Z with respect to the MLS cover. Also, $V_X$ is independent and each vertex in $V_X$ has the same neighborhood in both $G_1$ and $G_2$. According to \cref{lemma:standardform_LCr_lc}, if $G_1$ and $G_2$ are LC$_r$-equivalent then there is a single $r$-local complementation over vertices of $V_X$ together with a series of local complementations on vertices of type $\bot$ that transform $G_1$ into $G_2$. We first focus on the single $r$-local complementation (which commutes with the local complementations on vertices of type $\bot$, as vertices of type X and vertices of type $\bot$ share no edge) and thus consider all the possible graphs that can be reached from $G_1$ by mean of a single $r$-local complementation over vertices of $V_X$. 
Notice that such an $r$-local complementation only toggles edges for which both endpoints are in $V_Z$. Given a multiset $S$, the edges toggled in $G_1\star^r S$ can be represented by a vector $\omega^{(S)}\in {{\mathbb F}_2}^{\{u,v \in V_Z ~|~ u \neq v\}}$ such that for any $u,v \in V_Z$, $\omega^{(S)}_{u,v} = u\sim_{G_1} v ~\oplus~ u\sim_{G_1\star^r S} v$. The actions of all the possible $r$-local complementations on $G_1$ can thus be described as the set $\Omega =\{\omega^{(S)}  |  \text{$S$ is an $r$-incident multiset of vertices of type X}\}$.

\begin{lemma} \label{lemma:omega}
    $\Omega$ is a vector space and a basis $\mathcal B$ of $\Omega$ can be computed in runtime $O(r n^{r+2.38})$.
\end{lemma}

\begin{proof}
    A multiset $S$ of vertices of type X can be represented as a vector in $({\Zp{2^r}})^{V_X}$ which entries are $S(u) \bmod 2^{r}$, as the multiplicity can be considered modulo $2^r$ in the context of an $r$-local complementation. Let $\Sigma$ be the subset of $({\Zp{2^r}})^{V_X}$ that corresponds to all $r$-incident independent multisets of vertices in $V_X$. $\Sigma$ is a vector space since the property of $r$-incidence is preserved under the addition of two vectors, as $V_X$ is independent. $\Sigma$ is actually the space of solutions to a set of $O(n^{r+1})$ equations\footnote{There are precisely $\binom{|V_X|}{r+1}+\binom{|V_X|}{r}+\cdots+\binom{|V_X|}{3}+\binom{|V_X|}{2}$ equations.} given by the conditions of $r$-incidence. With some multiplications by powers of 2, all these equations are expressible as equations modulo $2^{r}$. For every set  $K\subseteq V\setminus \supp(S)$ of size between $2$ and $r+1$, the corresponding equation is $$\sum_{u \in \Lambda_{G_1}^K}2^{|K|-2+\delta(|K|-2)}S(u) = 0\bmod 2^{r}$$
    Notice that if $r=1$, the space of solutions is the entire space $({\Zp{2^r}})^{V_X}$, as there are no incidence constraints on usual local complementations. Summing up, to compute a basis of $\Sigma$, we solve a system of $O(n^{r+1})$ equations modulo $2^{r}$ with $O(n)$ variables. One can obtain a generating set $\{S_1,S_2,...,S_t\}$ of $\Sigma$ of size $t \ls n$ in $O(r n^{r+2.38})$ basic operations using an algorithm based on the Howell transform \cite{storjohann2000algorithms}. 

    Let $f$ be the function that associates with each element $S \in \Sigma$, its action $f(S) \in \Omega$ on the edges: for any $u,v \in V_Z$, $f(S)_{u,v} = u\sim_{G_1} v ~\oplus~ u\sim_{G_1\star^r S} v$. Notice that $f$ is linear, as for any $S, S' \in \Sigma$:
    \begin{align*}
        f(S+S') &= u\sim_{G_1} v ~\oplus~ u\sim_{G_1\star^r (S+S')} v
        = u\sim_{G_1} v ~\oplus~ u\sim_{(G_1\star^r S) \star^r S'} v\\
        &= u\sim_{G_1} v ~~\oplus~~ u\sim_{G_1\star^r S} v ~~\oplus~~ \left(S' \bullet\Lambda_{G_1\star^r S}^{u,v} = 2^{r-1}\bmod 2^{r}\right)\\
        &= u\sim_{G_1} v ~~\oplus~~ u\sim_{G_1\star^r S} v ~~\oplus~~ \left(S' \bullet\Lambda_{G_1}^{u,v} = 2^{r-1}\bmod 2^{r}\right)\\
        &= u\sim_{G_1 \star^r S} v ~~\oplus~~ u\sim_{G_1\star^r S'} v\\
        &= u\sim_{G_1} v ~\oplus~ u\sim_{G_1 \star^r S} v ~~\oplus~~ u\sim_{G_1} v ~\oplus~ u\sim_{G_1\star^r S'} v = f(S)+f(S')
    \end{align*} 

    This directly implies that $\Omega$ is a vector space: if $\omega, \omega' \in \Omega$, by definition there exist $S, S' \in \Sigma$ such that $f(S) = \omega$ and $f(S') = \omega$, moreover $\omega + \omega' = f(S) + f(S') = f(S+S') \in \Omega$. Also, let us prove that $\Omega$ is generated by $\{f(S_1),f(S_2),...,f(S_t)\}$. Take a vector $\tilde \omega \in \Omega$, by definition there exists $\tilde S \in \Sigma$ such that $\tilde \omega = f(\tilde S)$. $\tilde S$ can be expressed as a linear combination of vectors from the generating set, i.e.~$\tilde S = \sum_{i \in [1,t]}a_i S_i$ where $a_i \in \Zp{2^r}$. Then, for any $u,v \in V_Z$, $ (\tilde \omega)_{u,v} = f\left(\sum_{i \in [1,t]}a_i S_i\right)_{u,v} = \sum_{i \in [1,t]}a'_i f(S_i)_{u,v}$ by linearity of $f$, where $a'_i \in \mathbb F_2$ such that $a'_i = a_i \bmod 2$ when $a_i$ and $a'_i$ are viewed as integers. In other words, $\tilde w = \sum_{i \in [1,t]}a'_i f(S_i)$, implying that $\{f(S_1),f(S_2),...,f(S_t)\}$ is a generating set of $\Omega$. Using Gaussian elimination, one can easily obtain a basis $\mathcal B$ of $\Omega$ from $\{f(S_1),f(S_2),...,f(S_{t})\}$.
\end{proof}

Thanks to the exhaustive description of all possible $r$-local complementations on $G_1$, we are now ready to reduce LC$_r$-equivalence to LC-equivalence with some additional constraint. We denote by $G_1^{\#}$ (resp. $G_2^{\#}$) the graph obtained from $G_1$ (resp. $G_2$) by the following procedure. First, remove the vertices of $V_X$. Then, for each vector $\omega \in \mathcal B$, for each $u,v \in V_Z$ such that $\omega_{u,v} = 1$, add a vertex adjacent only to $u$ and $v$ and call it $p_{u,v}^{\omega}$, and let $\mathcal P^\omega= \{p_{u,v}^{\omega} ~|~ \omega_{u,v}=1\}$. In the following, we refer to the vertices added by this procedure as "new vertices".

\begin{lemma} \label{lemma:new_graphs}
    $G_1$ and $G_2$ are LC$_r$-equivalent if and only if there exists a sequence of (possibly repeating) vertices $a_1, \cdots, a_m$ such that $G^{\#}_2 = G^{\#}_1 \star a_1 \star \cdots \star a_m$ satisfying the following additional constraints:
    \begin{itemize}
        \item the sequence contains no vertex of $V_Z$;
        \item for each $\omega \in \mathcal B$, either the sequence contains every vertex of $\mathcal P^{\omega}$ exactly once, or it contains none. 
    \end{itemize}
\end{lemma}

\begin{proof}
    Suppose $G_1=_{LC_r} G_2$. According to \cref{lemma:standardform_LCr_lc}, $G_1$ and $G_2$ are related by a single $r$-local complementation over a multiset $S$ whose support lies in $V_X$, along with a sequence of local complementations on the vertices of $V \sm (V_X \cup V_Z)$. We note $G_2 = G_1 \star^r S \star u_1 \star \cdots u_k$. The edges toggled by an $r$-local complementation over $S$ are described by an element $\omega \in \Omega$. Let us decompose $\omega$ as a linear combination of basis vectors of $\Omega$: $\omega = \omega_1 + \cdots + \omega_t$, where each $\omega_i \in \mathcal B$. In $G^{\#}_1$ and $G^{\#}_2$, let $V_\omega = \bigcup_{i \in [1,t]} \mathcal P^{\omega_i}$.
    Then, $G^{\#}_2 = G^{\#}_1  \star^1 V_\omega \star u_1 \star \cdots u_k$. Note that the 1-local complementation over the set $V_\omega$ corresponds to the composition of local complementations on each element of $V_\omega$. Thus, $G_1$ is mapped to $G_2$ by a sequence of local complementations that satisfy the additional constraints.

    Conversely, suppose there exists a sequence of (possibly repeating) vertices $a_1, \cdots, a_m$ such that $G^{\#}_2 = G^{\#}_1 \star a_1 \star \cdots \star a_m$ satisfying the additional constraints. As local complementations on new vertices commute with each other and with local complementations on vertices of $V \sm (V_X \cup V_Z)$, one can take apart the vertices of the sequence among the new vertices, so there exist a set $V_0$ of new vertices and vertices $u_i$ in $V \sm (V_X \cup V_Z)$ such that $G^{\#}_2 = G^{\#}_1  \star^1 V_0 \star u_1 \star \cdots u_k$. The additional constraints imply that $V_0$ is a union of sets of vertices corresponding respectively to some elements $\omega_i \in \mathcal B$. Let $\omega \in \Omega$ be the sum of these elements. By construction, there exists a multiset $S$ in the original graphs whose action is described by $\omega$, implying $G_2 = G_1 \star^r S \star u_1 \star \cdots u_k$. Thus, $G_1 =_{LC_r} G_2$.
\end{proof}

There exists an efficient algorithm that decides whether two graphs are LC-equivalent with such additional constraints using our generalization of Bouchet's algorithm.

\begin{lemma} \label{lemma:LC_new_graphs}
    Deciding whether there exists a sequence of (possibly repeating) vertices $a_1, \!\cdots\!,\! a_m$ such that $G^{\#}_2 = G^{\#}_1 \star a_1 \star \cdots \star a_m$, satisfying the additional constraints described in \cref{lemma:new_graphs}, can be done in runtime $O(n^4)$.
\end{lemma}

\begin{proof}
    If the vector space $\Omega$ is of dimension zero (i.e.~$\Omega$ only contains the null vector), then there is no additional constraint, thus one can apply the usual Bouchet algorithm that decides LC-equivalence of graphs.

    If $\Omega$ is not of dimension zero, then $G^{\#}_1$ and $G^{\#}_2$ both have at least one even-degree vertex (since every "new vertex" is of degree 2). Using the notations of \cref{subsec:bouchet}, let us define the following linear constraints on $V^4$:  $\forall u \in V_Z$, $\forall \omega \in \mathcal B$, $\forall v,v'\in \mathcal P^{\omega}$, 
    \begin{itemize}
        \item $ u \in \overline B$;
        \item $v \in \overline C$;
        \item $v \in B$ if and only if $v' \in B$.
    \end{itemize} 
    According to \cref{prop:extendedBouchet}, a solution to the system of equations composed of $(i)$, $(ii)$ (see \cref{prop:Bouchet}) and the additional linear constraints can be computed in runtime $O(n^4)$ when it exists. Following the algorithm in the proof of \cref{prop:Bouchet}, such a solution yields a sequence of local complementations satisfying the additional constraints described in \cref{lemma:new_graphs}. Conversely, such a sequence of local complementations can be converted into a valid solution to the system of equations.
\end{proof}

Summing up, we have an algorithm that decides, for a fixed level $r$,  the LC$_r$-equivalence of graphs in polynomial runtime.

\begin{theorem} \label{thm:algolcr}
    There exists an algorithm that decides if two graph states are LC$_r$-equivalent with runtime $O(r n^{r+2.38} + n^{6.38})$, where $n$ is the number of qubits.
\end{theorem}

The algorithm reads as follows:
\begin{enumerate}
    \item Put $G_1$ and $G_2$ in standard form with respect to the same MLS cover if possible, otherwise output NO.
    \item Check whether each vertex has the same type in $G_1$ and $G_2$, and whether every vertex of type X has the same neighborhood in both graphs, otherwise output NO.
    \item Compute a basis of the vector space $\Omega$.
    \item Compute the graphs ${G}^{\#}_1$ and ${G}^{\#}_2$.
    \item Decide whether ${G}^{\#}_1$ and ${G}^{\#}_2$ are LC-equivalent with the additional constraints described in \cref{lemma:new_graphs}. Output YES if this is the case,  NO otherwise.
\end{enumerate}

Notice that the algorithm is exponential in $r$, in particular it does not provide an efficient algorithm to decide LU-equivalence of graph states. To address this issue, we provide in the next section some upper bounds on the level of a generalized local complementation.

\section{Bounds for generalized local complementation}

In this section, we prove an upper bound on the level of a valid generalized local complementation: roughly speaking we show that if $G\star^r S$ is valid then $r$ is at most logarithmic in the order $n$ of the graph $G$. This bound is however not true in general as 
whenever $G\star^r S$ is valid, we have $G\star^r S = G\star^{r+1} (2S)$, according to \cref{prop:monotonicity}. 
To avoid these pathological cases, we thus focus on genuine $r$-incident independent multisets:

\begin{definition}
Given a graph $G$, an $r$-incident independent multiset $S$ is \emph{genuine} if there exists a set $K \se V \sm \supp(S)$ such that $|K|>1$ and $\sum_{N_{G}(u)=K}S(u)$ is odd\footnote{With a slight abuse of notation, $\sum_{N_{G}(u)=K}S(u)$ is the sum over all $u\in V$ s.t. $N_{G}(u)=K$.}. 
\end{definition}

\begin{proposition} \label{prop:nontrivial}
If $G\star^r S$ is valid and there is no $(r-1)$-incident independent multiset $S'$ such that $G\star^{r-1} S' = G \star^r S$, then $S$ is a genuine $r$-incident independent multiset.    
\end{proposition}

\begin{proof} By contradiction, assume $S$ is an $r$-incident independent multiset that is not genuine. Let $S'$ be the multiset obtained from $S$ by choosing, for every set $K \se V \sm \supp(S)$ s.t. $\{u \in \supp(S)~|~N_{G}(u)=K\}$ is not empty, a single vertex $u \in \supp(S)$ s.t. $N_{G}(u)=K$, and setting $S'(u)=\sum_{N_{G}(u)=K}S(u)$ and for any other vertex $v \in \supp(S)$ s.t. $N_{G}(v)=K$, $S'(v)=0$. It is direct to show that $S'$ is $r$-incident and that $G\star^r S = G\star^r S'$. Then, let $S'/2$ be the multiset obtained from $S'$ by dividing by 2 the multiplicity of each vertex in $\supp(S')$. It is direct to show that $S'$ is $(r-1)$-incident and that $G\star^r S= G\star^{r-1} S'/2$.   
\end{proof}

Genuine $r$-incidence can only occur for multisets whose support is of size at least exponential in $r$.

\begin{lemma} \label{lemma:exp_support}
    Given an integer $r>1$, if $S$ is a genuine $r$-incident independent multiset of a graph $G$, then $|\supp(S)| \gs 2^{r+2}-r-3$.
\end{lemma}

\begin{proof}
    Let $m>1$ be the smallest integer such that there exists a set $K_0 \se V \sm \supp(S)$ of size $m$ such that $S \bullet\Lambda_G^{K_0}$ is odd.   
    Note that by hypothesis there exists such an integer. Indeed, let $K_\text{max}$ be the biggest subset (by inclusion) of $V \sm \supp(S)$ such that $\sum_{N_{G}(u)=K_\text{max}}S(u)$ is odd: then $S \bullet\Lambda_G^{K_\text{max}}$ is odd. Thus, by definition of the $r$-incidence, $m\gs r+2$.

    Let $G' = G[\supp(S) \cup K_0]$ be the graph obtained from $G$ by removing the vertices that are neither in the support of $S$, nor in $K_0$. By definition, $S$ is also $r$-incident in $G'$. Also, $S \bullet\Lambda_{G'}^{K_0}$ is odd, and for every set $K \varsubsetneq K_0$ s.t. $|K|>1$, $S \bullet\Lambda_{G'}^{K}$ is even.

    Let us prove that for any $K \se K_0$ s.t. $|K|>1$, $\sum_{N_{G'}(u)=K}S(u)$ is odd, by induction over the size of $K$. First notice that $\sum_{N_{G'}(u)=K_0}S(u) = S \bullet\Lambda_{G'}^{K_0}$ is odd. Then, let $K_1 \varsubsetneq K_0$ s.t. $|K_1|>1$.
    \begin{align*}
        &S \bullet\Lambda_{G'}^{K_1} = \sum_{K_1\se K\se K_0}\sum_{N_{G'}(u)=K}S(u)
        = \sum_{N_{G'}(u)=K_1}S(u) + \sum_{K_1\varsubsetneq K\se K_0}\sum_{N_{G'}(u)=K}S(u)\\
        &= \sum_{N_{G'}(u)=K_1}S(u) + |\{K\se K_0~|~K_1\varsubsetneq K\}| \bmod 2 \text{~~by hypothesis of induction}\\
        &= \sum_{N_{G'}(u)=K_1}S(u) + 1 \bmod 2 
    \end{align*}
Thus, $\sum_{N_{G'}(u)=K_1}S(u)$ is odd. As a consequence, for any $K \se K_0$ s.t. $|K|>1$, there exists at least one vertex $u\in \supp(S)$ s.t. $N_{G'}(u)=K$. Then, $|\supp(S)| \gs |\{K\se K_0~|~|K|>1\}| = 2^m - m - 1 \gs 2^{r+2} - (r+2) -1 = 2^{r+2}-r-3$.
\end{proof}

Likewise, $r$-local complementations that cannot be implemented by $(r-1)$-local complementations can only occur over multisets with sufficiently many vertices outside their support.

\begin{lemma}\label{lemma:sizeZ}
    If $G\star^r S$ is valid and there is no $(r-1)$-incident independent multiset $S'$ such that $G\star^{r-1} S' = G \star^r S$, then $|V\sm\supp(S)| \gs r+3$.
\end{lemma}

\begin{proof}
    According to \cref{prop:nontrivial}, $S$ is genuine. Let $K_\text{max}$ be one of the biggest subset (by inclusion) of $V \sm \supp(S)$ such that $\sum_{N_{G}(u)=K_\text{max}}S(u)$ is odd: then $S \bullet\Lambda_G^{K_\text{max}}$ is odd. If $|V \sm \supp(S)| \ls r+1$, then  $|K_{max}|\ls r+1$, contradicting the $r$-incidence of $S$. Thus, $|V\sm\supp(S)| \gs |K_{max}| \gs r+2$. 

    Now, suppose $|V\sm\supp(S)| = r+2$, i.e.~$K_{max}=V\sm\supp(S)$. Let us prove that for any $K \se V \sm \supp(S)$ s.t. $|K|>1$, $\sum_{N_{G}(u)=K}S(u)$ is odd, by induction over the size of $K$. First, notice that $\sum_{N_{G}(u)=V \sm \supp(S)}S(u)$ is odd.
    Then, let $K_0 \varsubsetneq V \sm \supp(S)$ s.t. $|K_0|>1$.
    \begin{align*}
        &S \bullet\Lambda_G^{K_0} =\!\!\!\! \sum_{K_0\se K\se V \sm \supp(S)}\sum_{N_{G}(u)=K}S(u)
        = \sum_{N_{G}(u)=K_0}S(u) + \sum_{K_0\varsubsetneq K\se V \sm \supp(S)}\sum_{N_{G}(u)=K}S(u)\\
        &= \sum_{N_{G}(u)=K_0}S(u) + |\{K\se V \sm \supp(S)~|~K_0 \varsubsetneq K\}| \bmod 2 \text{~~by hypothesis of induction}\\
        &= \sum_{N_{G}(u)=K_0}S(u) + 1 \bmod 2 
    \end{align*}
    Thus, $\sum_{N_{G}(u)=K_0}S(u)$ is odd, as $S \bullet\Lambda_G^{K_0}$ is even by $r$-incidence of $S$. As a consequence, for any $K \se V \sm \supp(S)$ s.t. $|K|>1$, there exists at least one vertex $u\in \supp(S)$ s.t. $N_{G}(u)=K$. 

    Let $S'$ be the multiset obtained from $S$ by choosing, for every $K \se V \sm \supp(S)$ s.t. $|K|>1$, a single vertex $u \in \supp(S)$ s.t. $N_{G}(u)=K$, and setting $S'(u)=S(u)-1$ (the multiplicity of other vertices remain unchanged). Let us prove that $S'$ is $r$-incident and $G \star^r S = G\star^{r} S'$. First, $S'$ is $r$-incident. Indeed, let an integer $k\in [0,r)$, let $K_1\subseteq V\setminus \supp(S')$ be a set of size $k+2$, and let $k'=k-|K_1\cap \supp(S)|$. $S\bullet \Lambda_G^{K_1}$ is a multiple of $2^{r-k'-\delta(k')}$ by $r$-incidence of $S$, so is $S'\bullet \Lambda_G^{K_1}$, as $S'\bullet \Lambda_G^{K_1} = S\bullet \Lambda_G^{K_1} -|\{K \se V \sm \supp(S)~|~K_1\sm \supp(S)\se K\}| = S\bullet \Lambda_G^{K_1} -2^{r-k'}$. Then, if $u$ or $v$ is in $\supp(S)$, $u\sim_{G\star^r S} v ~\Leftrightarrow~ u\sim_{G\star^{r} S'} v ~\Leftrightarrow~ u\sim_{G} v$. If $u,v \in V \sm \supp(S)$:
    \begin{align*}
        u\sim_{G\star^r S} v &~\Leftrightarrow~\left(u\sim_{G} v ~~\oplus~~ S \bullet\Lambda_G^{u,v} = 2^{r-1}\bmod 2^{r}\right)\\
        &~\Leftrightarrow~\left(u\sim_{G} v ~~\oplus~~ S'\bullet \Lambda_G^{u,v} +2^{r} = 2^{r-1}\bmod 2^{r}\right)\\
        &~\Leftrightarrow~\left(u\sim_{G} v ~~\oplus~~ S'\bullet \Lambda_G^{u,v} = 2^{r-1}\bmod 2^{r}\right)
        ~\Leftrightarrow~ u\sim_{G\star^{r}S'} v
    \end{align*}
    Thus, $G\star^r S = G\star^r S'$. Notice also that $S'$ is not genuine. Thus, by \cref{prop:nontrivial} there exists an $S''$ such that $G\star^{r-1} S'' = G \star^r S' = G \star^r S$.
\end{proof}

\cref{lemma:exp_support,lemma:sizeZ} together give a simple bound involving only the order of the graph.

\begin{proposition} \label{prop:boundlevel}
    If $G\star^r S$ is valid and there is no $(r-1)$-incident independent multiset $S'$ such that $G\star^{r-1} S' = G \star^r S$, then $n \gs 2^{r+2}$, where $n$ is the order of $G$.
\end{proposition}

Put differently, any $r$-local complementation over a graph of order at most $2^{r+2}-1$ can be implemented by $(r-1)$-local complementations:

\begin{corollary} \label{cor:LCr_LCr-1}
    If two graphs of order at most $2^{r+2}-1$ are LC$_r$-equivalent, then they are LC$_{r-1}$-equivalent.
\end{corollary}

In other words, two LC$_r$-equivalent but not LC$_{r-1}$-equivalent graphs are of order at least $2^{r+2}$. This implies the following strengthening of \cref{thm:LU_imply_LCr}.

\begin{corollary} \label{cor:LU_LCr}
    If two graphs of order at most $2^{r+3}-1$ are LU-equivalent, they are LC$_r$-equivalent. 
\end{corollary}

\begin{proof}
    Suppose that $G_1$ and $G_2$ of order $n \ls 2^{r+3}-1$ are LU-equivalent. According to \cref{thm:LU_imply_LCr}, $G_1$ and $G_2$ are LC$_n$-equivalent. If $n \ls r$ then $G_1$ and $G_2$ are trivially LC$_r$-equivalent. Otherwise, according to \cref{cor:LCr_LCr-1}, $G_1$ and $G_2$ are LC$_r$-equivalent by direct induction. 
\end{proof}

\cref{cor:LU_LCr} provides a logarithmic bound on the level of generalized local complementations to consider for LU-equivalence, hence the following strengthening of \cref{thm:LU_imply_LCr}.

\begin{theorem} \label{thm:LU_imply_LCr_log}
    If two graph states on $n>7$ qubits are LU-equivalent, then they are LC$_r$-equivalent for some $r \ls \lceil \log_2(\frac{n+1}8)\rceil$. (Two LU-equivalent graph states on $n\ls 7$ qubits are LC-equivalent.)
\end{theorem}

This bound leads to a quasi-polynomial time algorithm for LU-equivalence, as described in the next section. In \cref{subsec:LULC_19}, we elaborate on the consequences of \cref{thm:LU_imply_LCr_log} (or equivalently, \cref{cor:LU_LCr}) on the minimal order of graphs that are LU- but not LC-equivalent.

\section{An algorithm to recognize LU-equivalent graph states}

According to \cref{thm:algolcr}, we have an algorithm that recognizes two LC$_r$-equivalent graphs of order $n$ in runtime $O(r n^{r+2.38} + n^{6.38})$. According to \cref{thm:LU_imply_LCr_log}, $G_1$ and $G_2$ are LU-equivalent if and only if they are LC$_r$-equivalent, where $r \ls\log_2(n+1)-3$. Thus, our algorithm that decides LC$_r$-equivalence translates directly to an algorithm that decides LU-equivalence.

\begin{theorem}
    There exists an algorithm that decides if two graph states are LU-equivalent with runtime $O(n^{\log_2(n)})$, where $n$ is number of qubits.
\end{theorem}
\chapter{Sufficient conditions for LU=LC}

\label{chap:conditionsLULC}

We say that LU=LC holds for a graph $G$ if $G =_{LU} G' \Leftrightarrow G =_{LC} G'$ for any graph $G'$. LU=LC implies that:
\begin{itemize}
    \item \!the LU-orbit of a graph state can be explored with local complementations;
    \item \!LU-equivalence can be decided efficiently \!(using Bouchet's algorithm for LC-equivalence).
\end{itemize}
The 27-vertex graphs that are LU-equivalent but not LC-equivalent are an example of graphs for which LU=LC does not hold. However, many useful classes of graph states satisfy LU=LC. We review in \cref{subsec:LULC_somegraphs} and \cref{subsec:LULC_msc} known classes of graphs for which LU=LC holds. In \cref{subsec:LULC_glc} we use techniques adapted from \cref{chap:glc} to show that LU=LC holds for yet another useful class of graph states, some instances of repeater graph states, which was conjectured in \cite{tzitrin2018local}. Additionally, in \cref{subsec:LULC_19}, using results from \cref{chap:algo}, we prove that LU=LC holds for graph states containing up to 19 qubits.

\section{Simple graph families that satisfy LU=LC} \label{subsec:LULC_somegraphs}

Below we list some simple classes of graphs that satisfy LU=LC. LU=LC holds for:
\begin{itemize}
    \item complete graphs \cite{VandenNest05};
    \item bipartite complete graphs \cite{tzitrin2018local};
    \item large enough rectangular grids \cite{Sarvepalli_2010} (follows from \cref{prop:msc} in next section);
    \item graphs with no cycle of length 3 or 4 \cite{Zeng07}.
\end{itemize}

A more detailed review of classes of graphs satisfying LU=LC, containing additional, more intricate conditions, can be found in \cite{tzitrin2018local}.

\section{A condition based on minimal local sets} \label{subsec:LULC_msc}

In \cite{VandenNest05}, it is shown that LU=LC holds for graphs satisfying the so-called minimal support condition, which we formulate here in terms of types given by a MLS cover.

\begin{proposition}[Minimal Support Condition \cite{VandenNest05}] \label{prop:msc}
    LU=LC holds for $G$ if every vertex is of type $\bot$ with respect to some MLS cover.
\end{proposition}

The original formulation of the minimal support condition is the following: LU=LC holds for a stabilizer state $\ket \psi$ if the 3 Paulis X, Y and Z occur on every qubit in the group generated by the stabilizers of $\ket \psi$ with minimal support.

\section{A condition based on generalized local complementation} \label{subsec:LULC_glc}

%Our results
The results of \cref{chap:glc} imply a new criterion for LU=LC based on the standard form. 
This criterion is actually a sufficient and necessary condition, meaning it can be used to prove both that LU=LC holds for some graph, or the converse. This criterion makes use of the maximal MLS cover $\mathcal M_{max}$ that contains every minimal local set of the graph. Formally, given a graph $G$, $\mathcal M_{max} = \{L \se V ~|~ L \text{ is a minimal local set of $G$}\}$. 
We begin by proving the criterion for graphs in standard form.

\begin{lemma} \label{lemma:lu-lc}
    LU=LC holds for graphs in standard form (with respect to the maximal MLS cover $\mathcal M_{max}$) if and only if any $r$-local complementation over the vertices of type X can be implemented by local complementations.
\end{lemma}

\begin{proof}
    Let $G_1$ be a graph in standard form and $G_2$ be an arbitrary graph such that $G_1 =_{LU} G_2$. By \cref{thm:LU_imply_LCr}, $G_1 =_{LC_r} G_2$ for some integer $r$. By \cref{lemma:standardform}, there exists $G'_2 =_{LC} G_2$ in standard form. By \cref{lemma:standardform_LCr_lc} there exists $G''_2 =_{LC} G'_2$ in standard form\footnote{$G''_2$ is in standard form because it is obtained from $G'_2$ by local complementations on vertices of type $\bot$.} such that $G_1$ and $G''_2$ are related by a single $r$-local complementation over the vertices of type X (the vertices of type X with respect to $\mathcal M_{max}$ are the same in $G_1$ and $G''_2$ thanks to \cref{lemma:same_types}). If the $r$-local complementation can be implemented using local complementations, then $G_1$ and $G''_2$ are LC-equivalent, so are $G_1$ and $G_2$.

    On the contrary, let $G_2$ be the result of an $r$-local complementation over a graph $G_1$. In particular, $G_1=_{LU} G_2$. Hence, if LU=LC holds for $G_1$, then $G_1$ and $G_2$ are LC-equivalent, i.e. the $r$-local complementation can be implemented by local complementations. 
\end{proof}

Below we generalize the criterion to arbitrary graphs (not necessarily in standard form).

\begin{proposition} \label{prop:lu-lc}
    Given a graph $G$, the following are equivalent: \begin{itemize}
        \item LU=LC holds for $G$.
        \item For \textbf{some} graph LC-equivalent to $G$ that is in standard form (with respect to the maximal MLS cover $\mathcal M_{max}$), any $r$-local complementation over the vertices of type X can be implemented by local complementations.
        \item For \textbf{any} graph LC-equivalent to $G$ that is in standard form (with respect to the maximal MLS cover $\mathcal M_{max}$), any $r$-local complementation over the vertices of type X can be implemented by local complementations;
    \end{itemize}
\end{proposition}

\begin{proof}
    Any graph can be put in standard form by means of local complementations by \cref{lemma:standardform}. \cref{lemma:lu-lc} along with the fact that the LU=LC property is invariant by LC-equivalence, proves the proposition.
\end{proof}

Checking that $r$-local complementations can be implemented by local complementations is not easy in general, nonetheless it is convenient in many cases, for example when there are few vertices of type X or that they have low degree. Namely, this criterion is stronger than the minimal support condition, as graphs with only vertices of $\bot$ have no vertex of type X. Furthermore, it has been left as an open question in \cite{tzitrin2018local} whether LU=LC holds for some instances of repeater graph states. We use our new criterion to easily prove that this is the case. For this purpose, we prove that LU=LC holds for a broader class of graph.

\begin{proposition} \label{prop:lu-lc-degree1}
        LU=LC holds for graphs where each vertex is either a leaf i.e.~a vertex of degree 1, or is adjacent to a leaf.        
\end{proposition}

\begin{proof}

    Let $G$ be a graph where each vertex is either a leaf or adjacent to a leaf. We suppose without loss of generality that $G$ is connected, as LU=LC holds for a graph if and only if LU=LC holds for its connected components. Moreover, we suppose without loss of generality that the graph $G$ is of order at least 3, else the result is trivial.

    Let us partition the vertex set into $L$ "the leafs" and $P$ "the parents": $|L| \gs |P|$, any vertex in $L$ is adjacent to exactly one vertex in $P$, and any vertex in $P$ is adjacent to at least one vertex in $L$. Let us prove that the vertices in $L$ are of type X while the vertices in $P$ are of type Z. For any vertex $u \in L$, $\{u\} \cap N_G(u)$ is a minimal local set of dimension 1 generated by $\{u\}$ (as the unique neighbor of $u$ is not of degree 1). Moreover, a minimal local set cannot be generated by a set containing a vertex in $P$. Indeed, if a local set is generated by a set containing a vertex $v \in P$ related to a vertex $u\in L$, the local set is not minimal as it strictly contains the minimal local set $\{u,v\}$. Thus, every minimal local set is generated by a set in $L$. 

    Fix any ordering of the vertices where for any leaf $u$ related to a vertex $v$, $u < v$. Then, $G$ is in standard form (with respect to the maximal MLS cover $\mathcal M_{max}$). The vertices of type X are exactly the leafs. Thus, any $r$-local complementation over the vertices of type X has no effect on $G$. According to \cref{prop:lu-lc}, this implies that LU=LC holds for $G$.
\end{proof}

Such graphs include some instances of \emph{repeater graph states} \cite{azuma2023quantum}. A \emph{complete-graph-based repeater graph state} is a graph state whose corresponding graph of order $2n$ is composed of a complete graph of order $n$, along with $n$ leafs appended to each vertex.  A \emph{biclique-graph-based repeater graph state} is a graph state whose corresponding graph of order $4n$ is composed of a symmetric biclique (i.e. a symmetric bipartite complete graph) graph of order $2n$, along with $2n$ leafs appended to each vertex. Complete-graph-based repeater graph states and biclique-graph-based repeater graph states are illustrated in Figure \ref{fig:repeater}. %\Ncom{TODO}

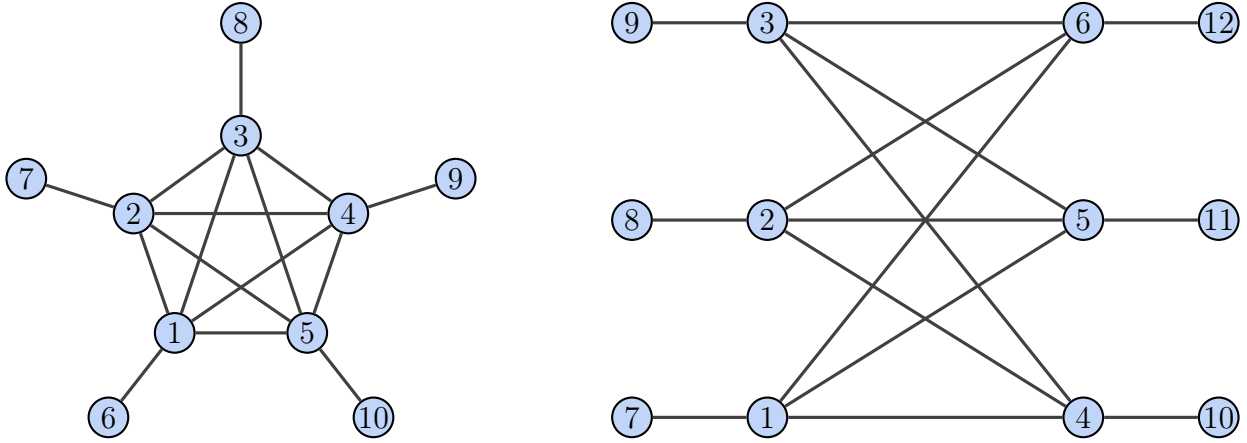
\begin{figure}[H]
    \centering
    
    \scalebox{1}{
    \begin{tikzpicture}[scale = 1]    
        \begin{scope}[every node/.style=vertices]
            \node (U1) at (-0.881,-1.124) {1};
            \node (U2) at (-1.427,0.464) {2};
            \node (U3) at (0,1.5) {3};
            \node (U4) at (1.427,0.464) {4};
            \node (U5) at (0.881,-1.124) {5}; 
            \node (U6) at (-0.881*2,-1.124*2) {6};           
            \node (U7) at (-1.427*2,0.464*2) {7}; 
            \node (U8) at (0,1.5*2) {8};  
            \node (U9) at (1.427*2,0.464*2) {9};
            \node (U10) at (0.881*2,-1.124*2) {10};            
        \end{scope}
        \begin{scope}[every node/.style={},
                        every edge/.style=edges]              
            \path [-] (U1) edge node {} (U2);
            \path [-] (U1) edge node {} (U3);
            \path [-] (U1) edge node {} (U4);
            \path [-] (U1) edge node {} (U5);
            \path [-] (U2) edge node {} (U3);
            \path [-] (U2) edge node {} (U4);
            \path [-] (U2) edge node {} (U5);
            \path [-] (U3) edge node {} (U4);
            \path [-] (U3) edge node {} (U5);
            \path [-] (U4) edge node {} (U5);

            \path [-] (U1) edge node {} (U6);
            \path [-] (U2) edge node {} (U7);
            \path [-] (U3) edge node {} (U8);
            \path [-] (U4) edge node {} (U9);
            \path [-] (U5) edge node {} (U10);

        \end{scope}
    \end{tikzpicture}\qquad\qquad\raisebox{0cm}{
        \begin{tikzpicture}[xscale = 0.6, yscale = 0.75]
        
        \begin{scope}[every node/.style=vertices]
            \node (U1) at (0,-3) {1};
            \node (U2) at (0,0.5) {2};
            \node (U3) at (0,4) {3};
            \node (U4) at (7,-3) {4}; 
            \node (U5) at (7,0.5) {5}; 
            \node (U6) at (7,4) {6}; 

            \node (U7) at (-3,-3) {7};
            \node (U8) at (-3,0.5) {8};
            \node (U9) at (-3,4) {9};
            \node (U10) at (10,-3) {10}; 
            \node (U11) at (10,0.5) {11}; 
            \node (U12) at (10,4) {12}; 
        \end{scope}
        
        \begin{scope}[every edge/.style=edges]              
            \path [-] (U1) edge node {} (U4);
            \path [-] (U1) edge node {} (U5);
            \path [-] (U1) edge node {} (U6);
            \path [-] (U2) edge node {} (U4);
            \path [-] (U2) edge node {} (U5);
            \path [-] (U2) edge node {} (U6);
            \path [-] (U3) edge node {} (U4);
            \path [-] (U3) edge node {} (U5);
            \path [-] (U3) edge node {} (U6);

            \path [-] (U1) edge node {} (U7);
            \path [-] (U2) edge node {} (U8);
            \path [-] (U3) edge node {} (U9);
            \path [-] (U4) edge node {} (U10);
            \path [-] (U5) edge node {} (U11);
            \path [-] (U6) edge node {} (U12);
        \end{scope}
    \end{tikzpicture}}
    }
    
    \caption{(Left) The complete-graph-based repeater graph state of order 10. (Right) The biclique-graph-based repeater graph state of order 12.}
    \label{fig:repeater}
\end{figure}

Complete-graph-based repeater graph states are the all-photonic repeaters introduced in \cite{azuma2015all}. Biclique-graph-based repeater graph states are a variant introduced in \cite{russo2018photonic} that is more efficient in terms of number of edges.
It was left as an open question in \cite{tzitrin2018local} whether LU=LC holds for complete-graph-based or biclique-graph-based repeater graph states (although it was proved for some variants). These graph states satisfy the condition of \cref{prop:lu-lc-degree1}, hence we can answer this question by the positive:

\begin{corollary}
    LU=LC holds for complete-graph-based repeater graph states and biclique-graph-based repeater graph states.
\end{corollary}

\section{LU=LC for graph states on up to 19 qubits} \label{subsec:LULC_19}

It is known that there exists a pair of 27-vertex graphs that are not LC-equivalent, but  LU-equivalent, more precisely they are LC$_2$-equivalent \cite{Ji07,Tsimakuridze17}. To this day it is still  an open question whether this is a minimal example (in number of vertices). In other words, does a pair of graphs that are LU-equivalent but not LC-equivalent on 26 vertices or fewer exist? In theory, one could check every pair of graphs of order up to 26, but the rapid combinatorial explosion in the number of graphs as the number of vertices increases, makes it unfeasible in practice.

The best bound known so far\footnote{In \cite{burchardt2024algorithm} it is proved that the number of LC- and LU-orbits of \textbf{unlabeled} graphs of order up to 11 is the same.} is LU=LC holds for graphs of order up to 8 \cite{CABELLO20092219}. \cref{cor:LU_LCr} already implies a substantial improvement over this bound: LU=LC holds for graphs of order up to 15. Furthermore, for graphs of order up to 31, LU=LC$_2$, i.e.~if two graphs of order up to 31 are LU-equivalent, they are LC$_2$-equivalent. Thus, asking whether LU=LC holds for graphs of order up to 26 is equivalent to asking whether LC$_2$=LC holds for graphs of order up to 26. One direction is to study when a 2-local complementation over a multiset $S$ can be implemented using only usual local complementations over vertices in the support of $S$. If this were to be the case for every graph of order up to 26, it would show that the 27-vertex counterexample is minimal in number of vertices. In the following we study the structure of 2-local complementation to prove that LU=LC holds for graphs of order up to 19.

\begin{proposition} \label{prop:LULC19}
    LU=LC holds for graph states on up to 19 qubits.
\end{proposition}

In this section we prove \cref{prop:LULC19}. The idea of the proof is to show that on any graph of order at most 19, any 2-local complementation can be implemented by usual local complementations. Some parts of the proof require exhaustive search and are thus assisted by computer. The code is available at \citepub{codelulc19}.

First, we show that 2-local complementation has a simple structure compared to the general $r$-local complementation. Namely, it is sufficient to consider $2$-local complementations over sets, rather than multisets. 

\begin{proposition} \label{prop:decomposition}
Any 2-local complementation can be decomposed into 1- and 2-local complementations over sets. 
\end{proposition}

\begin{example}
    In Figure \ref{fig:generalized_lc}, the 2-local complementation over $S = \{1,1,2,3\}$ can be decomposed into a 2-local complementation over the set $\{2,3\}$ and a 1-local complementation over the set $\{1\}$.
\end{example}

\begin{proof}
    Given a 2-local complementation over a multiset $S$, we assume without loss of generality that multiplicities are defined modulo $2^2 = 4$. Let  $S_1$ be the set of vertices that have multiplicity $2$ or $3$ in $S$. Notice that $G\star^2 S = G\star^2 S\star^1 S_1\star^1 S_1$ as $S_1$ is an independent set, $S_1$ is 1-incident and generalized local complementations are self inverse. So $G\star^2 S = G\star^2 S\star^2 (S_1 + S_1)\star^1 S_1 = G\star^2(S+ S_1 + S_1)\star^1 S_1$, where the multiplicity in  
    $S+ S_1 + S_1$ is either $0$ or $1$ modulo $4$. Thus, the $2$-local complementation over the multiset $S$ can be decomposed into 2- and 1-local complementations over sets $\supp(S+ S_1+S_1)$ and $S_1$ respectively. 
\end{proof}

As our proof contains computer-assisted exhaustive search, and generating every graph with a fixed number of vertices is very expensive, we first need to drastically decrease the size of the space to explore when studying 2-local complementation.

\begin{lemma} \label{lemma:lifting}
    Let $S$ be a 2-incident independent multiset of a graph $G=(V,E)$ and suppose that there exists no set $A \se \supp(S)$ such that $G \star^2 S = G \star^1 A$. Then there exists a graph $G'=(V',E')$ bipartite with respect to a bipartition $S', V' \sm S'$ of the vertices such that:
    \begin{itemize}
        \item $S'$ is 2-incident;
        \item $S'$ contains no twins;
        \item $S'$ contains no vertex of degree 0 or 1;
        \item $|S'| \ls |\supp(S)|$;
        \item $|V' \sm S'|\ls |V \sm \supp(S)|$;
        \item there exists no set $A \se S'$ such that $G' \star^2 S' = G' \star^1 A$.
    \end{itemize}
\end{lemma}

\begin{proof}
    The proof is constructive, in the sense that we construct $G'$ and $S'$ from $G$ and $S$. According to \cref{prop:decomposition}, there exist set $S_1, S_2 \se \supp(S)$ such that $S_2$ is 2-incident and $G \star^2 S = G \star^2 S_2 \star^1 S_1$. Also, the existence of a set $A \se S_2$ such that $G \star^2 S_2 = G \star^1 A$ implies $G \star^2 S = G \star^1 (A \Delta S_1)$. Let $G' = G$ and $S' = S_2$, then apply the following operations:
    \begin{enumerate}
        \item Remove the edges between vertices of $V' \sm S'$;
        \item Remove each vertex of $S'$ of degree 0 or 1;
        \item If there exists a pair of twins $u,v$ in $S'$, remove $u$ and $v$. Repeat until $S'$ contains no twins.
    \end{enumerate}
    Notice that each operation preserves the 2-incidence of $S'$ and that  $G' \star^2 S' = G' \star^1 A$ for no set $A \se S'$.
\end{proof}

According to \cref{lemma:sizeZ}, if there are at most 4 vertices outside the support of some 2-incident independent multiset $S$, then a 2-local complementation over $S$ can be implemented by usual local complementations. We use computer-assisted generation \citepub{codelulc19} to extend the result.

\begin{lemma} \label{lemma:2lc6}
    Let $S$ be a $2$-incident independent multiset of a graph $G=(V,E)$. If $|V\sm\supp(S)|\ls 5$, or if $|V \sm \supp(S)|= 6$ and $|\supp(S)|\ls 20$, then a 2-local complementation over $S$ can be implemented by local complementations over a subset of $\supp(S)$. 
\end{lemma}

\begin{proof}
Following \cref{lemma:lifting}, given an integer $k$, let $\mathcal G_k$ be the class of graphs that are bipartite with respect to a bipartition $S, V \sm S$ of the vertices such that:
\begin{itemize}
    \item $S$ is 2-incident;
    \item $S$ contains no twins;
    \item $S$ contains no vertex of degree 0 or 1;
    \item $|V \sm S| = k$.
\end{itemize}
It is easy to generate each graph of $\mathcal G_k$, although the number of elements in $\mathcal G_k$ grows double exponentially fast with $k$. $S$ can be defined as a list of words in $\{0,1\}^k$ of weight at least 2. More precisely each vertex of $S$ is uniquely associated with a set of $V \sm S$ of size at least 2, its neighborhood. Furthermore, the 2-incidence of $S$ implies that $S$ is uniquely determined by the set of its vertices of degree at least 4. Indeed, starting from a set containing only vertices of degree at least 4, the conditions of the form "$S\bullet \Lambda_G^K = 0 \bmod 2^{r-k-\delta(k)}$" translate into a procedure to find which vertices of degree 3 then 2 need to be added to the set so that $S$ is 2-incident. This proves that there is a bijection between $\mathcal G_k$ and lists of words in $\{0,1\}^k$ of weight at least 4. Thus, the size of $\mathcal G_k$ is exactly given by the formula $$ |\mathcal G_k| = 2^{\binom{k}{4} + \binom{k}{5} + \cdots + \binom{k}{k}}$$
For $k=1,2,3,4,5$ and $6$, the size of $\mathcal G_k$ is respectively $1,1,1,2,2^6 = 64$, and $2^{22} \sim 4 \times 10^6$ which is suitable for computation. But, even for $k$ as low as seven, the size of $\mathcal G_7$ is $2^{64} \sim 2 \times 10^{19}$.

For every $k$ from 1 to 6, we generate each graph $G$ of $\mathcal G_k$, along with the set $S$ defined above. Notice that a local complementation on a vertex $u$ of $S$ toggles the connectivity of some pairs of vertices of $V \sm S$, here the pairs where each end is a neighbor of $u$. In other words, to each vertex $u$ of $S$ we associate a vector in $\mathbb F_2^{\binom{k}{2}}$ corresponding to the action of the local complementation of $u$ on the graph. The set of the vectors corresponding to each vertex of $S$ spans an $\mathbb F_2$-vector space $\mathcal L$ describing the action of local complementation on vertices of $S$ on the graph. Using Gaussian elimination, we are able to compute a basis of $\mathcal L$. Furthermore, we compute the vector $x$ in $\mathbb F_2^{\binom{k}{2}}$ corresponding to the action of a 2-local complementation over $S$ on the graph. Checking if the action of a 2-local complementation over $S$ can be implemented by local complementations on vertices of $S$ amounts to checking if $x$ belongs to the vector space $\mathcal L$, which can be done efficiently using Gaussian elimination.

For $k \in [1,3]$, the set $S$ corresponding to the only graph $G \in \mathcal G_k$ is empty, hence a 2-local complementation over $S$ leaves $G$ invariant, i.e.~$G \star^2 S = G$. For $k = 4$, $S$ is either empty or contains 11 vertices; in both case it is easy to check that $G \star^2 S = G$. For $k = 5$, the computation shows that for each graph $G \in \mathcal G_5$, $G \star^2 S = G$. Now, fix $k=6$. For each graph $G \in \mathcal G_6$ such that the corresponding set $S$ contains at most 16 vertices, $G \star^2 S = G$. For each graph $G \in \mathcal G_6$ such that the corresponding set $S$ contains at most 20 vertices, a 2-local complementation over the corresponding set $S$ can be implemented by local complementations over vertices of $S$, i.e.~there exists a set $A \se S$ such that $G \star^2 S = G \star^1 A$. The property does not hold if $S$ is of size 21, for instance we recover the well-known 27-vertex counterexample to the LU=LC conjecture as described in \cite{Tsimakuridze17}.

According to \cref{lemma:lifting}, this is enough to prove the lemma.
\end{proof}

According to \cref{lemma:exp_support} and \cref{prop:nontrivial}, if the support of some 2-incident independent multiset $S$ is of size at most 10, then a 2-local complementation over $S$ can be implemented by usual local complementations. To extend this result to 2-incident independent multisets whose supports are of size at most 12, we first study the case of twin-less sets. 

\begin{lemma} \label{lemma:lessthan12}
    Let $S$ be a 2-incident independent set of a graph $G=(V,E)$ such that $S$ does not contain any twins and $|S| \ls 12$. Then, $G \star^2 S = G$.
\end{lemma}

\begin{proof}
    Given a set $S$, we define the strict neighborhood of $S$ in the graph $G$ as the set of vertices not in $S$ that are adjacent to at least one vertex in $S$: $\xi_G(S) = \{u \in V \sm S ~|~ N_G(u) \cap S \neq \emptyset\}$. Let us prove, by induction on $|\xi_G(S)|$, that for any 2-incident independent set $S$ such that $S$ does not contain any twins and $|S| \ls 12$, $G \star^2 S = G$.

    First, by brute force, we check that the property is true whenever $|\xi(S)| \ls 6$ (see details in the proof of \cref{lemma:2lc6}). 

    Now, suppose the property true for an integer $t \gs 6$. Let $S$ be a 2-incident independent set such that $S$ does not contain any twins, $|S| \ls 12$, and $|\xi_G(S)| = t+1$. Let us prove that $G \star^{2} S = G$. Suppose by contradiction that $G \star^{2} S \neq G$. Then, there exist $u,v \in \xi(S)$ such that the edge between $u$ and $v$ is toggled by the 2-local complementation over $S$. 
    
    Let $w \in \xi(S) \sm \{u,v\}$, and let $G' = G \sm w$, i.e.~$G'$ is the graph obtained from $G$ by removing the vertex $w$. $S$ is still a 2-incident independent set in $G'$ (but now it may contain twins). Notice that $\xi_{G'}(S) = \xi_{G}(S) -1$. The case where $S$ contains no twins contradicts the induction hypothesis. Suppose that $S$ does contain two twins, say $a$ and $b$. Then, their neighborhood in the original graph $G$ differs only by the vertex $w$ (thus there can only be pairs of twins, e.g. no triplets). Furthermore, $S \sm \{a,b\}$ is 2-incident and $G' \star^2 S = (G' \star^2 S) \star^1 \{a\}$. Let $S'$ be the set obtained from $S$ by removing every pair of twins and every vertex of degree 1. $S'$ is 2-incident and $|S'|\ls 10$. A non-empty 2-incident independent set without twins or vertex of degree 1 is genuine, and thus by \cref{lemma:exp_support} has at least 11 vertices. Thus, $S' = \emptyset$.

    $w$ was chosen arbitrarily in $w \in \xi(S) \sm \{u,v\}$. Thus, for any $w\in \xi(S) \sm \{u,v\}$, the set obtained from $S$ by removing every pair of twins and every vertex of degree 1 in the graph $G \sm w$, is empty. 

    Let us prove that $S$ contains a vertex $a$ such that $\xi(S) \sm \{u,v\} \se N_G(a)$. Suppose by contradiction that this is not the case and consider a vertex $b \in S$ such that $|N_G(b) \cap (\xi(S) \sm \{u,v\})|$ is maximum. Note that $|N_G(b) \cap (\xi(S) \sm \{u,v\})| \gs 2$, else  $G \star^2 S = G$, as $S$ is 2-incident. By hypothesis there is a vertex $w \in \xi(S) \sm \{u,v\}$ such that $w \notin N_G(b)$. In $G \sm w$, $b$ is not of degree 1 so it has a twin $z$, implying that $N_G(z) = N_G(b) \cup \{w\}$. This is a contradiction with the fact that $|N_G(b) \cap (\xi(S) \sm \{u,v\})|$ is maximum.

    As $|\xi_G(S)| \gs 7$, $|N_G(a) \cap (\xi(S) \sm \{u,v\})| \gs 5$. For any $w\in \xi(S) \sm \{u,v\}$, in $G \sm w$ each vertex of degree at least 2 has a twin. Thus, in $G$, for any $w\in \xi(S) \sm \{u,v\}$, there is a vertex $a_w \in S$ such that $N_G(a_w) = N_G(a) \sm \{w\}$. Additionally, for any such vertex $a_w$, for any $w'\in \xi(S) \sm \{u,v,w\}$, there is a vertex $a_{w,w'} \in S$ such that $N_G(a_{w,w'}) = N_G(a_w) \sm \{w\}$. Thus, $S$ contains at least $1+5+\binom{5}{2} = 16$ vertices, contradicting the hypothesis that $|S|\ls 12$.
\end{proof}

According to \cref{prop:decomposition}, any 2-local complementation can be decomposed into 1- and 2-local complementations over sets. Furthermore, one can check that if a 2-incident independent set $S$ contains two twins $u$ and $v$, then a 2-local complementation over $S$ has the same effect as a 2-local complementation over $S\sm\{u,v\}$ followed by a local complementation on $u$. Thus, the action of a 2-local complementation can be described by a 2-local complementation over a twin-less set followed by usual local complementations. Then, \cref{lemma:lessthan12} can be applied on the twin-less set to yield the following result:

\begin{lemma} \label{lemma:2lc12}
    Let $S$ be a 2-incident independent multiset of a graph $G$ such that $|\supp(S)|\ls 12$. Then, a 2-local complementation over $S$ can be implemented by local complementations over a subset of $\supp(S)$.
\end{lemma}

According to \cref{lemma:2lc6} and \cref{lemma:2lc12}, if a 2-incident independent multiset $S$ satisfies $|\supp(S)|\ls 12$ or $|V\sm\supp(S)|\ls 5$, or alternatively if $|\supp(S)|\ls 20$ and $|V\sm\supp(S)|= 6$, then a 2-local complementation over $S$ can be implemented by usual local complementations. Thus, for graphs of order up to 19, any 2-local complementation can be implemented by usual local complementations, implying LU=LC (as LU=LC$_2$ for graphs of order up to 31). In other words, a 2-local complementation that cannot be implemented by usual local complementation is possible only on a graph of order at least 20. This proves \cref{prop:LULC19}: LU=LC holds for graph states on up to 19 qubits.
\chapter{Vertex-minor universal graphs}

\label{chap:vmu}

In this chapter we study graphs satisfying the property of vertex-minor universality. A graph is said $k$-vertex-minor universal if any smaller graph, defined on any $k$ vertices, can be induced by local complementation and vertex deletions. For graph states whose underlying graphs are $k$-vertex minor universal, any $k$-qubit graph state can be induced over any subset of $k$ qubits using only local Clifford operations, local Pauli measurements, and classical communication. We give bounds on vertex-minor universality, and show the existence of $\Theta(\sqrt n)$-vertex-minor universal graphs, which is asymptotically optimal in terms of number of vertices. The construction is probabilistic and involves random bipartite graphs.

\section{Vertex-minors and vertex-minor universality}

\begin{definition}
    Let $G=(V,E)$ and $F=(V',E')$ where $V' \se V$. $F$ is a vertex-minor of $G$ if it can be reached from $G$ by some sequence of local complementations and vertex-deletions. Equivalently, $F$ is a vertex-minor of $G$ if there exists $G' =_{LC} G$ such that $G'[V'] = F$ (in other words, the vertex-deletions can be performed after the local complementations).
\end{definition}

The equivalence between these two definitions is discussed e.g. in \cite{DHW:howtotransform}. A complete and up-to-date survey on
vertex-minors can be found in \cite{OumSurvey}.

\begin{definition}
    A graph $G$ is $k$-vertex-minor universal if any graph on any $k$ vertices is a vertex-minor of $G$.
\end{definition}

Examples are provided in \cref{sec:vmu_smallgraphs}.

Vertex-minor universality implies strong quantum properties for the corresponding graph state. If $F$ is a vertex-minor of $G$ then the graph state $\ket{F}$ can be obtained from $\ket G$ using only local Clifford operations, local Pauli measurements and classical communications. The converse is true when $F$ has no isolated vertices~\cite{DWH:transfo}, but is false in general. For instance, $K_2$ (the complete graph with two vertices) is not $2$-vertex-minor universal since no local complementation can turn it into an empty graph. However, with \emph{e.g.} an X-measurement on each qubit (and possible Z-corrections), one can map the corresponding graph state $\ket{K_2}$ to the disconnected graph state $\ket + \ket +$. To be able to state a characterization, a solution is to introduce \emph{destructive} measurements (i.e. the measured qubit is removed from the system and can no longer be used).

\begin{proposition}\label{prop:vm}
   Given two graphs $G= (V,E)$ and $F= (V',E')$ such that $V' \se V$, $F$ is a vertex-minor of $G$ if and only if $\ket{F}$ can be obtained from $\ket G$ (on the qubits corresponding to $V'$) using only local Clifford operations, local destructive Pauli measurements, and classical communications.
\end{proposition}

\begin{proof} Notice that a similar statement -- involving non-destructive measurements and only valid when $F$ does not contain isolated vertices -- has been proved in \cite{DWH:transfo} (Theorem 2.2). We provide here a direct proof of \cref{prop:vm} that is actually slightly simpler thanks to the use of destructive measurements. In the following proof all measurements are destructive.\\
$(\Rightarrow)$ Local complementations can be implemented by means of local Clifford unitaries (see \cref{prop:implementation_lc}), and vertex deletions by means of Z-measurements together with classical communications and Pauli corrections (see \cite{DWH:transfo} or \cref{prop:meas}).\\
$(\Leftarrow)$ We prove the property by induction on the number of measurements. If there are no measurement the property is true (see \cref{prop:lclc}). Otherwise, let $u$ be the first qubit to be measured. Assume $u$ is measured according to $P$ and $C_u$ is the Clifford operator applied on $u$ before the measurement. $C_u^\dagger P C_u$ is proportional to some Pauli operator $P_0\in \{X,Y,Z\}$: 
\begin{itemize}
    \item[$(i)$] If $P_0=Z$, then the measurement of $u$ can be interpreted as a vertex deletion and leads to $\ket{G\setminus u}$ up to Pauli corrections. By the induction hypothesis, $F$ is a vertex minor of $G\setminus u$, thus of $G$.
    \item[$(ii)$] If $P_0=Y$, then the measurement of $u$ can be interpreted as a Z-measurement on $\ket{G\star u}$ (up to local Clifford operations on some other qubits), thus according to $(i)$, $F$ is a vertex minor of $G\star u$, so is of $G$. 
    \item[$(iii)$] If $P_0=X$ and $N_G(u)\neq \emptyset$, then  the measurement $u$ can be interpreted as a Y-measurement on $\ket{G\star v}$ with $v\in N_G(u)$ (up to local Clifford operations on qubits different from $u$), thus according to $(ii)$ $F$ is a vertex minor of $G\star v$, so is of $G$.
    \item[$(iv)$] If $P_0=X$ and $N_G(u)=\emptyset$, then $\ket G = \ket {G\setminus u}\otimes \ket +_u$ so after the measurement of $u$ the state is $\ket {G\setminus u}$, thus, by the induction hypothesis, $F$ is a vertex minor of $G\setminus u$, so is of $G$. \qedhere
\end{itemize}
\end{proof}

\begin{corollary}
    The two following statements are equivalent:
    \begin{itemize}
        \item The graph $G$ is $k$-vertex-minor universal;
        \item Over any $k$ qubits of the graph state $\ket G$, any graph state can be induced using only local Clifford operations, local destructive Pauli measurements, and classical communications.
    \end{itemize}
\end{corollary}

As any stabilizer state is LC-equivalent to a graph state (see \cref{prop:stabilizer_LC}), this is actually equivalent to being able to induce any stabilizer state over $k$ qubits.

\begin{remark}
    Vertex-minor universality was introduced as a generalization of the notion of pairability \cite{bravyi2024generating}. A quantum state $\ket \psi$ is said $k$-pairable if for every $k$ disjoint pairs of qubits  $\{a_1, b_1\},\ldots,\{a_k, b_k\}$, there exists a LOCC (local operations and classical communication) protocol that starts with $\ket \psi$ and ends up with a state where each of those $k$ pairs of qubits shares an EPR-pair. When restricting the allowed operations to local Clifford operations, local  Pauli measurements, and classical communication, a graph state $\ket G$ is $k$-pairable if any graph corresponding to a perfect matching on any $2k$ vertices, is a vertex-minor of $G$. Pairability is actually a weaker notion than vertex-minor universality, as if $G$ is $2k$-vertex-minor universal then $\ket G$ is $k$-pairable.
\end{remark}

\section{Bounds on vertex-minor universality}

\subsection{Bound on the minimal degree up to local complementation}\label{subsec:bound_dloc}

The vertex-minor universality of a graph is bounded by its minimum degree up to local complementation $\dloc$ (see \cref{def:dloc}).

\begin{proposition}
    \label{prop:deltaloc_vmu}
    If a graph $G$ is $k$-vertex-minor universal then $k<\dloc(G)+2$.
\end{proposition}

\begin{proof}
By contradiction, assume there exists a graph $G$ that is $(\dloc(G)+2)$-vertex-minor universal. By \cref{prop:dloc_localset}, $G$ contains a local set $L=D\cup Odd_G(D)$ of size $\dloc(G)+1$. By hypothesis, there exists a sequence of local complementations that maps $G$ to a graph $G'$ such that $G'[L \cup u]=F$, where $F$ is the graph defined on $L \cup u$ ($u$ being an arbitrary vertex of $V \sm L$), such that its only edge is between $u$ and an arbitrary vertex $v \in L$.
By \cref{cor:localset_invariant}, a local set is invariant by local complementation. Thus, $L$ is still a local set in $G'$, i.e. there exists $D'$ s.t.~$L=D'\cup Odd_{G'}(D')$. If $v \in D'$ then $u \in Odd_{G'}(D')$. Also, $v \notin Odd_{G'}(D')$ because v is not adjacent to any other vertex of $L$. This contradicts that $L$ is still a local set.
\end{proof}

\subsection{Bound on the number of qubits}

Given any $k$, a $k$-vertex-minor universal graph has at least a quadratic order in $k$:

\begin{proposition}\label{prop:bound_vmu}
    If a graph $G$ of order $n$ is $k$-vertex-minor universal then $$k< \sqrt{2n\log_2(3)}+2$$
\end{proposition}

\begin{proof}
If $F$ of order $k$ is a vertex-minor of $G$ of order $n$, then $\ket F$ can be obtained from $\ket G$ by means of local Pauli measurements on $n-k$ qubits and  local Clifford unitaries on $k$ qubits. There are 3 possible Pauli measurements per qubit, so $3^{n-k}$ in total. Notice that for a fixed choice of Pauli measurements, different local Clifford transformations on the remaining $k$ qubits can only generate graph states corresponding to graphs that are equivalent up to local complementation. Moreover, it is known that there are at least $2^{\frac{k^2-5k}2-1}$ different graphs on $k$ vertices up to local complementation \cite{bahramgiri2007enumerating}. As a consequence, if $G$ is $k$-vertex-minor universal, we must have $3^{n-k}\gs 2^{\frac{k^2-5k}2-1}$. Using numerical analysis, this implies $k< \sqrt{2n\log_2(3)}+2$.
\end{proof}

This bound was pointed out by Sergey Bravyi in a personal communication.

\section{Vertex-minor universality of some small graphs} \label{sec:vmu_smallgraphs}

To illustrate the concept of vertex-minor universality, we study the vertex-minor universality of some small graphs. The results of this section were observed by numerical analysis. The code consists in an exploration of the orbit of a given graph by local complementation (either by using essentially a breadth-first search, or by applying local complementations on random vertices, which yields results faster), to find the ${n \choose k}2^{k \choose 2}$ possible induced subgraphs of order $k$. First, we provide minimal examples of $k$-vertex-minor universal graphs for $k$ up to $4$, illustrated in Figure \ref{fig:examples_vmu}.

\begin{figure}[H]

\centering
\scalebox{1}{

%Graph 1
\begin{tikzpicture}[scale = 0.5]
\begin{scope}[shift={(0,0)},scale=1.4,every node/.style=vertices]
    \node (U1) at (0,0) {1};
    \node (U2) at (3,5) {2};
    \node (U3) at (6,0) {3};
\end{scope}
\begin{scope}[every edge/.style=edges] 
    \path [-] (U1) edge node {} (U2);
    \path [-] (U1) edge node {} (U3);
    \path [-] (U2) edge node {} (U3);
\end{scope}

%Graph 2
\begin{scope}[shift={(13,0)},scale=1,every node/.style=vertices]
    \node (U4) at (0,0) {1};
    \node (U3) at (4,0) {6};
    \node (U5) at (-2,3.464) {2};
    \node (U2) at (6,3.464) {5};
    \node (U0) at (0,6.928) {3}; 
    \node (U1) at (4,6.928) {4};  
\end{scope}
\begin{scope}[every edge/.style=edges]     
    
    \path [-] (U2) edge node {} (U3);
    \path [-] (U3) edge node {} (U4);
    \path [-] (U4) edge node {} (U5);
    \path [-] (U5) edge node {} (U0);
    \path [-] (U0) edge node {} (U1);
    \path [-] (U1) edge node {} (U2);
\end{scope}

%Graph 3
\begin{scope}[shift={(26,2.853*1.3)},scale=1.3,every node/.style=vertices]
    \node (U0) at (-0.927,2.853) {5};
    \node (U1) at (0.927,2.853) {6};
    \node (U2) at (2.427,1.763) {7};
    \node (U3) at (3,0) {8};
    \node (U4) at (2.427,-1.763) {9};
    \node (U5) at (0.927,-2.853) {10}; 
    \node (U6) at (-0.927,-2.853) {1};  
    \node (U7) at (-2.427,-1.763) {2};
    \node (U8) at (-3,0) {3};
    \node (U9) at (-2.427,1.763) {4};  
\end{scope}
\begin{scope}[every edge/.style=edges]              
    \path [-] (U0) edge node {} (U1);
    \path [-] (U1) edge node {} (U2);
    \path [-] (U2) edge node {} (U3);
    \path [-] (U3) edge node {} (U4);
    \path [-] (U4) edge node {} (U5);
    \path [-] (U5) edge node {} (U6);
    \path [-] (U6) edge node {} (U7);
    \path [-] (U7) edge node {} (U8);
    \path [-] (U8) edge node {} (U9);
    \path [-] (U9) edge node {} (U0);
    \path [-] (U0) edge   node {} (U5);
    \path [-] (U1) edge   node {} (U6);
    \path [-] (U2) edge   node {} (U7);
    \path [-] (U3) edge  node {} (U8);
    \path [-] (U4) edge   node {} (U9);
\end{scope}
\end{tikzpicture}
}   
\caption{(Left) The graph $K_3$ is 2-vertex-minor universal. (Middle) The graph $C_6$ is 3-vertex-minor universal. (Right) The "wheel" graph of order 10 is 4-vertex-minor universal.}
\label{fig:examples_vmu}
\end{figure}
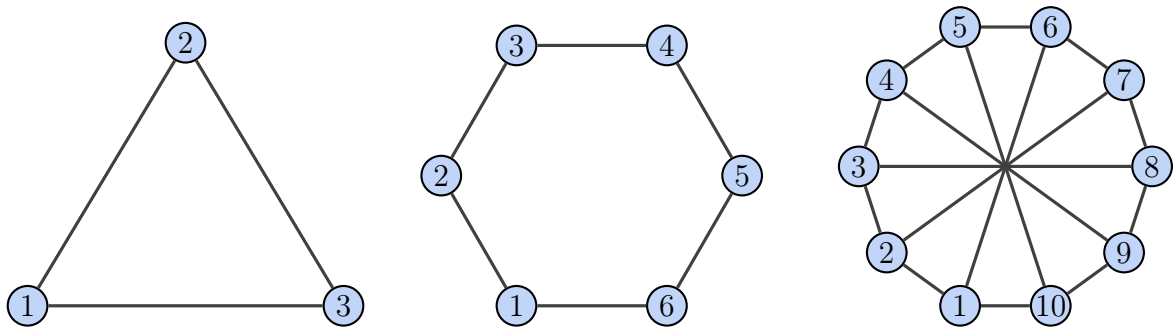

\subparagraph{2-vertex-minor universality.}
As we saw above, $K_2$ is not 2-vertex-minor universal: $K_3$ (see Figure \ref{fig:examples_vmu}, left) is actually the smallest 2-vertex-minor universal graph, because every edge can be toggled by a local complementation on the opposite vertex.
\subparagraph{3-vertex-minor universality.}
We observed that $C_6$, the cycle of order 6 (see Figure \ref{fig:examples_vmu}, middle), is a 3-vertex-minor universal graph, and there exists no 3-vertex-minor universal graph of order smaller or equal to 5. Indeed, using the database from \cite{Adcock20} along with \cref{prop:deltaloc_vmu}, it appears that the only graph with order 5 or less having a minimum degree up to local complementation larger or equal to 2 (thus being a candidate for 3-vertex-minor universality) is $C_5$. We have observed however that $C_5$ is not 3-vertex-minor universal by exploring its orbit by local complementation composed of 132 graphs, and checking that none of them contains an independent set of size 3.
\subparagraph{4-vertex-minor universality.}
We observed that the 10-vertex "wheel" graph from \cite{bravyi2024generating} (see Figure \ref{fig:examples_vmu}, right) is 4-vertex-minor universal. Conversely, no graph of order $9$ or less is 4-vertex-minor universal \cite{bravyi2024generating}. We also observed that the Petersen graph (see Figure \ref{fig:petersen}), of order 10, is also 4-vertex-minor universal.

\vspace{1em}

\cref{prop:deltaloc_vmu} motivates looking at graphs with large minimum degree up to local complementation to find graphs with large vertex-minor universality. Paley graphs\footnote{Given a prime $q = 1 \bmod 4$, the Paley graph of order $q$ is a graph which vertices are elements of the finite field $\mathbb{F}_q$, such that two vertices share an edge if their difference is a square in $\mathbb{F}_q$.} are good candidates in this regard, as their minimum degree up to local complementation scales as the square root of their order \cite{Javelle12}. Numerical analysis indicates that the Paley graph of order 13 is 4-vertex-minor universal (but not 5-vertex-minor universal), the Paley graph of order 17 is 5-vertex-minor universal (but not 6-vertex-minor universal), and the Paley graph of order 29 is 5-vertex-minor universal. Although the computation becomes too long to check, the Paley graph of order 29 may be 6-vertex-minor universal, as any perfect matching on any 6 vertices is a vertex-minor.

\section{An optimal probabilistic construction}

In this section we prove that the bound given in \cref{prop:bound_vmu} is tight asymptotically, i.e. there exist $k$-vertex-minor universal graphs whose order grows quadratically with $k$. 

\begin{theorem}
    For any constant $\alpha>2$, there exists $k_0 \in \mathbb N$ such that for any $k>k_0$, there exists a $k$-vertex-minor universal graph $G$ of order at most $\alpha k^2$.
    \label{thm:existence_vmu}
\end{theorem}

The remaining of this section is a proof of \cref{thm:existence_vmu}. First we bound the probability that some graph of order $k$ is not a vertex-minor of a random bipartite graph $G$, in \cref{lemma:proba_full_rank}. Then we bound the probability that such a random bipartite graph is $k$-vertex-minor universal, in \cref{lemma:proba_vmu}, by defining some algorithm that tries to generate any graph as a vertex-minor of $G$. Finally, we prove that there exists a $k$-vertex-minor universal bipartite graph of quadratic order in $k$.
More precisely, the probability of a random bipartite graph of quadratic order being $k$-vertex-minor universal goes to $1$ exponentially fast in $k$. In this section, we note $L(G),R(G)$ the bipartition of a bipartite graph $G = (V,E)$ (implying $V=L(G)\cup R(G)$ and $L(G) \cap R(G) = \emptyset$).

\begin{proposition}
    Fix constants $\epsilon > 0$, $c>2$, and $c'>  \frac{1+\epsilon}{\ln(2)}$. There exists $k_0 \in \mathbb N$ such that for any $k>k_0$, the random bipartite graph $G$ (the probability of an edge existing between two vertices, one in $L(G)$ and one in $R(G)$, is $1/2$, independently of the other edges) with $|L(G)|= \lfloor c'k\ln(k) \rfloor$ and $|R(G)| = \lfloor c k^2 \rfloor$, is $k$-vertex-minor universal with probability at least $1 - e^{-\epsilon k \ln(k)}$. 
    \label{prop:proba_exp_vmu}
\end{proposition}

\cref{prop:proba_exp_vmu} will be proved alongside \cref{thm:existence_vmu} in this section. Notation-wise, given a set $A$ and an integer $k$, ${A \choose k}$ refers to $\{B \se A ~|~|B|=k\}$.

\begin{lemma}
    \label{lemma:proba_full_rank}
    Consider a random bipartite graph $G$ with $|L(G)|\gs k$, $|R(G)| \gs 4{k \choose 2}+5$: the probability of an edge existing between two vertices (one in $L(G)$ and one in $R(G)$) is $1/2$, independently of the other edges. Take $k \in \mathbb{N}$ and consider a set of vertices $K \in {L(G) \choose k}$. The probability that there exists a graph defined on $K$ that is not a vertex-minor of $G$ is upper bounded by $e^{-\frac{\left(\frac{|R(G)|}{4}-{k \choose 2}+1\right)^2}{\left(\frac{7|R(G)|}{4}-{k \choose 2}+1\right)}}$.
\end{lemma}

\begin{proof}
    For some $j \in \mathbb N\setminus\{0\}$ and $A \in {R(G) \choose j}$, consider the incidence matrix $M_A$ of size $j \times {k \choose 2}$, whose column  $i$ represents the pairs of vertices of $K$ that are in the neighborhood of the $i^{th}$ vertex of $A$, in the sense that its entries are 1 if the pair of vertices $u$,$v$ is in its neighborhood, 0 else. Note that if there exists some $A \in {R(G) \choose {k \choose 2}}$ whose incidence matrix $M_A$ is of full column-rank, then any $2^{{k \choose 2}}$ graph defined on $K$ is a vertex-minor of $G$. Indeed, column number $i$ represents the edges (resp. non-edges) of $K$ to be toggled by a local complementation on the $i^{th}$ vertex of $A$. So now we will bound the probability of such a set $A$ existing within $R(G)$.

    For this purpose we will greedily try to construct the set $A \in {R(G) \choose {k \choose 2}}$, one vertex after the other, by considering each vertex in $R(G)$ one by one, and we will lower bound the probability of the event "there exists some $A \in {R(G) \choose {k \choose 2}}$ whose incidence matrix $M_A$ is of full column-rank" by the probability of success of the algorithm. The algorithm works as follows. Arbitrarily order the vertices of $R(G)$. At each step (say that we have $j$ vertices in $A$ at some step), suppose the corresponding matrix of incidence (of size $j \times {k \choose 2}$) full column-rank. We consider the next vertex $u \in R(G)$ in the list: if adding its corresponding vector to $M_A$ increases its column-rank, then we add $u$ to $A$, else we remove $u$ from the vertices to consider. The algorithm stops (and succeeds) if $M_A$ has ${k \choose 2}$ columns and is full column-rank. Let us show that the probability of a vertex $u$ increasing the column-rank of $M_A$ (if $j < {k \choose 2}$) is lower-bounded by $1/4$. 

    If $M_A$ is of rank $j < {k \choose 2}$, there exists a non-zero vector $W$ (i.e. a set of pairs of vertices of $K$) that is orthogonal to all $j$ first vectors. $W$ can be seen as the characteristic function of the edges of some graph $F$ on the vertices of $L(G)$. Adding a vertex $u$ to $A$ increases the rank of $M_A$ if the vector $U$ of incidence of $u$ in $K$ is such that $U \cdot W = 1 \mod 2$ (because then $U$ is not in the span of $M_A$). Note that, if $F$ has exactly one edge, then there is exactly probability $\frac{1}{4}$ that $U \cdot W = 1 \mod 2$ (in this case the two ends of the unique edge of $F$ are adjacent to $u$, which happens with probability $\frac{1}{2} \times \frac{1}{2}$). As $F$ has at least one edge, it has at least one vertex of non-zero degree $z$. Let us draw randomly the neighborhood of $u$: first we draw among the vertices of $F\sm\{z\}$, then we add $z$ with probability $\frac{1}{2}$. The probability that an odd number of neighbors of $z$ are neighbors of $u$ is $1/2$, so drawing $z$ changes the parity of the number of edges in $F$ whose ends are both neighbors of $u$, with probability $1/2$. At the end of the day there is a probability of at least $\frac{1}{4}$ that $U \cdot W = 1 \mod 2$, so that $u$ increases the column-rank of $M_A$.
    
    Finally, the algorithm fails if we encounter more than $|R(G)|-{k \choose 2}+1$ vertices that did not increase the column-rank of $M_A$. Let us introduce a random variable $T$ that follows the distribution $B(|R(G)|,3/4)$. The probability that the algorithm fails is upper bounded by $\Pr(T \gs |R(G)|-{k \choose 2}+1)$. We will use the Chernoff bound: With $\mu = \mathbb{E}[T] = \frac{3|R(G)|}{4}$, for any $\delta > 0$, $\Pr(T \gs (1+\delta)\mu) \ls e^{-\frac{\delta^2}{2+\delta}\mu}$. As we need $(1+\delta)\mu = |R(G)|-{k \choose 2}+1$, we take $\delta = \frac{|R(G)|-{k \choose 2}+1-\mu}{\mu}$. From $|R(G)| \gs 4{k \choose 2}+5$ it follows that $\delta > 0$.
    The Chernoff bound then gives $$ \Pr\left(T \gs |R(G)|\!-\!{k \choose 2}\!+\!1\right) \ls e^{\!-\frac{\left(\frac{|R(G)|-{k \choose 2}+1-\mu}{\mu}\right)^2}{\left(\frac{|R(G)|-{k \choose 2}+1+\mu}{\mu}\right)}\mu} \!\!\!\!\!= e^{-\frac{\left(|R(G)|-{k \choose 2}+1-\mu\right)^2}{\left(|R(G)|-{k \choose 2}+1+\mu\right)}} = e^{-\frac{\left(\frac{|R(G)|}{4}-{k \choose 2}+1\right)^2}{\left(\frac{7|R(G)|}{4}-{k \choose 2}+1\right)}} $$

    So the probability of the existence of $A \se {R(G) \choose k}$ whose incidence matrix $M_A$ if of full column-rank is lower bounded by $ 1 -e^{-\frac{\left(\frac{|R(G)|}{4}-{k \choose 2}+1\right)^2}{\left(\frac{7|R(G)|}{4}-{k \choose 2}+1\right)}} $.
\end{proof}

\begin{lemma}
    \label{lemma:proba_vmu}
    Consider a random bipartite graph $G$ with $|L(G)|\gs k$, $|R(G)|\gs4{k \choose 2}+5$: the probability of an edge existing between two vertices (one in $L(G)$ and one in $R(G)$) is $1/2$, independently of the other edges. The probability that $G$ is $k$-vertex-minor universal is lower bounded by $$  1 - \left( \frac{k}{2^{|L(G)|-k+1}} + e^{-\frac{\left(\frac{|R(G)|}{4}-{k \choose 2}+1\right)^2}{\left(\frac{7(|R(G)|-k)}{4}-{k \choose 2}+1\right)}}\right)\times {n \choose k} $$
\end{lemma}

\begin{proof}
    Given a set $K \in {V \choose k}$, we consider the bad event $A_K$: "there exists a graph defined on $K$ that is not a vertex-minor of $G$". The probability that $G$ is $k$-vertex-minor universal is, by definition, $\Pr\left( \bigcap_{K \in {V \choose k}} \overline{A_K}\right)$. Suppose that each probability $\Pr(A_K)$ is upper bounded by some $p$. Then, using the union bound, $$\Pr\left( \bigcap_{K \in {V \choose k}} \overline{A_K}\right) = 1 - \Pr\left(\bigcup_{K \in {V \choose k}} A_K \right) \gs 1 - \sum_{K \in {V \choose k}} \Pr(A_K) \gs 1 - p\times {|L(G)|+|R(G)| \choose k}$$

    Now, fix some $K \in {V \choose k}$. Let us upper bound $\Pr(A_K)$. For this purpose, let us show in the following how one can induce with high probability any graph on $K$ as a vertex-minor, by first turning the graph $G$ into a graph $G'$ defined on a subset of $V$, that is bipartite (as $G$), and such that the vertices of $K$ are all "on the left". Let $L_{\overline K}=L(G)\sm K$, $R_K=R(G)\cap K$. The algorithm works as follows. Roughly speaking, we use pivotings to move the vertices of $R_K$ from the right side to the left side. Given $a\in R(G)$ and $b\in L(G)$, pivoting an edge $ab$ in a bipartite graph $G$ produces the bipartite graph $G\wedge ab=G*a*b*a$ where the edges between $N_G(a)\setminus \{b\}$ and $N_G(b)\setminus \{a\}$ are toggled and vertices $a$ and $b$ are exchanged (in other words $N_{G\wedge ab}(b)=N_G(a)\Delta \{a,b\}$, $N_{G\wedge ab}(a)=N_G(b)\Delta \{a,b\}$, so that the graph is bipartite according to the partition $R':=R(G)\Delta\{a,b\}$, $L':=L(G)\Delta\{a,b\}$). Once all vertices of $R_K$ are moved to the left by means of pivotings, we then consider the induced subgraph obtained by removing the vertices that have been moved from the left to the right side. We obtain a bipartite graph such that all vertices of $K$ are on the left side, the idea is then to apply \cref{lemma:proba_full_rank} to show that with high probability one can induce any graph on $K$ as a vertex-minor (using only local complementation on vertices on the right side). So we have to prove that the constructed graph behaves as a random bipartite graph: each edge exists independently with probability $1/2$.

    To this end we provide a little more details on the algorithm. Given the initial random bipartite graph $G$, we proceed as follows: given a vertex $a\in R_K$, if there is no edge between $a$ and $L_{\overline K}$, the algorithm fails. Otherwise, we consider an arbitrary vertex $b\in N_G(a)\cap L_{\overline K}$ and perform a pivoting on $ab$, then remove vertex $b$, leading to a graph $G\wedge ab\setminus b$ that is bipartite according to $R':=R\setminus \{a\}$, $L':=L(G)\Delta\{a,b\}$. We show in the following that this bipartite graph is random, i.e. each edge exists independently with probability $1/2$.    
    \begin{itemize}
        \item For any $u\in R'$ we have $a\sim_{G\wedge ab} u$ if and only if $b\sim_{G} u$, so $\Pr(a\sim_{G\wedge ab} u)=\frac12$.
        \item For any $u\in R', v\in L'\setminus \{a\}$ we have $u\sim_{G\wedge ab} v$  if and only if  ($u\sim_{G} v ~~\oplus~~ (u\in N_G(b)\wedge v\in N_G(a))$). As the event $u\in N_G(b)\wedge v\in N_G(a)$ is independent of the existence of an edge between $u$ and $v$ in $G$, we have $\Pr(u\sim_{G\wedge ab} v)=\frac12$.
    \end{itemize}
    Regarding independence, notice that the existence of an edge $(u,v)$ in $G\wedge ab\sm b$ is independent of the existence of all edges but $(u,v)$ in $G$. The independence of the existence of the edges in $G$ guarantees the independence in $G\wedge ab \sm b$.

    To sum up, starting from a random bipartite graph and a vertex $a\in R_K$, if there is an edge between $a$ and some vertex $b\in L_{\overline K}$, we move $a$ to the left side (by means of a pivoting) and remove $b$: the remaining graph is a random bipartite graph. The algorithm consists in repeating this process until $R_K$ is empty, leading to a random bipartite graph. If the procedure succeeds, we end up with a bipartite graph $G'$ with $|R(G')|=|R(G)|-|R_K|$, $|L(G')|=|L(G)|$ such that $R(G') \se R(G)$ and $L(G') \se L(G) \cup R_K$. Recall that we obtained $G'$ from $G$ using only local complementations and vertex-deletions, so any vertex-minor of $G'$ is a vertex-minor of $G$. At the end of the day, $\Pr(A_K)$ is upper bounded by the probability that $(i)$ the algorithm fails or $(ii)$ there exists a graph defined on $K$ that is not a vertex-minor of $G'$.
    \begin{itemize}
        \item[$(i)$] Let us upper bound the probability that the algorithm fails. The algorithm is composed of $|R_K|$ steps. At each step, the algorithm fails if there is no edge between $a$ and $L_{\overline K}$. The bipartite graph, at this point, is random: each edge between $a$ and the vertices of $L_{\overline K}$ exists independently with probability 1/2. The probability that there is no edge between $a$ and $L_{\overline K}$ is $\frac{1}{2^{|L_{\overline K}|}}$.  In general, $|L_{\overline K}|$ is lower bounded by $|L(G)|-k+1$, and $|R_K|$ is upper bounded by $k$. At the end of the day, using the union bound, the algorithm fails with probability at most $\frac{k}{2^{|L(G)|-k+1}}$. 
        \item[$(ii)$] Suppose the algorithm succeeds. We end up with a random bipartite graph $G'$. Using \cref{lemma:proba_full_rank}, the probability that there exists a graph defined on $K$ that is not a vertex-minor of $G'$ is upper bounded by $$e^{-\frac{\left(\frac{|R(G)|-|R_K|}{4}-{k \choose 2}+1\right)^2}{\left(\frac{7(|R(G)|-|R_K|)}{4}-{k \choose 2}+1\right)}} \ls e^{-\frac{\left(\frac{|R(G)|}{4}-{k \choose 2}+1\right)^2}{\left(\frac{7(|R(G)|-k)}{4}-{k \choose 2}+1\right)}}$$
    \end{itemize}

    So, at the end of the day, $$ \Pr(A_K) \ls \frac{k}{2^{|L(G)|-k+1}} + e^{-\frac{\left(\frac{|R(G)|}{4}-{k \choose 2}+1\right)^2}{\left(\frac{7(|R(G)|-k)}{4}-{k \choose 2}+1\right)}}$$

    And then,
    \begin{equation*}\Pr\left( \bigcap_{K \in {V \choose k}} \overline{A_K}\right)  \ls 1 - \left( \frac{k}{2^{|L(G)|-k+1}} + e^{-\frac{\left(\frac{|R(G)|}{4}-{k \choose 2}+1\right)^2}{\left(\frac{7(|R(G)|-k)}{4}-{k \choose 2}+1\right)}}\right)\times {n \choose k} \qedhere
    \end{equation*}

\end{proof}

\begin{remark}
    \cref{lemma:proba_vmu} has concrete applications on its own right: one can infer a lower bound on the probability of generating a $k$-vertex-minor universal graph, for any choice of $k$ and $n$. Figure \ref{fig:99} is a table containing orders for which a randomly generated bipartite graph is $k$-vertex-minor universal with at least 99\% probability, for particular values of $k$ ranging from 3 to 100.
\end{remark}

\begin{figure}[h!]
    \centering
    \begin{tabular}{|c|c|c|c|c|c|c|c|c|c|c|c|c|c|}
        \hline
        k               & 3  & 4  & 5  & 6  & 7  & 8  & 9  & 10  & 11  & 12  & 13  & 14  & 15  \\
        \hline
        n  & 47  & 68  & 93  & 123  & 156  & 194  & 235  & 281  & 331  & 385  & 443  & 505  & 571 \\
        \hline
        $|L(G)|$  & 25  & 32  & 39  & 47  & 55  & 63  & 71  & 79  & 88  & 96  & 105  & 113  & 122 \\
        \hline
        $|R(G)|$ & 22  & 36  & 54  & 76  & 101  & 131  & 164  & 202  & 243  & 289  & 338  & 392  & 449 \\
        \hline
    \end{tabular}

    \bigskip

    \begin{tabular}{|c|c|c|c|c|c|c|c|c|c|c|c|c|c|}
        \hline
        k    & 20  & 25  & 30  & 35  & 40  & 50  & 60  & 70  & 80  & 90  & 100 \\
        \hline
        n  & 962  & 1456  & 2049  & 2743  & 3539  & 5431  & 7726  & 10422  & 13519  & 17019  & 20920 \\
        \hline
        $|L(G)|$  & 167  & 215  & 263  & 313  & 364  & 468  & 575  & 684  & 795  & 908  & 1023 \\
        \hline
        $|R(G)|$  & 795  & 1241  & 1786  & 2430  & 3175  & 4963  & 7151  & 9738  & 12724  & 16111  & 19897 \\
        \hline
    \end{tabular}
    \caption{Parameters for which a randomly generated bipartite graph is $k$-vertex-minor universal with at least 99\% probability. }
    \label{fig:99}
\end{figure}

Now we are ready to conclude. Fix some constants $c>2$ and $c'>\frac{1}{\ln(2)}$. Let $G$ be a random bipartite graph $G$ with $|L(G)| = \lfloor c'k\ln(k) \rfloor$ and $|R(G)| = \lfloor ck^2 \rfloor$: the probability of an edge existing between two vertices (one in $L(G)$ and one in $R(G)$) is $1/2$, independently of the other edges.\\

Recall $n = |V| = |L(G)| + |R(G)| = \lfloor c'k\ln(k) \rfloor + \lfloor ck^2 \rfloor$. Using \cref{lemma:proba_vmu}, the probability that $G$ is $k$-vertex-minor universal is lower bounded by $$  1 - \left( \frac{k}{2^{|L(G)|-k+1}} + e^{-\frac{\left(\frac{|R(G)|}{4}-{k \choose 2}+1\right)^2}{\left(\frac{7(|R(G)|-k)}{4}-{k \choose 2}+1\right)}}\right)\times {n \choose k} $$

Let us prove that this probability is positive with our choice of parameters, for some big enough $k$. It is sufficient to have: $$(i)~~  \frac{k}{2^{|L(G)|-k+1}}{n \choose k}< \dfrac{1}{2} \text{~~and~~} (ii)~~ e^{-\frac{\left(\frac{|R(G)|}{4}-{k \choose 2}+1\right)^2}{\left(\frac{7(|R(G)|-k)}{4}-{k \choose 2}+1\right)}}{n \choose k}< \dfrac{1}{2} $$

Let us show that these equations are satisfied for any large enough $k$. Recall that ${n \choose k} \ls 2^{nH_2(k/n)}$ where $H_2(x) = -x\log_2(x)-(1-x)\log_2(1-x)$ is the binary entropy.

$(i)$: It is sufficient that $\log_2(k)+nH_2(k/n)-|L(G)|+k-1 < -1$. 

$\log_2(k)+nH_2(k/n)-|L(G)|+k-1  \sim_{k\to \infty} n \frac{k}{n}\log_2(\frac{k}{n}) - c'k\ln(k) = k (\log_2(k) - \log_2(n)) - c'k\ln(k) \sim_{k\to \infty} \frac{1}{\ln(2)}k\ln(k) - c'k\ln(k)$. The choice of $c'$ guarantees that for any large enough $k$, $(i)$ is satisfied.

$(ii)$: It is sufficient that $nH_2(k/n)\ln(2)-\frac{\left(\frac{|R(G)|}{4}-{k \choose 2}+1\right)^2}{\left(\frac{7(|R(G)|-k)}{4}-{k \choose 2}+1\right)} < -\ln(2)$. $\frac{\left(\frac{|R(G)|}{4}-{k \choose 2}+1\right)^2}{\left(\frac{7(|R(G)|-k)}{4}-{k \choose 2}+1\right)}$  $\sim_{k\to \infty} \frac{\left(\frac{ck^2}{4}-\frac{k^2}{2}\right)^2}{\left(\frac{7ck^2}{4}-\frac{k^2}{2}\right)} = k^2 \dfrac{(c-2)^2}{4(7c-2)}$. We saw above that $nH_2(k/n)\ln(2) \sim_{k\to \infty} k\ln(k)$. The choice of $c$ guarantees that for any large enough $k$, $(i)$ is satisfied.

This proves that, for any large enough $k$, $G$ of order $\lfloor c'k\ln(k) \rfloor + \lfloor ck^2 \rfloor$, is $k$-vertex-minor universal with non-zero probability. Taking $\alpha > c$, for any large enough $k$, $\lfloor c'k\ln(k) \rfloor + \lfloor ck^2 \rfloor \ls \alpha k^2$, proving \cref{thm:existence_vmu}.

Furthermore, we just saw that side $(i)$ of the equation dominates $(ii)$ asymptotically. Thus, the probability of $G$ being $k$-vertex-minor universal is roughly lower bounded by $1 - 2^{\frac{1}{\ln(2)}k\ln(k) - c'k\ln(k)} = 1 - e^{-(\ln(2)c'-1)k\ln(k)}$ as $k$ grows. Then, for any $\epsilon > 0$ such that $\epsilon < ln(2)c'-1$, for any large enough $k$, $G$ of order $\lfloor c'k\ln(k) \rfloor + \lfloor ck^2 \rfloor$, is $k$-vertex-minor universal with probability at least $1 - e^{-\epsilon k \ln(k)}$, proving \cref{prop:proba_exp_vmu}.

\chapter{Conclusion}

I believe the main contribution of this thesis is the introduction of the generalized local complementation, or $r$-local complementation, as an exact characterization of LU-equivalence between graph states. The most obvious application was the design of a quasi-polynomial algorithm for deciding LU-equivalence. Also, LU=LC holds for graph states on up to 19 qubits is a nice result, as 8 qubits was the state of the art for many years. 
Below are possible research directions for local equivalences between graph states.

\begin{itemize}
    \item Does there exist a polynomial time algorithm to decide LU-equivalence of graph states? Unfortunately, it does not seem that $r$-local complementation can help that much (besides providing a quasi-polynomial time algorithm), as even just verifying that an $r$-local complementation is valid
    relies on the verification of a non-polynomial number of equations. 
    Consequently, whether deciding LU-equivalence is in NP is still open.
    \item Is the 27-vertex counterexample to the LU=LC conjecture minimal? Trying to prove LU=LC holds for graphs with up to 26 vertices takes too much computational power with only the results from \cref{chap:conditionsLULC}. Thus, new mathematical results that reduce the search space, for example by making use of symmetries, are necessary.
    \item Local Clifford operations (including Pauli measurements) on graph states are harnessed by the very powerful formalism of vertex-minors. Is a similar treatment possible for arbitrary local operations? One can imagine for example a notion of vertex-minor where the local complementation is replaced by the $r$-local complementation. 
    \item It is common to study the LC-orbit of graph states, i.e. the orbit under local complementation. Studying the LU-orbit, i.e. the orbit under $r$-local complementation, is arguably even more natural, as the orbit represents the set of graph states that have the same entanglement (in the sense of LU-equivalence). Maybe some results about local complementation orbits, for example regarding their size \cite{bahramgiri2007enumerating, dahlberg2020counting}, can be lifted to $r$-local complementation orbits. Another recent problem of interest is finding a graph with a minimal number of edges in the LC-orbit \cite{kumabe2025complexitygraphstatepreparationclifford, davies2025preparinggraphstatesforbidding, sharma2025minimizingnumberedgeslcequivalent}. It is not clear whether the number of edges can be further lowered when considering the LU-orbit instead of the LC-orbit.
    \item Generalized local complementation, when looked at under the scope of weighted hypergraph states, is nothing but simultaneous X-rotations inducing weighted hyperedges that happen to cancel out. Similar techniques than those in this thesis may be used to study the LC-equivalence and LU-equivalence of (weighted) hypergraph states, for example by designing (hyper)graphical operations that capture the action of local operations.
\end{itemize}

This has been a real pleasure working on graph states during these last three years. They are overall beautiful combinatorial objects. Besides, I believe them to be a nice bridge between physicists and graph theorists \cite{chitambar2025quantumgraphstatesbridging}, linking physical concepts such as local operations over quantum states (unitary operations, Clifford operations or single-qubit measurements), to graphical concepts such as local complementation, vertex-minors, and now generalized local complementation.

%%% Appendicies of thesis  %%%%%%%%%%%%%%%%%%%%%%%%%%%%%%%%%%%%%%%%%%%%%%%%%%%%%%%%%%%%%%%%%%%%%%%%

\appendix
%\include{appendixa}
%\include{appendixb}

%%% Bibliography  %%%%%%%%%%%%%%%%%%%%%%%%%%%%%%%%%%%%%%%%%%%%%%%%%%%%%%%%%%%%%%%%%%%%%%

%List of publications
%\apptocmd{\bibliography}{\addcontentsline{toc}{chapter}{\protect\textbf{Publications and scripts}}}{}{}

%\apptocmd{\bibliographypub}{\addcontentsline{toc}{chapter}{\protect\textbf{Bibliography}}}{}{}

%\addcontentsline{toc}{chapter}{\protect\textbf{Bibliography}}

%\bibliographystylepub{unsrturl}
%\bibliographypub{mypublications.bib}

%%% add bibliography to table of contents
%\apptocmd{\bibliography}{\addcontentsline{toc}{chapter}{\protect\textbf{\bibname}}}{}{}

%%% Change name of section (before: Bibliography)
\renewcommand\bibname{References}
%\vspace{-50pt}

%%% To adjust space between bibliography items 
%\setlength\bibsep{4pt plus 1pt minus 1pt}
%   change 4pt to something else; don't drop last two lengths (they are stretchable "glue")

\bibliographystyle{unsrturl}
\bibliography{ref.bib}

@inproceedings{claudet2024covering,
	archiveprefix = {arXiv},
	author = {Nathan Claudet and Simon Perdrix},
	booktitle = {Proceedings of the 50th workshop on Graph Theory (WG 2024) },
	eprint = {2402.10678},
	primaryclass = {quant-ph},
	title = {Covering a Graph with Minimal Local Sets},
	year = {2024},
	doi={10.1007/978-3-031-75409-8_10}
}

@inproceedings{claudet2024local,
	archiveprefix = {arXiv},
	author = {Nathan Claudet and Simon Perdrix},
	booktitle = {Proceedings of the 42nd {S}ymposium on {T}heoretical {A}spects of {C}omputer Science ({STACS} 2025)},
	eprint = {2409.20183},
	primaryclass = {quant-ph},
	title = {Local equivalence of stabilizer states: a graphical characterisation},
	year = {2025},
	doi ={10.4230/LIPIcs.STACS.2025.27}
}

@inproceedings{claudet2025deciding,
	archiveprefix = {arXiv},
	author = {Nathan Claudet and Simon Perdrix},
	booktitle = {Proceedings of the 52nd International Colloquium on Automata, Languages, and Programming (ICALP 2025)},
	eprint = {2502.06566},
	primaryclass = {quant-ph},
	title = {Deciding Local Unitary Equivalence of Graph States in Quasi-Polynomial Time},
	year = {2025},
	doi ={10.4230/LIPIcs.ICALP.2025.59}
}

@inproceedings{Cautres2024,
	author = {Cautr\`{e}s, Maxime and Claudet, Nathan and Mhalla, Mehdi and Perdrix, Simon and Savin, Valentin and Thomass\'{e}, St\'{e}phan},
	title =	{Vertex-Minor Universal Graphs for Generating Entangled Quantum Subsystems},
	booktitle = {Proceedings of the 51st International Colloquium on Automata, Languages, and Programming (ICALP 2024)},
	doi = {10.4230/LIPIcs.ICALP.2024.36},
	year = {2024},
	eprint={2402.06260}
}

@misc{claudet2023smallkpairablestates,
      title={Small k-pairable states}, 
      author={Nathan Claudet and Mehdi Mhalla and Simon Perdrix},
      year={2023},
      eprint={2309.09956},
      archivePrefix={arXiv},
      primaryClass={quant-ph}
}

@misc{codelulc19,
  author       = {Nathan Claudet},
  title={{LU=LC} up to 19 qubits: \href{https://github.com/nathanclaudet/LU-equals-LC-up-to-19qubits}{github.com/nathanclaudet/LU-equals-LC-up-to-19qubits}},
  year = {2025}
}

@article{kotzig1968eulerian,
  title={{E}ulerian lines in finite 4-valent graphs and their transformations},
  author={Kotzig, Anton},
  journal={Theory of Graphs},
  pages={219--230},
  year={1968},
  publisher={Academic Press}
}

@misc{kotzig1977quelques,
  title={Quelques remarques sur les transformations $\kappa$},
  author={Kotzig, Anton},
  journal={seminaire Paris},
  year={1977}
}

@inproceedings{amy2024exact,
author="Amy, Matthew
and Glaudell, Andrew N.
and Kelso, Shaun
and Maxwell, William
and Mendelson, Samuel S.
and Ross, Neil J.",
title="Exact Synthesis of Multiqubit {C}lifford-Cyclotomic Circuits",
booktitle="Proceedings of the 16th Conference on Reversible Computation (RC 2024)",
year="2024",
eprint={2311.07741},
doi={10.1007/978-3-031-62076-8_15}
}

@article{Tsimakuridze17,
	author = {Nikoloz Tsimakuridze and Otfried G{\"u}hne},
	doi = {10.1088/1751-8121/aa67cd},
	journal = {Journal of Physics A: Mathematical and Theoretical},
	month = {Apr},
	number = {19},
	pages = {195302},
	publisher = {{IOP} Publishing},
	title = {Graph states and local unitary transformations beyond local {C}lifford operations},
	volume = {50},
	year = 2017,
	eprint={1611.06938}
}

@article{VandenNest04,
	author = {Maarten Van den Nest and Jeroen Dehaene and Bart De Moor},
	doi = {10.1103/physreva.69.022316},
	eprint = {quant-ph/0308151},
	journal = {Physical Review A},
	month = {Feb},
	number = {2},
	publisher = {American Physical Society ({APS})},
	title = {Graphical description of the action of local {C}lifford transformations on graph states},
	volume = {69},
	year = {2004}
}

@article{Bouchet1987,
title = {Isotropic Systems},
journal = {European Journal of Combinatorics},
volume = {8},
number = {3},
pages = {231-244},
year = {1987},
issn = {0195-6698},
doi = {10.1016/S0195-6698(87)80027-6},
author = {André Bouchet}
}

@article{Bouchet1991,
	author = {André Bouchet},
	day = {01},
	doi = {10.1007/BF01275668},
	issn = {1439-6912},
	journal = {Combinatorica},
	month = {Dec},
	number = {4},
	pages = {315-329},
	title = {An efficient algorithm to recognize locally equivalent graphs},
	url = {https://doi.org/10.1007/BF01275668},
	volume = {11},
	year = {1991}
}

@article{VdnEfficientLC,
	author = {Maarten Van den Nest and Dehaene, Jeroen and De Moor, Bart},
	doi = {10.1103/PhysRevA.70.034302},
	issue = {3},
	journal = {Physical Review A},
	month = {Sep},
	numpages = {3},
	pages = {034302},
	publisher = {American Physical Society},
	title = {Efficient algorithm to recognize the local {C}lifford equivalence of graph states},
	volume = {70},
	year = {2004},
	eprint={quant-ph/0405023}
}

@article{Hein06,
	author = {Hein, Marc and D{\"u}r, Wolfgang and Eisert, Jens and Raussendorf, Robert and Maarten Van den Nest and Briegel, Hans J.},
	doi = {10.3254/978-1-61499-018-5-115},
	eprint = {quant-ph/0602096},
	journal = {Quantum computers, algorithms and chaos},
	month = {Mar},
	title = {Entanglement in Graph States and its Applications},
	volume = {162},
	year = {2006}
}

@article{Zeng08,
  title = {Semi-{C}lifford operations, structure of {$\mathcal C_k$} hierarchy, and gate complexity for fault-tolerant quantum computation},
  author = {Zeng, Bei and Chen, Xie and Chuang, Isaac L.},
  journal = {Physical Review A},
  volume = {77},
  issue = {4},
  pages = {042313},
  numpages = {12},
  year = {2008},
  month = {Apr},
  publisher = {American Physical Society},
  doi = {10.1103/PhysRevA.77.042313},
  eprint={0712.2084}
}

@article{Zeng07,
	author = {Bei Zeng and Hyeyoun Chung and Andrew W. Cross and Isaac L. Chuang},
	doi = {10.1103/physreva.75.032325},
	journal = {Physical Review A},
	month = {Mar},
	number = {3},
	publisher = {American Physical Society ({APS})},
	title = {Local unitary versus local {C}lifford equivalence of stabilizer and graph states},
	eprint={quant-ph/0611214},
	volume = {75},
	year = 2007
}

@article{Cui2016,
  title = {Diagonal gates in the {C}lifford hierarchy},
  author = {Cui, Shawn X. and Gottesman, Daniel and Krishna, Anirudh},
  journal = {Physical Review A},
  volume = {95},
  issue = {1},
  pages = {012329},
  numpages = {7},
  year = {2017},
  month = {Jan},
  publisher = {American Physical Society},
  doi = {10.1103/PhysRevA.95.012329},
  eprint={1608.06596}
}

@article{Verstraete2003,
  title = {Normal forms and entanglement measures for multipartite quantum states},
  author = {Verstraete, Frank and Dehaene, Jeroen and De Moor, Bart},
  journal = {Physical Review A},
  volume = {68},
  issue = {1},
  pages = {012103},
  numpages = {7},
  year = {2003},
  month = {Jul},
  publisher = {American Physical Society},
  doi = {10.1103/PhysRevA.68.012103},
  eprint={quant-ph/0105090}
}

@article{Hein04,
	author = {Marc Hein and Jens Eisert and Hans J. Briegel},
	doi = {10.1103/physreva.69.062311},
	eprint = {quant-ph/0307130},
	journal = {Physical Review A},
	month = {Jun},
	number = {6},
	publisher = {American Physical Society ({APS})},
	title = {Multiparty entanglement in graph states},
	volume = {69},
	year = 2004
}

@article{gross2007lu,
	author = {David Gross and Maarten Van den Nest},
	bibsource = {dblp computer science bibliography, https://dblp.org},
	doi = {10.26421/QIC8.3-4-3},
	journal = {Quantum Information and Computation},
	number = {3},
	pages = {263--281},
	title = {The {LU-LC} conjecture, diagonal local operations and quadratic forms over {GF(2)}},
	url = {https://doi.org/10.26421/QIC8.3-4-3},
	volume = {8},
	year = {2008},
	eprint={0707.4000}
}

@article{Rains97,
	author = {Eric M. Rains},
	bibsource = {dblp computer science bibliography, https://dblp.org},
	doi = {10.1109/18.746807},
	journal = {{IEEE} Transactions on Information Theory},
	number = {1},
	pages = {266--271},
	title = {Quantum Codes of Minimum Distance Two},
	url = {https://doi.org/10.1109/18.746807},
	volume = {45},
	year = {1999}
}

@inproceedings{CattaneoP15,
	author = {David Cattan{\'{e}}o and Simon Perdrix},
	booktitle = {Proceedings of the 26th International Symposium on Algorithms and Computation, ({ISAAC} 2015)},
	doi = {10.1007/978-3-662-48971-0\_23},
	title = {Minimum Degree Up to Local Complementation: Bounds, Parameterized Complexity, and Exact Algorithms},
	year = {2015},
	eprint={1503.04702}}

@article{mhalla2012graph,
	author = {Mehdi Mhalla and Simon Perdrix},
	journal = {International Journal of Unconventional Computing},
	number = {1-2},
	pages = {153--171},
	title = {Graph States, Pivot Minor, and Universality of ({X},{Z})-Measurements},
	volume = {9},
	year = {2013},
	eprint={1202.6551}
}

@inproceedings{gravier2013quantum,
	author = {Gravier, Sylvain and Javelle, J{\'e}r{\^o}me and Mhalla, Mehdi and Perdrix, Simon},
	booktitle = {Proceedings of the 8th {M}athematical and {E}ngineering {M}ethods in {C}omputer {S}cience International Doctoral Workshop ({MEMICS} 2012)},
	title = {Quantum secret sharing with graph states},
	doi={10.1007/978-3-642-36046-6_3},
	year = {2013},
	url={https://hal.science/hal-00933722/document}
}

@article{markham2008graph,
	author = {Markham, Damian and Sanders, Barry C.},
	journal = {Physical Review A},
	number = {4},
	pages = {042309},
	publisher = {APS},
	title = {Graph states for quantum secret sharing},
	doi={10.1103/PhysRevA.78.042309},
	volume = {78},
	year = {2008},
	eprint={0808.1532}
}

@InProceedings{Javelle2013,
author="Javelle, J{\'e}r{\^o}me
and Mhalla, Mehdi
and Perdrix, Simon",
title="New Protocols and Lower Bounds for Quantum Secret Sharing with Graph States",
booktitle="Proceedings of the 7th Conference on the Theory of Quantum Computation, Communication, and Cryptography (TQC 2012)",
year="2013",
pages="1--12", 
eprint={1109.1487},
doi={10.1007/978-3-642-35656-8_1}
}

@misc{schlingemann2001stabilizer,
	author = {Schlingemann, Dirk},
	eprint={quant-ph/0111080},
	title = {Stabilizer codes can be realized as graph codes},
	year = {2001}
}

@article{schlingemann2001quantum,
	author = {Schlingemann, Dirk and Werner, Reinhard F.},
	journal = {Physical Review A},
	number = {1},
	pages = {012308},
	publisher = {APS},
	title = {Quantum error-correcting codes associated with graphs},
	doi={10.1103/PhysRevA.65.012308},
	volume = {65},
	year = {2001},
	eprint={quant-ph/0012111}
}

@article{OumCliqueWidth,
	author = {{Sang-il} Oum},
	doi = {https://doi.org/10.1016/j.jctb.2005.03.003},
	issn = {0095-8956},
	journal = {Journal of Combinatorial Theory, Series B},
	number = {1},
	pages = {79-100},
	title = {Rank-width and vertex-minors},
	volume = {95},
	year = {2005}
}

@article{dahlberg2020transforming,
	author = {Dahlberg, Axel and Helsen, Jonas and Wehner, Stephanie},
	doi = {10.22331/q-2020-10-22-348},
	eprint = {1907.08019},
	issn = {2521-327X},
	journal = {{Quantum}},
	month = {Oct},
	pages = {348},
	publisher = {{Verein zur F{\"{o}}rderung des Open Access Publizierens in den Quantenwissenschaften}},
	title = {Transforming graph states to {B}ell-pairs is {NP}-{C}omplete},
	volume = {4},
	year = {2020}
}

@article{DHW:howtotransform,
	author = {Axel Dahlberg and Jonas Helsen and Stephanie Wehner},
	doi = {10.1088/2058-9565/aba763},
	eprint = {1805.05306},
	journal = {Quantum Science and Technology},
	month = {Sep},
	number = {4},
	pages = {045016},
	publisher = {IOP Publishing},
	title = {How to transform graph states using single-qubit operations: computational complexity and algorithms},
	volume = {5},
	year = {2020}
}

@article{DWH:transfo,
	author = {Dahlberg, Axel and Wehner, Stephanie},
	doi = {10.1098/rsta.2017.0325},
	eprint = {1805.05305},
	journal = {Philosophical Transactions of the Royal Society A: Mathematical, Physical and Engineering Sciences},
	number = {2123},
	pages = {20170325},
	title = {Transforming graph states using single-qubit operations},
	volume = {376},
	year = {2018}
}

@article{hahn2019quantum,
	author = {Hahn, Frederik and Pappa, Anna and Eisert, Jens},
	doi = {10.1038/s41534-019-0191-6},
	eprint = {1805.04559},
	journal = {npj Quantum Information},
	number = {1},
	pages = {1--7},
	publisher = {Nature Publishing Group},
	title = {Quantum network routing and local complementation},
	volume = {5},
	year = {2019}
}

@inproceedings{Perdrix06,
	author = {H{{\o}}yer, Peter and Mhalla, Mehdi and Perdrix, Simon},
	booktitle = {Proceedings of the 17th International Symposium on Algorithms and Computation (ISAAC 2006)},
	doi = {10.1007/11940128\_64},
	month = {Dec},
	title = {Resources Required for Preparing Graph States},
	url = {https://hal.archives-ouvertes.fr/hal-01378771},
	year = {2006}
}

@article{Adcock20,
	author = {Adcock, Jeremy C. and Morley-Short, Sam and Dahlberg, Axel and Silverstone, Joshua W.},
	doi = {10.22331/q-2020-08-07-305},
	eprint = {1910.03969},
	issn = {2521-327X},
	journal = {Quantum},
	month = {Aug},
	pages = {305},
	publisher = {Verein zur Forderung des Open Access Publizierens in den Quantenwissenschaften},
	title = {Mapping graph state orbits under local complementation},
	volume = {4},
	year = {2020}
}

@article{VandenNest05,
	author = {Maarten Van den Nest and Jeroen Dehaene and Bart De Moor},
	doi = {10.1103/physreva.71.062323},
	journal = {Physical Review A},
	month = {Jun},
	number = {6},
	publisher = {American Physical Society ({APS})},
	title = {Local unitary versus local {C}lifford equivalence of stabilizer states},
	volume = {71},
	year = 2005,
	eprint={quant-ph/0411115}
}

@article{Ji07,
author = {Ji, Zhengfeng and Chen, Jianxin and Wei, Zhaohui and Ying, Mingsheng},
title = {The {LU-LC} conjecture is false},
year = {2010},
issue_date = {Jan 2010},
publisher = {Rinton Press, Incorporated},
address = {Paramus, NJ},
volume = {10},
number = {1},
issn = {1533-7146},
journal = {Quantum Information and Computation},
month = {Jan},
pages = {97–108},
numpages = {12},
eprint={0709.1266},
doi = {QIC10.1-2-8.html}
}

@misc{5pb,
	author = {Krueger, Olaf and Werner, Reinhard F.},
	copyright = {Assumed arXiv.org perpetual, non-exclusive license to distribute this article for submissions made before January 2004},
	eprint={quant-ph/0504166},
	publisher = {arXiv},
	title = {Some Open Problems in Quantum Information Theory},
	year = {2005}
}

@inproceedings{Javelle12,
	author = {J{\'e}r{\^o}me Javelle and Mehdi Mhalla and Simon Perdrix},
	booktitle = {Proceedings of the 38th workshop on Graph Theory (WG 2012)},
	doi = {10.1007/978-3-642-34611-8_16},
	title = {On the Minimum Degree Up to Local Complementation: Bounds and Complexity},
	year = {2012},
	eprint={1204.4564}
}

@article{Bouchet1993,
	author = {André Bouchet},
	doi = {10.1016/0012-365X(93)90357-Y},
	issn = {0012-365X},
	journal = {Discrete Mathematics},
	number = {1},
	pages = {75-86},
	title = {Recognizing locally equivalent graphs},
	volume = {114},
	year = {1993}
}

@article{bravyi2024generating,
	author = {Bravyi, Sergey and Sharma, Yash and Szegedy, Mario and de Wolf, Ronald},
	journal = {Quantum},
	pages = {1348},
	publisher = {Verein zur F{\"o}rderung des Open Access Publizierens in den Quantenwissenschaften},
	title = {Generating $ k $ EPR-pairs from an $ n $-party resource state},
	doi={10.22331/q-2024-05-14-1348},
	volume = {8},
	year = {2024},
	eprint={2211.06497}
}

@article{meignant2019distributing,
	author = {Meignant, Cl\'ement and Markham, Damian and Grosshans, Fr\'ed\'eric},
	doi = {10.1103/PhysRevA.100.052333},
	eprint = {1811.05445},
	issue = {5},
	journal = {Physical Review A},
	month = {Nov},
	numpages = {6},
	pages = {052333},
	publisher = {American Physical Society},
	title = {Distributing graph states over arbitrary quantum networks},
	volume = {100},
	year = {2019}
}

@inproceedings{fischer2021distributing,
	author = {Fischer, Alex and Towsley, Don},
	booktitle = {Proceedings of the 2021 IEEE International Conference on Quantum Computing and Engineering (QCE)},
	doi = {10.1109/QCE52317.2021.00049},
	eprint = {2009.10888},
	pages = {324--333},
	title = {Distributing graph states across quantum networks},
	year = {2021},
	eprint={2009.10888}
}

@misc{bahramgiri2007enumerating,
	author = {Bahramgiri, Mohsen and Beigi, Salman},
	eprint = {math/0702267},
	title = {Enumerating the classes of local equivalency in graphs},
	year = {2007}
}

@article{Beigi2006,
	author = {Bahramgiri, Mohsen and Beigi, Salman},
	month = {Oct},
	title = {Graph States Under the Action of Local {C}lifford Group in Non-Binary Case},
	eprint = {quant-ph/0610267},
	year = {2006}}

@article{Ketkar2006,
	author = {Ketkar, Avanti and Klappenecker, Andreas and Kumar, Santosh and Sarvepalli, Pradeep},
	doi = {10.1109/TIT.2006.883612},
	journal = {IEEE Transactions on Information Theory},
	month = {Dec},
	pages = {4892 - 4914},
	title = {Nonbinary Stabilizer Codes Over Finite Fields},
	eprint={quant-ph/0508070},
	volume = {52},
	year = {2006}
}

@inproceedings{marin2013,
	author = {Marin, Anne and Markham, Damian and Perdrix, Simon},
	doi = {10.4230/LIPIcs.TQC.2013.308},
	booktitle = {Proceedings of the 8th Conference on the Theory of Quantum Computation, Communication and Cryptography (TQC 2013)},
	month = {Apr},
	title = {Access Structure in Graphs in High Dimension and Application to Secret Sharing},
	year = {2013},
	eprint={1304.7105}
}

@article{Kante2007,
	author = {Kante, Mamadou and Rao, Micha{\"e}l},
	doi = {10.1007/s00224-012-9399-y},
	journal = {Theory of Computing Systems},
	month = {Sep},
	title = {The Rank-Width of Edge-Coloured Graphs},
	eprint={0709.1433},
	volume = {52},
	year = {2007}
}

@article{Bouchet1989,
	author = {André Bouchet},
	doi = {10.1111/j.1749-6632.1989.tb22439.x},
	journal = {Annals of the New York Academy of Sciences},
	number = {1},
	pages = {81-93},
	title = {Connectivity of Isotropic Systems},
	volume = {555},
	year = {1989}
}

@inproceedings{Fon-Der-Flaass1996,
author="Fon-Der-Flaass, D. G.",
title="Local Complementations of Simple and Directed Graphs",
booktitle="Discrete Analysis and Operations Research",
year="1996",
doi="10.1007/978-94-009-1606-7_3",
url="https://doi.org/10.1007/978-94-009-1606-7_3"
}

@article{Sarvepalli_2010,
   title={Local equivalence, surface-code states, and matroids},
   volume={82},
   ISSN={1094-1622},
   url={http://dx.doi.org/10.1103/PhysRevA.82.022304},
   DOI={10.1103/physreva.82.022304},
   number={2},
   journal={Physical Review A},
   publisher={American Physical Society (APS)},
   author={Sarvepalli, Pradeep and Raussendorf, Robert},
   year={2010},
   month={Aug},
   eprint={0911.1571}
}

@article{tzitrin2018local,
	author = {Tzitrin, Ilan},
	journal = {Physical Review A},
	number = {3},
	pages = {032305},
	publisher = {APS},
	title = {Local equivalence of complete bipartite and repeater graph states},
	doi={10.1103/PhysRevA.98.032305},
	volume = {98},
	year = {2018},
	eprint={1805.05968}
}

@article{azuma2015all,
	author = {Azuma, Koji and Tamaki, Kiyoshi and Lo, Hoi-Kwong},
	journal = {Nature communications},
	number = {1},
	pages = {1--7},
	publisher = {Nature Publishing Group},
	title = {All-photonic quantum repeaters},
	doi={10.1038/ncomms7787},
	volume = {6},
	year = {2015},
	eprint={1309.7207}
}

@article{russo2018photonic,
	author = {Russo, Antonio and Barnes, Edwin and Sophia E. Economou},
	journal = {Physical Review B},
	number = {8},
	pages = {085303},
	publisher = {APS},
	title = {Photonic graph state generation from quantum dots and color centers for quantum communications},
	doi={10.1103/PhysRevB.98.085303},
	volume = {98},
	year = {2018},
	eprint={1801.02754}
}

@article{azuma2023quantum,
	author = {Azuma, Koji and Economou, Sophia E. and Elkouss, David and Hilaire, Paul and Jiang, Liang and Lo, Hoi-Kwong and Tzitrin, Ilan},
	journal = {Reviews of Modern Physics},
	doi={10.1103/RevModPhys.95.045006},
	number = {4},
	pages = {045006},
	publisher = {APS},
	title = {Quantum repeaters: From quantum networks to the quantum internet},
	volume = {95},
	year = {2023},
	eprint={2212.10820}
}

@article{Li2022,
author={Li, Bikun
and Economou, Sophia E.
and Barnes, Edwin},
title={Photonic resource state generation from a minimal number of quantum emitters},
journal={npj Quantum Information},
year={2022},
month={Feb},
day={03},
volume={8},
number={1},
pages={11},
issn={2056-6387},
doi={10.1038/s41534-022-00522-6},
url={https://doi.org/10.1038/s41534-022-00522-6},
eprint={2108.12466}
}

@article{Ghanbari2024,
  title = {Optimization of deterministic photonic-graph-state generation via local operations},
  author = {Ghanbari, Sobhan and Lin, Jie and MacLellan, Benjamin and Robichaud, Luc and Roztocki, Piotr and Lo, Hoi-Kwong},
  journal = {Physical Review A},
  volume = {110},
  issue = {5},
  pages = {052605},
  numpages = {11},
  year = {2024},
  month = {Nov},
  publisher = {American Physical Society},
  doi = {10.1103/PhysRevA.110.052605},
  eprint={2401.00635}
}

@article{CABELLO20092219,
	author = {Ad{\'a}n Cabello and Antonio J. L{\'o}pez-Tarrida and Pilar Moreno and Jos{\'e} R. Portillo},
	doi = {10.1016/j.physleta.2009.04.055},
	issn = {0375-9601},
	journal = {Physics Letters A},
	number = {26},
	pages = {2219-2225},
	title = {Entanglement in eight-qubit graph states},
	volume = {373},
	year = {2009},
	eprint={0812.4625}
}

@article{briegel2009measurement,
	author = {Briegel, Hans J. and Browne, David E. and D{\"u}r, Wolfgang and Raussendorf, Robert and Maarten Van den Nest},
	journal = {Nature Physics},
	number = {1},
	pages = {19--26},
	publisher = {Nature Publishing Group UK London},
	title = {Measurement-based quantum computation},
	doi={10.1038/nphys1157},
	volume = {5},
	year = {2009},
	eprint={0910.1116}
}

@article{raussendorf2001one,
	author = {Raussendorf, Robert and Briegel, Hans J.},
	journal = {Physical Review Letters},
	number = {22},
	pages = {5188},
	publisher = {APS},
	title = {A one-way quantum computer},
	doi={10.1103/PhysRevLett.86.5188},
	volume = {86},
	year = {2001}
}

@article{raussendorf2003measurement,
	author = {Raussendorf, Robert and Browne, Daniel E. and Briegel, Hans J.},
	journal = {Physical review A},
	number = {2},
	pages = {022312},
	publisher = {APS},
	title = {Measurement-based quantum computation on cluster states},
	eprint={quant-ph/0301052},
	doi={10.1103/PhysRevA.68.022312},
	volume = {68},
	year = {2003}
}

@article{van2005edge,
	author = {Van den Nest, Maarten and De Moor, Bart},
	publisher = {Citeseer},
	title = {Edge-local equivalence of graphs},
	year = {2005},
	eprint={math/0510246}
}

@article{zeng2011transversality,
  title={Transversality versus universality for additive quantum codes},
  author={Zeng, Bei and Cross, Andrew and Chuang, Isaac L.},
  journal={IEEE Transactions on Information Theory},
  volume={57},
  number={9},
  pages={6272--6284},
  year={2011},
  publisher={IEEE},
  eprint={0706.1382},
  doi={10.1109/TIT.2011.2161917}
}

@misc{burchardt2024algorithm,
      title={Algorithm to Verify Local Equivalence of Stabilizer States}, 
      author={Adam Burchardt and Jarn de Jong and Lina Vandré},
      year={2024},
      eprint={2410.03961},
      archivePrefix={arXiv},
      primaryClass={quant-ph} 
}

@phdthesis{storjohann2000algorithms,
  title={Algorithms for matrix canonical forms},
  author={Storjohann, Arne},
  year={2000},
  school={ETH Zurich},
  url={https://cs.uwaterloo.ca/~astorjoh/diss2up.pdf}
}

@article{Bouchet1988,
title = {Graphic presentations of isotropic systems},
journal = {Journal of Combinatorial Theory, Series B},
volume = {45},
number = {1},
pages = {58-76},
year = {1988},
issn = {0095-8956},
doi = {10.1016/0095-8956(88)90055-X},
author = {André Bouchet}
}

@article{Bouchet1987digraph,
author = {André Bouchet},
title = {Digraph Decompositions and {E}ulerian Systems},
journal = {SIAM Journal on Algebraic Discrete Methods},
volume = {8},
number = {3},
pages = {323-337},
year = {1987},
doi = {10.1137/0608028}
}

@article{Bouchet2001,
title = {Multimatroids {III.} {T}ightness and Fundamental Graphs},
journal = {European Journal of Combinatorics},
volume = {22},
number = {5},
pages = {657-677},
year = {2001},
issn = {0195-6698},
doi = {10.1006/eujc.2000.0486},
author = {André Bouchet}
}

@article{Bouchet1994,
title = {Circle Graph Obstructions},
journal = {Journal of Combinatorial Theory, Series B},
volume = {60},
number = {1},
pages = {107-144},
year = {1994},
issn = {0095-8956},
doi = {https://doi.org/10.1006/jctb.1994.1008},
author = {André Bouchet}
}

@article{Bouchet1988TransformingTB,
  title={Transforming trees by successive local complementations},
  author={André Bouchet},
  journal={Journal of Graph Theory},
  year={1988},
  volume={12},
  pages={195-207},
  doi={10.1002/jgt.3190120210}
}

@Article{Bouchet1987reducing,
author={André Bouchet},
title={Reducing prime graphs and recognizing circle graphs},
journal={Combinatorica},
year={1987},
month={Sep},
day={01},
volume={7},
number={3},
pages={243-254},
issn={1439-6912},
doi={10.1007/BF02579301}
}

@article{OUM2006,
title = {Approximating clique-width and branch-width},
journal = {Journal of Combinatorial Theory, Series B},
volume = {96},
number = {4},
pages = {514-528},
year = {2006},
issn = {0095-8956},
doi = {https://doi.org/10.1016/j.jctb.2005.10.006},
author = {{Sang-il} Oum and Paul Seymour}
}

@book{nielsenchuang,
  author = {Nielsen, Michael A. and Chuang, Isaac L.},
  publisher = {Cambridge University Press},
  title = {Quantum Computation and Quantum Information},
  year = 2000
}

@article{OumSurvey,
title = {Vertex-minors of graphs: A survey},
journal = {Discrete Applied Mathematics},
volume = {351},
pages = {54-73},
year = {2024},
issn = {0166-218X},
doi = {10.1016/j.dam.2024.03.011},
url = {https://dimag.ibs.re.kr/home/donggyu/wp-content/uploads/sites/16/2023/04/2023-survey-Vertex-minors-of-graphs.pdf},
author = {Donggyu Kim and {Sang-il} Oum}
}

@article{Feynman1982,
author={Feynman, Richard P.},
title={Simulating physics with computers},
journal={International Journal of Theoretical Physics},
year={1982},
month={Jun},
day={01},
volume={21},
number={6},
pages={467-488},
issn={1572-9575},
doi={10.1007/BF02650179}
}

@article{Shor97,
author = {Shor, Peter W.},
title = {Polynomial-Time Algorithms for Prime Factorization and Discrete Logarithms on a Quantum Computer},
journal = {SIAM Journal on Computing},
volume = {26},
number = {5},
pages = {1484-1509},
year = {1997},
doi = {10.1137/S0097539795293172},
eprint = {quant-ph/9508027}
}

@Article{Bell2014secret,
author={Bell, B. A.
and Markham, Damian
and Herrera-Mart{\'i}, D. A.
and Marin, Anne
and Wadsworth, W. J.
and Rarity, J. G.
and Tame, M. S.},
title={Experimental demonstration of graph-state quantum secret sharing},
journal={Nature Communications},
year={2014},
month={Nov},
day={21},
volume={5},
number={1},
pages={5480},
issn={2041-1723},
doi={10.1038/ncomms6480},
url={https://doi.org/10.1038/ncomms6480},
eprint={1411.5827}
}

@article{Keet2010,
  title = {Quantum secret sharing with qudit graph states},
  author = {Keet, Adrian and Fortescue, Ben and Markham, Damian and Sanders, Barry C.},
  journal = {Physical Review A},
  volume = {82},
  issue = {6},
  pages = {062315},
  numpages = {11},
  year = {2010},
  month = {Dec},
  publisher = {American Physical Society},
  doi = {10.1103/PhysRevA.82.062315},
  eprint={1004.4619}
}

@article{steane1996multiple,
  title={Multiple-particle interference and quantum error correction},
  author={Steane, Andrew},
  journal={Proceedings of the Royal Society of London. Series A: Mathematical, Physical and Engineering Sciences},
  volume={452},
  number={1954},
  pages={2551--2577},
  year={1996},
  publisher={The Royal Society London},
  eprint={quant-ph/9601029},
  doi={10.1098/rspa.1996.0136}
}

@article{Calderbank1996,
  title = {Good quantum error-correcting codes exist},
  author = {Calderbank, Arthur Robert and Shor, Peter W.},
  journal = {Physical Review A},
  volume = {54},
  issue = {2},
  pages = {1098--1105},
  numpages = {0},
  year = {1996},
  month = {Aug},
  publisher = {American Physical Society},
  doi = {10.1103/PhysRevA.54.1098},
  eprint={quant-ph/9512032}
}

@article{Cabello2011,
  title = {Optimal preparation of graph states},
  author = {Cabello, Ad\'an and Danielsen, Lars Eirik and L\'opez-Tarrida, Antonio J. and Portillo, Jos\'e R.},
  journal = {Physical Review A},
  volume = {83},
  issue = {4},
  pages = {042314},
  numpages = {7},
  year = {2011},
  month = {Apr},
  publisher = {American Physical Society},
  doi = {10.1103/PhysRevA.83.042314},
  eprint={1011.5464}
}

@article{gottesman1998heisenberg,
  title={The {H}eisenberg representation of quantum computers},
  author={Gottesman, Daniel},
  eprint={quant-ph/9807006},
  year={1998}
}

@article{Aaronson2004,
  title = {Improved simulation of stabilizer circuits},
  author = {Aaronson, Scott and Gottesman, Daniel},
  journal = {Physical Review A},
  volume = {70},
  issue = {5},
  pages = {052328},
  numpages = {14},
  year = {2004},
  month = {Nov},
  publisher = {American Physical Society},
  doi = {10.1103/PhysRevA.70.052328},
  eprint={quant-ph/0406196}
}

@article{Mannalath2023,
  title = {Multiparty entanglement routing in quantum networks},
  author = {Mannalath, Vaisakh and Pathak, Anirban},
  journal = {Physical Review A},
  volume = {108},
  issue = {6},
  pages = {062614},
  numpages = {15},
  year = {2023},
  month = {Dec},
  publisher = {American Physical Society},
  doi = {10.1103/PhysRevA.108.062614},
  eprint={2211.06690}
}

@phdthesis{khesin2025quantum,
    title={Quantum Computing from Graphs},
    author={Andrey Boris Khesin},
    year={2025},
    eprint={2501.17959},
    archivePrefix={arXiv},
    primaryClass={quant-ph}
}

@misc{gottesman1997stabilizercodesquantumerror,
      title={Stabilizer Codes and Quantum Error Correction}, 
      author={Daniel Gottesman},
      year={1997},
      eprint={quant-ph/9705052},
      archivePrefix={arXiv},
      primaryClass={quant-ph}
}

@book{hardy1952inequalities,
  title={Inequalities},
  author={Hardy, Godfrey Harold and Littlewood, John Edensor and P{\'o}lya, George},
  year={1952},
  publisher={Cambridge university press}
}

@article{Bunch1974,
author={Bunch, James R.
and Hopcroft, John E.},
title={Triangular Factorization and Inversion by Fast Matrix Multiplication},
journal={Mathematics of Computation},
year={1974},
month={Jan},
publisher={American Mathematical Society},
volume={28},
number={125},
pages={231-236},
doi={10.2307/2005828}
}

@article{Ibarra198245,
title = {A generalization of the fast {LUP} matrix decomposition algorithm and applications},
journal = {Journal of Algorithms},
volume = {3},
number = {1},
pages = {45-56},
year = {1982},
issn = {0196-6774},
doi = {10.1016/0196-6774(82)90007-4},
author = {Oscar H. Ibarra and Shlomo Moran and Roger Hui}
}

@misc{chitambar2025quantumgraphstatesbridging,
      title={Quantum Graph States: Bridging Classical Theory and Quantum Innovation, Workshop Summary}, 
      author={Eric Chitambar and Kenneth Goodenough and Otfried Gühne and Rose McCarty and Simon Perdrix and Vito Scarola and Shuo Sun and Quntao Zhang},
      year={2025},
      eprint={2508.04823},
      archivePrefix={arXiv},
      primaryClass={quant-ph}
}

@misc{kumabe2025complexitygraphstatepreparationclifford,
      title={Complexity of graph-state preparation by {C}lifford circuits}, 
      author={Soh Kumabe and Ryuhei Mori and Yusei Yoshimura},
      year={2025},
      eprint={2402.05874},
      archivePrefix={arXiv},
      primaryClass={quant-ph}
}

@misc{davies2025preparinggraphstatesforbidding,
      title={Preparing graph states forbidding a vertex-minor}, 
      author={James Davies and Andrew Jena},
      year={2025},
      eprint={2504.00291},
      archivePrefix={arXiv},
      primaryClass={quant-ph}
}

@misc{sharma2025minimizingnumberedgeslcequivalent,
      title={Minimizing the number of edges in {LC}-equivalent graph states}, 
      author={Hemant Sharma and Kenneth Goodenough and Johannes Borregaard and Filip Rozpędek and Jonas Helsen},
      year={2025},
      eprint={2506.00292},
      archivePrefix={arXiv},
      primaryClass={quant-ph}, 
}

@article{dahlberg2020counting,
  title={Counting single-qubit {C}lifford equivalent graph states is {\#}{P}-complete},
  author={Dahlberg, Axel and Helsen, Jonas and Wehner, Stephanie},
  journal={Journal of Mathematical Physics},
  volume={61},
  number={2},
  year={2020},
  publisher={AIP Publishing},
  eprint={1907.08024},
  doi={10.1063/1.5120591}
}

@article{Englbrecht2020,
  title = {Symmetries and entanglement of stabilizer states},
  author = {Englbrecht, Matthias and Kraus, Barbara},
  journal = {Physical Review A},
  volume = {101},
  issue = {6},
  pages = {062302},
  numpages = {17},
  year = {2020},
  month = {Jun},
  publisher = {American Physical Society},
  doi = {10.1103/PhysRevA.101.062302},
  eprint={2001.07106}
}

@article{Briegel2001,
  title = {Persistent Entanglement in Arrays of Interacting Particles},
  author = {Briegel, Hans J. and Raussendorf, Robert},
  journal = {Physical Review Letters},
  volume = {86},
  issue = {5},
  pages = {910--913},
  numpages = {0},
  year = {2001},
  month = {Jan},
  publisher = {American Physical Society},
  doi = {10.1103/PhysRevLett.86.910},
  eprint={quant-ph/0004051}

}

@INPROCEEDINGS{Grassl2002,
  author={Grassl, Markus  and Klappenecker, Andreas  and Rotteler, Martin },
  booktitle={Proceedings of the 2002 IEEE International Symposium on Information Theory (ISIT 2002)}, 
  title={Graphs, quadratic forms, and quantum codes}, 
  year={2002},
  volume={},
  number={},
  pages={45-},
  doi={10.1109/ISIT.2002.1023317},
  eprint={quant-ph/0703112}}

\end{document}